\pdfoutput=1

\documentclass[11pt,a4paper]{article}
\usepackage[T1]{fontenc}
\usepackage{hyperref}
\usepackage{geometry,amsmath,amsfonts,amssymb,amsthm}
\usepackage{slashed}
\usepackage{epsfig}
\usepackage{latexsym}
\usepackage{graphicx}
\usepackage{amssymb}
\usepackage{array}
\usepackage{verbatim}
\usepackage{bm}
\usepackage{lmodern}
\usepackage{dsfont}
\usepackage[footnotesize]{caption}
\usepackage{subfigure}
\usepackage{diagbox}

\usepackage{cancel}
\usepackage{cite}

\usepackage{float}

\usepackage{multirow}
\usepackage{tikz}

\usepackage{bbold}

\usepackage{rotating}
\usepackage{ifthen}
\usepackage[utf8]{inputenc}

\numberwithin{equation}{section} 
\usepackage{varioref}
\usepackage{prettyref}

\usepackage{microtype}

\newrefformat{fig}{Fig.~\ref{#1}}
\newrefformat{subsec}{Sec.~\ref{#1}}
\newrefformat{eq}{Eq.~\eqref{#1}}
\newrefformat{fn}{Footnote~\ref{#1}}
\newrefformat{app}{App.~\ref{#1}}
\newrefformat{conj}{Conjecture~\ref{#1}}
\newrefformat{lem}{Lemma~\ref{#1}}
\newrefformat{exa}{Example~\ref{#1}}
\newrefformat{def}{Definition~\ref{#1}}
\newrefformat{thm}{Theorem~\ref{#1}}
\newrefformat{tab}{Table~\ref{#1}}

\usepackage{textcomp}

\renewcommand\[{\begin{equation}}
\renewcommand\]{\end{equation}}
\renewenvironment{align*}{\align}{\endalign}

\newtheorem{thm}{Theorem}
\newtheorem{lem}[thm]{Lemma}

\newtheorem{example}[thm]{Example}
\newtheorem{defn}[thm]{Definition}

\hypersetup{pdfencoding=auto, psdextra}

\hypersetup{ pdfborder={0 0 .7 [1 2]} }

 \edef\oldeverydisplay{\the\everydisplay}
\usepackage[nodisplayskipstretch]{setspace}
\everydisplay\expandafter{\oldeverydisplay}
\newcommand{\dAlemb}{\Box}

\DeclareMathOperator{\sn}{sn}
\DeclareMathOperator{\am}{am}

\begin{document}
\begin{titlepage}

\thispagestyle{plain}

\vskip 1.5cm
\flushright{DESY 18-106}
\vskip 1.5cm

\begin{center}
{\LARGE\bf Emerging chromo-natural inflation}

\vskip 1cm
{\large  Valerie~Domcke$^1$, Ben Mares,  Francesco Muia$^2$, Mauro Pieroni$^3$}\\[3mm]
{\it{
${}^{1}$ Deutsches Electronen Synchrotron (DESY), 22607 Hamburg, Germany  \\
${}^2$ ICTP, Strada Costiera 11, Trieste 34014, Italy,\\
${}^{3}$ Instituto de F\'{\i}sica Te\'orica UAM/CSIC C/ Nicol\'as Cabrera 13-15 \\Universidad Aut\'onoma de Madrid Cantoblanco, Madrid 28049, Spain  \\[3mm]
}}

\end{center}

\vskip 2cm

\begin{abstract}
The shift-symmetric coupling of a pseudo-scalar particle driving inflation to gauge fields provides a unique way of probing cosmic inflation. For the case of an $\mathrm{SU}(2)$ gauge group, we suggest a possible mechanism for the development of a classical isotropic background gauge field from the standard quantum mechanical vacuum in the far past. Over the course of inflation, the theory dynamically evolves from an approximately abelian regime into an inherently non-abelian regime, with distinct predictions for the scalar and tensor power spectra. The latter regime closely resembles a setup known as chromo-natural inflation, although our main focus here is on a new part of the parameter space which has received little attention so far. For single-field slow-roll inflation models, large scales may exit the horizon in the abelian regime, ensuring agreement with the observations of the anisotropies in the cosmic microwave background, whereas smaller scales experience the non-abelian effects. This results in a strong enhancement of the stochastic gravitational wave background at small scales,  e.g.\ at frequencies accessible with ground-based interferometers. For the scalar power spectrum, a similar enhancement arises due to non-linear contributions.
\end{abstract}

\end{titlepage}

\setcounter{page}{2}

\tableofcontents
\newpage

\section{Introduction\label{sec:introduction}}
The paradigm of cosmic inflation, proposed to explain the puzzling homogeneity and flatness of the Hot Big Bang Universe~\cite{Guth:1980zm}, has been strikingly successful in predicting the anisotropies of the Cosmic Microwave Background (CMB), measured to great precision by the Planck satellite~\cite{Ade:2015lrj}. This paradigm however, leaves open questions. What guarantees the required flatness of the inflationary potential? How is the inflation sector coupled to the Standard Model (SM) of particle physics? The lack of observable predictions on far sub-horizon scales makes it very difficult to find satisfactory and testable answers to these questions. In this context, a special role is played by pseudo-scalar inflation models, in which the inflaton $\phi$ (the particle driving inflation) couples to the field strength tensor $F_{\mu \nu}$ of massless gauge fields through the derivative coupling $\phi F_{\mu \nu} \tilde F^{\mu \nu}$. This coupling is compatible with a shift-symmetry of the inflaton $\phi$ protecting the flatness if the inflationary potential, it provides an immediate way to couple the inflation sector to a gauge field sector (which could be the SM or a hidden sector) and it leads to distinctive signatures, including a strongly enhanced chiral gravitational wave background~\cite{Cook:2011hg,Dimastrogiovanni:2012ew,Adshead:2013qp}.

The phenomenology of these models, both for abelian and non-abelian gauge fields, has recently received a lot of interest. In both cases, the gauge field sector experiences a tachyonic instability during inflation, leading to an explosive particle production which impacts the predictions of inflation. 
For abelian gauge fields this instability is controlled by the inflaton velocity, implying large effects towards the end of inflation in single-field slow-roll inflation models whereas the CMB scales can be largely unaffected, see Ref.~\cite{Barnaby:2010vf} for an overview. The phenomenology of this model includes a strongly enhanced and non-Gaussian contribution to the scalar and tensor power spectra~\cite{Barnaby:2010vf,Barnaby:2011qe,Barnaby:2011vw,Shiraishi:2013kxa,Cook:2013xea}, which may lead to a distortion of the CMB black body spectrum~\cite{Meerburg:2012id}, primordial black hole (PBH) production~\cite{Linde:2012bt, Domcke:2017fix,Garcia-Bellido:2016dkw} and an enhanced chiral gravitational wave signal in the frequency band of LIGO and LISA~\cite{Cook:2011hg,Barnaby:2011qe,Barnaby:2011vw,Anber:2012du,Domcke:2016bkh,Bartolo:2016ami}. Furthermore, the effective friction induced by the gauge field allows for inflation on rather steep potentials~\cite{Anber:2009ua}. The interplay of the gauge fields with the production of charged fermions has been studied in~\cite{Domcke:2018eki} and validity of the perturbative analysis has been scrutinized in~\cite{Ferreira:2015omg, Peloso:2016gqs}. 
The coupling to non-abelian $\mathrm{SU}(2)$ gauge fields, dubbed chromo-natural inflation (CNI) in~\cite{Adshead:2012kp}, allows for inflationary solutions on steep potentials in the presence of a non-vanishing isotropic background gauge field configuration. An analysis of the perturbations~\cite{Dimastrogiovanni:2012st,Dimastrogiovanni:2012ew,Adshead:2013qp,Adshead:2013nka} revealed an enhanced tensor power spectrum, however to the degree of excluding the model as an explanation for the anisotropies in the CMB. The same conclusion holds in the regime where the scalar field can be integrated out~\cite{Namba:2013kia}, referred to as  gauge-flation~\cite{Maleknejad:2011jw,Maleknejad:2011sq} (see also~\cite{Adshead:2012qe,SheikhJabbari:2012qf}). 
Modifications to the original model can evade this conclusion by employing different inflation potentials~\cite{Caldwell:2017chz,DallAgata:2018ybl}, by enlarging the field content of the model~\cite{Dimastrogiovanni:2016fuu,McDonough:2018xzh} or by considering a spontaneously broken gauge symmetry~\cite{Adshead:2016omu}.

In this paper we study the possibility of a dynamical emergence of CNI, under plausible assumptions that we will discuss in due course. In CNI, the gauge field background is assumed to be homogeneous, isotropic and have a sufficiently large vacuum expectation value, so that the background evolution of the inflaton is dominated by the gauge friction term.  We show how such an isotropic background may develop from the regular Bunch--Davies initial conditions in the far past, providing justification for what is commonly taken for granted in CNI.\footnote{We emphasize that the mechanism presented in this paper is not a definitive solution to the problem of generating a background in CNI models: our arguments are based on a separation of length scales whose validity varies throughout the parameter space and should be explicitly verified in a dedicated lattice simulation.} For small gauge field amplitudes, the non-abelian $\mathrm{SU}(2)$ dynamics reduce to three copies of an abelian gauge group. As the inflaton velocity increases over the course of inflation, the tachyonic enhancement of the gauge fields in the abelian regime triggers a classical, inherently non-abelian background evolution. In this background, only a single helicity component of the gauge field features a regime of tachyonic instability. Contrary to the abelian case, each Fourier mode experiences this instability only for a finite time interval. 
We provide analytical results which make only minimal assumptions about the values of the parameters involved.
For the explicit parameter example which we study numerically, we find the gauge friction term to be subdominant in the non-abelian regime, contrary to the usual assumption in CNI. We emphasize that the transition from an effectively abelian to a non-abelian regime is generic in single field axion inflation models, and naturally removes the tension of the original CNI model with the Planck data by delaying the enhancement of the tensor power spectrum to smaller scales. Moreover, this dynamical transition implies that the catastrophic instability in the scalar sector, arising in part of the parameter space as pointed out in~\cite{Dimastrogiovanni:2012ew}, is generically avoided.

Throughout most of the paper we restrict ourselves to the linearized system of perturbations (see also~\cite{Dimastrogiovanni:2012st,Dimastrogiovanni:2012ew,Adshead:2013qp,Adshead:2013nka}). We however point out the importance of higher-order contributions to the scalar perturbation sector, taking into account that two enhanced helicity 2 gauge field perturbation can source helicity 0 (i.e.\ scalar) modes. We estimate the impact of this on the scalar power spectrum, finding an enhancement which is exponentially sensitive to the inflaton velocity, similar to what was found in the abelian case~\cite{Linde:2012bt}.\footnote{While this paper was being finalized, Refs.~\cite{Dimastrogiovanni:2018xnn,Papageorgiou:2018rfx} appeared, which also study the effects of the nonlinear coupling between the helicity $2$ and helicity $0$ perturbations. We briefly comment on these completely independent results in Sec.~\ref{sec:example}, finding overall good agreement within the expected uncertainties.}

The remainder of this paper is organized as follows. We begin with an executive summary in Sec.~\ref{sec:ExecutiveSummary}, to help  guide the reader through the different points discussed in this paper, followed by an overview on our notation in Sec.~\ref{sec:notation}. In Sec.~\ref{sec:axioninflation}, we review some of the key results and equations of abelian and non-abelian axion inflation, setting the notation for the following sections. Sec.~\ref{sec:background} is dedicated to the study of the emerging non-trivial homogeneous isotropic gauge field background. In Sec.~\ref{sec:linearized} we study the linearized system of perturbations in a general homogeneous isotropic gauge field background. This is applied to a specific parameter example in Sec.~\ref{sec:example}, showing explicitly the transition from the abelian to the non-abelian regime. We compute the resulting scalar and tensor power spectrum, taking into account non-linear contributions. We conclude in Sec.~\ref{sec:conclusions}. Six appendices deal with the derivation of the linearized perturbation equations, including the gravitational modes not included in the main text (App.~\ref{app:fulleom}), the explicit gauge field basis used in our linearized analysis (App.~\ref{app:basis}) details on the computation of the non-linear contributions to the scalar power spectrum (App.~\ref{app:variance_computation}), technical details supporting the analysis of the gauge field background (App.~\ref{app:sec3}), mathematical properties of homogeneous isotropic gauge fields (App.~\ref{app:gaugefields}) and analytical approximations of the Whittaker function describing the enhanced perturbation mode of the non-abelian regime (App.~\ref{app:asymptotics}).

\subsection{Executive summary \label{sec:ExecutiveSummary}}

To help guide the reader through the different aspects of our analysis, we give a preview of our key equations and results in this section, skipping all technical details. These results will be derived in the subsequent sections.

Our main focus will lie on the linearized regime of $\mathrm{SU}(2)$ axion inflation. The pseudo-scalar (axion-like) inflaton $\phi$ is coupled to the field strength tensor of the $\mathrm{SU}(2)$ gauge fields through the derivative coupling $\phi F_{\mu \nu} \tilde F^{\mu \nu}$. In the linearized regime, the $\mathrm{SU}(2)$ gauge field\footnote{We adopt a common abuse of notation by referring to the gauge potential $A^a_\mu$ as the \emph{gauge field}.} $A_\mu^a$ can be decomposed into a homogeneous isotropic background $f(\tau)$ and perturbations $\delta A_\mu^a$:\footnote{See Sec.~\ref{sec:notation} for our index conventions.} 
\begin{equation}
 A_\mu^a(\tau, \vec x) = f(\tau) \delta^a_\mu + \delta A^a_\mu (\tau, \vec x) \,.
\end{equation}
The classical evolution of the background gauge field is governed by
\begin{equation}
 \frac{\textrm{d}^2}{\textrm{d} \tau^2}(e f) + 2 (e f)^3 - \frac{2 \xi}{(- \tau)} (e f)^2 = 0 \,,
\end{equation}
where $e$ denotes the $\mathrm{SU}(2)$ gauge coupling and $\xi$, encoding the velocity of the inflaton and defined in Eq.~\eqref{eq:rev_xi}, is typically taken to be ${\cal O}(1 - 10)$ during the last 60 e-folds of inflation. 

For a slowly evolving inflaton, $\xi \simeq const.$, the classical background evolution is focused around two attractor solutions\footnote{In Sec.~\ref{subsec:dynamical_background} and~\ref{subsec:gaugefluctuations} we comment on the difference between the background evolution studied in this paper and the ‘magnetic drift regime’ of~\cite{Adshead:2012kp,Dimastrogiovanni:2012ew, Adshead:2013qp,  Adshead:2013nka}.},
\begin{equation}
 e f(\tau) = c_i \, \xi (- \tau)^{-1} \quad \text{with  } c_0 = 0 \,, \quad c_2 = \tfrac{1}{2} (1 + \sqrt{1 - 4/\xi^2}) \,,
\end{equation}
where the latter is only possible for $\xi \geq 2$. 
Beyond this classical motion, the background is also sourced by the fluctuations $\delta A^a_\mu$. These dominate the background evolution around the $c_0$-solution, and eventually trigger the transitions from the $c_0$ to the $c_2$ solution. For details see Sec.~\ref{sec:background}.

Out of the six physical degrees of freedom of the gauge field, the most important is the helicity $+2$ mode $w_{+2}$, which couples directly to the metric tensor mode, sourcing chiral gravitational waves (see also~\cite{Dimastrogiovanni:2012ew,Adshead:2013qp}). In the $c_2$ background solution, its equation of motion 
\begin{equation}
 \frac{\textrm{d}^2}{ \textrm{d} x^2} w_{+2}(x) + \left(1 - \frac{2 \xi}{x} + 2 \left(\frac{\xi}{x} - 1 \right) \frac{c_2 \xi}{x} \right) w_{+2}(x) = 0 \,,
 \label{eq:eom2Intro}
\end{equation}
(where $x = - k \tau$ with $k$ the momentum of the Fourier mode $w_{+2}$) has an exact solution in terms of the Whittaker function in the limit of constant $\xi$:
\begin{equation}
w_{+2}^{(e)}(x)=e^{(1+c_{2})\pi\xi/2}W_{-i \kappa,\,-i\mu}\left(-2ix\right) \,,
\label{eq:AnalyticalSolutionIntro}
\end{equation}
with $\kappa = (1 + c_2) \xi$ and  $\mu = \xi \sqrt{2 c_2 - (2 \xi)^{-2}}$. Due to a tachyonic instability in Eq.~\eqref{eq:eom2Intro} in between $x_\text{max,min} = (1 + c_2 \pm \sqrt{1 + c_2^2})\xi$, this solution is strongly enhanced just before horizon crossing. At and after horizon crossing, Eq.~\eqref{eq:AnalyticalSolutionIntro} is well approximated by
\begin{equation}
  \, w_{+2}(x) \simeq 2 e^{(\kappa - \mu) \pi} \sqrt{\frac{ x}{ \mu}} \cos\left[\mu \ln(2 x) + \theta_0 \right] \,.
\end{equation}
With this solution at hand, we can approximately analytically solve the 
coupled system of helicity $+2$ gauge fields and gravitational waves 
(see Eq.~\eqref{eq:+2mode}), obtaining for the gravitational wave amplitude $w_{+2}^{(\gamma)}$ after freeze-out on super-horizon scales
\begin{equation}
 x \, w_{+2}^{(\gamma)}(x) \big|_{x \geq 1} \simeq -  \frac{2 H \xi^{5/2}}{e } 2^{3/4} e^{(2 - \sqrt{2}) \pi \xi} \,,
\end{equation}
and consequently for the amplitude of the chiral stochastic gravitational wave background (see Eq.~\eqref{eq:GWapprox} for details),
\begin{equation}
 \Omega_{\rm GW} \simeq \frac{1}{24} \Omega_r \left( \frac{\xi^3 H}{\pi M_P} \right)^2_{\xi = \xi_\text{cr}}  \left( \frac{2^{7/4} H}{e} \xi^{-1/2} e^{(2 - \sqrt{2}) \pi \xi}  \right)^2_{\xi = \xi_\text{ref}} \,.
\end{equation}

The scalar perturbations are not enhanced at the linear level in the parameter space in the focus of this work. However, non-linear contributions, sourced by two enhanced helicity $+2$ gauge field modes, yield an exponentially enhanced contribution to the scalar power spectrum. We report analytical estimates for the resulting contribution to the scalar power spectrum in Eq.~\eqref{eq:Ds_2nd}. 

Combining the results on the background evolution and the analysis of the perturbations, the following picture emerges: At early times, deep in de-Sitter space with small values of $\xi$, the non-abelian axion inflation model reduces to the abelian regime. Two factors are necessary to trigger the transition to the inherently non-abelian regime: The $c_2$ solution of the classical background emerges at $\xi \geq 2$ and the gauge field fluctuations have to reach a sufficient amplitude to trigger initial conditions for the classical motion which actually lead to the $c_2$ solution. We emphasize that the linearized description of this transition is based on two assumptions, which we will justify in Secs.~\ref{sec:axioninflation} and \ref{sec:linearized}, respectively: (i) the gauge fields sourced in the abelian regime are approximately homogeneous over a Hubble patch and (ii) the gauge field fluctuations in the non-abelian regime are small compared to this homogeneous background.\footnote{A quantification of the resulting uncertainties on our final results most likely requires a lattice simulation of the full non-linear theory in de Sitter space. Current state-of-the-art techniques~\cite{Adshead:2015pva,Cuissa:2018oiw} can however only evolve this system for a few Hubble times, insufficient to address this question. We hope that this work will trigger future research in this direction. }

As a proof of concept, we study a parameter example in Sec.~\ref{sec:example} in which the CMB scales exit the horizon in the abelian regime at relatively small $\xi$ (thus ensuring agreement with all CMB observations), whereas smaller scales exit the horizon after the transition to the inherently non-abelian regime. The resulting scalar and tensor power spectra are strongly enhanced at small scales, see Fig.~\ref{fig:spectraPS} and \ref{fig:spectraGW}.

\subsection{Notation and conventions \label{sec:notation}}
We summarize here the main conventions used throughout this paper. The metric signature is $(-,+,+,+)$ and we mostly employ conformal time $\tau$ instead of cosmic time $t$. Derivatives with respect to the conformal time are denoted by a prime, while derivatives with respect to the cosmic time are denoted by a dot. We often use the dimensionless variable
\begin{equation}
x = - k \tau \,.
\end{equation}
The first (second) derivative of the functional 
$\mathcal{S}(\phi)$
with respect to the field $\phi$ is denoted by 
$\mathcal{S}_{, \phi}$ ($\mathcal{S}_{, \phi \phi}$). 
The Fourier transform of the function  
$F(t,\vec{x})$
(or similarly for
$F(\tau,\vec{x})$)
is given by
\begin{equation}
F(t,\vec{x}) = \int \frac{\textrm{d}^3 \vec{k}}{\left(2 \pi\right)^{3/2}} \, \tilde{F}(t,\vec{k}) e^{-i \vec{k} \cdot \vec{x}} \,.
\end{equation}
(Anti-)Symmetrization is defined as
\begin{equation}
S_{(ij)} = \frac{S_{ij} + S_{ji}}{2} \,, \qquad A_{[ij]} = \frac{A_{ij} - A_{ji}}{2} \,.
\end{equation}
Greek letters refer to space-time indices ($\mu = 0, 1, 2, 3$), roman letters from the beginning of the alphabet refer to gauge indices (e.g. $a = 1, 2, 3$ for a $\mathrm{SU}(2)$ gauge group) and roman letters from the middle of the alphabet refer to spatial indices ($i = 1, 2, 3$). We use the usual conventions for $\mathrm{SU}(N)$ gauge fields $\mathbf{A}_\mu = A^a_\mu \mathbf{T}_a$. The field strength tensor is defined as 
\begin{equation}
	F^a_{\mu\nu} \mathbf{T}_a \equiv \mathbf{F}_{\mu\nu} \equiv \frac{i}{e}[\mathbf{D}_{\mu},\mathbf{D}_{\nu} ] = \partial_\mu \mathbf{A}_{\nu} - \partial_\nu \mathbf{A}_{\mu} - i e[\mathbf{A}_{\mu},\mathbf{A}_{\nu}]  \; ,
\end{equation}
where $e$ is the coupling constant, $\mathbf{T}_a$ is the $a$-th generator of the group, $\mathbf{A}_{\mu} \equiv A^a_{\mu} \mathbf{T}_a$ and where we have used the definition of covariant derivative:
\begin{equation}
	\mathbf{D}_{\mu} \equiv \partial_{\mu} - i e \mathbf{A}_{\mu} \equiv \partial_{\mu} - i e A^a_{\mu} \mathbf{T}_a\; .
\end{equation} 
With the commutation relation
\begin{equation}
	 [\mathbf{T}_a, \mathbf{T}_b] = i \varepsilon_{abc}  \mathbf{T}_c \; ,
\end{equation}
 the field strength can be expressed as
\begin{equation}
	\label{eq:field_strength}
	F^a_{\mu\nu} = \partial_\mu A^a_{\nu} - \partial_\nu A^a_{\mu} + e \varepsilon^{abc} A^b_{\mu}A^c_{\nu} \; .
\end{equation}
The dual tensor to the field strength is defined as
\begin{equation}
\tilde{F}_a^{\mu \nu} = \frac{\varepsilon^{\mu \nu \rho \sigma}}{2 \sqrt{-g}} F^a_{\rho \sigma} \,,
\end{equation}
where we use the convention $\varepsilon^{0123} = 1$ for the anti-symmetric tensor. Additional conventions related to the computation of the equations of motion in the ADM formalism are reported in App.~\ref{app:fulleom}.

\section{The role of gauge fields during inflation \label{sec:axioninflation}}

\subsection{The abelian limit \label{sec:abelian}}

In the limit of small gauge couplings and/or small gauge field amplitudes, any non-abelian $\mathrm{SU}(N)$ gauge group will (approximately) act as $N^2 - 1$ copies of an abelian group. Let us thus, also for later reference, begin by briefly reviewing the case of a pseudoscalar inflaton $\phi$ coupled to an abelian gauge field $A_\mu$~\cite{Turner:1987bw,Garretson:1992vt,Anber:2006xt} (for recent analyses see e.g.\ \cite{Barnaby:2010vf,Barnaby:2011qe,Domcke:2016bkh,Jimenez:2017cdr}),
\label{eq:rev_action_pseudoscalar}
\begin{equation}
\mathcal{S}= \int \textrm{d}^4 x \sqrt{|g|} \left[M_P^2 \frac{R}{2} -\frac{1}{2} \, \partial_\mu \phi \, \partial^\mu \phi - V(\phi) - \frac{1}{4} F_{\mu \nu} F^{\mu \nu} - \frac{\alpha}{4 \Lambda} \phi F_{\mu \nu} \tilde{F}^{\mu \nu} \right ]\, .
\end{equation}
Here $V(\phi)$ denotes the inflaton potential, $F_{\mu \nu} \, (\tilde F^{\mu \nu})$ is the (dual) field-strength tensor of the abelian gauge group and $\alpha/\Lambda$ encodes the coupling between the inflaton and the gauge field.\footnote{Identifying $\phi$ as an axion of a global $\mathrm{U}(1)$ symmetry constrains the coupling $\alpha/\Lambda$. For $\alpha = e^2/(4 \pi)$, the scale $\Lambda$ indicates the scale at which coupling of the axion to the chiral anomaly becomes relevant and the effective theory should be replaced by a more fundamental theory. This scale $\Lambda$ should lie above the Hubble scale of inflation. On the other hand, the scalar potential breaking the axion shift symmetry through non-perturbative contributions is periodic in $2 \pi \Lambda$, with slow-roll inflation requiring $\Lambda \gtrsim M_P$ (the extra friction arising from the last term in Eq.~\eqref{eq:rev_action_pseudoscalar} cannot evade this conclusion for the parameter values considered here). A UV-completion is thus far from obvious (see Ref.~\cite{Agrawal:2018mkd} for recent progress), and we here retain the effective field theory point of view, treating $V(\phi)$ and $\alpha/\Lambda$ as independent free parameters.}

Since $F_{\mu \nu} \tilde{F}^{\mu \nu}$ is CP-odd, it will prove useful to work with the Fourier-modes of the gauge field in the chiral basis, 
\begin{equation}
	\label{eq:Fourier_A}
	 \vec{A}(\tau, \vec{x}) = \int \frac{\textrm{d}^3 \vec{k}}{(2\pi)^{3/2}} \left[ \sum_{\lambda = \pm} \tilde{A}_{\lambda}(\tau, \vec{k}) \vec{e}_\lambda (\vec{k}) \hat{a}(\vec{k}) e^{i\vec{k}\cdot\vec{x}} + \tilde{A}_{\lambda}^*(\tau, \vec{k}) \vec{e}^*_{\lambda^\prime} (\vec{k}) \hat{a}^\dagger (\vec{k}) e^{-i\vec{k}\cdot\vec{x}} \right] \; ,
\end{equation}
with the polarization vectors fulfilling $\vec e_\lambda(\vec k) \cdot \vec k = 0$, $\vec e_\lambda(\vec k) \cdot \vec e_{\lambda'}(\vec k) = \delta_{\lambda \lambda'}$ and  $i \vec k \times \vec e_\lambda(\vec k) = \lambda  k\vec e_{\lambda}(\vec k)$ with $k = |\vec k|$. $\hat a^{(\dagger)}$ denotes the annihilation (creation) operator and $\tilde{A}_\lambda(\vec k)$ the corresponding Fourier coefficients. Here $\vec k$ and $\vec x$ denote the co-moving wave vector and coordinates,
\begin{equation}
 \textrm{d} s^2 = g_{\mu \nu }\textrm{d}x^\mu \textrm{d}x^\nu = a^2 (\tau) (-\textrm{d} \tau^2 + \textrm{d} \vec{x}^2 ) = a^2 (\tau) \eta_{\mu \nu }\textrm{d}x^\mu \textrm{d}x^\nu \,,
\end{equation}
with $a(\tau)$ the metric scale factor and $\eta_{\mu \nu}$ the Minkowski metric. Adopting temporal gauge, we have moreover set $A_0 = 0$. The equations of motion for the homogeneous inflaton field and for the gauge field then read,
\begin{align}
\label{eq:rev_motion}
\ddot \phi + 3 H \dot{\phi} + \frac{\partial V}{\partial \phi} & = \frac{\alpha}{\Lambda} \langle \vec{E} \cdot \vec{B} \rangle \ ,\\
\label{eq:rev_motionAfourier}
 \tilde{A}''_{\pm}(\tau,\vec{k})  + \left[ k^2 \mp \, k  \, \frac{2\xi}{-\tau} \right]\tilde{A}^{a}_{\pm}(\tau,\vec{k}) & =  \ 0 \ , 
\end{align}
where we have introduced the physical `electric' and `magnetic' fields as
\begin{equation}
	\label{eq:rev_electric_magnetic}
	\vec{E} \equiv -\frac{1 }{a^2} \frac{\textrm{d} \vec{A}}{\textrm{d} \tau}  \ , \qquad \qquad  \vec{B} \equiv \frac{1}{a^2} \vec{\nabla} \times \vec{A}  \ .
\end{equation}
The expectation values $\langle \bullet \rangle$ in Eq.~\eqref{eq:rev_motion} indicate the spatial average.
The parameter $\xi$, encoding the tachyonic instability in Eq.~\eqref{eq:rev_motionAfourier}, is given by
\begin{equation}
\xi \equiv \frac{\alpha \, \dot{\phi}}{2 \, \Lambda \, H} \ ,
\label{eq:rev_xi}
\end{equation}
with $H = \dot a/a$ denoting the Hubble rate during inflation. In the following we will consider $\dot \phi > 0$ and hence $\xi > 0$ without loss of generality.

In the slow-roll regime, $|\ddot \phi| \ll H |\dot \phi|, | V_{,\phi}|$, we can neglect the change of $\xi$ on the time-scales relevant in Eq.~\eqref{eq:rev_motionAfourier}. This enables us to approximately decouple the equations and solve the equation of motion for the gauge fields analytically, with a parametric dependence on the parameter $\xi$,
\begin{equation}
 \tilde A_\lambda(\tau,k) = \frac{1}{\sqrt{2k}} e^{\lambda \pi \xi/2} W_{- i \lambda \xi, 1/2}(2 i k \tau) \,.
 \label{eq:rev_Whittaker}
\end{equation}
Here $W_{k,m}(z)$ is the Whittaker function. For $\lambda = +$ this describes an oscillatory function which starts to grow exponentially around horizon crossing ($k |\tau| \sim 1$), before becoming approximately constant on super-horizon scales, see Fig.~\ref{fig:Whittakerabelian}. The $\lambda = -$ mode does not exhibit this tachyonic instability and remains oscillatory.
The overall normalization is obtained by matching to the Bunch--Davies vacuum in the infinite past, namely 
\begin{equation}
\tilde{A}_{\lambda}(\tau,k)\approx\frac{e^{-ik\tau}}{\sqrt{2k}}\quad\textrm{as }\tau\to-\infty.\label{eq:BDvac}
\end{equation}

\begin{figure}
 \center
 \includegraphics{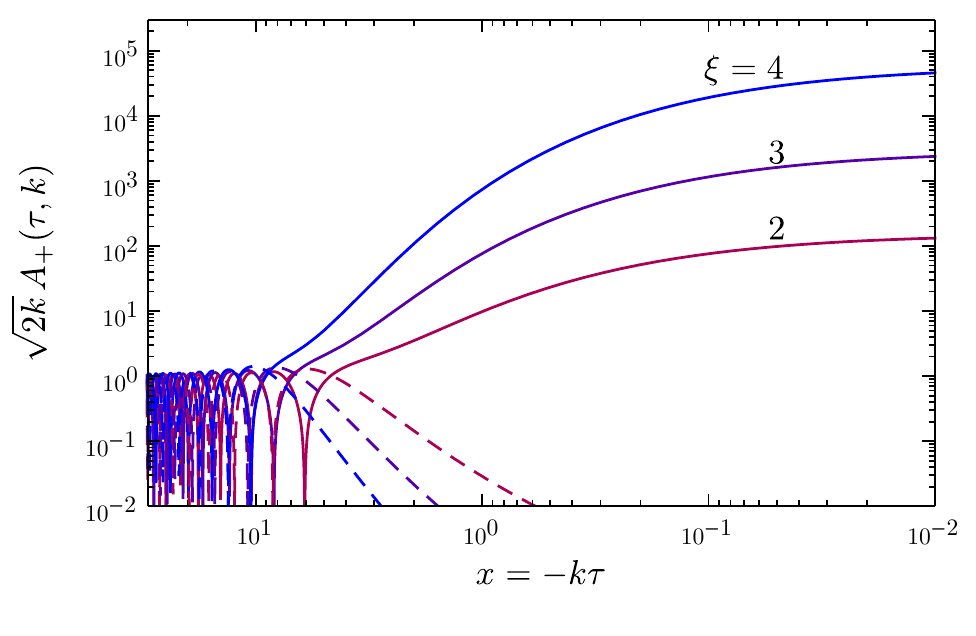}
 \caption{Enhanced helicity mode described by Eq.~\eqref{eq:rev_Whittaker} for different values of $\xi$. The solid (dashed) curves show the absolute value of growing (decaying) component. For a suitable choice of the complex phase of the Whittaker function (see also App.~\ref{app:asymptotics}), the growing (decaying) component is the real (imaginary) part.}
 \label{fig:Whittakerabelian}
\end{figure}

The explicit solution~\eqref{eq:rev_Whittaker} in turn enables us to explicitly evaluate the right-hand side of Eq.~\eqref{eq:rev_motion}. For $\xi \gtrsim 3$ this is well approximated by
\begin{equation}
 \langle \vec E \cdot \vec B \rangle  \simeq - 2.4 \cdot 10^{-4} \frac{H^4}{\xi^4} e^{2 \pi \xi} \,.
 \label{eq:rev_EB}
\end{equation}
Recalling the definition of $\xi$ in Eq.~\eqref{eq:rev_xi}, this enables us to (numerically) solve Eq.~\eqref{eq:rev_motion}. 
The resulting evolution of 
this system with the inflaton coupled to an abelian gauge field 
has been studied e.g.\ in Refs.~\cite{Anber:2009ua,Barnaby:2010vf,Cook:2011hg,Barnaby:2011qe,Anber:2012du,Linde:2012bt,Domcke:2016bkh}, obtaining the following key results:
\begin{itemize}
 \item The tachyonic enhancement of the $A_+$ modes leads to a significant backreaction in the equation of motion for $\phi$, which is exponentially sensitive to $\xi$. This can be interpreted as an additional friction term for the inflaton.
 \item In single field inflation models, $V_{,\phi}/V$ typically increases over the course of inflation, implying an increasing value of $\xi$. Constraints on non-gaussianities in the CMB impose $\xi_\text{CMB} \lesssim 2.5$ whereas the backreaction mentioned above dynamically limits the growth of $\xi$ over the course of 50-60 e-folds of inflation, typically leading to $\xi \lesssim 10$.\footnote{Note that for $\xi \gtrsim 4.5$, perturbative control has been shown to break down~\cite{Ferreira:2015omg, Peloso:2016gqs}.}
 \item The presence of the gauge fields leads to an additional source term for the scalar and tensor power spectra. Due to the increasing value of $\xi$ this effect is typically largest at small scales (i.e.\ towards the end of inflation).
\end{itemize}

For later reference, let us discuss in detail three quantities which will be relevant for the analysis carried out in the next parts of this work: the gauge field variance, the homogeneity scale and decoherence time.

\textit{Variance}. Isotropy ensures that averaged over the whole universe $\langle \vec A \rangle = 0$, but we may estimate the magnitude of the gauge fields in any Hubble patch by computing the variance,
\begin{align}
 \langle 0 | \vec A(\tau, \vec x) \vec A(\tau, \vec x) | 0 \rangle^{1/2} & = \left(\int  \frac{\textrm{d}^3 \vec{k} }{(2\pi)^3}  \left[ \tilde{A}_{+}(\vec{k})  \tilde{A}_{+}^*(\vec{k}) + \tilde{A}_{-}(\vec{k})  \tilde{A}_{-}^*(\vec{k} ) \right] \right)^{1/2} \nonumber \\
 & =  e^{\pi \xi/2} \sqrt{ \int \frac{ \textrm{d} k}{2\pi^2} \, \frac{k}{2} \left| W_{-i\xi,1/2}(2ik\tau)\right|^2 } \nonumber \\
 &  = \frac{aH}{2\pi} e^{\pi \xi/2} \sqrt{ \int_0^{x_{UV}} \textrm{d} x\, x \left|W_{-i\xi,1/2}(- 2 i x)\right|^2 } \nonumber \\
 & \simeq \frac{1}{(-\tau)}\,  0.008 \times e^{2.8 \xi} \,.	 \label{eq:variance_abelian}
\end{align}
Here we set the upper integration limit to $x_{UV} = 2 \xi$, so as to not count the vacuum contribution. In agreement with Ref.~\cite{Jimenez:2017cdr}, we find that all the integrals of this type performed in this paper are rather insensitive to the choice of $x_{UV}$ for $\xi \gtrsim 3$.

\textit{Homogeneity}. The energy density stored in the gauge fields can be computed as  $(\vec E^2 + \vec B^2)/2$ with
\begin{align}
 \frac{1}{2} \langle \vec E^2 \rangle & = \frac{1}{2 a^4} \int \frac{\textrm{d}^3 \vec{k}}{(2 \pi)^3} \bigg|\frac{\partial \tilde{A}_+^k(\tau)}{\partial \tau}\bigg|^2  = \frac{H^4}{8 \pi^2}e^{\pi \xi} \int_0^{2 \xi} \textrm{d} x \,\,x^3 \bigg| \frac{\partial W_{-i \xi, 1/2}(- 2 i x)}{\partial x} \bigg|^2 \,,\\
  \frac{1}{2} \langle \vec B^2 \rangle & = \frac{1}{2 a^4} \int \frac{\textrm{d}^3 \vec{k}}{(2 \pi)^3} k^2 \bigg|\tilde{A}_+^k(\tau)\bigg|^2  = \frac{H^4}{8 \pi^2}e^{\pi \xi} \int_0^{2 \xi} \textrm{d} x \,\,x^3 \bigg|  W_{-i \xi, 1/2}(- 2 i x)\bigg|^2 \,.
\end{align}
We can now determine for each value of $\xi$, the value of $x_{90}$ for which $90\%$ of the energy is contained in modes with $x < x_{90}$, see  left panel of Fig.~\ref{fig:propertiesabelian}. For any $x > x_{90}$, we can then safely model the gauge field as a homogeneous background field. We conclude that for (the phenomenologically interesting) large values of $\xi \gtrsim 3$, the homogeneity scale lies at $x_{90} \lesssim {\cal O}(1)$, so that on sub-horizon scales, this gauge field acts like a homogeneous background field. This approximation becomes better for larger values of $\xi$. For reference, the dashed line in Fig.~\ref{fig:propertiesabelian} indicates the value of $x$ for which $95\%$ of the energy is contained within $x < x_{95}$. 
\begin{figure}
\subfigure{
\includegraphics{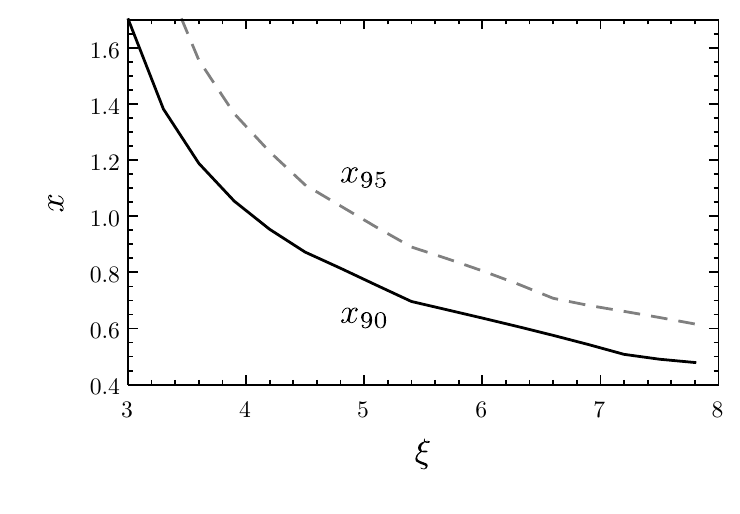}
}
\hfill
\subfigure{
\includegraphics{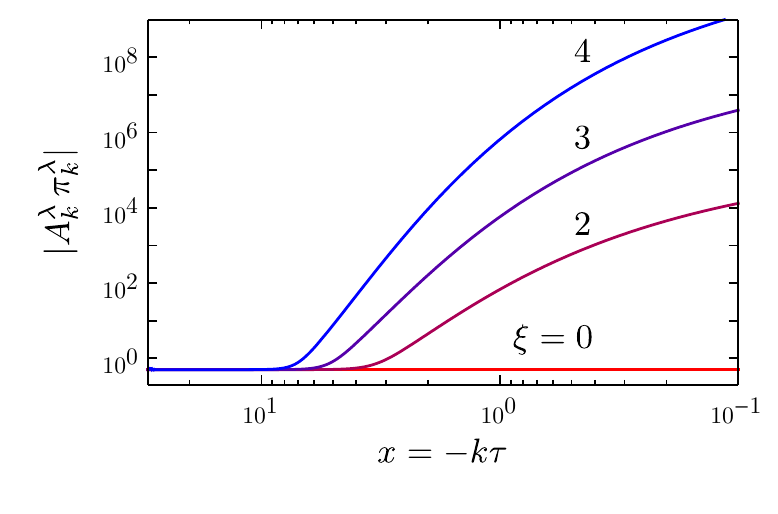}
}
\caption{Properties of the gauge fields in the abelian regime. \textit{Left panel:} Scale of homogeneity for different values of $\xi$, defined such that $90\%$ ($95\%$) of the energy is contained in modes with $- k \tau = x < x_{90}$ ($x < x_{95}$). \textit{Right panel}: $|A^\lambda_k \pi^\lambda_k|$ for $\xi = 0, 2, 3, 4$. Decoherence occurs once $|A^\lambda_k \pi^\lambda_k| \gg 1$.}
  \label{fig:propertiesabelian}
\end{figure}

\textit{Decoherence.} For any given mode decoherence is reached if $|\tilde{A}^\lambda_k \pi^\lambda_k| \gg 1$~\cite{Guth:1985ya}. Using the free-field expression for the conjugate momentum, $\pi^\lambda_k = \partial_0 \tilde{A}^\lambda_k$, the right panel of Fig.~\ref{fig:propertiesabelian} demonstrates that decoherence is reached at $x \sim \xi$. As a further check, in order to establish the transition to the classical behaviour, we computed the number of particles $n_k$ in each mode (see~\cite{Brax:2010ai}) and we checked at which point the regime $n_k \gg 1$ is reached. The results agree with those shown in Fig. \ref{fig:propertiesabelian}\,: decoherence is reached at $x \sim \xi$.

In summary, we find that in the abelian limit, any Hubble patch develops a classical, approximately homogeneous gauge field background, whose average magnitude grows exponentially with $\xi$ as indicated in Eq.~\eqref{eq:variance_abelian}. In the next section, we will highlight the key changes to this picture in the non-abelian regime. 

\subsection{Non-abelian regime}

Let us now consider the same action as in Eq.~\eqref{eq:rev_action_pseudoscalar}, but now in the case of an $\mathrm{SU}(2)$ gauge group,
\begin{align}
\mathcal{S} & = \int \textrm{d}^4 x \sqrt{|g|} \left[M_P^2 \frac{R}{2} -\frac{1}{2} \, \partial_\mu \phi \, \partial^\mu \phi - V(\phi) - \frac{1}{4} F^a_{\mu \nu} F_a^{\mu \nu} - \frac{\alpha}{4 \Lambda} \phi F^a_{\mu \nu} \tilde F_a^{\mu \nu} \right ] \nonumber \\
& \equiv \int d^4 x \, \sqrt{-g} \, \left[\mathcal{L}_{\rm EH} + \mathcal{L}_{\phi} + \mathcal{L}_{\rm YM} + \mathcal{L}_{\rm CS}\right] \,. \label{eq:action}
\end{align}
The resulting equations of motion for the homogeneous inflaton field $\phi(\tau)$ and the gauge fields $A^a_\mu(\tau,\vec x)$ read:
\begin{equation}
	\ddot{\phi} + 3H \dot{\phi} + V_{,\phi} + \frac{\alpha}{4 \Lambda} \, \frac{\varepsilon^{\mu \nu \rho \sigma} }{ a^3(t)} \, F^a_{\mu\nu} F^a_{\rho\sigma} = 0 \;,
	\label{eq:phib}
\end{equation}
and
\begin{equation}
\begin{aligned}
\label{eq:eq_motion_final}
	\eta^{\nu\sigma} \left\{\dAlemb A_\sigma^a - \partial_\sigma \left(\partial_\mu \eta^{\mu\rho} A_\rho^a \right)+ e \varepsilon^{a b c} \eta^{\mu\rho} \left[ 2A_\rho^b \partial_\mu A_\sigma^c + \left( \partial_\mu A_\rho^b\right) A_\sigma^c - A_\mu^b \partial_\sigma A_\rho^c \right] + \right. \\
	+ \left. e^2 \eta^{\mu\rho} \left[ A_\rho^a \left( A_\mu^b A_\sigma^b \right) - A_\sigma^a \left( A_\mu^b A_\rho^b \right) \right] \right\} + \frac{\alpha}{2 \Lambda} \phi^\prime \varepsilon^{0 \nu j k} \left[2 \partial_j A^a_k + e \varepsilon^{abc} A^b_j A^c_k\right] = 0 \; ,
	\end{aligned}
\end{equation}
where we have introduced the $\dAlemb$-operator defined as usual as $\dAlemb \equiv g^{\mu\nu} \partial_{\mu} \partial_{\nu}$ which here is expressed in co-moving coordinates.

The non-linear equation~\eqref{eq:eq_motion_final} is highly sensitive to the presence of a gauge field background as described in Sec.~\ref{sec:abelian}. An exact treatment of the system requires solving the non-linear coupled system of equations of motions in an exponentially expanding background, a very challenging task. Instead, we will work in a linear approximation (as in Refs.~\cite{Dimastrogiovanni:2012ew,Adshead:2013qp,Adshead:2013nka}), expanding the gauge fields around a homogeneous background, denoted by $A^{(0)}(\tau)$, so that 
\begin{equation}
\begin{aligned}
\label{eq:background_plus_linear}
A(\tau, \vec x) = A^{(0)}(\tau) + \delta A(\tau, \vec x).
\end{aligned}
\end{equation}
We will discuss the (classical) evolution of the background in Sec.~\ref{sec:background} and the (quantum) evolution of the fluctuations in Sec.~\ref{sec:linearized}. This treatment is valid as long as the evolution of the background is indeed governed by the classical equation of motion, i.e.\ as long as the growth of the fluctuations does not overcome the classical motion.  A similar condition ensures the washout of the initial inhomogeneities: As discussed above, the energy stored in the gauge fields enhanced during the abelian regime is peaked on super-horizon scales. This physical scale arises from a dynamical equilibrium between a continuous re-sourcing of the background gauge field by (enhanced) horizon-crossing modes and the red-shifting of longer wavelength modes. Consequently, a suppression of the growth of the gauge-field fluctuations in the non-abelian regime (with respect to their abelian counterpart) diminishes the supply of modes sourcing the peak in the Fourier spectrum, leading to a red-shift of this inhomogeneity scale to larger, far super-horizon scales. To make the overall picture clear from the start, we highlight in the following some of the key results, the derivation of these will follow in Secs~\ref{sec:background} and \ref{sec:linearized}, correspondingly. 

\begin{figure}[t]
\centering
 \includegraphics[width = 0.7 \textwidth]{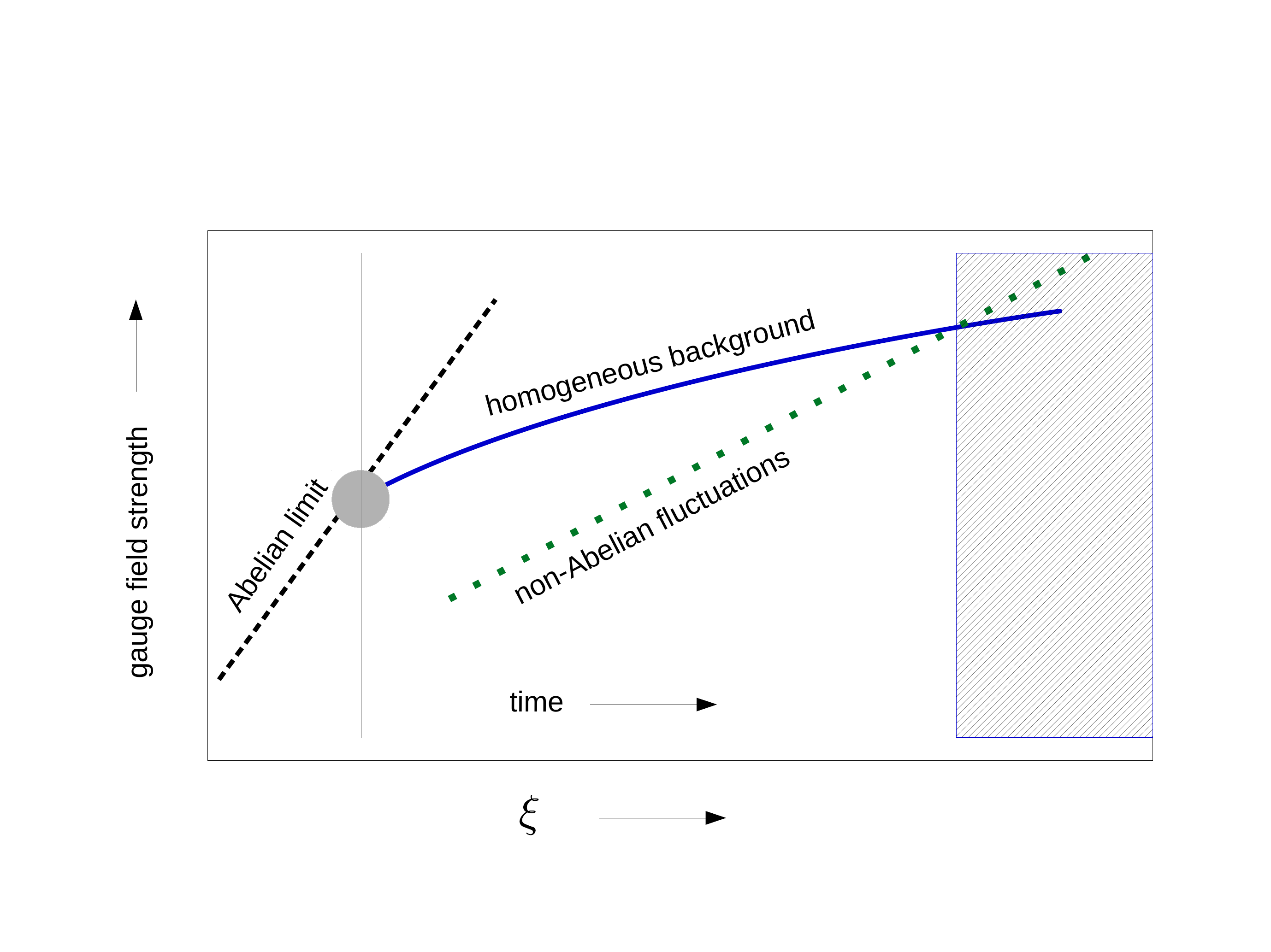}
 \caption{Sketch of the evolution of the average magnitude of the gauge fields from the abelian to the non-abelian regime. The vertical line marks the transition from the abelian limit to the full non-abelian theory, the gray circle indicates the requirement of matching the initial conditions accordingly. In the non-abelian regime, the fluctuations grow slower, but may nevertheless at some point overcome the classically evolving background. This region of parameter space is beyond the scope of the present paper, as indicated by gray shaded region. }
 \label{fig:sketch}
\end{figure}

 We find that the background field dynamically evolves towards an isotropic configuration with two distinct asymptotic behaviours. On the one hand, for small initial conditions, the co-moving background evolves towards a constant value, and thus remains small compared to tachyonically enhanced fluctuations, see Eq.~\eqref{eq:variance_abelian}. 
 In this regime, we are essentially back in the abelian limit, i.e.\ the fluctuations are well described by Eq.~\eqref{eq:rev_Whittaker} with 3 enhanced  and 3 oscillating modes.\footnote{ One may worry about the justification of the linearization~\eqref{eq:background_plus_linear} in this regime. From Eq.~\eqref{eq:eq_motion_final}, we note that in the limit $A^{(0)}(\tau) \rightarrow 0$, a necessary condition for the linearization to be valid is $e \, (\delta A)^2  \ll \partial_\mu \delta A$, or in other words $ e \, \delta A \ll k = x/(- \tau)$, indicating the regime where the non-abelian terms become irrelevant. For modes crossing the horizon  ($x = 1$), this condition holds if
 \begin{equation}
  0.008 \times \exp(2.8 \, \xi) \ll 1/e \,,
  \label{eq:Conditionabelian}
 \end{equation}
where we have inserted Eq.~\eqref{eq:variance_abelian}. For far super-horizon modes, the non-abelian terms become more important. However, at this point due to a red-shift in momentum and a decay in the amplitude, the contribution of these modes to e.g.\ the variance of the energy density is negligible. Note that the condition~\eqref{eq:Conditionabelian} is not sufficient to justify the linearization of the equation of motion for the inflaton~\eqref{eq:phib}. In the abelian regime, the last term contains at least two powers of $\delta A_\mu^a$, and its relative importance will depend on the coupling strength $\alpha/\Lambda$. We will return to the importance of these non-linear effects in detail in Sec.~\ref{sec:example}.} On the other hand, for sufficiently large initial conditions (and only if $\xi \geq 2$), there is an asymptotic solution for the background which, in terms of the comoving gauge field $A^{(0)}$, grows as $1/|\tau|$. 
We stress that this background is driven by classical motion and, contrary to the approximately homogeneous gauge field formed in the abelian case, it is not sourced by super-horizon fluctuations. In this regime, the background significantly modifies the equation of motion for the fluctuations. Consequently, we find that only a single gauge field mode is enhanced, and the enhancement is moreover significantly suppressed compared to the abelian case. Given the strong gauge field production in the abelian regime and the increasing value of $\xi$ over the course of inflation, eventually the growing background solution will be triggered. The point at which this happens depends on the gauge coupling $e$ and the CP-violating coupling $\alpha/\Lambda$. A sketch of this overall picture is given in Fig.~\ref{fig:sketch}. 

\section{The non-abelian homogeneous gauge field background \label{sec:background}}

In this section we study the classical evolution of the homogeneous non-abelian gauge-field background. In Sec.~\ref{subsec:Isotropic} we discuss three distinct types of solutions. Among these, of particular interest is the ``$c_{2}$\nobreakdash-type'' of solution which, in physical coordinates, describes a background gauge field whose magnitude, for any fixed $\xi$, approaches a positive constant. This is similar to the background field assumed in CNI (see~\cite{Dimastrogiovanni:2012st,Dimastrogiovanni:2012ew,Adshead:2013qp,Adshead:2013nka}). (In comoving coordinates, the background field grows in proportion to the scale factor $a(\tau)$, or equivalently\footnote{For the time scales we are considering, $H$ is effectively constant.} in proportion to $(-\tau)^{-1}$.) We will show that this solution is only possible for $\xi\geq2$, and it is stable under perturbations. The key result of this section is to describe the initial conditions necessary to reach this type of solution. To this end, we discuss the different types of solutions both at early and at late times. We close \prettyref{subsec:Isotropic} by showing a phase-space diagram of these solutions, illustrating the different types of solutions as well as their behaviour at early and late times. Based on this in-depth study of the non-equilibrium behaviour of the classical equation of motion we will conclude that 
\begin{itemize}
\item Once the magnitude of the initial conditions reaches a particular threshold, the classical equation of motion for the gauge field background evolves with high probability towards a $c_{2}$\nobreakdash-type homogeneous and isotropic background solution. 
\end{itemize}
These initial conditions in turn are understood to be sourced by the enhanced gauge field fluctuations generated before this $c_{2}$\nobreakdash-type solution developed. We will return to these quantum fluctuations in \prettyref{sec:linearized}. For now, we will only note that in the far past, these fluctuations are well described by the abelian limit discussed in Sec.~\ref{sec:abelian}.

The analysis of Sec.~\ref{subsec:Isotropic} will assume an isotropic gauge field background. We will justify this in Sec.~\ref{subsec:Anisotropic-background-fields} by demonstrating that the homogeneous background evolves towards isotropy. We will further see in Secs.~\ref{sec:linearized} and \ref{sec:example}, that this background suppresses the quantum gauge field fluctuations. We therefore conclude that after the homogeneous background is triggered, the dynamics of the gauge field background are accurately captured by the \emph{classical} equation of motion for a homogeneous and isotropic gauge field.

We conclude this discussion in Sec.~\ref{subsec:dynamical_background} by including the dynamical evolution of the inflaton background. Technical details and mathematical proofs are relegated to Appendices~\ref{app:sec3} and \ref{app:gaugefields}.

\subsection{\label{subsec:Isotropic}Equation of motion for an isotropic gauge field background}

In this section, we consider in detail the equation of motion for the non-abelian gauge field background $A^{(0)}$. (For context, see the discussion around \prettyref{eq:background_plus_linear}.) This is the zeroth order part of our approximation, so we ignore for now the inhomogeneous first-order perturbations $\delta A$ which we will add later in Sec.~\ref{sec:linearized}. We make the following explicit assumptions on the background field: 
\begin{itemize}
\item The background gauge field $A^{(0)}$ is homogeneous and isotropic. (We will show in \prettyref{subsec:Anisotropic-background-fields} that isotropy is a valid assumption in our regime of interest. For a discussion on homogeneity in the abelian limit, see \prettyref{sec:abelian}.) 
\item The inflaton field $\phi$ is homogeneous and evolves in the slow-roll regime. In particular, we consider $\xi$ to be constant. (See~\prettyref{eq:rev_xi} and the subsequent comments.) 
\end{itemize}
Any $\mathrm{SU}(2)$ gauge field $A^{(0)}(\tau)$ which is homogeneous and isotropic is (after applying a gauge transformation) of the form (see e.g.~\cite{Verbin:1989sg,Maleknejad:2012fw}),
\begin{equation}
(A^{(0)})_{0}^{a}=0,\quad(A^{(0)})_{i}^{a}=f(\tau)\,\delta_{i}^{a}.\label{eq:ansatz-A}
\end{equation}
We provide a rigorous proof of this statement as \prettyref{thm:isotropy} in \prettyref{app:iso-gauge}. We emphasize that although this particular choice of $A^{(0)}$ happens to be in temporal gauge, no gauge-fixing constraints have been imposed on $A^{(0)}+\delta A$. 

The corresponding equation of motion for $f(\tau)$ is 
\begin{equation}
\frac{\mathrm{d}^{2}}{\mathrm{d}\tau^{2}}ef(\tau)+2\left(ef(\tau)\right)^{3}-\frac{2\,\xi}{-\tau}(ef(\tau))^{2}=0,\label{eq:chromonatural_eom_final}
\end{equation}
obtained by inserting \prettyref{eq:ansatz-A} into \prettyref{eq:eq_motion_final}. Our task is now to analyze the qualitative behaviour of solutions to this ordinary differential equation, where $e$ and $\xi$ are constants.

It is helpful to observe the following symmetries of this equation. 
\begin{itemize}
\item There is always a factor of $e$ wherever $f(\tau)$ appears. Consequently, we focus our analysis on the quantity ``$ef(\tau)$'' instead of ``$f(\tau)$.'' The coupling constant $e$ is nothing but a scale factor. 
\item The substitution 
\begin{align*}
ef(\tau) & \mapsto-ef(\tau)\\
\xi & \mapsto-\xi
\end{align*}
preserves solutions of \prettyref{eq:chromonatural_eom_final}. Thus solutions with $\xi<0$, are identical (up to a sign) to solutions with $\xi>0$. As in \prettyref{sec:abelian} we assume without loss of generality that $\xi\geq0$. 
\item \label{enu:symmetry}For any positive real number $\lambda$, the transformation 
\begin{equation}
ef(\tau)\mapsto\lambda\,ef(\lambda\tau)\label{eq:symmetry3}
\end{equation}
preserves solutions of \prettyref{eq:chromonatural_eom_final}. The most straightforward consequence is that in \prettyref{fig:scaling-symmetry}, we may replace the axis labels $(\tau,ef(\tau))$ by $(H_{*}\tau,H_{*}^{-1}ef(\tau))$ for any constant $H_{*}$. For instance, one might take $H_{*}$ to be the value of the Hubble parameter at the end of inflation. Alternatively, for convenience, in this section (and this section only) we will work in units where $\tau$ and $ef(\tau)$ are dimensionless.\\
More generally, this transformation can be understood in terms of the physical quantity $g(N)$ defined as follows:  
\begin{equation}
g(N)\equiv-\tau\,ef(\tau),\ \textrm{ where }\tau=\tau_{0}\,e^{N},\ \textrm{and}\ \tau_{0}\equiv\tau(N=0).\label{eq:g-def}
\end{equation}
Here $N$ is the usual measure of e-folds during inflation with $dN=-Hdt=d\tau/\tau$, so that \prettyref{eq:chromonatural_eom_final} becomes 
\begin{equation}
\frac{\mathrm{d}^{2}}{\mathrm{d}N^{2}}g(N)-3\frac{\mathrm{d}}{\mathrm{d}N}g(N)+2g(N)\left(g(N)^{2}-\xi g(N)+1\right)=0\,.\label{eq:g-of-N}
\end{equation}
This is an autonomous \footnote{\emph{Autonomous} means that the time variable doesn't explicitly appear in the equation of motion. For example, the quartic oscillator equation $w''(x)+2w(x)^{3}=0$ is autonomous, while the Airy equation $w''(x)-xw(x)=0$ is not. } equation, so solutions are invariant under time translations of the form $g(N)\mapsto g(N+\ln\lambda)$. With comoving quantities, these time translations correspond precisely to the transformation~\eqref{eq:symmetry3}. When this transformation is applied in the limit $\lambda\to0^{+}$, it corresponds to the limiting behaviour as $\tau\to0^{-}$ (i.e.\ to the infinite future). As part of our analysis in the next subsection, we illustrate in \prettyref{fig:scaling-symmetry} how the transformation acts on the both the comoving quantity $ef(\tau)$ and physical quantity $g(N)$. Note that Refs.~\cite{Adshead:2012kp,Dimastrogiovanni:2012ew,Adshead:2013qp,Adshead:2013nka} work directly with the physical gauge field background. When studying the infinite future it is more convenient to work with $g(N)$, otherwise we find it more convenient to study $ef(\tau)$. 
\end{itemize}

\subsubsection{Three distinct types of solutions\label{subsec:solutions}}

\subsubsection*{Typical behaviour of solutions }

Before rigorously analyzing the behaviour of solutions, we begin with an informal discussion of the two most common types of solution to \prettyref{eq:chromonatural_eom_final}. Typical examples of these are depicted as solid black lines in \prettyref{fig:scaling-symmetry}. Details and proofs will be provided below. 
\begin{figure}[t]
\centering{}\includegraphics{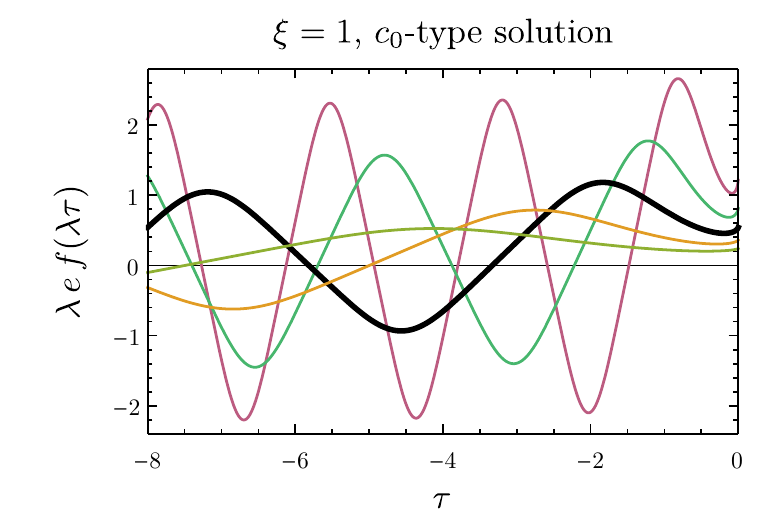}\hfill{}\includegraphics{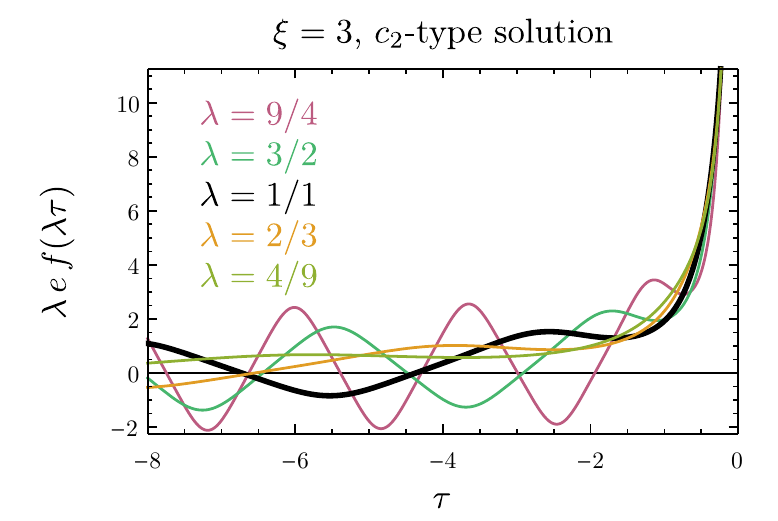}\\
\includegraphics{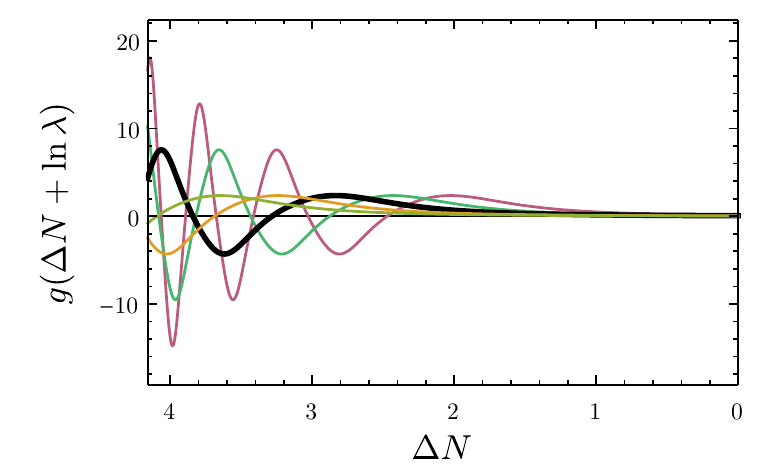}\hfill{}\includegraphics{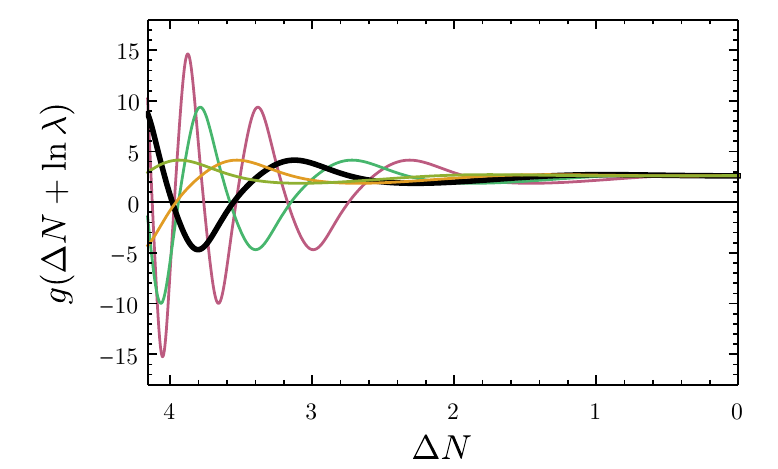}\caption{\label{fig:scaling-symmetry} Typical solutions for the classical gauge field background. The thicker black curve ($\lambda=1$) is obtained by numerically solving Eq.~\eqref{eq:chromonatural_eom_final} for specific (but generic) initial conditions. The coloured curves are obtained by applying the transformation \eqref{eq:symmetry3} for various $\lambda$. The top panels show the solutions in comoving coordinates and conformal time. The bottom panels show the same solutions in physical coordinates and e-folds, where \eqref{eq:symmetry3} simply corresponds to a time-shift. (To be definite, we choose $\Delta N$ to denote the number of e-folds before $\tau_{0}=-\tfrac{1}{8}$.) The left panels show a typical example of a $c_{0}$\protect\nobreakdash-type background solution. The right panels show a typical $c_{2}$\protect\nobreakdash-type background solution. In the far future, the bounded $c_{0}$\protect\nobreakdash-type (resp.~growing $c_{2}$\protect\nobreakdash-type) background solutions in comoving coordinates correspond in physical coordinates to decaying (resp.~bounded) solutions.}
\end{figure}

For large negative values of $\tau$ (the far past), solutions for $ef(\tau)$ are typically oscillatory of a fixed amplitude. We caution the reader that in our model, this oscillatory behaviour does not actually occur in the far past. This is because at early times, the gauge field background is dominated by small but growing fluctuations from super-horizon modes, and so the classical equation of motion breaks down there. \label{ref:no-osc}

On the other hand, we shall be primarily concerned with what happens as $\tau\to0^{-}$ (the infinite future), and the influence of the initial conditions on this behaviour. Based on the two parameters which determine the initial conditions, we can divide solutions into three categories based on their behaviour as $\tau\to0^{-}$:
\begin{description}
\item [{$c_{0}$\nobreakdash-type~solutions}] \label{c0typesol} The function $ef(\tau)$ remains bounded, and $ef(\tau)$ converges to a finite value as $\tau\to0^{-}$. In this case, the physical gauge field background $f(\tau)/a(\tau)$ approaches zero and will remain small compared to the tachyonically enhanced gauge field fluctuations, see Eq.~\eqref{eq:variance_abelian}. 
\item [{$c_{2}$\nobreakdash-type~solutions}] The function $ef(\tau)$ is unbounded as $\tau\to0^{-}$. In this case, the growth of $ef(\tau)$ is always proportional to $(-\tau)^{-1}$. These are the background solutions which will be most relevant throughout this work, and which are responsible for the inherently non-abelian regime of CNI. The physical gauge field background $f(\tau)/a(\tau)$ approaches a positive constant.
\item [{$c_{1}$\nobreakdash-type~solutions}] These solutions form the ``saddle points'' between $c_{0}$\nobreakdash-type and $c_{2}$\nobreakdash-type solutions. They arise only with finely-tuned initial conditions. Just like $c_{2}$\nobreakdash-type solutions, their growth is proportional to $(-\tau)^{-1}$ as $\tau\to0^{-}$, however with a smaller proportionality constant. 
\end{description}
Asymptotic formulas for these three families of solutions are given in \prettyref{app:t_to_0}.

Our two central questions are as follows: 
\begin{enumerate}
\item Given initial conditions for a solution to \prettyref{eq:chromonatural_eom_final}, will the solution be $c_{0}$\nobreakdash-type or $c_{2}$\nobreakdash-type? 
\item For which initial conditions is the solution oscillatory? When so, at what time do the oscillations stop? \label{enu:central-questions} 
\end{enumerate}

\subsubsection*{Ansatz}

We can write down up to three explicit solutions to \prettyref{eq:chromonatural_eom_final} with the ansatz 
\begin{equation}
ef(\tau)=c\,\xi/(-\tau)\,,\qquad\textrm{ equivalently }g(N)=c\,\xi,\label{eq:ansatz-f}
\end{equation}
where $c$ is a constant. Solutions of this form arise by rescaling any general solution of \prettyref{eq:chromonatural_eom_final} to its $\tau\to0^{-}$ limit, namely by applying the transformation \eqref{eq:symmetry3} in the limit as $\lambda\to0$. (This fact is part of \prettyref{thm:bg-future}, and can be readily verified from the formulas of \prettyref{app:t_to_0}.) Indeed, functions of the form of this ansatz are precisely the fixed points of \eqref{eq:symmetry3}.

We obtain a solution to \prettyref{eq:chromonatural_eom_final} when $c$ is one of 
\begin{equation}
c_{0}=0\,,\quad c_{1}=\tfrac{1}{2}(1-\sqrt{1-4/\xi^{2}})\,,\quad c_{2}=\tfrac{1}{2}(1+\sqrt{1-4/\xi^{2}})\,,\label{eq:c-definitions}
\end{equation}
motivating the nomenclature for $c_{i}$\nobreakdash-type solutions introduced above. Note that since $f(\tau)$ must be real, the $c_{1}$ and $c_{2}$ solutions exist only when $\xi\geq2$. In this case, 
\[
0+\xi^{-2}<c_{1}\leq\tfrac{1}{2}\leq c_{2}<1-\xi^{-2}\,,
\]
and asymptotically as $\xi\to\infty$ we have 
\begin{equation}
c_{1}=0+\xi^{-2}+O(\xi^{-4}),\quad c_{2}=1-\xi^{-2}+O(\xi^{-4})\,.\label{eq:c-asymptotics}
\end{equation}
The reader will find it especially useful to keep in mind that $c_{2}\approx1$ for large $\xi$.

We shall see in \prettyref{subsec:phase-space} that the $c_{0}$ and $c_{2}$ solutions are stable under all small perturbations of the initial conditions. Thus they both have a two-parameter basin of attraction. The $c_{1}$-solution is stable under just one direction of perturbations, so it is just part of a one-parameter family. \prettyref{app:t_to_0} contains explicit asymptotic formulas for these families. The structure of these families is explained in \prettyref{subsec:phase-space}.

The $c_{2}$ solution (which exists only when $\xi\geq2$) plays a central role in our story because it is an explicit stable non-abelian solution: 
\begin{equation}
ef(\tau)=c_{2}\,\xi/(-\tau)=\frac{\tfrac{1}{2}\left(1+\sqrt{1-4/\xi^{2}}\right)\xi}{-\tau}\,.\label{eq:c2solution}
\end{equation}

\begin{description}
\item [{Note}] We refer to the three particular solutions 
\begin{equation}
ef(\tau)=c_{i}\,\xi/(-\tau),\quad i\in\left\{ 1,2,3\right\} \label{eq:ci-solutions-123}
\end{equation}
respectively as \emph{the $c_{0}$ solution} (or simply \emph{the zero solution}), \emph{the $c_{1}$ solution} and \emph{the $c_{2}$ solution}. In contrast there are three \textbf{families} of $c_{i}$\nobreakdash-\textbf{type} \emph{solutions}, of which the $c_{i}$ solutions are respective members. A $c_{i}$\nobreakdash-type solution approaches the corresponding $c_{i}$ solution in the infinite future. More details on the families of $c_{i}$\nobreakdash-type solutions are given in \prettyref{app:t_to_0}. 
\end{description}

\subsubsection*{Oscillatory behaviour}

We remind the reader that although the oscillatory regime for the classical background field equation \prettyref{eq:chromonatural_eom_final} which we describe in this subsection extends to the infinite past, our model does not obey this classical equation at early times (see page~\pageref{ref:no-osc}). Nevertheless we will see in this section how the mathematical analysis of the oscillatory regime in the infinite past provides a nice criterion for determining which initial conditions lead to either $c_{0}$\nobreakdash-type solutions or $c_{2}$\nobreakdash-type solutions.

The oscillatory behaviour of solutions is explained by the following theorem: 
\begin{thm}
\label{thm:oscillatory}Any particular solution $ef(\tau)$ to \prettyref{eq:chromonatural_eom_final} has two associated constants: 
\begin{itemize}
\item $\omega\geq0$, 
\item $u_{0}\in[0,5.244)$. 
\end{itemize}
These constants depend on the solution, so they are determined once initial conditions are fixed. The solution can be written in the form 
\begin{align}
ef(\tau) & =\omega\cdot\mathrm{sn}(\omega\tau+u_{0})+\epsilon(\tau)\,,\label{eq:approx-quartic}
\end{align}
for some function $\epsilon(\tau)$ which is $\mathcal{O}((\xi+\xi^{2})/(-\tau))$ as $\tau\to-\infty$. Here $\mathrm{sn}(u)$ denotes the Jacobi $\mathrm{sn}(u|m)$ function with elliptic parameter $m=-1$ (see \prettyref{app:Jacobi} for details). We recall that the Jacobi $\mathrm{sn}$ function with argument $m$ is periodic with quarter-period given by the complete elliptic integral $K(m)$. The precise range for the periodic parameter $u_{0}$ is thus $[0,4K(-1))$. The constant $\omega$ is always uniquely determined by initial conditions. The constant $u_{0}$ is uniquely determined when $\omega\neq0$. The parameters $(\omega,u_{0})$ transform under \eqref{eq:symmetry3} as 
\begin{equation}
(\omega,u_{0})\mapsto(\lambda\omega,u_{0})\,.\label{eq:ic-transform}
\end{equation}
\end{thm}

Moreover, we have numerically verified the stronger statement that for \emph{all} $\tau$, 
\begin{equation}
\left|\epsilon(\tau)\right|\leq\frac{4\xi}{-3\tau}\,.\label{eq:epsilon-ineq}
\end{equation}
We prove \prettyref{thm:oscillatory} in \prettyref{app:pf-approx-w}. A rigorous proof of \prettyref{eq:epsilon-ineq} is likely possible using similar techniques, but it is beyond the scope of this paper. 

\prettyref{thm:oscillatory} tells us that when $\omega>0$, $ef(\tau)\approx\omega\cdot\mathrm{sn}(\omega\tau+u_{0})$ when $\left|\epsilon(\tau)\right|\ll\omega$. In particular combining this with \prettyref{eq:epsilon-ineq}, oscillation occurs at early times when 
\begin{equation}
\tau\ll-\frac{\xi}{\omega}\,,\label{eq:osc-from-omega}
\end{equation}
which we take as the definition of the oscillatory regime. In the case $\omega=0$, \eqref{eq:epsilon-ineq} implies that $|ef(\tau)|\leq\tfrac{4}{3}\xi/(-\tau)$ so that there are no oscillations.

We remark that \prettyref{thm:oscillatory} and \prettyref{eq:epsilon-ineq} are consistent with the results obtained from the ansatz~\eqref{eq:ansatz-f}. Namely the $c_{i}$ solutions correspond to $\omega=0$ and $\epsilon(\tau)=ef(\tau)=c_{i}\xi/(-\tau)$. (This satisfies \prettyref{eq:epsilon-ineq} because $\left|c_{i}\right|\leq\tfrac{4}{3}$.) To explain why $\omega=0$ is necessary for the $c_{i}$ solutions, recall that the $c_{i}$ solutions are fixed points of the transformation \eqref{eq:symmetry3}. Thus by \prettyref{eq:ic-transform} we have $\omega=\lambda\omega$ for all positive $\lambda$, and hence $\omega=0$. 

\prettyref{eq:osc-from-omega} is unfortunately not very practical for determining which initial conditions lead to oscillation, because $\omega$ is difficult to compute from given initial conditions. As a remedy, the following theorem suggests a very simple criterion in terms of initial conditions $ef(\tau_{1})$ and $ef'(\tau_{1})$ at time $\tau_{1}$. It introduces a function $\omega_{ef}(\tau)$ which serves as an approximation to the constant $\omega$. 
\begin{thm}
\label{thm:approx-w}(Criterion for oscillation) Let $ef(\tau)$ be a particular solution to \prettyref{eq:chromonatural_eom_final}. Define the associated function 
\begin{align}
\omega_{ef}(\tau) & \equiv\sqrt[4]{(ef'(\tau))^{2}+\left(ef(\tau)\right)^{4}}.\label{eq:omega_ef}
\end{align}
 As explained in \prettyref{app:Jacobi}, this approximates the envelope of $ef(\tau)$ as it oscillates. Then 
\[
\omega\equiv\lim_{\tau\to-\infty}\omega_{ef}(\tau)
\]
 coincides with the parameter $\omega$ specified in \prettyref{thm:oscillatory}. Furthermore, the solution is oscillatory (i.e.\ $\omega>0$) if there exists any time $\tau_{1}$ such that either 
\begin{itemize}
\item $\omega_{ef}(\tau_{1})>0$ when $0\leq\xi<2$, or 
\item $\omega_{ef}(\tau_{1})>\tfrac{4}{3}\xi/(-\tau_{1})$. 
\end{itemize}
\end{thm}

We prove this theorem in \prettyref{app:pf-approx-w}. 

Based on the second bullet point, we identify the transition time $\tau_{1}$ between the oscillatory regime and non-oscillatory regime as occurring when 
\begin{equation}
\sqrt[4]{(ef'(\tau_{1}))^{2}+\left(ef(\tau_{1})\right)^{4}}\approx\frac{4\xi}{-3\tau_{1}}\approx\frac{\xi}{-\tau_{1}}\,.\label{eq:BoundaryOscillatory}
\end{equation}
 This answers the second question of page~\pageref{enu:central-questions}. As will become clear from the discussion in Sec.~\ref{subsec:phase-space}, we can use this to estimate the necessary amplitude of the gauge field fluctuations which is required to trigger a $c_{2}$\nobreakdash-type solution.

Here we pause to take account of the two notions of ``oscillatory'' that we have developed so far. Firstly, a solution is, according to \prettyref{thm:oscillatory} and \prettyref{eq:epsilon-ineq}, oscillatory in the far past if the constant $\omega$ associated with the solution is positive. In that case, the solution oscillates when $\tau\ll-\frac{\xi}{\omega}$. Thus $-\omega\tau\sim\xi$ sets the scale for the transition. In contrast, \prettyref{thm:approx-w} provides a particular criterion which is well-suited for checking whether initial conditions at some time $\tau_{1}$ have corresponding solutions which begin in this oscillatory regime: $\omega_{ef}(\tau_{1})>\tfrac{4}{3}\xi/(-\tau_{1})$ with $\omega_{ef}(\tau)$ defined in \prettyref{eq:omega_ef}. 

\begin{figure}[t]
\centering{}\includegraphics{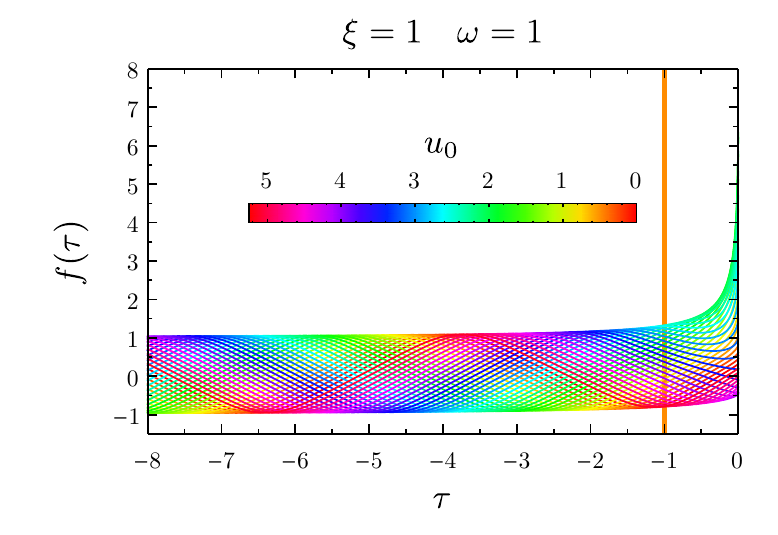}\hfill{}\includegraphics{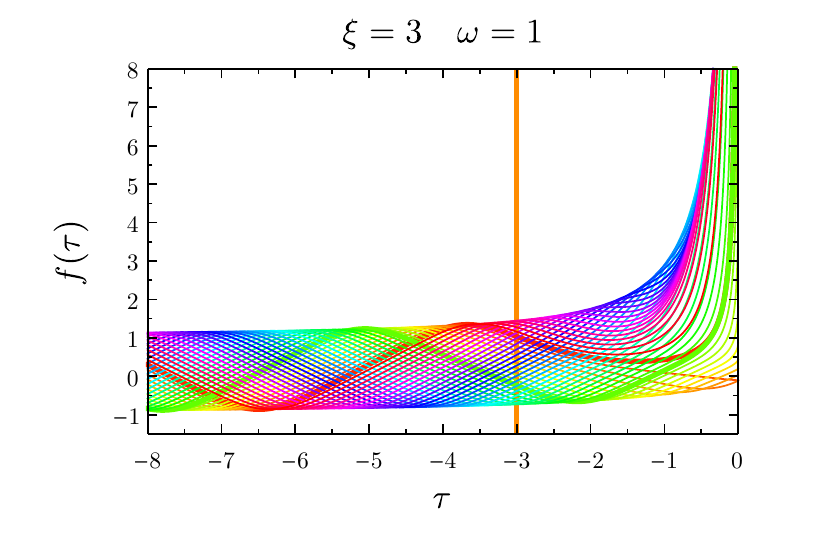}\\
\includegraphics{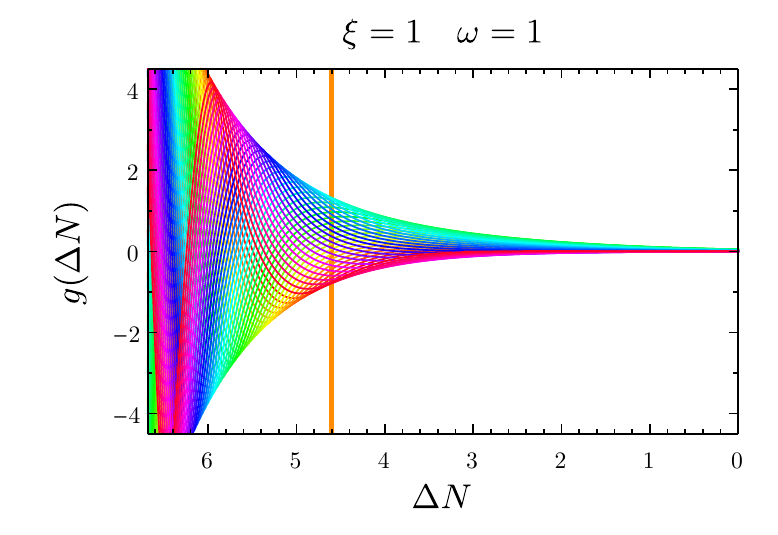}\hfill{}\includegraphics{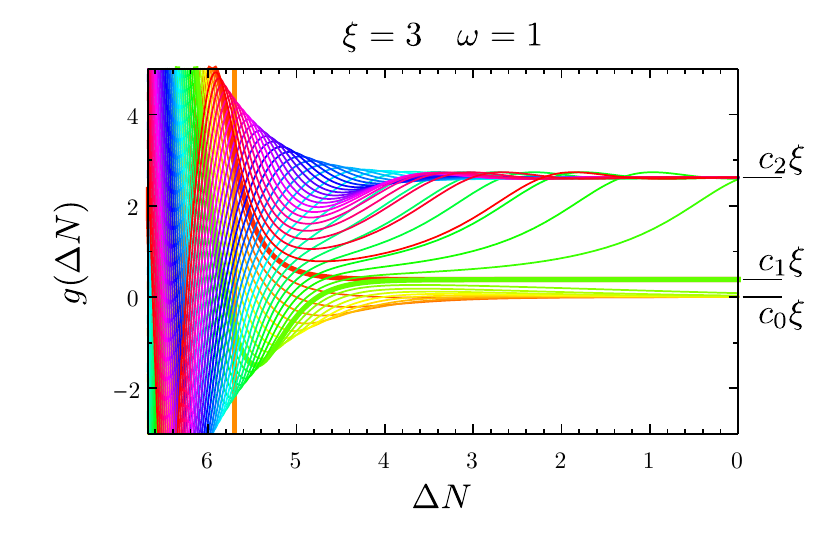} \caption{\label{fig:omega_eq_1} Systematic study of all possible solutions to the classical background equation~\eqref{eq:chromonatural_eom_final} for two different values of $\xi$. All solutions have been normalized to unit amplitude ($\omega=1$) and the phase $u_{0}$ is indicated by colour. The lower panels show the same solutions in physical coordinates. This illustrates that the limiting values of $g(N)$ as $\tau\to0^{-}$ are $c_{i}\xi$, and the value of the phase $u_{0}$ determines which of the $c_{i}\xi$ is reached. The two special $c_{1}$\protect\nobreakdash-type solutions are indicated by thicker lines. The orange vertical line indicates the transition time $\tau=-\xi/\omega$. We have chosen $\Delta N$ to denote the number of e-folds before $\tau_{0}=-10^{-2}$. }
\end{figure}
For solutions with $\omega\neq0$, we may normalize the amplitude of oscillations to $\omega=1$ by applying the transformation \eqref{eq:symmetry3} with $\lambda=\omega^{-1}$. (In terms of the physical quantity $g(N)$, this entails time-shifting the solutions so that they all exit the oscillatory regime at the same point in time.) \prettyref{fig:omega_eq_1} illustrates how the solutions of \prettyref{eq:chromonatural_eom_final} (normalized to $\omega=1$) depend on the remaining free phase parameter $u_{0}$. We note that the upper-left panel with $\xi=1$ does not admit unbounded solutions as $\tau\rightarrow0^{-}$, whereas the upper-right panel ($\xi=3$) admits both bounded and unbounded solutions. The value of $u_{0}$ is colour-coded, and we point out that for $\xi=3$, the $c_{0}$\nobreakdash-type solutions have colours which range only from orange-red to yellow-green. More precisely, this is the interval $u_{0}\in\left(0.136,1.409\right)$, which is approximately one fourth of the phase range. When $\xi=3$ and the phase $u_{0}$ is random, the probability of a $c_{0}$\nobreakdash-type solution is 24.3\%, and the probability of a $c_{2}$\nobreakdash-type solution is $75.7$\%.

The distinct categories of solutions are particularly evident in the lower panels of \prettyref{fig:omega_eq_1} depicting the physical gauge field amplitude. The limiting values in the infinite future are discrete: 
\begin{thm}
\label{thm:bg-future}When $\xi\neq2$, all solutions of \prettyref{eq:chromonatural_eom_final} satisfy
\[
\lim_{\tau\to0^{-}}-\tau\,ef(\tau)=c_{i}\xi
\]
for some $i\in\left\{ 0,1,2\right\} $, where $c_{i}$ is defined in \prettyref{eq:c-definitions}. Furthermore, the $\lambda\to0$ limit of the transformation \eqref{eq:symmetry3} applied to any solution of \prettyref{eq:chromonatural_eom_final} is the corresponding $c_{i}$ solution. 
\end{thm}

The proof is provided in \prettyref{app:pf-bg-future}, and it uses the machinery developed in \prettyref{subsec:phase-space}.

\begin{figure}[t]
\begin{centering}
\includegraphics{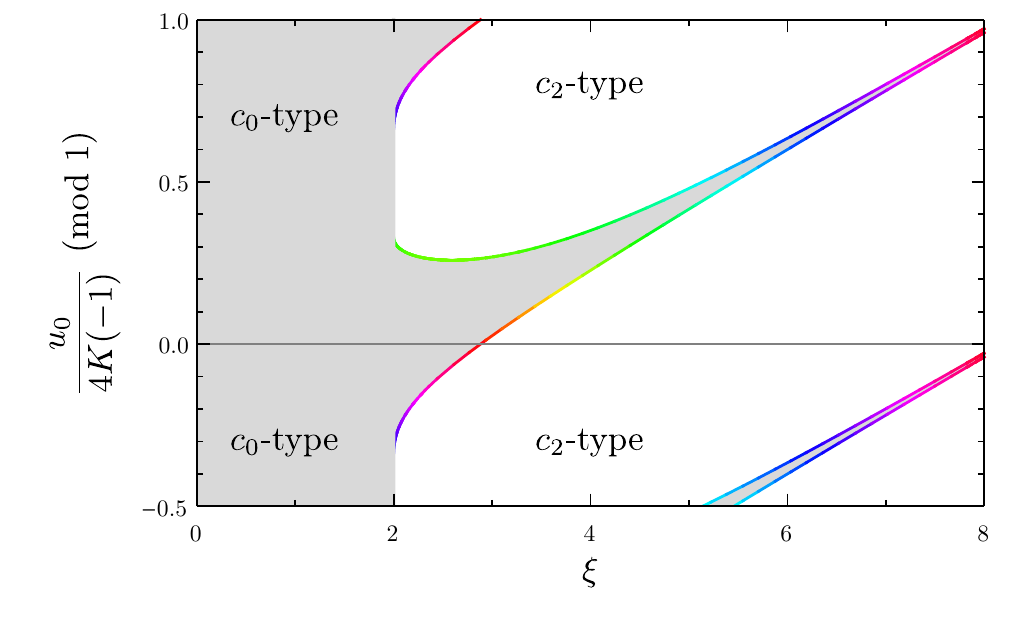} 
\par\end{centering}
\caption{\label{fig:u0-of-xi}Parameter space leading to different types of solutions for the classical background gauge field. The coloured curves indicate the phase (normalized mod 1) of the two $c_{1}$\protect\nobreakdash-type solutions as a function of $\xi$. Note that the $y$-axis has period 1. The grey and white regions respectively indicate $c_{0}$- and $c_{2}$\protect\nobreakdash-type solutions. This figure illustrates that $c_{0}$\protect\nobreakdash-type solutions are rare when $\xi$ is large. }
\end{figure}
This theorem gives us a way to classify solutions into three distinct categories. Solutions with generic initial conditions are always of type $c_{0}$ or $c_{2}$. Fine-tuning the phase $u_{0}$ to achieve a solution exactly between $c_{0}$ and $c_{2}$ leads to exactly two values of the phase which correspond to ``$c_{1}$\nobreakdash-type solutions.'' The corresponding red and green curves in the lower panels of \prettyref{fig:omega_eq_1} are indicated with thicker lines. The following theorem formalizes the notion that the $c_{1}$\nobreakdash-type solutions are the boundary between $c_{0}$\nobreakdash-type and $c_{2}$\nobreakdash-type solutions, the proof of which is also provided in \prettyref{app:pf-bg-future}.
\begin{thm}
\label{thm:phase-intervals}For each $\xi>2$, there exists exactly two distinct values of $u_{0}$ corresponding to $c_{1}$\nobreakdash-type solutions. The two complementary phase intervals correspond respectively to $c_{0}$\nobreakdash-type and $c_{2}$\nobreakdash-type solutions. 
\end{thm}

\prettyref{fig:u0-of-xi} visualizes the phase intervals leading to the respective $c_{0}$\nobreakdash-type and $c_{2}$\nobreakdash-type solutions, generalizing the above results to the entire $\xi$-range of interest. We note in particular that for $\xi\gtrsim4$, the $c_{0}$\nobreakdash-type solutions become highly unlikely for random initial conditions. This will be a crucial ingredient in answering the first question on page~\pageref{enu:central-questions}. 

\subsubsection{\label{subsec:phase-space}Phase space diagram}

\subsubsection*{Change of variables }

As we saw in \prettyref{eq:g-def}, there is a change of variables which puts \prettyref{eq:chromonatural_eom_final} into the form of an autonomous system. Thus the dynamics are captured by a 2-dimensional phase space diagram. This enables us to re-phrase the results obtained in \prettyref{subsec:solutions} in a more intuitive way.

Rather than choose $(g(N),g'(N))$ as phase space coordinates, we find the following choice more convenient: 
\begin{equation}
q(\tau)\equiv-\tau\,ef(\tau)\,,\quad p(\tau)\equiv(-\tau)^{2}ef'(\tau)\,.\label{eq:cov}
\end{equation}
The equations of motion under these new coordinates then become 
\[
\frac{\mathrm{d}q}{\mathrm{d}\tau}=\frac{p(\tau)-q(\tau)}{-\tau},\quad\frac{\mathrm{d}p}{\mathrm{d}\tau}=\frac{-2\left(q(\tau)^{3}-\xi\,q(\tau)^{2}+p(\tau)\right)}{-\tau}\,.
\]
The denominator of $-\tau$ can be eliminated via the substitution $\mathrm{d}\tau=\tau\,\mathrm{d}N$, rendering the system autonomous: 
\begin{align}
\frac{\mathrm{d}q}{\mathrm{d}N} & =q-p\,,\quad\frac{\mathrm{d}p}{\mathrm{d}N}=2(q^{3}-\xi\,q^{2}+p)\,.\label{eq:q-p-eom}
\end{align}
Just as for the physical quantity $g(N)$ defined in \prettyref{eq:g-def}, the transformation \eqref{eq:symmetry3} also acts on $q(N)$ and $p(N)$ as $N$-translation 
\begin{equation}
N\mapsto N+\ln\lambda\,.\label{eq:lambda-time-translation}
\end{equation}

This differential equation is solved by the flow lines of the vector field 
\begin{equation}
\left(q-p,2\left(q^{3}-\xi\,q^{2}+p\right)\right)\,,\label{eq:vec-field}
\end{equation}
in the $q$-$p$ plane.

We now begin a complete classification of solutions to \prettyref{eq:chromonatural_eom_final} based on an analysis of this vector field \eqref{eq:vec-field}. For simplicity we exclude the degenerate case when $\xi=2$ exactly.

The zeroes of this vector field are readily verified to be 
\begin{equation}
\mathbf{c}_{i}\equiv\left(q,p\right)=\left(c_{i}\xi,c_{i}\xi\right),\label{eq:ci-points}
\end{equation}
for $c_{i}$ defined in \eqref{eq:c-definitions}, and the corresponding constant trajectories are, up to the change of variables \eqref{eq:cov}, the $c_{i}$ solutions of \prettyref{eq:ci-solutions-123}. Therefore for $\xi<2$, $\mathbf{c}_{0}$ is the unique zero of \eqref{eq:vec-field}. As $\xi$ passes through the value $2$, the pair of zeroes $\mathbf{c}_{1}$ and $\mathbf{c}_{2}$ is created at the point $(1,1)$. Thus for $\xi>2$ there are three zeroes in total. The zeroes at $\mathbf{c}_{0}$ and $\mathbf{c}_{2}$ are stable, while $\mathbf{c}_{1}$ is a saddle point. Thus the stable trajectories of $\mathbf{c}_{0}$ and $\mathbf{c}_{2}$ form two-parameter families, while the stable trajectories of $\mathbf{c}_{1}$ form only a one-parameter family. These families correspond to the ``$c_{i}$\nobreakdash-type solutions'' described on page~\pageref{c0typesol}.

\subsubsection*{Visualizing solutions with a phase-like diagram}

Using the change of variables from \prettyref{eq:cov}, we can visualize the structure of solutions to \prettyref{eq:chromonatural_eom_final} in a very effective manner. Solutions to \prettyref{eq:chromonatural_eom_final} can be plotted as trajectories in the $q$-$p$ plane. Two solutions parameterize the same trajectory if and only if they are related by a shift in the time variable $N$ according to \prettyref{eq:lambda-time-translation}. In the first row of \prettyref{fig:trajectories-panel} we plot various such trajectories. Since the phase $u_{0}$ defined by \prettyref{thm:oscillatory} is invariant under \eqref{eq:symmetry3}, each trajectory has a well-defined phase which is indicated by colour in the first row of \prettyref{fig:trajectories-panel}, with the same colour coding as in Fig.~\ref{fig:omega_eq_1}. Oscillation is represented by the spirals in the top right panel of \prettyref{fig:trajectories-panel}. Solutions spiral inwards along a trajectory of fixed colour, and each crossing of the $p$-axis corresponds to a zero of the solution.

\begin{figure}
\centering{}\includegraphics{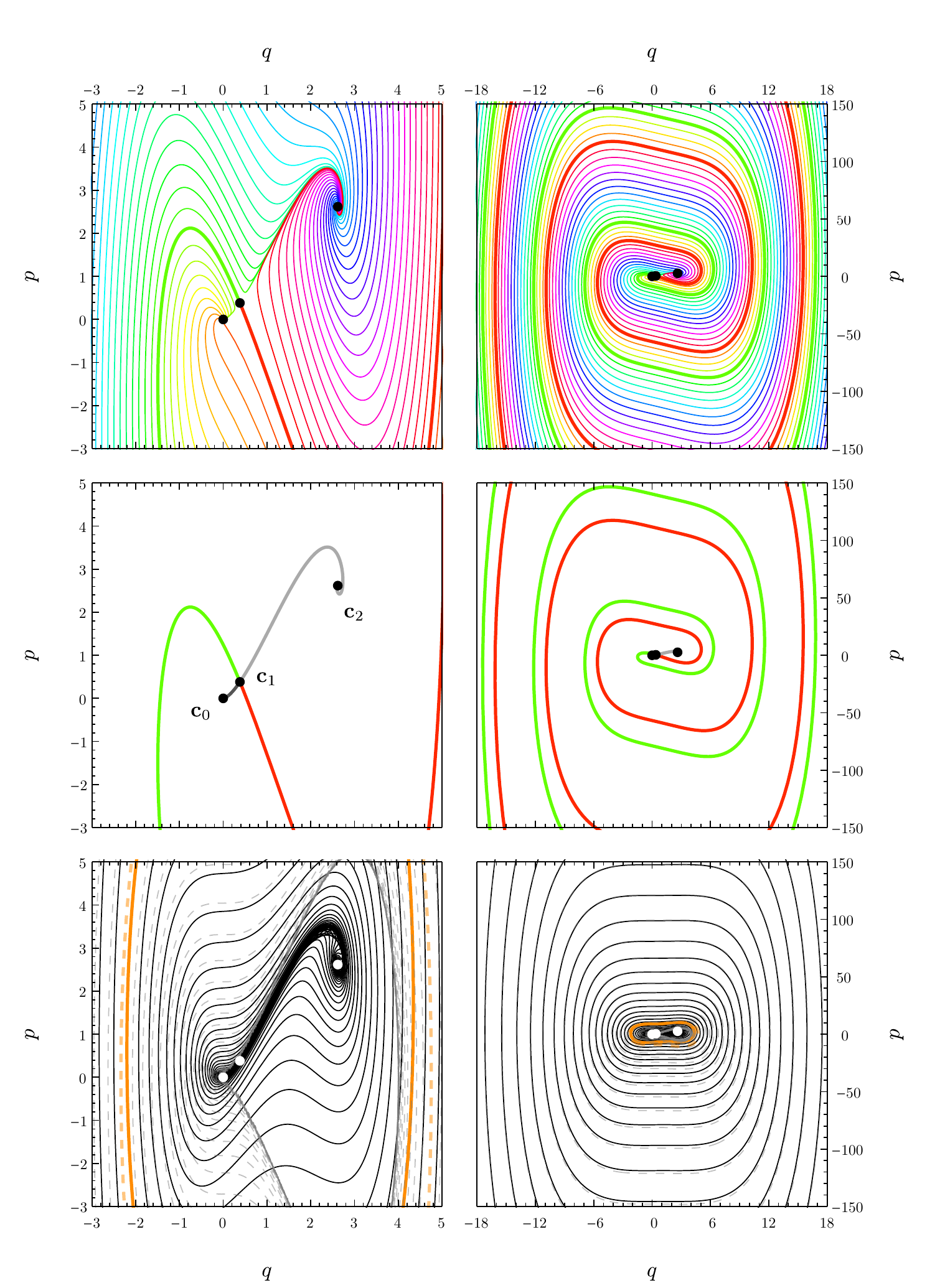} \caption{\label{fig:trajectories-panel} Evolution of the background gauge field in phase space, depicted by trajectories of the vector field \prettyref{eq:vec-field} for $\xi=3$. The dots are zeroes of the vector field, corresponding to the $c_{i}$ solutions of Sec.~\ref{subsec:solutions}. The left column shows the non-oscillatory regime around the zeroes, whereas the right column is a zoomed-out view showing the oscillatory regime. The second row shows some special trajectories, the third row depicts contours of constant $-\omega\tau$. See the text for further details.}
\end{figure}

Already at this point, since no two trajectories are allowed to cross, we observe that the boundary between the basins of attraction for $\mathbf{c}_{0}$ and $\mathbf{c}_{2}$ is given precisely by the two trajectories of $c_{1}$\nobreakdash-type solutions (plus of course their limit point $\mathbf{c}_{1}$.) These are depicted as red and green lines in the second row of Fig.~\ref{fig:trajectories-panel}.

We can construct a natural coordinate system on the $q$-$p$ plane by taking a coordinate complementary to the phase $u_{0}$. The complementary invariant $\omega$ of solutions is not a suitable candidate, because it transforms nontrivially under \eqref{eq:symmetry3} according to \prettyref{eq:ic-transform}. The quantity $-\omega\tau$ is however invariant under \eqref{eq:symmetry3}, and level curves are shown in the last row of \prettyref{fig:trajectories-panel} (with a spacing of $N/10$, see Eq.~\eqref{eq:g-def}). For a given trajectory, these level curves correspond to fixed-time contours, and hence the speed of approach to the respective $c_{i}$ solution is encoded in the spacing of these level curves. The contour $-\omega\tau=\xi$ is highlighted in orange, indicating the transition between the oscillatory and the non-oscillatory regime (see \prettyref{eq:osc-from-omega}). Moreover, we indicate the level curves of $D(N)^{1/4}$ (see Eq.~\eqref{eq:def-D}) as dashed lines, showing the excellent agreement between $-\omega\tau$ and the auxiliary function $D(N)^{1/4}$ in the oscillatory regime. These dashed level curves accumulate to the level curve $D(N)=0$, which plays a key role in the analysis of \prettyref{app:sec3}.

The resulting coordinate system is degenerate at $\omega=0$. In the left panel of the middle row of Fig.~\ref{fig:trajectories-panel}, we show the $\omega=0$ solutions corresponding to the two ``instanton-type'' trajectories (vacuum-to-vacuum transitions which tunnel from the $c_{1}$ solution in the infinite past to either the $c_{0}$ solution or $c_{2}$ solution in the infinite future) as grey curves. The asymptotic formula for the corresponding solutions is \prettyref{eq:instanton-type} with $\rho>0$ and $\rho<0$ respectively. These instanton-type solutions, together with all the $c_{i}$ solutions (which are the limit points), are all of the only non-oscillatory ($\omega=0$) solutions to \prettyref{eq:chromonatural_eom_final}. 
\begin{figure}
\begin{centering}
\includegraphics{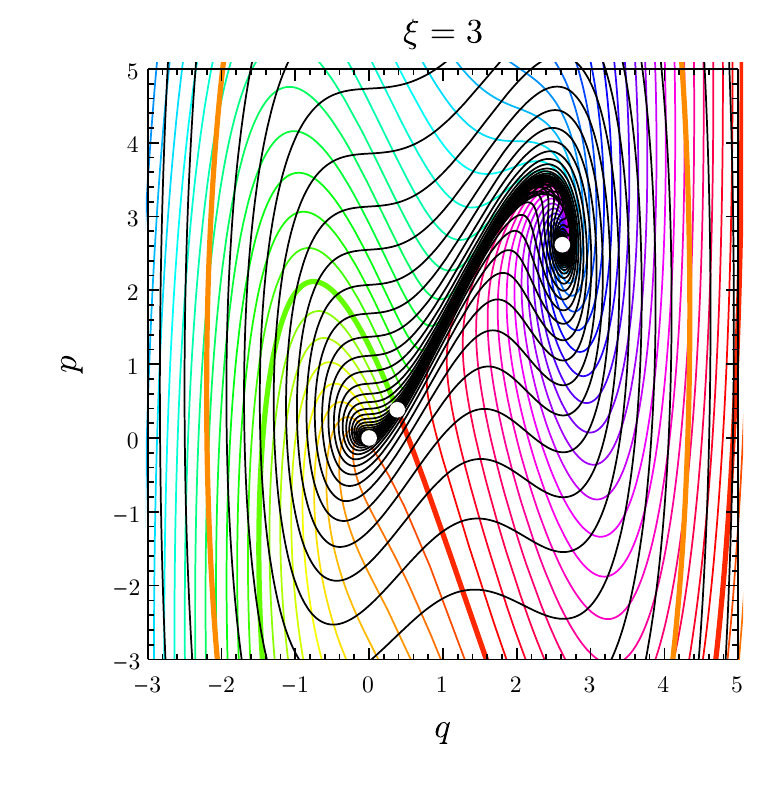}$\qquad$\includegraphics{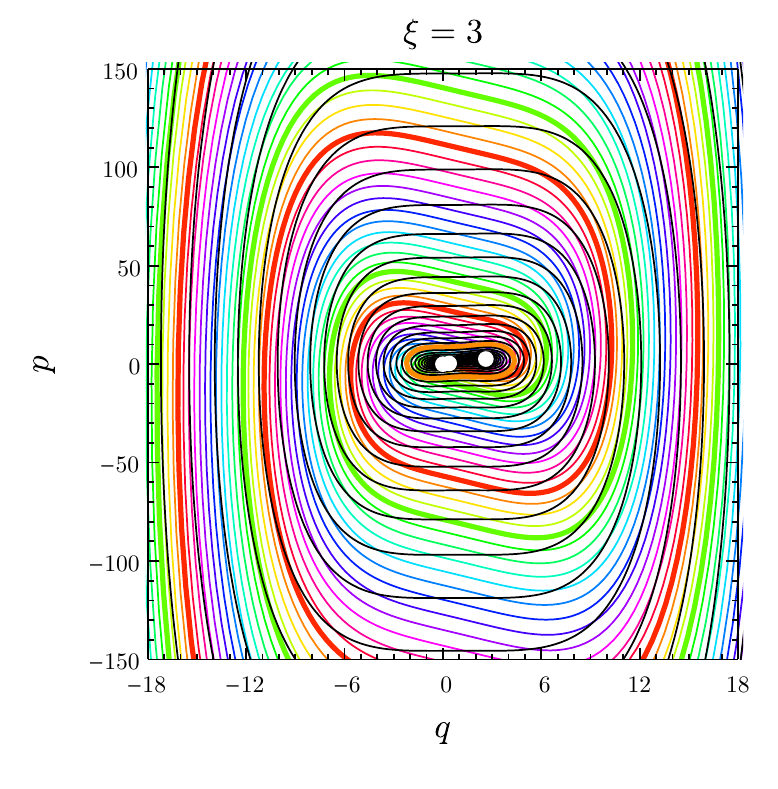} 
\par\end{centering}
\caption{\label{fig:rainbow} Summary of the evolution of the background gauge field. The coloured lines are trajectories of the classical background field evolution in phase space, with the colour coding corresponding to different phases $u_{0}$ as in Fig.~\ref{fig:omega_eq_1}. The black curves are contours of constant $-\omega\tau$. }
\end{figure}

In this way, we understand the structure of all the solutions to \prettyref{eq:chromonatural_eom_final} and how they fit together. The results are summarized in Fig.~\ref{fig:rainbow}, which simply combines the first and last row of Fig.~\ref{fig:trajectories-panel}.

We see how if generic initial conditions are chosen to be oscillatory, they will spiral inwards towards either $\mathbf{c}_{0}$ or $\mathbf{c}_{2}$. Finally, recall from \prettyref{fig:u0-of-xi} that $\mathbf{c}_{2}$ is favoured, overwhelmingly so as $\xi$ increases. This explains why Eq.~\eqref{eq:BoundaryOscillatory} can be used as a criterion for the required magnitude of the initial conditions necessary to trigger a $c_{2}$\nobreakdash-type solution.

We note that $c_{1}\approx0$ for large $\xi$, and hence the basin of attraction for $\mathbf{c}_{2}$ actually comes very close to $\mathbf{c}_{0}$. Thus it is quite likely that the abelian gauge field fluctuations will trigger a $c_{2}$\nobreakdash-type solution even before the oscillatory regime is entered. But given that the fluctuations grow exponentially with $\xi$, it is sufficient to simply have an order-of-magnitude estimate for the transition time. From \prettyref{eq:BoundaryOscillatory} we conclude that the transition occurs when 
\begin{equation}
e\langle A_{\text{ab}}^{2}\rangle^{1/2}\sim\xi/(-\tau).\label{eq:transition-time}
\end{equation}

\subsection{Anisotropic background gauge fields \label{subsec:Anisotropic-background-fields}}

Until now, throughout our analysis of the homogeneous background we have assumed isotropy, so that
\begin{equation}
(A^{(0)})_{i}^{a}=\delta_{i}^{a}\,f(\tau).\label{eq:iso-bg}
\end{equation}
 Here we study the effect of anisotropies in the background gauge field, assuming de-Sitter space. (Note that anisotropic CNI cosmologies have been studied e.g. in \cite{Maleknejad:2013npa}.) We verify that all anisotropies of a homogeneous gauge-field background decay (in physical coordinates) into one of the previously-studied isotropic solutions. 

First we show there are no anisotropic analogues of the $c_{i}$ solutions. Next we consider homogeneous anisotropic perturbations of \eqref{eq:iso-bg} to linear order. Finally, as a non-perturbative verification, we numerically solve the full nonlinear equations of motion for a homogeneous anisotropic background. This justifies our previous assumption of isotropy.

\subsubsection*{No anisotropic steady states }

We make the anisotropic analogue of our ansatz \prettyref{eq:ansatz-f}, namely 
\begin{equation}
e\,(A^{(0)})_{i}^{a}(\tau)=C_{i}^{a}\xi/(-\tau)\,.\label{eq:ansatz-aniso}
\end{equation}
\prettyref{eq:eq_motion_final} yields the following equations: 
\begin{equation}
(2\xi^{-2}+\sigma_{2}^{2}+\sigma_{3}^{2})\sigma_{1}-2\sigma_{2}\sigma_{3}=0\textrm{ and cyclic permutations,}\label{eq:svdansatz}
\end{equation}
where $\sigma_{1}$, $\sigma_{2}$ and $\sigma_{3}$ are the singular values of $C_{i}^{a}$. As can be verified by solving this with a computer algebra system, the only real solutions are equivalent to the three isotropic solutions we already found in \prettyref{eq:ansatz-f}.

\subsubsection*{Anisotropic perturbations of the background gauge field to linear order }

We wish to consider first-order perturbations of \prettyref{eq:iso-bg} which are anisotropic, and thus of the form 
\[
(A^{(0)})_{i}^{a}=f(\tau)\delta_{i}^{a}+P_{i}^{a}(\tau)\,\epsilon+\mathcal{O}(\epsilon^{2})\,.
\]
 As explained in \prettyref{app:global-symmetries}, we may decompose $P_{i}^{a}$ into irreducible representations of the diagonal $\mathrm{SO}(3)$ subgroup of $\mathrm{SO}(3)_{\mathrm{gauge}}\times\mathrm{SO}(3)_{\mathrm{spatial}}$ as 
\[
P_{i}^{a}=s(\tau)\delta_{i}^{a}+v^{j}(\tau)\varepsilon_{ija}+T_{i}^{a}(\tau).
\]
 Since $f(\tau)$ already accounts for the diagonal degree of freedom, we impose that $s(\tau)=0$. We now substitute this ansatz into our twelve equations of motion. It's implicit here that the three $(A^{(0)})_{0}^{a}$ components are determined by the equation of motion (as constraint equations). Expanding out the remaining nine equations of motion, we obtain \prettyref{eq:chromonatural_eom_final}, together with the rank-five equation for the perturbations 
\begin{align*}
T''+\frac{2\xi}{-\tau}ef(\tau)\,T & =0\,.
\end{align*}
We find no equations of motion involving $v^{j}(\tau)$, indicating that they are the gauge degrees of freedom. Accordingly, the remaining three equations are equivalent to $0=0$.

In the case of the $c_{0}$-solution $f(\tau)=0$, the general solution is $T(\tau)=T_{0}+\tau T_{1}$ so that $T(\tau)\to T_{0}$ as $\tau\to0^{-}$. The corresponding physical quantity thus decays as $(-\tau)T_{0}+\mathcal{O}(\tau^{2})$ as $\tau\to0$. 

In any $c_{2}$\nobreakdash-type solution, as the isotropic component $f(\tau)$ grows in the positive direction, WKB theory dictates that $T$ decays in proportion to $\left(-\tau^{-1}f(\tau)\right)^{-1/4}$ (see \prettyref{eq:exp-approx}). Thus when $ef(\tau)$ is any $c_{2}$\nobreakdash-type solution, $T$ decays in proportion to $\sqrt{-\tau}$, and the corresponding physical quantity decays as $(-\tau)^{3/2}$. In the case of the $c_{2}$ solution, the exact solution for $T(\tau)$ is 
\begin{align*}
T(\tau) & =T_{0}\sqrt{-\tau}\exp\left(\pm i\,\mu_{2}\log(-\tau)\right)+\textrm{h.c}\,,\\
\mu_{2} & \equiv\xi\sqrt{2c_{2}-(2\xi)^{-2}}\,,
\end{align*}
where $T_{0}$ is a complex symmetric traceless tensor determined by the initial conditions.

\subsubsection*{Numerical solutions of full nonlinear anisotropic background}

We showed above that any homogeneous background solution which has anisotropies at sub-leading order must evolve towards isotropy. While this supports the hypothesis that any homogeneous background tends toward isotropy, it does not prove anything about highly anisotropic backgrounds. For this we resort to numerical simulation. Specifically, we numerically solve the fully anisotropic\footnote{For the fully anisotropic case, while one may diagonalize the spatial components of $A(\tau_{\mathrm{ini}})$ at the initial time $\tau_{\mathrm{ini}}$, the spatial components of $A(\tau_{\mathrm{ini}})$ and $A'(\tau_{\mathrm{ini}})$ are not simultaneously diagonalizable.} system \prettyref{eq:eq_motion_final} for the twelve functions\footnote{We have twelve functions $(A^{(0)})_{\mu}^{a}$ subject to three constraint equations and six independent dynamical equations. To get a well-formed system, one must add three gauge-fixing constraints. We found it convenient to impose temporal gauge $\left(A^{(0)}\right)_{0}=0$.} $(A^{(0)})_{\mu}^{a}$ which determine the background field $A^{(0)}$.

\begin{figure}[t]
\centering{}\subfigure{ \includegraphics{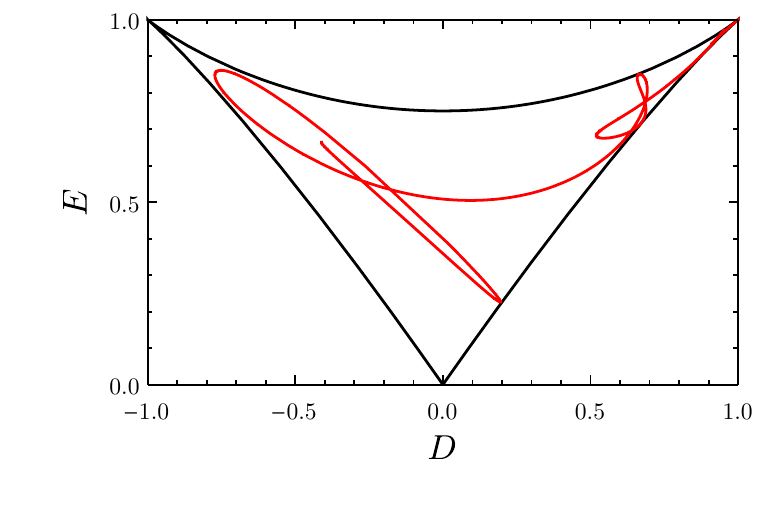} } \hfill{}\subfigure{ \includegraphics{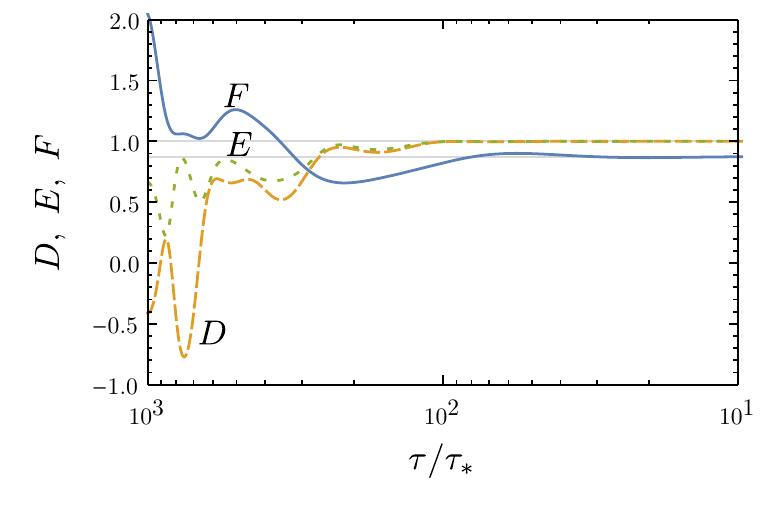} } \caption{\label{fig:c2-iso} A $c_{2}$\protect\nobreakdash-type solution with random anisotropic initial conditions evolving towards isotropy for $\xi=3$. Left panel: the positively orientated isotropic configuration corresponds to the top-right corner. Right panel: For the same parameter point, time evolution (relative to an arbitrary time $\tau_{*}$) of the quantities $D$, $E$ and $F$ as defined in the text. The horizontal lines denote the asymptotic values $1$ and $c_{2}\approx0.87$ characterizing an isotropic $c_{2}$ solution. }
\end{figure}

In order to understand the resulting numerical solutions, we need a way to visualize their properties. As a generalization of $ef(\tau)$ to the non-isotropic case, we define for any nonzero $3\times3$ matrix $A$: 
\begin{align}
F(A)\equiv\frac{-\tau|A|}{\sqrt{3}\xi}\ \textrm{ where }|A|\equiv\sqrt{A_{i}^{a}A_{i}^{a}}.\label{eq:F-norm-fn}
\end{align}
Then in the special case that $A^{(0)}$ is isotropic, $F(A^{(0)}(\tau))=-\tau\left|ef(\tau)\right|/\xi$. Thus if $A^{(0)}$ corresponds to an isotropic $c_{i}$\nobreakdash-type solution, then $\lim_{\tau\to0^{-}}F(A^{(0)}(\tau))=c_{i}$ in accordance with \prettyref{thm:bg-future}. Next we must quantify the degree to which $A^{(0)}$ is anisotropic. We define in \prettyref{app:quant-aniso} two further parameters $D(A)$ and $E(A)$ for this purpose, which are invariant under rotation, gauge symmetry, and multiplication by a positive scalar. Up to a normalization factor, $D(A)\in\left[-1,1\right]$ is $(\det A)/\left|A\right|^{3}$, while the definition of $E(A)$ is more involved. The pair of values $(D(A),E(A))$ determines a point in the triangular-shaped region in the left panel of \prettyref{fig:c2-iso} (see also \prettyref{fig:DE-triangle}). The matrix $A$ is isotropic when $D(A)=\pm1$, or equivalently when $E(A)=1$. When $D(A)=+1$ (resp.\ $-1$) the gauge field is positively (resp.\ negatively) oriented.\footnote{We say that an isotropic gauge field $A_{i}^{a}=f(\tau)\delta_{i}^{a}$ is \emph{positive} when $f(\tau)$ is positive. This has the following physical significance. An isotropic gauge field identifies an orthonormal basis of the Lie algebra with $|f(\tau)|$ times an orthonormal basis of 3-space (via contraction). The Lie algebra carries a natural orientation where the structure constants are $+i\varepsilon^{abc}$. For 3-space, the chiral term in our Lagrangian picks out a preferred orientation (which corresponds to the standard orientation when $\xi>0$). The relative orientation thus is the sign of $\xi f(\tau)$. Since we assume $\xi>0$, the relevant sign is that of $f(\tau)$.}

In \prettyref{fig:c2-iso} we show a typical example of a $c_{2}$\nobreakdash-type solution with random anisotropic initial conditions evolving towards isotropy. As expected for all $c_{2}$\nobreakdash-type solutions, $(D,E)\to(1,1)$ in the infinite future, indicating positively-oriented isotropy. (In contrast, $(D,E)$ need not approach $(1,1)$ for $c_{0}$\nobreakdash-type solutions since $F(A)\to0$ and zero is isotropic.) In our numerical simulations, we observe that within a few e-folds, all solutions converge towards an isotropic solution of the form \prettyref{eq:ansatz-aniso}, namely either a $c_{0}$\nobreakdash-type solution or a $c_{2}$\nobreakdash-type solution. The proportion of $c_{2}$\nobreakdash-type solutions was even higher than predicted in \prettyref{fig:u0-of-xi}. We conclude that 
\begin{itemize}
\item the $c_{2}$\nobreakdash-type solutions are stable against small anisotropic perturbations; 
\item sufficiently large anisotropic initial conditions with $\xi>2$ usually lead to $c_{2}$\nobreakdash-type solutions; 
\item the continuous sourcing of the background field through the enhanced abelian super-horizon modes will therefore inevitably lead to an isotropic $c_{2}$\nobreakdash-type background solution. 
\end{itemize}

\subsection{Coupled gauge field - inflaton background \label{subsec:dynamical_background}}

Previously in this section, we took $\xi$ to be a constant, external parameter in the equation of motion for the homogeneous gauge field background. We now turn to the complete dynamical background evolution, including also the evolution of the homogeneous inflaton field $\phi(\tau)$ and hence the (slow) evolution of $\xi$. This leads to the coupled system of equations
\begin{align}
f^{\prime\prime}(\tau)+2e^{2}f^{3}(\tau)-e\frac{\alpha}{\Lambda}\,\phi^{\prime}f^{2}(\tau)\,\  & =0 \,, \label{eq:fbackground2} \\
\phi''(\tau)+2aH\phi'(\tau)+a^{2}V_{,\phi}(\phi)+\frac{3\alpha e}{\Lambda a^{2}}f^{2}(\tau)f'(\tau) & =0 \,.
\label{eq:phibackground}
\end{align}
In single-field slow-roll inflation, $\xi$ typically increases over the course of inflation. This slowly evolving value of $\xi$ slightly modifies some of the results of the previous subsections (e.g.\ the precise values for the range of phases which lead to the $c_2$-solution in Fig.~\ref{fig:u0-of-xi} may be shifted), but the overall picture remains valid. After inserting the $c_2$-solution, Eq.~\eqref{eq:phibackground}  can be expressed as
\begin{equation}
 \frac{\dot \phi}{H} + \frac{V_{,\phi}}{V} +\frac{ \alpha}{e^2 \Lambda} H^2 (c_2 \, \xi)^3 = 0\,,
 \label{eq:phibackground2}
\end{equation}
where we have neglected the slow-roll suppressed term $\ddot \phi$. Assuming that the last term is sub-dominant, $\dot \phi/H$, $V_{,\phi}/V$ and $\xi$ are all proportional to $\sqrt{\varepsilon}$, with $\varepsilon = \dot \phi^2/(2 H^2) \simeq (V'/V)^2/2$ denoting the first slow-roll parameter. Moreover, for a quadratic or cosine potential as is usually considered in axion inflation models, $H^2$ is proportional to $1/\varepsilon$.  In summary, the time-dependence of all terms in Eq.~\eqref{eq:phibackground} is governed by the square root of the first slow-roll parameter. In particular, if the last term is sub-dominant at any point in time (after the $c_2$-solution has been reached), it will always remain sub-dominant. For the parameter example of Sec.~\ref{sec:example}, we find precisely this situation.

We note that this is a different regime than the `magnetic drift' regime studied in Refs.~\cite{Adshead:2012kp,Dimastrogiovanni:2012ew,Adshead:2013qp,Adshead:2013nka}. There, the friction term was taken to be large compared to the Hubble friction, $\xi \alpha H/(e \Lambda)  \gg 1$. Also in this regime, there is a local attractor for the gauge field background which scales as $f(\tau) \sim 1/\tau$, with however a different constant of proportionality. Within the non-abelian regime, the difference in our results with respect to these earlier works on CNI, in particular concerning the stability of the scalar sector, can be traced back to the fact that we do not restrict our analysis to this magnetic drift regime.

\section{Linearized equations of motion \label{sec:linearized}}

We now turn to the inhomogeneous equations of motion, adding perturbations to the homogeneous quantities discussed in the previous section.  This includes the perturbations of the gauge field, the inflaton and the metric. We start by deriving the linearized equations of motions for all relevant degrees of freedom in Sec.~\ref{sec:setup}, reproducing the results first obtained in Refs.~\cite{Dimastrogiovanni:2012ew,Adshead:2013qp}. The helicity basis, introduced in Sec.~\ref{GaugeAndBasis}, proves to be convenient to identify the physical degrees of freedom and simplify the system of equations. This becomes is particularly evident in Sec.~\ref{sec:eom_gauge_fields} which discusses the resulting equations of motion for the pure Yang--Mills sector. We can immediately identify the single enhanced mode and even give an exact analytical expression for the mode function in the limit of constant $\xi$. Finally, Sec.~\ref{sec:AllFluc} includes also the inflaton and metric tensor fluctuations. {Limitations of the linearized treatment of the perturbations are pointed out in Sec.~\ref{sec:AllFluc}. They primarily affect the helicity 0 sector and we will return to this point in more detail in Sec.~\ref{subsec:powerspectra}.} The results obtained in Secs.~\ref{sec:setup}, \ref{GaugeAndBasis} and \ref{sec:AllFluc} are in agreement with the findings of Refs.~\cite{Dimastrogiovanni:2012ew,Adshead:2013nka,Adshead:2013qp,Namba:2013kia}. Any differences in the results can be traced back to the different parameter regime for the background gauge field evolution, see discussion below Eq.~\eqref{eq:phibackground2}. In addition, we here provide analytical results for the simplified system of Sec.~\ref{sec:eom_gauge_fields}, setting the stage for semi-analytical estimates of the scalar and tensor power spectrum.  Throughout Sec.~\ref{sec:eom_gauge_fields} and \ref{sec:AllFluc}, the parameter $\xi$ is taken to be constant. Its evolution will be considered in Sec.~\ref{sec:example}.

\subsection{Setup for the linearized analysis \label{sec:setup}}

In this section we will derive the system of first-order differential equations for all gauge degrees of freedom and the inflaton fluctuations, assuming a homogeneous and isotropic gauge field background, for further details see App.~\ref{app:fulleom}.

The starting point is the action reported in Eq.~\eqref{eq:action}. We work in the ADM formalism~\cite{Arnowitt:1962hi}, i.e.\ we write the metric as
\begin{equation}
\textrm{d} s^2 = - N^2 \textrm{d} \tau^2 + h_{ij} \left(\textrm{d} x^i + N^i \textrm{d} \tau\right) \, \left(\textrm{d} x^j + N^j \textrm{d} \tau\right) \,.
\end{equation}
We decompose
\begin{equation}
h_{ij} = a^2 \left[\left(1+A\right) \delta_{ij} + \partial_i \partial_j B + \partial_{(i} C_{j)} + \gamma_{ij}\right] \,,
\end{equation}
where $\partial_i C_{i}=0$ and $\gamma_{ij}$ is transverse-traceless, i.e.\ $\gamma_{ii} = \partial_i \gamma_{ij} = 0$. There are four degrees of freedom arising from coordinate reparameterization: two scalar and two vector. In the scalar sector, we impose \textit{spatially flat gauge} which sets
\begin{equation}
A = B = 0 \,.
\end{equation}
In the vector sector, {we choose the gauge in such a way that} $C_i = 0$ (see Eq.~(A.115) of \cite{Baumann:2009ds}), which implies
\begin{equation}
h_{ij} = a^2 \left(\delta_{ij} + \gamma_{ij}\right) \,.
\end{equation}
As we numerically checked that the lapse ($N$) and shift ($N^i$) contributions to the subsequent equations do not affect the results, we discard them in this section\footnote{Hence, we take $N = a$ and $N^i = 0$.} for pedagogical reasons, and we refer to Appendix \ref{app:fulleom} for the complete expressions. 
{Note that the discarded vector $N^i$ contains two physical (but non-radiative) degrees of freedom from the metric.}\\

We expand the gauge fields and the inflaton field as (see also Eqs.~\eqref{eq:background_plus_linear} and \eqref{eq:ansatz-A})
\begin{align}
 A^a_i(\tau, \vec x) & = f(\tau) \, \delta^a_i + \delta A^a_i(\tau, \vec x) \,, \\
 \phi(\tau,x) &  = \langle\phi(\tau)\rangle + \delta \phi (\tau,x) \,,
\end{align}
where $\delta^a_i$ is the Kronecker delta function, $f(\tau) \, \delta^a_i$ and $\langle\phi(\tau)\rangle$ comprise the homogeneous background, while $\delta A^a_i(\tau, \vec x)$ and $\delta \phi(\tau,x)$ denote the quantum fluctuations around the homogeneous background. In order to infer the equations of motion which are linear in the fluctuations, we need to expand the Lagrangian up to quadratic order in all the field fluctuations. To make the computation easy to follow, we split it and the results into various terms arising from $\mathcal{S}_i = \int \textrm{d}^4 x \, \sqrt{-g} \, \mathcal{L}_i$, following Eq.~\eqref{eq:action}. The quadratic terms take the following form:
\begin{align}
 \delta^2 \mathcal{S}_\phi  = & \int \textrm{d}^4 x\, \left[\frac{a^2}{2} \left(\left(\delta \phi'\right)^2  - \left(\partial_i \delta \phi \right)^2 - a^2 V_{,\phi\phi} \left(\delta \phi\right)^2 \right)\right]\,, \\
\delta^2 \mathcal{S}_{\rm YM}  = & \int \textrm{d}^4x \, \left[- \frac{1}{2} \delta A^a_0 \partial_i \partial_i  \delta A^a_0 + \delta A^a_0 \partial_0 \left(\partial_i \delta A^a_i - e f \varepsilon^{abi} \delta A^b_i \right) + \right. \nonumber \\
& + \delta A^a_0 \left(2 e f' \varepsilon^{abi} \delta A^b_i + e f \varepsilon^{abi} \left(\partial_i\delta A^b_0\right) + e^2 f^2 \delta A^a_0\right) - \nonumber \\
& - \frac{1}{2}  \delta A^a_i  \left(\delta A^a_i\right)'' + \frac{1}{2} \delta A^a_j \left(\partial_i \partial_i \delta A^a_j\right) + \frac{1}{2} (\partial_i \delta A^a_i)^2 - \nonumber \\
& - e^2 f^2 \left(\left(\delta A^a_a\right)^2 + \frac{1}{2} \left(\delta A^b_i\right)^2 - \frac{1}{2} \delta A^b_i \delta A^i_b \right) - e f \varepsilon^{abc} \left( \delta^b_i \partial_i \delta A^a_k \delta A^c_k + \delta^c_k \partial_i \delta A^a_k \delta A^b_i \right) + \nonumber \\
& + \frac{\left(f'\right)^2 - e^2 f^4}{4} \gamma^{jk} \gamma^{kj} - f' \gamma^{aj} \partial_0 \delta A^{(a}_{j)} - e f^2 \varepsilon^{abc} \gamma^{ij} \left(\delta^b_{(i} \partial_{j)} \delta A^a_{c} - \delta^b_{(i} \partial_{c} \delta A^a_{j)} \right) \bigg] \,,
\label{eq:QuadraticLYMMainText}
\end{align}
 \begin{align}
\delta^2 \mathcal{S}_{\rm CS}  = & \int \textrm{d}^4 x \, \left[-\frac{\alpha}{\Lambda} \langle \phi \rangle \left(\varepsilon^{ijk} \left(\delta A^a_i\right)' \partial_j \delta A^a_k + 2 e f \left(\delta A^a_{[a}\right)' \delta A^k_{k]} + e f'  \delta A^b_{[b} \delta A^c_{c]}  \right) - \right. \nonumber \\
& \left. - \frac{\alpha}{ \Lambda} \delta \phi \left[f' \varepsilon^{ajk}  \partial_j \delta A^a_k + 2 e f f' \delta A^a_a +  e f^2 \left(\delta A^a_a\right)' -  e f^2 \partial_i \delta A^i_0 \right]\right] \,,
\label{eq:CSQuadraticMainText} \\
\label{eq:EH_to_use}
\delta^2 \mathcal{S}_{EH}  = & \int \textrm{d}^4 x \, \frac{a^2}{2} \left[\frac{ \gamma_{ij} \partial_l \partial_l \gamma_{ij} }{4} + \frac{\gamma^\prime_{ij} \gamma^\prime_{ij}}{4} \right] \,. 
\end{align}

As we will see in Sec. \ref{sec:GCG}, the term proportional to $\delta A^a_0$ in the first line of Eq.~\eqref{eq:QuadraticLYMMainText} vanishes after imposing the generalized Coulomb condition that reads (see Sec.~\ref{sec:GCG} for further details):
\begin{equation}
\partial_i \delta A^a_i - e f \varepsilon^{abi} \delta A^b_i = 0 \,.
\end{equation}

The equation of motion for $\delta A^a_0$ gives Gauss's law, which reads
\begin{align}
\label{eq:GaussLaw}
0 & = \left(\partial_i \partial_i  - 2 e^2 f^2 \right) \delta A^a_{0} - \left(\partial_i \delta A^a_i - e f \varepsilon^{abi} \delta A^b_i\right) - \nonumber \\
& - 2 e f' \varepsilon^{abi} \delta A^b_i - 2 e f \varepsilon^{abi} \left(\partial_i \delta A^b_0\right) + \frac{\alpha e}{\Lambda} f^2 \delta^{ai} \partial_i \delta \phi \,.
\end{align}

We write the linear equation of motion for the inflaton fluctuations in terms of the variable $(a \delta \phi)$ for later convenience
\begin{align}
0 & = - \left(a \delta \phi\right)'' + \partial_i \partial_i \left(a \delta \phi\right) + \frac{a''}{a} \left(a \delta \phi\right) - 2 \mathcal{H} \left(a \delta \phi\right) + 2 \mathcal{H}^2 \left(a \delta \phi\right) - a^2 V_{,\phi\phi} \left(a \delta \phi\right) - \nonumber \\
& - \frac{\alpha}{\Lambda a} \left[f' \varepsilon^{ijk} \partial_j \delta A^i_k + 2 e f f' \delta A^a_a + e f^2 \left(\delta A^a_a\right)' - e f^2 \partial_i \delta A^i_0 \right] \,,
\label{eq:EOMInflatonFluc}
\end{align}
where $\mathcal{H} = \frac{a'}{a}$. The linear equations of motion for the dynamical gauge field degrees of freedom are
\begin{align}
0 &= \delta A_i^{a \prime \prime } - \partial_j \partial_j \delta A_i^{a}  + \partial_i \left( -\delta A_0^{a \prime} - \partial_j \delta A^{a}_j \right) + \nonumber \\
& - e \varepsilon^{a b c } \left[ -2 \delta A_0^{b} \delta^c_i f^\prime  + 2 f \delta^b_j \partial_j \delta A_i^{c} + f \delta^c_i  (- \delta A_0^{b \prime} + \partial_j \delta A^{b}_j) - f \delta^b_j \partial_i \delta A_j^{c} \right] - \nonumber \\
& - e^2 f^2 \left[ \delta^a_j \delta^b_j \delta A_i^{b}+ \delta^a_j\delta^b_i \delta A_j^{b}+ \delta^b_j\delta^b_i \delta A_j^{a} - 3 \delta A_i^{a} - 2 \delta^a_i\delta^b_j \delta A_b^{j} \right] - \nonumber \\
& - \frac{\alpha}{2 \Lambda}\left[ \phi^\prime \varepsilon_{ i j k} \left[2 \partial_j A^a_k + 2 e f \varepsilon^{abc} \delta A_j^{b} \delta^c_k\right] + 2 f^\prime \varepsilon^{aji}\partial_j \delta \phi + 2 e f^2 \delta^a_i \delta \phi' \right] - \nonumber \\
& - f'' \gamma^{a}_i + f' \left(\gamma^a_i\right)' + e f^2  \varepsilon^{ajk} \partial_k \gamma_{ij} + e^2 f^3 \gamma^a_{i}  \equiv {\mathbf L}(\delta A,\phi,\gamma)\,,
\label{eq:eom_i_linear}
\end{align}
where for later convenience we have defined the linear operator $\mathbf L$. Finally, we give the equation of motion for the metric fluctuations in terms of the variable $(a \gamma_{ij})$ 
\begin{align}
\frac{a}{4}\left[(a \gamma_{ij})'' + \left(-\partial_l\partial_l - \frac{a''}{a} \right) (a \gamma_{ij})\right] & = \frac{f^{\prime \ 2} - e^2 f^4}{2 a}(a \gamma_{ij} ) - f^\prime \partial_0 \delta A^{(i}_{j)}  + f^\prime \partial_{(i} \delta A^{j)}_0 + \nonumber \\
& + e f^2 \gamma^{ij} \left[\varepsilon^{aic} \partial_{[j} \delta A^a_{c]} + \varepsilon^{ajc} \partial_{[i} \delta A^a_{c]} \right] + e^2  f^3  \delta A^{(i}_{j)} \,.
\label{eq:eomGWs}
\end{align}
{We point out that the right-hand side of this equation is given by the transverse traceless component of the anisotropic energy momentum tensor, and hence this equation is equivalent to the linearized Einstein equations used in gravitational wave physics~\cite{Maggiore:1999vm}. }

\subsection{Choice of gauge and basis \label{GaugeAndBasis}}

In the following we explain our choice of basis for dealing with the gauge field fluctuations, which will greatly simplify the analysis. After introducing the generalization of Coulomb gauge to a non-vanishing gauge field background, we decompose the 12 degrees of freedom of the gauge fields into helicity eigenstates. We further identify the degrees of freedom associated with gauge transformations and constraint equations, leaving us with six physical degrees of freedom. The explicit form of these basis vectors is given in App.~\ref{app:basis}.

\subsubsection{Generalized Coulomb gauge}
\label{sec:GCG}

In Eq.~\eqref{eq:ansatz-A}, we chose a particular representative $(A^{(0)})^a_i=f(\tau)\delta^a_i$ for our homogeneous and isotropic background field.  This is just one representative from the corresponding gauge-equivalence class.  When considering physical fluctuations around this background configuration, we restrict ourselves to fluctuations which are orthogonal to the space spanned by gauge-equivalent configurations
\begin{equation}
 U f(\tau) \delta^a_i \,  U^\dagger + \frac{i}{e} U \partial_i U^\dagger\,, \quad \forall \; U  = \exp(i \xi^a \textbf{T}_a)
 \label{eq:GaugeSpace}
\end{equation}
where $\xi^a$ denotes a infinitesimal gauge transformation parameter.  This condition should apply on each time slice. This orthogonality condition reads
\begin{equation}
 0 = \langle \mathbf{D}_i^{(A^{(0)})} \xi^a \textbf{T}_a | \delta A_i^c \textbf{T}_c \rangle = \text{Tr} \int (\mathbf{D}_i^{(A^{(0)})} \xi^a \textbf{T}_a) \cdot (\delta A_i^c \textbf{T}_c ) \, \textrm{d}^3 \vec{x} \quad \forall \, \xi^a \,,    
 \label{eq:orthogonal}
\end{equation} 
where 
\begin{equation}
 \mathbf{D}_\mu^{(A)} \xi^a \textbf{T}_a = \partial_\mu \xi^a \textbf{T}_a - i e [A_\mu^b \textbf{T}_b, \xi^a \textbf{T}_a]
 \label{eq:infinitesimal_gauge}
\end{equation}
denotes the gauge-covariant derivative. After some algebra, Eq.~\eqref{eq:orthogonal} becomes
\begin{equation}
 - \frac{1}{2} \int \xi^a (\partial_i \delta A^a_i + e \varepsilon_{a b c} (A^{(0)})^b_i \, \delta A^c_i )\, \textrm{d}^3 \vec{x}  = 0  \quad \forall \, \xi^a  \quad \Rightarrow \quad 
 D_i^{(A^{(0)})} \delta A^a_i = 0 \,.
\end{equation}
Inserting Eq.~\eqref{eq:ansatz-A} for $A^{(0)}$ we obtain the gauge fixing condition 
\begin{equation}
 \mathbf{C}(\delta A)^{\, a} \equiv  \partial_i(\delta A^a_i) + e f(\tau) \varepsilon^{a i c}  \delta A^c_i = 0 \,.
 \label{eq:GenCoulombGauge}
\end{equation}
In the following, we will in fact not fix the gauge, but we will choose a basis in which the 6 physical degrees of freedom (obeying Eq.~\eqref{eq:GenCoulombGauge}) and the 3 gauge degrees of freedom (contained in the subspace~\eqref{eq:GaugeSpace}) are explicit and orthogonal. This preserves gauge invariance as a consistency check at any point of the calculation, while clearly separating physical and gauge degrees of freedom. Together with the constraint equation~\eqref{eq:GaussLaw}, this splits the 12 degrees of freedom contained in the $3 \times 4$ matrix $\delta A^a_\mu$ into 6 physical, 3 gauge and 3 non-dynamical degrees of freedom, as expected for a massless  $\mathrm{SU}(2)$ gauge theory.

\subsubsection{The helicity basis}
\label{sec:helbas}

In the absence of a background gauge field ($f(\tau) = 0$), Eq.~\eqref{eq:action} is invariant under two independent global $\mathrm{SO}(3)$ rotations: one acting on the spatial index and the other acting on the  $\mathrm{SU}(2)$ index of the gauge field $A^a_i(\tau, \vec x)$. In the presence of the background Eq.~\eqref{eq:ansatz-A}, this symmetry is reduced to a single $\mathrm{SO}(3)$ symmetry, which is the diagonal subgroup of $\mathrm{SO}(3)_\text{gauge} \times \mathrm{SO}(3)_\text{spatial}$ (see  \prettyref{app:global-symmetries} for details). The Fourier decomposition introduces a preferred direction $\vec k$, which without loss of generality we will choose to be along the $x$-axis, $\vec k = k \hat e_1$. This breaks the diagonal $\mathrm{SO}(3)$ symmetry down to an $\mathrm{SO}(2)$ symmetry of rotations around $\vec k$.  The generator of this symmetry is a helicity operator of massless particles. This generator is given explicitly by\footnote{ The helicity operator can be extended to act on the full $3\times4$ matrix $\delta A$ by defining $\mathbf{H}(\delta  A_0)^a  =  i \varepsilon^{1 c a } \delta A^c_0$.} 
\begin{align}
  \mathbf{H}(\overrightarrow{\mathbf{\delta  A}}) & \equiv [\mathbf{T}_1, \overrightarrow{\mathbf{\delta  A}}] + i \hat e_1 \times  \overrightarrow{\mathbf{\delta  A}} \,, \nonumber 
 \\
  \Rightarrow \mathbf{H}(\delta  A)^a_i & = i \varepsilon^{1 c a } \delta A^c_i + i \varepsilon^{i1j} \delta A^a_j  \,,
     \label{eq:helicity_operator}
\end{align}
where $\overrightarrow{\mathbf{\delta  A}} = \delta A^a_i \mathbf{T}_a$. 
Expressing the linearized system of equations of motion in terms of the linear operator $\mathbf L$, see Eq.~\eqref{eq:eom_i_linear}, 
 $\mathbf{L}(\delta A, \delta \phi, \gamma) = 0 $,
the symmetry properties above imply that this linear operator must commute with the helicity operator, $[\mathbf{L}, \mathbf{H} ] = 0$. It will thus be useful to decompose $\delta A$ into helicity eigenstates, which will lead to a block-diagonal structure for $\mathbf{L}$.  This formalism is best known in the context of metric perturbations under the name ``SVT decomposition.''

Let us look at the eigenvalues and multiplicities of these states. With respect to the diagonal $\mathrm{SO}(3)$ group, the $3 \times 3$ matrix {$\delta A$} decomposes as $\mathbf{3} \otimes \mathbf{3} = \mathbf{1} \oplus \mathbf{3} \oplus \mathbf{5}$: a scalar (S), a vector (V), and a tensor (T). The corresponding helicities are
\begin{equation}
 (S) : 0 \,, \qquad (V) : -1, 0, + 1 \,, \qquad (T) : -2, -1, 0, +1, + 2 \,,
\end{equation}
implying multiplicities 3, 2 and 1 for the helicities $0$, $\pm 1$ and $\pm 2$ respectively. These nine degrees of freedom correspond to the six physical and three gauge degrees of freedom mentioned in the previous subsection. Since the gauge transformation acts only on the gauge indices and not the spatial indices (see e.g.\ Eq.~\eqref{eq:infinitesimal_gauge}), the three gauge degrees of freedom form a vector (helicities $-1,0,1$). The helicities of the remaining six physical degrees of freedom must thus be $-2, -1 , 0 (\times 2), +1, +2$. Hence in this basis, the linear operator $\mathbf L$ (and hence our equations of motion for $\delta A^a_i$) decomposes into four decoupled equations (for $\pm 1$ and $\pm 2$) and two (generically) coupled equations for the two helicity 0 modes. In the following we describe the basis we use for the gauge, constraint and physical degrees of freedom. The corresponding explicit basis vectors can be found in App.~\ref{app:basis}.  In Sec.~\ref{sec:AllFluc} we will include also the inflaton (helicity 0) and tensor metric (helicity $\pm2$) fluctuations, which will couple to the helicity 0 and $\pm2$ gauge field modes, respectively.

Let us first consider the pure gauge degrees of freedom, which can be decomposed in terms of  basis vectors $\hat g$  (see Appendix.~\ref{app:basis}) as 
\begin{equation}
 (\delta \tilde{A}^a_\mu)_\text{gauge}(\tau, \vec k) = \sum_b (\hat g_b)^a_\mu w_b^{(g)}(\tau, \vec k) 
 =  \sum_{\lambda}( \hat g_\lambda)^a_\mu \, w^{(g)}_\lambda(\tau, \vec k)\,,
\end{equation}
where $\hat g_b$ $(b = \{1,2,3\})$ denotes the basis vectors of the gauge degrees of freedom in $\mathrm{SU}(2)$ space, $\hat g_\lambda$ with $\lambda = \{-,0,+\}$ denotes the basis vectors in terms of helicity states and we denote the corresponding coefficients by $w_{b}^{(g)}$ and $w_{\lambda}^{(g)}$,
respectively. Introducing a helicity basis for the elements of the Lie algebra, $w_b^{(g)} \mathbf{T}_b = w_\lambda^{(g)} \mathbf{T}_\lambda$, with
\begin{equation}
 \mathbf{T}_\pm =  (\mathbf{T}_2 \pm i \mathbf{T}_3)/\sqrt{2} \,, \qquad  \mathbf{T}_0 = \mathbf{T}_1 \,,
\end{equation}
the infinitesimal gauge transformation~\eqref{eq:infinitesimal_gauge}  defines the basis vectors $\hat g^{(g)}_\lambda$,
\begin{equation}
 \mathbf{D}_{\mu}^{(A^{(0)})} \left( w_\lambda^{(g)} \mathbf{T}_\lambda \right) = (\hat g_\lambda)^a_\mu w_\lambda^{(g)} \mathbf{T}_a \,.
 \label{eq:gaugebasis0}
\end{equation}
The explicit form of the three basis vectors $\hat g_\lambda$ which satisfy Eqs.~\eqref{eq:gaugebasis0} and are eigenstates of \eqref{eq:helicity_operator} are given in App.~\ref{app:basis}. We note that in  any background which is a fixed point of Eq.~\eqref{eq:symmetry3} (e.g.\ if the background follows the $c_2$-solution), the $k$ and $\tau$ dependence of the basis vectors is fully encoded in $x = - k \tau$ only.

So far, we have considered only the spatial components of $\delta A$. The time components $\delta A_0$ are subject to the constraint equations~\eqref{eq:GaussLaw}. We can solve these explicitly and substitute the solution back into the equation of motion for the spatial components. 
However in practice we will find it more convenient to introduce basis vectors also for these constraint degrees of freedom, extending the differential operator $\mathbf L$ to a differential-algebraic operator. The explicit form of the  corresponding `constraint' basis vectors in the helicity basis is given in App.~\ref{app:basis}.

The remaining eigenspace of the helicity operator~\eqref{eq:helicity_operator} is spanned by the basis vectors of the physical degrees of freedom $\hat e_\lambda$, see App.~\ref{app:basis} for the explicit form.
As anticipated, we find two states with helicity 0, and one state each with helicity $-2, -1, +1, +2$. One can immediately verify explicitly that the basis vectors presented here have the desired qualities, i.e. they are  orthonormal, eigenfunctions of $\mathbf{H}$ with eigenvalues giving the helicity, $\mathbf L(\hat g_\lambda) = 0$ (gauge invariance) and  $\mathbf{C}(\hat e_\lambda) = 0$ (compatibility with generalized Coulomb gauge, see \eqref{eq:GenCoulombGauge}). The choice of basis derived here closely resembles the basis used in Refs.~\cite{Dimastrogiovanni:2012ew,Adshead:2013nka,Adshead:2013qp,Namba:2013kia}. The main difference is that we explicitly separate the $\pm \lambda$ states and normalize our basis vectors. As we will see in the next section, this simplifies the resulting equations of motion (in particular when considering only the degrees of freedom of the gauge sector).

\subsection{Equations of motion for the gauge field fluctuations}
\label{sec:eom_gauge_fields}

In this section we will compute the equations of motion for the gauge field fluctuations in the canonically normalized helicity basis introduced above (see also App.~\eqref{app:basis}) and discuss their key properties. In Sec.~\ref{sec:AllFluc} we will extend this to the inflaton and metric tensor fluctuations.

Inserting $\delta A$ in terms of the helicity basis,
\begin{equation}\label{eq:w-lambda-def}
 \delta \tilde{A}^a_\mu(\tau,k) =   \sum_{\lambda}( \hat e_\lambda)^a_\mu \, \frac{ w^{(e)}_\lambda(x) }{\sqrt{2k}} +  \sum_{\lambda}( \hat f_\lambda)^a_\mu \, \frac{ w^{(f)}_\lambda(x) }{\sqrt{2k}} + \sum_{\lambda}( \hat g_\lambda)^a_\mu \, \frac{ w^{(g)}_\lambda(x) }{\sqrt{2k}} \,,
\end{equation}
into the first order equations of motion (see Sec.~\ref{sec:setup} and App.~\ref{app:fulleom}) we  obtain the equations of motions for the coefficients $w^{(i)}_\lambda$ with $i = \{e, f, g \}$ denoting the physical, constraint and gauge degrees of freedom, respectively. Here we have absorbed a factor of $\sqrt{2k}$  (originating from the normalization of the Bunch--Davies vacuum, c.f.\ Eq.~\eqref{eq:BDvac}) into $w_\lambda^{i}$. As we will see below,  for the background solutions of interest, this will render $w^{(i)}_\lambda$ a function of $x = - k \tau$ only. The three equations for the gauge degrees of freedom simply read $0 = 0$, reflecting gauge invariance. For the helicity $\pm2$ modes we obtain
\begin{align*}
\frac{\textrm{d}^2 }{\textrm{d}x^2} w^{(e)}_{-2}(x)+\left(1+\frac{2\xi}{x}+2\left(\frac{\xi}{x}+1\right)y_k(x)\right)w^{(e)}_{-2}(x) & =0\,, \label{eq:-2mode}\\
\frac{\textrm{d}^2}{\textrm{d}x^2}w^{(e)}_{+2}(x) +\left(1-\frac{2\xi}{x}+2\left(\frac{\xi}{x}-1\right)y_k(x)\right)w^{(e)}_{+2}(x) & =0\,,
\label{eq:+2mode}
\end{align*}
with 
\begin{equation}
\label{eq:y_definition}
 y_k(x) \equiv \frac{e f(\tau)}{k} \,.
\end{equation}

We can now appreciate some of the advantages of the canonically normalized helicity basis. The equations of motion for the $\pm 2$ modes are fully decoupled, and moreover contain no terms involving the first derivatives $w_\lambda'(x)$. This makes them amenable to WKB analysis. We immediately see that for $\xi \geq 0$ and $y(x) \geq 0$ the $-2$ mode always has a positive effective squared mass, whereas the $+2$ mode can be tachyonic. Consider momentarily the limit where $\xi$ is constant and $f(\tau)$ is one of the three fixed points of the symmetry~\eqref{eq:symmetry3},  $y_k(x) = y(x) = c_i\xi/x$ for some $i\in\{0,1,2\}$, where $c_i$ is defined in \eqref{eq:c-definitions}. In this case, the solutions of Eq.~\eqref{eq:+2mode} are Whittaker functions:  
 \[ w_{+2}^{(e)}(\tau)= e^{(1+c_{i})\pi\xi/2} W_{-i(1+c_{i})\xi,\,-i\sqrt{2\xi^{2}c_{i}-1/4}}\left(2ik\tau\right) \,,  \label{eq:AnalyticalSolution}  \]  
with the normalization set by the Bunch--Davies vacuum \eqref{eq:BDvac} in the infinite past. For  $c_i=c_0=0$, this solution coincides with the abelian solution, Eq.~\eqref{eq:rev_Whittaker}. The region of tachyonic instability for the helicity $+2$ mode as well as some useful approximative expressions for Eq.~\eqref{eq:AnalyticalSolution} will be discussed below.

Next we turn to the $\pm 1$ modes. Here we need to consider the two equations for the dynamical degrees of freedom and two constraint equations. For shorter notation, we introduce two reparameterizations of $y_k(x)$, 
\begin{equation}
y_k(x)=\tfrac{1}{2}\left(\tan\theta_{-}(x)-1\right)=\tfrac{1}{2}\left(\tan\theta_{+}(x)+1\right).
\end{equation}
with $\theta_\pm \in (- \pi/2, \pi/2)$. With this, the equations for the dynamical and constraint degrees of freedom read
\begin{align}
0 = & \, \frac{\textrm{d} w^{(e)}_{\pm1}(x)}{\textrm{d}x^2} \pm \sqrt{2}  i \sec \theta_{\pm}\left( \tfrac{\textrm{d}}{\textrm{d}x}\theta_{\pm}\right) w^{(f)}_{\pm1}(x) \,  + \nonumber \\
& +\left(-\left(\tfrac{\textrm{d}}{\textrm{d}x}\theta_{\pm}\right)^{2}+\tfrac{1}{2}\left(1\pm\sin(2\theta_{\pm})\right)+\left(\frac{1}{2}\mp\frac{\xi}{x}\right)\left(2\cos^{2}\theta_{\pm}\mp\tan\theta_{\pm}\right)\right) w^{(e)}_{\pm1}(x) \,, \nonumber \\
0 = & \, \sec \theta_{\pm} w^{(f)}_{\pm1}(x) \pm 2 \sqrt{2}  i w^{(e)}_{\pm1}(x) \tfrac{\textrm{d}}{\textrm{d}x}\theta_{\pm}  \,.
\end{align}
After inserting the constraint equations, this simplifies to
\begin{align}
 \frac{\textrm{d} w^{(e)}_{-1}(x)}{\textrm{d}x^2}+\left(3\left(\tfrac{\textrm{d}}{\textrm{d}x}\theta_{-}\right)^{2}+\tfrac{1}{2}\left(1-\sin(2\theta_{-})\right)+\left(\frac{1}{2}+\frac{\xi}{x}\right)\left(2\cos^{2}\theta_{-}+\tan\theta_{-}\right)\right)w^{(e)}_{-1}(x) & =0 \nonumber \\
  \frac{\textrm{d} w^{(e)}_{+1}(x)}{\textrm{d}x^2}+\left(3\left(\tfrac{\textrm{d}}{\textrm{d}x}\theta_{+}\right)^{2}+\tfrac{1}{2}\left(1+\sin(2\theta_{+})\right)+\left(\frac{1}{2}-\frac{\xi}{x}\right)\left(2\cos^{2}\theta_{+}-\tan\theta_{+}\right)\right) w^{(e)}_{+1}(x) & =0.
  \label{eq:pm1mode}
\end{align}
For the $c_2$ background attractor solution given in Eq.~\eqref{eq:c2solution}, the resulting effective masses are always positive. We will turn to a more detailed stability analysis in the next subsection. 

Finally let us consider the two helicity zero modes. Since these expressions are somewhat more lengthy, we only give the final expression after substituting the constraint equation
\begin{equation}
 -  \sec^2 \theta_0  w^{(f)}_{0}(x)- \sqrt{2}  \sec \theta_0  \tfrac{\textrm{d}}{\textrm{d}x}\theta_{0} w^{(e)}_{02}(x)  = 0\label{eq:0constraint} \,,
\end{equation}
where we have introduced $\theta_0 \in [- \pi/2, \pi/2]$ as
\begin{equation}
 y_k(x)=\tfrac{1}{\sqrt{2}}\tan\theta_{0}(x) \,.
\end{equation}
With this,
\begin{equation}
 \frac{\textrm{d} }{\textrm{d}x^2}\left(\begin{array}{c}
  w^{(e)}_{01}(x)\\
 w^{(e)}_{02}(x)
\end{array}\right)+M_{0}(x)\left(\begin{array}{c}
w^{(e)}_{01}(x)\\
w^{(e)}_{02}(x)
\end{array}\right)  = 0
\label{eq:0modes}
\end{equation}
with the $2\times2$ Hermitian mass matrix $M_{0}$ for the two $e_{0i}$ modes given by
\[
M_{0}=\left(\begin{array}{cc}
1-\sqrt{2}\frac{\xi}{x}\tan\theta_{0}+2\tan^{2}\theta_{0} & -\frac{2i}{\cos\theta_0}\,\left(\frac{\xi}{x}-\tfrac{1}{\sqrt{2}}\tan\theta_{0}\right)\\
\frac{2i}{\cos\theta_0}\left(\frac{\xi}{x}-\tfrac{1}{\sqrt{2}}\tan\theta_{0}\right) & \sin^{2}\theta_{0}+\cos^{-2}\theta_{0}-\frac{\xi}{\sqrt{2}x}\sin2\theta_{0}+3\left(\tfrac{\textrm{d}}{\textrm{d}x}\theta_{0}\right)^{2}
\end{array}\right)\,.
\label{eq:M0}
\]
{For the background solution of Eq.~\eqref{eq:c2solution} and for $\xi \gg 1$, the off-diagonal elements vanish. Furthermore, on far sub-horizon scales ($x \gg 1$) and far super-horizon scales  $x \ll 1$, the diagonals elements approach unity and $2 \xi^2/x^2$, respectively.}
One may be tempted to diagonalize the general expression of $M_{0}$, but the diagonalization
would be time-dependent and hence re-introduces first-derivatives of $w^{(e)}_{0i}$. {We note that the helicity 0 sector is particularly sensitive to non-linear contributions neglected in our analysis so far, arising from two enhanced helicity $2$ modes coupling to the helicity 0 modes, see also Eq.~\eqref{eq:caveat}. We will discuss this effect in more detail in Sec.~\ref{subsec:powerspectra}.}

In summary and as anticipated, the modes with helicity $\pm 1$ and $\pm2$ form four decoupled harmonic oscillators with the time-dependent mass terms specified in Eqs.~\eqref{eq:-2mode}, \eqref{eq:+2mode} and \eqref{eq:pm1mode}. The two helicity zero modes form  a system of coupled, mass-dependent harmonic oscillators given by Eq.~\eqref{eq:0modes}.

 \subsubsection*{Stability analysis}
 Let us look at these fluctuations in two different background limits (taking $\xi$ to be constant): $f(\tau) \rightarrow 0$ and $e f(\tau) = c_2 \,  \xi/(- \tau)$ (see Eq.~\eqref{eq:c-definitions}). In the former case, defining $e_{\pm0}(x):=(e_{01}(x)\mp ie_{02}(x))/\sqrt{2}$, the effective squared mass\footnote{Since we are considering ODEs as a function of $x$, $\tfrac{\textrm{d}^2}{\textrm{d}x^2} w(x) + m^2 w(x) = 0$, the `squared mass' is dimensionless quantity.} of $e_{\pm0}$, $e_{\pm1}$, and $e_{\pm2}$ is $1\mp2\xi/x$, so that as in the abelian case, the `$-$' modes are
unenhanced, while the `$+$' modes are enhanced for $x<2\xi$. In this case, the spatial components of the helicity basis simplify to
\begin{equation}
e_{\pm0}=\frac{1}{2}\left(\begin{array}{ccc}
0 & 0 & 0\\
0 & 1 & \pm i\\
0 & \mp i & 1
\end{array}\right),\quad e_{\pm1}=\frac{1}{\sqrt{2}}\left(\begin{array}{ccc}
0 & 1 & \pm i\\
0 & 0 & 0\\
0 & 0 & 0
\end{array}\right),\quad e_{\pm2}=\frac{1}{2}\left(\begin{array}{ccc}
0 & 0 & 0\\
0 & 1 & \pm i\\
0 & \pm i & -1
\end{array}\right).
\end{equation}

 \begin{figure}
 \centering
  \includegraphics{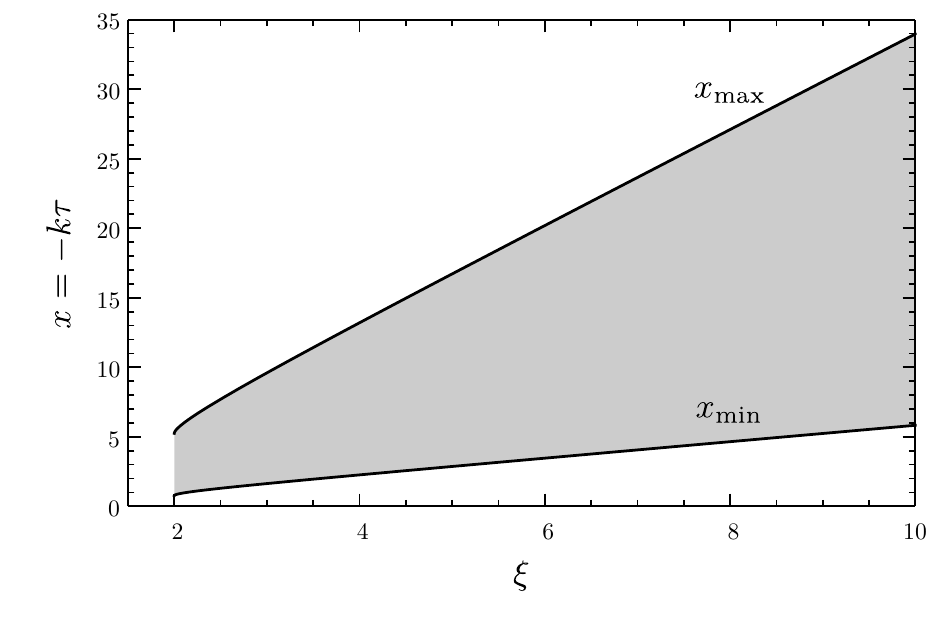}
  \caption{Tachyonic region of the $+2$ mode in a non-abelian, $c_2$-background (shaded in gray). Contrary to the abelian regime, the instability region is bounded from both sides and only affects a single mode. }
  \label{fig:stabilitytensor}
 \end{figure}

On the other hand, for $e f(\tau) = c_2 \,  \xi/(- \tau)$, the squared mass terms appearing in Eqs.~\eqref{eq:-2mode} and Eq.~\eqref{eq:pm1mode} for $w_{-2}^{(e)}(x)$ and $w_{\pm1}^{(e)}(x)$, respectively, are positive for all $x, \xi > 0$. Similarly, the matrix $M_0$ in Eq.~\eqref{eq:0modes} is positive definite if and only if $\xi > 3/\sqrt{2} \simeq 2.12$, as can be immediately checked from the sign of the trace and the determinant. The instability in the scalar sector for $\xi < 3/\sqrt{2}$ corresponds precisely to the catastrophic instability observed in \cite{Dimastrogiovanni:2012ew} for $m_g > 2 H$, where in our notation $m_g = \sqrt{2} e f(\tau)/a \mapsto \sqrt{2} c_2 \xi H$. Note however that a non-abelian background can only form for $\xi > 2$, and it is likely to form only for $\xi \gg 2$ (see Sec.~\ref{subsec:Isotropic} and in particular Fig.~\ref{fig:u0-of-xi}). Moreover, as we will see in Sec.~\ref{subsec:gaugefluctuations} (see in particular Fig.~\ref{fig:matching}) the transition from the abelian regime to the non-abelian regime occurs at $\xi \gtrsim 3$ for perturbative gauge couplings $e < 0.1$. As a consequence, despite the presence of a potentially dangerous instability  in the scalar sector, the corresponding region of the parameter space is naturally avoided by the mechanism described in this work.

 The only mode which can experience a tachyonic instability in a $c_2$-background is the $e_{+2}$ mode. The mass term for this mode is then given by
\begin{align}
 m_{+2}^2 = 1 - \frac{2 \xi}{x} + \frac{\xi (\xi - x)}{x^2} \left(1 + \sqrt{1 - 4/\xi^2} \right) \rightarrow 1 - \frac{4 \xi}{x} + \frac{2 \xi^2}{x^2} \,,
 \label{eq:mass-term-2p}
\end{align}
where in the last step we have assumed $\xi \gg 2$.
The region in which this mass term becomes tachyonic is shown as gray shaded region in Fig.~\ref{fig:stabilitytensor} and is given by
\begin{equation}
 x_\text{min} \equiv \left[ 1 + c_2 - \sqrt{1 + c_2^2}\right] \xi < x <  \left[1 + c_2 + \sqrt{1 + c_2^2}\right] \xi  \equiv x_\text{max}\,,
 \label{eq:def-xmin-xmax}
\end{equation}
which for $\xi \gg 2$ yields\footnote{As a word of caution, we note that in particular for small $\xi$, the lower part of this range can come close to the inhomogeneity scale of the initial conditions determined by the abelian regime, cf.~Fig.~\ref{fig:propertiesabelian}. In this case, inhomogeneities in the initial conditions may affect the instability band depicted in Fig.~\ref{fig:stabilitytensor}.}
\begin{equation}
 \xi (2 - \sqrt{2}) < x < \xi (2 + \sqrt{2}) \,.
\end{equation}

In Fig.~\ref{fig:TensorModeEvolution} we show the evolution of the helicity $+2$ mode in both regimes. The initial conditions are set by imposing the Bunch--Davies vacuum on far sub-horizon scales,
\begin{equation}
 w_{+2}^{(e)}(x) =   e^{i x} \qquad \text{for} \; x \gg 1 \,. \label{eq:BDsimple}
\end{equation}
Note that these solutions are only functions of $x = - k \tau$ and $\xi$. They are in particular independent of the value of the gauge coupling $e$ and the absolute time $\tau$ (although of course the slowly varying value of $\xi$ will introduce an implicit dependence on $\tau$).

\begin{figure}
\subfigure{
\includegraphics[width = 0.48 \textwidth]{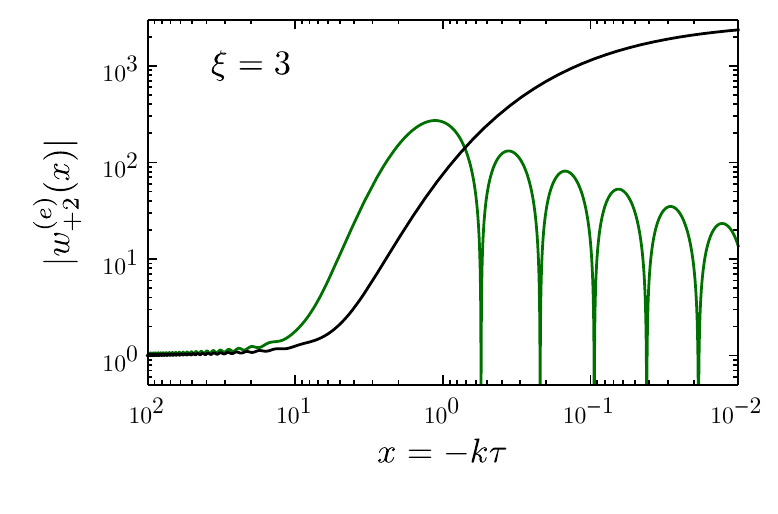}
}
\hfill
\subfigure{
\includegraphics[width = 0.48  \textwidth]{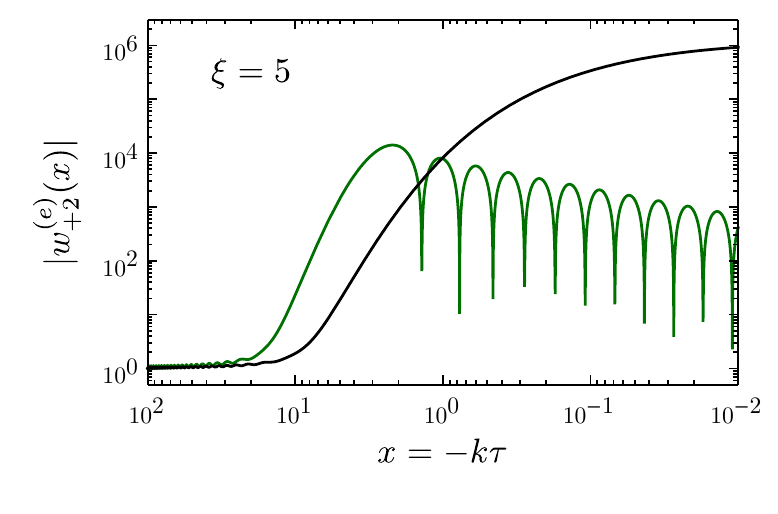}
}
\caption{Evolution of the helicity $+2$ mode for $\xi = 3$ (left panel) and $\xi = 5$ (right panel). The black curves correspond to the abelian regime ($f(\tau) \simeq 0$), the green curves to the non-abelian regime ($e \, f(\tau) = c_2 \xi /(- \tau)$. }
  \label{fig:TensorModeEvolution}
\end{figure}

{A key observation here is that in the presence of a vanishing or $c_0$-type background solution, the helicity $+2$ mode of the linearized non-abelian theory behaves very much like the enhanced helicity mode of the abelian theory, see Fig.~\ref{fig:Whittakerabelian}. With this in mind, we will refer to the time before the $c_2$-solution develops as the `abelian regime', in contrast to the `non-abelian regime' characterized by the inherently non-abelian effects induced by the  $c_2$-background solution.}

In summary, in the abelian regime ($f(\tau) = 0$), 3 modes become enhanced as soon as $x < 2 \xi$. In the non-abelian regime ($e f(\tau) = c_2 \, \xi/(- \tau)$), only a single mode is enhanced. The enhancement occurs earlier (as soon as $x \lesssim \xi(2 + \sqrt{2})$) compared to the abelian regime  but contrary to the abelian regime only lasts for some finite period of time (for $\Delta x \simeq 2  \sqrt{2} \xi $). As we will see below, these differences lead to a significant changes between the properties of gauge field fluctuations arising in the abelian and non-abelian regime. 
In particular, due to the helicity decomposition, the single enhanced mode of the non-abelian regime can only source (at the linear level) tensor perturbations (i.e.\ gravitational waves) but not scalar perturbations (i.e.\ no curvature perturbations). 

\subsubsection*{Approximate solutions for the enhanced helicity +2 mode}

The tachyonically enhanced modes in the abelian regime have been discussed in much detail in the literature (see Sec.~\ref{sec:abelian}). Here we focus on the enhanced mode in the inherently non-abelian regime, i.e.\ the helicity $+2$ mode in a $c_2$ gauge-field background.  In the limit of constant $\xi$, the exact solution to Eq.~\eqref{eq:+2mode} is given by Eq.~\eqref{eq:AnalyticalSolution},\footnote{To leading order in $1/\xi$, this expression agrees with the one given in~\cite{Adshead:2013nka}. The discrepancy at higher orders is due to the different background solution chosen (see also discussion in Sec.~\ref{subsec:dynamical_background}).}
\[
w_{+2}^{(e)}(\tau)= e^{\kappa \pi/2} W_{-i\kappa,\,-i\mu}\left(2ik\tau\right) \,,
\label{eq:whittakeragain}
\]
with $\kappa = (1 + c_2) \xi \simeq 2 \xi$ and $\mu = \xi \sqrt{2 c_2 - (2 \xi)^{-2}} \simeq \sqrt{2} \xi$. For the $c_2$ background solution, we derive useful asymptotic expressions in App.~\ref{app:asymptotics}, approximating the enhanced component of Eq.~\eqref{eq:AnalyticalSolution} on super-horizon scales and around the epoch of maximal enhancement, respectively:
\begin{align}
\, w_{+2} & \simeq 2 e^{(\kappa - \mu) \pi} \sqrt{\frac{ x}{ \mu}} \cos\left[\mu \ln(2 x) + \theta_0 \right]  \qquad  &&\text{for   } x \ll  x_\text{min} \,,
 \label{eq:NonabelianAsymptotics}\\
 w_{+2}(x)& \simeq\sqrt{4\pi}\ e^{(\kappa-\mu)\pi}\left(\frac{\zeta(x)}{V(x)}\right)^{1/4}\mathrm{Ai}\left(\zeta(x)\right) \qquad &&\text{for    } x \simeq x_\text{min} \,,
  \label{eq:NonabelianAsymptotics2}
\end{align}
with Ai$(x)$ denoting the Airy Ai function and
\begin{align}
  V(x)  & =-\left(1-\frac{2\kappa}{x}+\frac{\mu^{2}}{x^{2}}\right) \,, \qquad
\zeta(x)\approx\left(2\mu^{2}-2\kappa x_\text{min}\right)^{1/3}\ln\left(\frac{x}{x_\text{min}}\right) \,.
\label{eq:Airyapp}
 \end{align}
These expressions will prove useful to obtain analytical estimates. For details see App.~\ref{app:asymptotics}.

 \subsection{Including the inflaton and gravitational wave fluctuations \label{sec:AllFluc}}
 
With this understanding of the growth of the gauge field fluctuations, let us now include the scalar and metric tensor fluctuations. The former will couple to the helicity 0 gauge field modes, the latter to the helicity $\pm2$ modes. 

Let us start with the helicity 0 modes. After inserting the constraint equation which now reads
\begin{align}
 -    w^{(f)}_{0}(x)- \sqrt{2}  \cos \theta_0  \tfrac{\textrm{d}}{\textrm{d}x}\theta_{0} w^{(e)}_{02}(x) & = -  \frac{i k \alpha}{\sqrt{2} e \Lambda} \sin \theta^2_0 \delta \phi \,, \label{eq:0constraintphi}
\end{align}
the equations for the dynamical degrees of freedom read
\begin{align}
  \frac{\textrm{d} }{\textrm{d}x^2}\left(\begin{array}{c}
  w^{(e)}_{01}(x)\\
 w^{(e)}_{02}(x) \\
 a \delta \phi(x) 
\end{array}\right)
+ N^k_0(x)  \frac{\textrm{d} }{\textrm{d}x}\left(\begin{array}{c}
  w^{(e)}_{01}(x)\\
 w^{(e)}_{02}(x) \\
 a \delta \phi(x) 
\end{array}\right) + 
 \tilde M^k_{0}(x)\left(\begin{array}{c}
w^{(e)}_{01}(x)\\
w^{(e)}_{02}(x) \\
 a \delta \phi(x)
\end{array}\right)  = 0\,, \label{eq:fullscalar}
\end{align}
with 
\begin{align}
 N^k_0 & = \frac{\gamma x}{\sqrt{2}} \begin{pmatrix}
     0 & 0 &  \tan^2 \theta_0 \\
     0 & 0 & - \frac{i}{\sqrt{2}} \sin \theta_0 \tan^2 \theta_0 \\
     - \tan^2 \theta_0 & - \frac{i}{\sqrt{2}} \sin \theta_0 \tan^2 \theta_0  & 0
                                            \end{pmatrix},
\end{align}
\begin{align}
 \tilde M^k_0 & = \begin{pmatrix}
        \left(M_0\right)_{11} &  \left(M_0\right)_{12}  & \frac{\gamma}{\sqrt{2}} \tan^2 \theta_0 \\
      \left(M_0\right)_{21}  &
        \left(M_0\right)_{22} 	& -\frac{i \gamma}{2} \left(\sin \theta_0 \tan^2 \theta_0 - 2 x \cos \theta_0 \theta_0'\right) \\
        - \sqrt{2} \gamma x  \frac{\tan \theta_0}{\cos^2 \theta_0} \theta_0' & -\frac{i \gamma x}{16} \frac{\left(15 + \cos\left(4 \theta_0\right) \right)}{\cos^3 \theta_0} \theta_0' & m_{\phi \phi}^2
                     \end{pmatrix},  \label{eq:M0full}         
\end{align}
where $M_0(x)$ is given in Eq.~\eqref{eq:M0}, $\theta_0' = \textrm{d} \theta_0 / \textrm{d}x $, $\gamma = \alpha H  / (e \Lambda) $ and
\begin{equation}
 m_{\phi \phi}^2 = 1 - \frac{2}{x^2} + \frac{\alpha^2 H^2 x^2}{e^2 \Lambda^2} \frac{ \sin^6 \theta_0}{\sin^2\left(2 \theta_0\right)} + \frac{V_{, \phi \phi}}{H^2 x^2	}\,.
 \label{eq:mphi}
\end{equation}
As long as the gauge coupling is not very small, $e \gg \xi H  \alpha/\Lambda$, the coupling between the gauge field modes and the inflaton mode is suppressed around horizon crossing, and the two helicity 0 gauge field modes are to a good approximation described by the unperturbed system~\eqref{eq:0modes}. 
Recalling that $\tan \theta_0/\sqrt{2} = y(x) = e f(\tau)/k$, we note that all off-diagonal terms, including the entire matrix $N_0^k$, vanish in the absence of a background gauge field, $f(\tau) = 0$. In the case of a gauge field background following the $c_2$-solution, $\tan \theta_0/\sqrt{2} = c_2 \, \xi / x$, we note that all off-diagonal terms, including the matrix $N_0^k$, vanish for $x \rightarrow \infty$, i.e.\ in the infinite past, and the matrix $\tilde M_0^k$ reduces to the unit matrix, allowing us to impose Bunch--Davies initial conditions in the infinite past.

In the opposite regime, on far super-horizon scales, the second term in Eq.~\eqref{eq:mphi} is responsible for the freezing out of the $\delta \phi$ fluctuations. In the limit $\alpha \rightarrow 0$, $x \ll1$ and $V_{,\phi\phi} \rightarrow 0$, the equations of motion for $w_0^{(\phi)} = a \, \delta \phi$ simply reads
\begin{equation}
 \left(w_0^{(\phi)}\right)''(x)  - \frac{2}{x^2} \,  w_0^{(\phi)}(x) = 0 \,,
 \label{eq:freezeout}
\end{equation}
with the solution $x \, w_0^{(\phi)}(x) = A x^3 + B $ with the integration constants $A$ and $B$. For $x \rightarrow 0$ this leads to a decaying solution ($A = 1, B = 0$) and a constant solution ($A = 0, B = 1$). This is the usual freeze-out mechanism for scalar (and tensor) fluctuations. Note that the sign in Eq.~\eqref{eq:freezeout} is crucial to obtain a constant solution. {The last two terms in Eq.~\eqref{eq:mphi} could in principle interfere with this freeze-out mechanism, however the last term in ensured to be sub-dominant in slow-roll inflation and the second-last term only becomes large together with all the off-diagonal terms, in which case the full coupled system must be analyzed.} {We point out that the freeze-out of the inflaton perturbation, $(a \delta \phi) \propto 1/x$ also entails its decoupling from the helicity 0 gauge field modes on super-horizon scales. }

In the left panel of Fig.~\ref{fig:fluctuations} we show the evolution of these helicity 0 modes for a parameter example of the benchmark scenario of the next section. Here $w_0^{(\phi)}$ denotes the coefficient of the comoving scalar mode $(a \, \delta \phi)$. {We clearly see the freeze-out of the inflation fluctuations after horizon crossing. The oscillations visible on sub-horizon are induced by the time-dependence of the eigenstates of the system. We have verified that the sum of the absolute value squared of all three states is $x$-independent as expected in this regime.\footnote{{The time-dependence of these (interacting)  eigenstates induces some ambiguity when imposing the Bunch--Davies initial conditions at any finite value of $x$. We have verified that our final results are not affected by this.}}
}

\begin{figure}
\subfigure{
\includegraphics{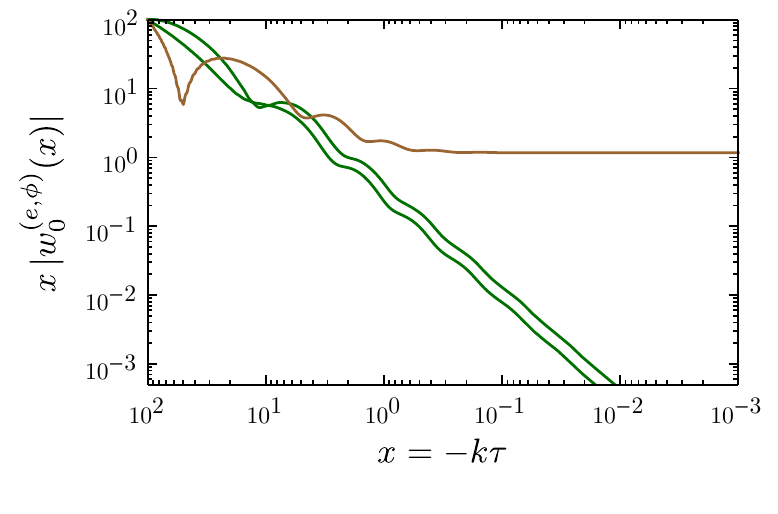}
}
\hfill
\subfigure{
\includegraphics{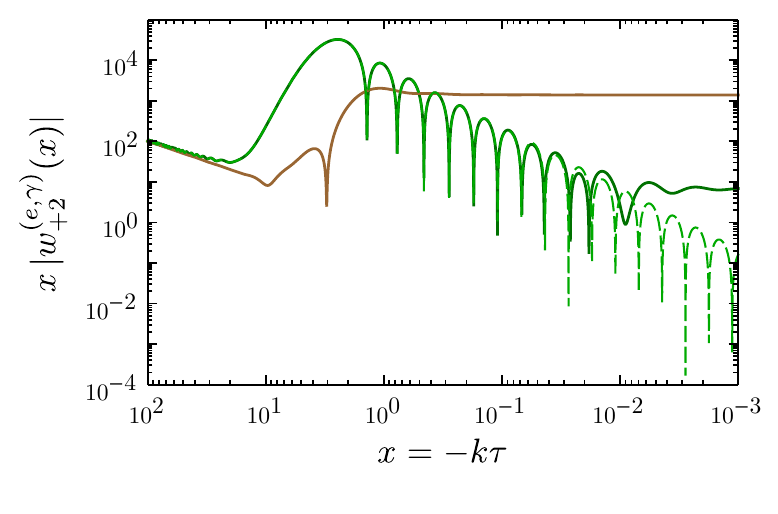}
}
\caption{Evolution of the (physical) scalar  and tensor fluctuations in a $c_2$-type background solution. \textbf{Left panel}: Helicity $0$ modes of the gauge field (dark green) and of the inflaton (brown). \textbf{Right panel}: Helicity $+2$ modes of the gauge field (dark green) and the metric tensor (brown). For reference, the dashed green curve shows the gauge field mode in the absence of a coupling to the metric tensor mode. Here we have set $\xi = 5$, $e = 5 \times 10^{-3}$, $H = 10^{-5} M_P$, $\alpha/\Lambda = 30$.  {Moreover, working in slow-roll approximation, we have set $V_{,\phi \phi} = 0$.}}
  \label{fig:fluctuations}
\end{figure}

{Our study so far is based on the linearized system of equations given in Sec.~\ref{sec:setup}, which forbids a coupling between the enhanced helicity $+2$ mode and the helicity $0$ modes. To higher orders in $\delta A$, this is no longer true since two tensor modes can combine into a scalar mode. Schematically, e.g.\ a  term bilinear in $A$ in the action can be expressed as
\begin{equation}
 \text{linear: }  f \cdot w_0^{(e)} + h.c. \,, \qquad \text{quadratic: }  |w_{+2}^{(e)}|^2  
 \label{eq:caveat}
\end{equation}
at the linearized level and to next order, respectively. With $\delta A \sim |w_{+2}^{(e)}| \gg |w_{0}^{(e)}|$, the condition $\delta A \ll f(\tau)$ is not sufficient to ensure that the linear term is the dominant one.  
In fact, this observation is well known in the case of abelian axion inflation, where the backreaction of the enhanced gauge fields mode occurs precisely true the $(\delta A)^2 \rightarrow \delta \phi$ process. Generalizing the procedure of Refs.~\cite{Linde:2012bt} (see also \cite{Barnaby:2011qe}) to the non-abelian case, we will estimate the contribution to the scalar power spectrum arising from the non-linear contributions in Sec.~\ref{subsec:powerspectra}.\footnote{{In the abelian limit, these non-linear contributions are also responsible for a friction-type backreaction of the produced gauge fields on the background equation for the inflaton, see Eq.~\eqref{eq:rev_motion}. On the contrary, in the non-abelian regime (and in particular for the parameter example studied in the next section), the corresponding contribution is subdominant to the gauge field background contribution, given by the last term in Eq.~\eqref{eq:phibackground}, as long as $\delta A \ll f(\tau)$.}}
}

Next we turn to the helicity $\pm2$ modes, i.e.\ the gravitational waves $\gamma_{ij}$ coupled to the $e_{\pm 2}$ gauge field modes. We express the metric tensor perturbations in the helicity basis as 
\begin{equation}
 a \, \gamma_{ij} = \frac{1}{2} \sum_{\lambda = \pm 2} \frac{w^{(\gamma)}_{\pm 2}(x) }{\sqrt{2k}}\begin{pmatrix}
                                                                    0 & 0 & 0\\
                                                                    0 & \mp i & 1 \\
                                                                    0 & 1 & \pm i
                                                                   \end{pmatrix} \,.
\end{equation}
The  equations of motion are then given by
\begin{align}
 \frac{\textrm{d} }{\textrm{d}x^2}  \left(\begin{array}{c}
   w^{(e)}_{\pm 2}(x)\\
    w^{(\gamma)}_{\pm 2}(x)  
\end{array}\right)
+ N_{\pm 2}(x)  \frac{\textrm{d} }{\textrm{d}x}\left(\begin{array}{c}
  w^{(e)}_{\pm 2}(x)\\
  w^{(\gamma)}_{\pm 2}(x) 
\end{array}\right) + 
  M_{\pm 2}(x)\left(\begin{array}{c}
w^{(e)}_{\pm 2}(x)\\
 w^{(\gamma)}_{\pm 2}(x) 
\end{array}\right)  = 0\,, \label{eq:fulltensor}
\end{align}
with
\begin{align}
 N_{\pm 2}(x) &  =  \frac{y'(x) H x}{e} \begin{pmatrix}
                    0 & - 1 \\
                   4  & 0
                   \end{pmatrix} \,, \\
M_{\pm 2}(x) & = \begin{pmatrix} 
              1  \mp \frac{2\xi}{x} + 2 (\frac{\xi}{x} \mp 1) y (x)  & \frac{H}{e} \left (-y' (x) - (2\xi \mp x) y (x)^2 + 
      x y (x)^3 \right) \\
                   -\frac{4 H x y (x)^2}{e}  (y (x)\mp 1) & 
                 1 - \frac{2}{x^2} +  \frac{2 H^2 x^2}{e^2} \left( y (x)^4 -  y' (x)^2\right)  
                 \end{pmatrix} \,.
\end{align}
where $y' \equiv \tfrac{\textrm{d} }{\textrm{d}x} y(x) $.
In the large-$\xi$ limit of the $c_2$-solution, $y(x) = \xi/x$, this becomes\footnote{In the following expressions we set $M_P = 1$.}
\begin{align}
 N_{\pm 2}(x) &  =  - \frac{ H \xi }{e x} \begin{pmatrix}
                    0 & -1 \\
                   4  & 0
                   \end{pmatrix} \,, \\
M_{\pm 2}(x) & = \begin{pmatrix} 
               1 \mp 4 \frac{\xi}{x} + 2 \frac{\xi^2}{x^2} & \frac{H \xi}{e x^2} \left( \pm 1 + x \xi \mp \xi^2 \right) \\
              \frac{4 H \xi^2}{e x^2 } \left(\pm x - \xi \right) &
                1 - \frac{2}{x^2}  + \frac{2 H^2 \xi^2}{ e^2 x^2} (-1 + \xi^2)
                 \end{pmatrix} \,.
\end{align}
We recognize the tachyonic instability in the $+2$ gauge field mode in the top left entry of $M_{+2}$, leading to an exponential growth for this mode for $(2-\sqrt{2})\xi \leq x \leq (2+\sqrt{2}) \xi$. In the bottom right corner we find the (helicity conserving) mass for the metric tensor mode. Here the first term accounts for the free oscillation on sub-horizon scales, whereas the second term is responsible for the freeze-out on super-horizon scales. The last term is the source term arising from the background gauge field, see the first term on the right-hand side of Eq.~\eqref{eq:eomGWs}, and contributes a positive mass term for $\xi > 2$. Similar to the helicity 0 case discussed above, the system reduces to $N_{\pm2} = 0$ and $M_{\pm2} = \mathbb{1}$ in the far past as $x \rightarrow \infty$. For sub-horizon modes ($x > 1$), the off-diagonal elements are small as long as $e M_P \ll  H$, ensuring that the gauge field modes are well described by Eqs.~\eqref{eq:-2mode} and \eqref{eq:+2mode}. See also the right panel of Fig.~\ref{fig:fluctuations}.

{Finally, let us perform an analytical estimate of the super-horizon amplitude of the enhanced gravitational wave mode in the large-$\xi$ limit. For $x \leq 1$, freeze-out of the gravitational wave implies $w_{+2}^{(\gamma)}(x) \propto 1/x$ and hence $\tfrac{d^2}{dx^2} w_{+2}^{(\gamma)}(x) = 2/x^2 w_{+2}^{(\gamma)}(x)$, thus precisely canceling the second term of the bottom right element of $M_{\pm2}$. For $x \leq 1$, the gauge field mode $w_{+2}^{(e)}$ is well described by the solution~\eqref{eq:NonabelianAsymptotics}, implying that $\overline{(w_{+2}^{(e)})'(x)} \sim 1/x \, \overline{w_{+2}^{(e)}(x)}$ where the bar indicates an averaging over the (co)sine. With this we see that the off-diagonal first derivative term is suppressed by a factor $\xi^2$ compared to the off-diagonal mass term, and the equation of motion for $w_{+2}^{(\gamma)}(x)$ at $x = 1$ reads
\begin{equation}
   \frac{4 H \xi^3}{e }  w_{+2}^{(e)}(x = 1) \simeq \left(1  + \frac{2 H^2 \xi^4}{ e^2}  \right) w_{+2}^{(\gamma)}(x = 1) \,.
\end{equation}
If furthermore $H \xi^2 \ll e$, we can immediately obtain the value of the gravitational wave mode at (and beyond) horizon crossing as
\begin{equation}
 x \, w_{+2}^{(\gamma)}(x) \big|_{x \lesssim 1} = -  \frac{4 H \xi^{3}}{e } w_{+2}^{(e)}(x = 1) \simeq -  \frac{2 H \xi^{5/2}}{e }  2^{3/4} e^{(2 - \sqrt{2}) \pi \xi} \,,
 \label{eq:GWanalytical}
\end{equation}
where in the last step we have inserted Eq.~\eqref{eq:NonabelianAsymptotics}, replacing the cosine with a factor of 1/2. For the parameter point of Fig.~\ref{fig:fluctuations} this yields $ \, x \, w_{+2}^{(\gamma)}(x) \big|_{x \lesssim 1} \simeq 3.7 \times 10^3$, agreeing with the full numerical solution up to an order one factor.} We emphasize that due to helicity conservation only one of the two metric tensor modes is enhanced in this manner, resulting in a chiral gravitational wave spectrum.

 In a similar manner, the `freeze-out'-like behaviour visible for the gauge field modes in Fig.~\ref{fig:fluctuations} can be traced back to the coupling to the metric tensor perturbation through to top-right element of $M_{+ 2}$. {On far super-horizon scales, the contribution from the frozen gravitational wave mode becomes comparable to the contribution from the decaying gauge field modes in the in the equation of motion for the gauge fields. 
 In this regime, the derivative terms are suppressed by a factor of $\xi^{-2}$ compared to the $M_{+2}$ terms.} The amplitude of $w^{(e)}_{+2}$ can then be estimated by comparing the two terms in the first line of $M_{+2}$  to get 
 \begin{equation}
  	w^{(e)}_{+2} (x \rightarrow 0)\simeq \frac{H\xi}{2 e} w^{(\gamma)}_{+2}(x \rightarrow 0)\, ,
  \end{equation} 
 For the parameter point of Fig.~\ref{fig:fluctuations} this yields  $w^{(e)}_{+2} (x \rightarrow 0)\simeq 10^{-2} \, \times \, w^{(\gamma)}_{+2}(x \rightarrow 0) $, in good agreement with the full numerical result.
For the parameter example of this paper, the contribution of the far super-horizon modes to both the energy and variance of the gauge fields is negligible, due to a suppression both in amplitude and momentum compared to the modes crossing the horizon at the same time. Consequently, these are well described by employing the solutions of Eq.~\eqref{eq:+2mode}. On the other hand, if the gauge coupling is very small, this description is no longer accurate and the gauge and gravity sector need to be treated as a fully coupled system. {In this regime, the gauge field/gravity interactions induce an exchange of energy between the $e_{\pm 2}$ modes and gravitational waves~\cite{Caldwell:2016sut, Caldwell:2017sto}.}

Both the scalar and tensor sector preserve the usual scaling behaviour of de Sitter space. In the limit of constant $H$ and $\xi$, we obtain a scale-invariant scalar and tensor power spectrum. The slow variation of $H$ and $\xi$ obtained in any realistic inflation model will lead to deviations from this exact scale invariance. We will discuss this in more detail in the next section.

 \section{A worked example \label{sec:example}}
 
 To illustrate the results obtained so far, we will discuss an explicit parameter example in this section. The most natural scalar potential for an axion is a periodic potential, breaking the shift symmetry of the axion down to a discrete symmetry due to non-perturbative effects,
 \begin{equation}
  V = V_0 \left[1 - \cos\left( \frac{\phi}{f_\phi} \right) \right] \simeq \frac{1}{2} m^2 \phi^2 - \frac{\lambda}{4} \phi^4 \,.
  \label{eq:ScalarPotential}
 \end{equation}
In the following we will take $m = 7.5 \times 10^{-6}~M_P$ and $\lambda = 1.1 \times 10^{-13}$ (corresponding to $V_0^{1/4} 
\simeq 8.3 \times 10^{-3}~M_P$ and $f_\phi \simeq 9.2~M_P$). This parameter choice ensures the correct normalization of the scalar power spectrum at CMB scales as well as a tensor-to-scalar ratio in agreement with the Planck data~\cite{Ade:2015lrj}.\footnote{We note that natural inflation, described by Eq.~\eqref{eq:ScalarPotential} is in some tension with the latest Planck data. For our purposes, the precise form of the potential is not relevant, so we stick with Eq.~\eqref{eq:ScalarPotential} for simpicity. For a discussion of the impact of different types of scalar potentials in abelian axion inflation, see Ref.~\cite{Domcke:2016bkh}.} 
The remaining parameters are then the gauge-field inflaton coupling $\alpha/\Lambda$ which directly controls the size of the parameter $\xi$ and the gauge coupling $e$. In the following we choose $\alpha/\Lambda = 30$ and $e = 5 \times 10^{-3}$. This parameter choice places the matching point between the abelian and the non-abelian regime within the observable last 55 e-folds of inflation while keeping a safe distance from the CMB scales. It serves to illustrate the main results as well as the difficulties encountered in a concrete realization of emerging chromo-natural inflation. A generalization of this setup is briefly discussed at the end of Subsection~\ref{subsec:gaugefluctuations}.
 
 The discussion in this section is organized as follows. In Subsection~\ref{subsec:gaugefluctuations} we will discuss the growth of the gauge field fluctuations with particular emphasis on the tachyonic modes as well as their backreaction on the homogeneous background field. Upon determining the range of validity of our linearized approach, we turn to the scalar and tensor power spectra in Subsection~\ref{subsec:powerspectra}.
 
 \subsection{Growth of gauge field fluctutions \label{subsec:gaugefluctuations}}
 
We first recall some key results about the homogeneous background evolution and the gauge field fluctuations from the previous sections:
\begin{itemize}

\item In single field inflation models, and in particular for the scalar potential considered here, the inflaton velocity $\dot \phi$ and hence the parameter $\xi$ increases during inflation.
 
 \item In the far past, when $\xi<2$, the only stable solution for a classical isotropic gauge field background is the zero solution. General solutions are described by small perturbations around the zero solution.
 
  \item As long as the homogeneous background  is sufficiently small, three of the six gauge field modes are tachyonically enhanced, corresponding to three copies of the abelian limit described in Sec.~\ref{sec:abelian} (see Fig.~\ref{fig:TensorModeEvolution}). In this abelian limit, the variance  $\langle A_\text{ab}^2 \rangle^{1/2}$ grows exponentially with $\xi$ and is well described by Eq.~\eqref{eq:variance_abelian}.  

  \item When $\xi > 2$, a stable, non-zero background solution develops (see Sec.~\ref{sec:background}).  We refer to this second solution as the ``$c_2$-solution.'' It becomes possible that at some point, large fluctuations arising from the tachyonically enhanced modes will push the background away from the zero solution and towards the $c_2$-solution. 
 
  \item The transition from an approximately-zero homogeneous background field to the $c_2$-solution occurs once the fluctuations become large enough to trigger the $c_2$-solution, $ e \langle A_\text{ab}^2 \rangle^{1/2} \sim \xi / (- \tau)$, see Eq.~\eqref{eq:transition-time}.\footnote{Eq.~\eqref{eq:BoundaryOscillatory} marks the boundary to the oscillatory regime, from where $c_i$-type solutions spiral inwards to their asymptotic $c_i$ values. As Fig.~\ref{fig:u0-of-xi} illustrated, for sufficiently large $\xi$ the $c_2$ solution becomes overwhelmingly likely. In the unlikely event that the classical background begins to evolve towards a $c_0$ solution, the gauge field background would be continued to be dominated by $\langle A_\text{ab}^2 \rangle^{1/2}$, growing according to Eq.~\eqref{eq:variance_abelian}. The resulting stochastic initial conditions will eventually trigger a $c_2$-type background. Numerically, the condition $ e \langle A_\text{ab}^2 \rangle^{1/2} \sim \xi / (- \tau)$ is basically equivalent to requiring that the magnitude of the fluctuations be of the same order as the $c_2$ solution and very similar to the requirement of Eq.~\eqref{eq:Conditionabelian}. This `matching condition' is conservative in the sense that even smaller fluctuations which reach the  $c_1$ saddle point solution (see Sec.~\ref{subsec:phase-space}) could (classically) evolve towards the $c_2$-solution. This is depicted by the dotted blue line in Fig.~\ref{fig:matching}. } This is depicted by the solid black line in Fig.~\ref{fig:matching}.

  \item As the background grows, we enter the non-abelian regime.  In this regime, the background field evolves towards the isotropic $c_2$-solution where only the helicity $+2$-mode is enhanced.
  \end{itemize}

\begin{figure}
 \centering
 \includegraphics[width = 0.6 \textwidth]{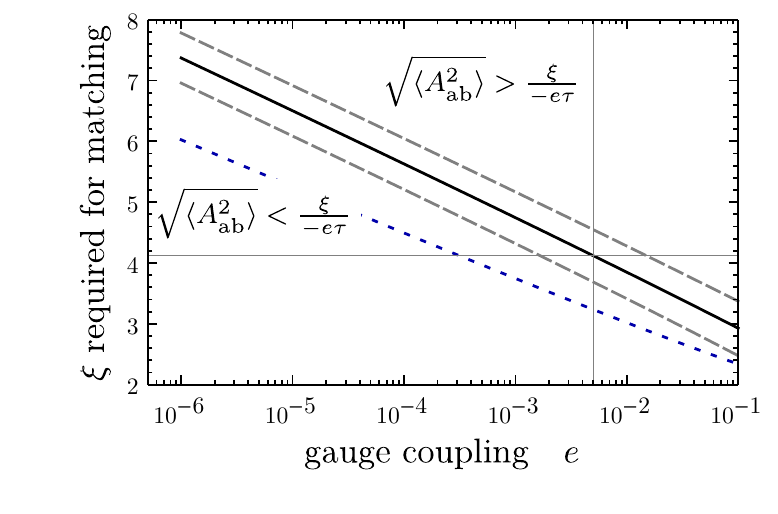}
 \caption{Values of $\xi$ required to match the abelian and non-abelian regime as a function of the gauge coupling $e$. The horizontal and vertical gray lines indicate the parameter point used in this section. {For reference, the dashed gray lines indicate modified matching conditions, $\langle A_\text{ab}^2 \rangle^{1/2} = {\cal N} \xi /(- \tau e)$ with ${\cal N} = \{1/3, 3 \}$ parametrizing the theoretical uncertainties in the matching condition, see text. The dotted blue line indicates where $\langle A_\text{ab}^2 \rangle^{1/2}$ can reach the unstable $c_1$ background solution  (see Sec.~\ref{subsec:phase-space}). }}
 \label{fig:matching}
\end{figure}

Based on these observations, our strategy will be the following: (i) As long as the abelian variance $\langle A_\text{ab}^2 \rangle^{1/2}$ as given in Eq.~\eqref{eq:variance_abelian} is much smaller than $\xi /(- \tau e)$ we work in the abelian limit with $f(\tau) = 0$. (ii) When $\langle A_\text{ab}^2 \rangle^{1/2} \simeq \xi /(- \tau e)$ we take $f(\tau)$ to be given by the $c_2$-solution~\eqref{eq:c2solution}. We match $\phi(\tau)$ and $\phi'(\tau)$ at this point, but turning on the term on the right-hand side of Eq.~\eqref{eq:rev_motion} will cause a discontinuity in $\phi''(\tau)$.\footnote{In the abelian regime, the corresponding term is given by the right-hand side of Eq.~\eqref{eq:rev_motion} after inserting Eq.~\eqref{eq:rev_EB}. For the parameter point discussed in this section, this is about a factor 10 smaller than the non-abelian expression at the matching point.} We ensure that for the parameter point we consider this term is sub-dominant, so as to limit any unphysical effects here. (iii) We compute the evolution of all degrees of freedom in this background, tracking each mode from the sub- to the super-horizon regime. Note that for smaller values of the gauge coupling, the transition from the abelian to the non-abelian regime at $ e \langle A_\text{ab}^2 \rangle^{1/2} = \, \xi/(- \tau)$ requires larger values of $\xi$, see Fig.~\ref{fig:matching}.

Several comments are in order. Firstly, our matching procedure from the abelian to the non-abelian regime should be seen as a rough order-of-magnitude estimate only, and the results for the evolution of the background and of the fluctuations in this transition regime should be treated with care accordingly. Secondly, the linearization we are using is justified~\footnote{Notice that this is the same criterion adopted in~\cite{Adshead:2016omu} . In particular Eq.~\eqref{eq:estimate-variance} is equivalent to Eq.(7.6) of~\cite{Adshead:2016omu} .} as long as $\langle \delta A^2 \rangle^{1/2} \ll f(\tau)$. Once the growth of the $e_{+2}$ mode overcomes the background evolution, a different treatment of the gauge field background becomes necessary, which is beyond the scope of this paper.

\begin{figure}[t]
 \centering
 \subfigure{
 \includegraphics[width = 0.48 \textwidth]{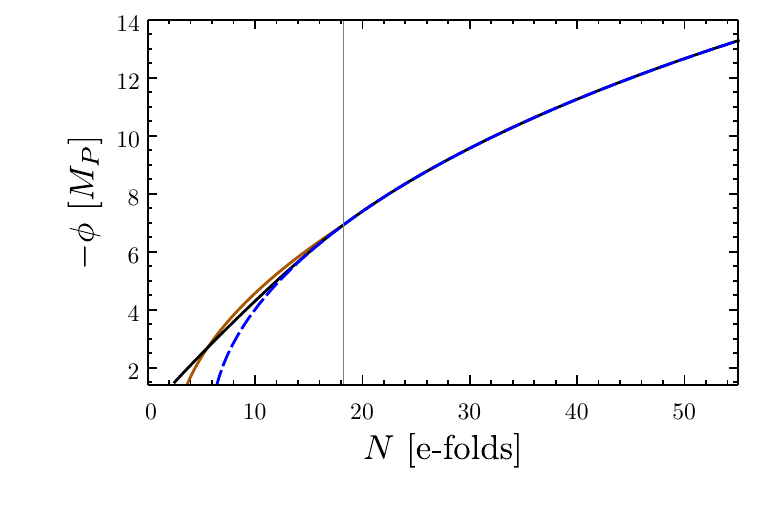}}
 \hfill 
 \subfigure{
  \includegraphics[width = 0.48 \textwidth]{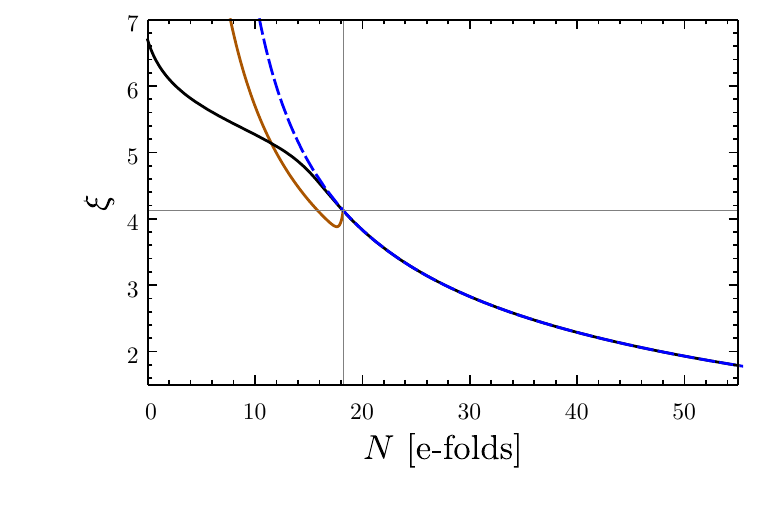}}
 \caption{Evolution of the inflaton field $\phi$ for $\alpha/\Lambda = 0$ in the absence of the inflaton - gauge field coupling (dashed blue), in the abelian limit ($\alpha/\Lambda = 30, \, e = 0$, solid black) and including non-abelian effects in the weak coupling limit ($\alpha/\Lambda = 30, \, e = 5 \times 10^{-3}$, solid orange). The gray lines indicate the matching between the abelian and non-abelian regime as detailed in the text.}
 \label{fig:phi-evolution}
\end{figure}

In Fig.~\ref{fig:phi-evolution}, we depict the evolution of the homogeneous inflaton background in the abelian (solid black curve) and in the non-abelian regime (solid orange curve), obtained by numerically solving Eq.~\eqref{eq:phibackground} with $f(\tau)$ set to 0 and to $c_2 \, \xi/(- e \tau)$ in the two regimes, respectively. The matching point $\langle A^2_\text{ab} \rangle^{1/2} = c_2 \, \xi/(- e \tau)$ is indicated by the horizontal and vertical gray lines. For reference, we also show the evolution in the absence of the inflaton - gauge field coupling (dashed blue curve). The x-axis of Fig.~\ref{fig:phi-evolution} is labeled in e-folds, $dN = - H dt$, where we use the convention that inflation ends at $N = 0$ and the CMB scales exit the horizon at $N \simeq 55$.

{The right panel of Fig.~\ref{fig:phi-evolution} shows the velocity of $\phi$, encoded by the parameter $\xi$. After the matching, the velocity drops abruptly, since turning on the background gauge field enhances the gauge-field induced term in the inflaton equation of motion. {The details are sensitive to the matching procedure we invoke, and for our computations of the scalar and tensor power spectra in Sec.~\ref{subsec:powerspectra} we will therefore exclude a few e-folds around this transition regime.}

More importantly, after this transition regime the last term in Eq.~\eqref{eq:phibackground} in only proportional to $\xi^3$ instead of being  exponentially sensitive to $\xi$ as in the abelian regime (see Eq.~\eqref{eq:rev_EB}). Consequently, the dominant terms in the equation for $\phi$ in this regime are the second and third term of Eq.~\eqref{eq:phibackground}, and, similar to the situation in the absence of the gauge field - inflaton coupling, $\xi \propto \sqrt{\varepsilon} \propto 1/\sqrt{N}$~\cite{Domcke:2016bkh} (see also Sec.~\ref{subsec:dynamical_background}). This in particular implies that in this regime we are not in the `magnetic drift regime' studied in \cite{Adshead:2013nka}, which is characterized by the gauge friction dominating over the Hubble friction. }

Next we consider the evolution of the gauge field fluctuations in this background. From the discussion in the previous section, we know that the helicity $+2$ mode captures the enhancement both in the abelian and in the non-abelian regime. For simplicity, we will restrict our discussion here to this mode, however we have checked numerically that including the full system does not lead to any significant changes. In de Sitter space, any mode with co-moving momentum $k$ exits the horizon at $k = a H = - x/\tau$, where to good approximation $3 H^2 M_P^2 = V(\phi)$. Setting $a = 1$ at the end of inflation, this implies that at e-fold $N$, the mode $k_N = \exp(-N) H$ exits the horizon. In Fig.~\ref{fig:modes-evolution} we show the evolution of six modes which exit the horizon in the abelian regime (left panel) and in the non-abelian regime (right panel). The gradual change within each panel is due to the evolving background, i.e.\ the slow increase of $\xi$.

\begin{figure}[t]
\subfigure{
\includegraphics[width = 0.48 \textwidth]{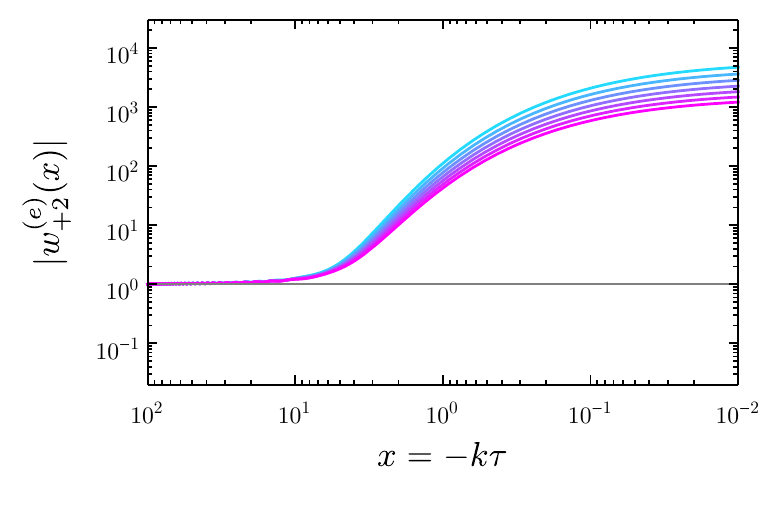}
}
\hfill
\subfigure{
\includegraphics[width = 0.48  \textwidth]{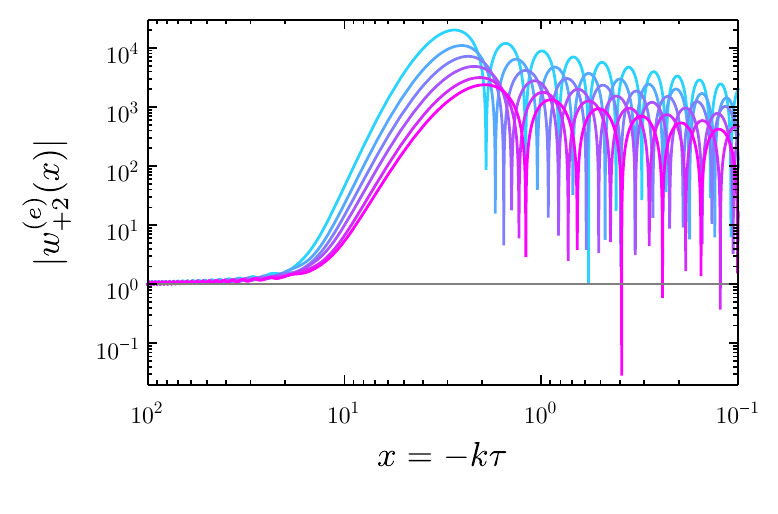}
}
\caption{Evolution of the helicity $+2$ mode in a dynamical background. \textbf{Left panel}: abelian regime, modes exiting at $N = 25,26\dots 31$. \textbf{Right panel}: Non-abelian regime, modes exiting at $N = 8,9,\dots 13$. In both panels, $k$ increases (i.e.\ the value of $N$ labelling the horizon exit decreases) from purple to blue. Parameters as in Fig.~\ref{fig:phi-evolution}.}
  \label{fig:modes-evolution}
\end{figure}

Since the change of $\xi$ is slow, we can estimate the variance of these fluctuations by {(see Eq.~\eqref{eq:variance_abelian})}
{ \begin{equation}
 \langle \delta A^2 \rangle_N = \int \frac{\textrm{d}^3 \vec{k}}{(2 \pi)^3 } \frac{1}{2 k} | w_{+2}^{(e)}\left(k, 
 \tau(N) \right) |^2  \simeq \frac{1}{\left(- \tau(N)\right)^2} \int \frac{x \,  \textrm{d}x}{4 \pi^2} | w_{+2,N}^{(e)}(x) |^2 \,,
 \label{eq:estimate-variance}
\end{equation} 
where $w_{+2,N}$ indicates the mode function of the wave vector $k_N$ (see Fig.~\ref{fig:modes-evolution}), which we use to approximate the full mode function at the e-fold $N$. } In the left panel of Fig.~\ref{fig:variance} we show the resulting variance (green dots), together with the homogeneous background solution $f(\tau)$ (solid blue). In the abelian regime, the semi-analytical expression Eq.~\eqref{eq:variance_abelian} (shown as a dashed green line) gives a good approximation over most of this regime. {The deviation at large $N$ simply reflects that our fitting formula~\eqref{eq:variance_abelian} is not optimized for very small values of $\xi$,} 
whereas the deviations a few e-folds before the matching point reflect that the super-horizon parts of these modes are affected by the non-abelian regime, which leads to a suppression of the variance. In the following analysis we will exclude the modes which exit the horizon within 3 e-folds before or after the matching point, so as to minimize artifacts introduced by the specific matching procedure. To illustrate the uncertainties involved, the gray curves indicate the variance we obtain in the transition region with the matching procedure above (dashed gray curve) and by imposing the matching at a later point (dotted gray curve), as indicated by the vertical dashed line.

In the non-abelian regime, the fluctuations are initially suppressed compared to the abelian case, due to a combination of two effects: { Firstly, as is evident from Fig.~\ref{fig:phi-evolution}, the parameter $\xi$, which controls the tachyonic instability in the helicity $+2$ gauge field mode, initially drops after switching on the background gauge field. Secondly, in the non-abelian regime the variance grows more slowly as a function of $\xi$, c.f.\ right panel of Fig.~\ref{fig:variance}.} The smallness of $\langle \delta A^2 \rangle$ compared to its counterpart in the abelian regime and compared to the classical background $f(\tau)$ is a crucial ingredient in justifying the ansatz of a linearized analysis around a homogeneous background field in the non-abelian regime. In the parameter example at hand, the non-Abelian fluctuations are at best mildly suppressed compared to the homogeneous background, implying significant uncertainties in our analysis of the non-abelian regime for this parameter point. The situation improves for smaller gauge - inflaton couplings $\alpha/\Lambda$ and for smaller gauge couplings $e$. In particular, we note the gauge friction dominated regime of CNI is free from this problem~\cite{Adshead:2016omu}. However, in these cases the systematic uncertainties associated with the matching procedure tend to be larger. Given these limitations, we refrain from tweaking the parameter point discussed here, emphasizing that its purpose is to illustrate our main line of thought as well as the encountered obstacles, leaving a more detailed investigation of the parameter space to future works. 

Once the fluctuations reach values close to the background field we cannot trust our treatment anymore, since the motion of the background field is now no longer determined by its classical motion.\footnote{We note that the perturbativity criterion employed in~\cite{Caldwell:2017chz} (see also~\cite{Adshead:2016omu,Adshead:2017hnc}), which measures the fluctuations per logarithmic frequency interval, $(\tfrac{\textrm{d} }{\textrm{d} \ln k} \langle \delta A^2 \rangle^{1/2})/f$ is less restrictive. This quantity does not exceed the percent level in the entire regime depicted in Fig.~\ref{fig:variance}. This may indicate the possibility of pushing the linearized analysis somewhat further than  the conservative cut-off implemented in this analysis. Moreover, we are somewhat overestimating the variance in Eq.~\eqref{eq:estimate-variance}, since we are integrating over the super-horizon part of the mode $ w_{+2, k_N}^{(e)}(\tau)$, whereas we should be taking the super-horizon contributions of the modes which crossed the horizon accordingly earlier, at correspondingly smaller $\xi$. } 
This regime calls for a dedicated lattice simulation to capture the non-linear effects, which is beyond the scope of this work.

\begin{figure}[t]
\subfigure{
\includegraphics{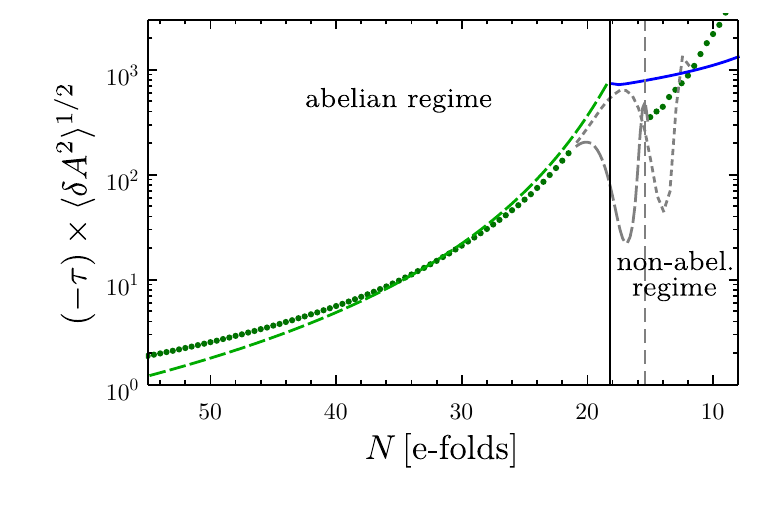}
}
\hfill
\subfigure{
\includegraphics{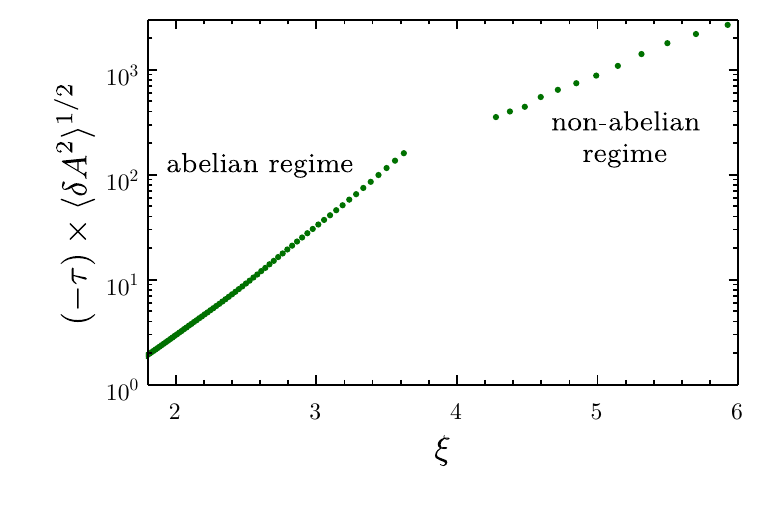}
}
\caption{Magnitude of the gauge fields. Background solution $f(\tau)$ (blue) and estimated variance of the fluctuations (green). The solid vertical line denotes the matching point. The gray dashed lines are auxiliary quantities as described in the text. Same parameters as in Fig.~\ref{fig:phi-evolution}.}
 \label{fig:variance}
\end{figure}

The results in this section were obtained for the parameter choice given below Eq.~\eqref{eq:ScalarPotential}. The parameters of the scalar potential, $m$ and $\lambda$, are directly related to the CMB observables and do not impact the discussion of this section much. On the other hand, the coupling parameters $\alpha/\Lambda$ and $e$ are crucial for our discussion. Increasing $\alpha/\Lambda$ increases $\xi$ at the CMB scales and consequently leads to an earlier transition to the non-abelian regime. Current observations constrain $\xi_\text{CMB} \lesssim 2.5$ in the contest of abelian axion inflation~\cite{Barnaby:2011qe}. Increasing the gauge coupling $e$ leads to a lower threshold of $\xi$ to trigger the non-vanishing gauge field background, see Fig.~\ref{fig:matching}. Correspondingly, the transition happens earlier and also more smoothly, since the gauge field source term in the inflaton equation of motion~\eqref{eq:phibackground}, proportional to $\xi^3/e^2$ when inserting the $c_2$-solution, is less important.

\subsection{Scalar and tensor power spectra \label{subsec:powerspectra}}

We now turn to the scalar and tensor power spectra of the benchmark model of the previous subsection, the key observables of any inflation model (for a review, see~\cite{Baumann:2009ds}):
\begin{align}
 \langle \zeta_{\vec k} \zeta_{\vec k}' \rangle & = (2 \pi)^3 \delta (\vec k + \vec k') {\cal P}_\zeta(k)  \,, \\
 \langle h_{\vec k}^\lambda h_{\vec k'}^{\lambda'} \rangle & = (2 \pi)^3 \delta (\vec k + \vec k') \delta_{\lambda \lambda'} {\cal P}_{h^\lambda}(k)  \,.
\end{align}
Identifying the Mukhanov variables which are canonically normalized on far sub-horizon scales as $v^\zeta = (a \, \delta \phi)$ and $v^h = (a \gamma)$ for scalars and tensors, respectively, the power spectra read
\begin{align}
  {\cal P}_\zeta(k) & =  \left( \frac{H}{\dot \phi} \right)^2 \left(\frac{|v^\zeta_k(x)|}{a}\right)^2 \,, \\
  {\cal P}_{h^\lambda}(k) & =  \left( \frac{2}{M_P} \right)^2 \left(\frac{|v^h_k(x)|}{a}\right)^2 \,,
\end{align}
where $\zeta$ denotes the gauge invariant curvature perturbation and $\lambda$ denotes the helicity of the gravitational wave $h^\lambda$.
Due to the freeze-out of $(a \delta \phi)$ and $(a \gamma)$ on super-horizon scales (see Sec.~\ref{sec:AllFluc}), it suffices to evaluate these power spectra at horizon crossing $(x = 1)$. Since at this point in time the coupling to the gauge fields can be very relevant, we perform this task numerically, solving the mode equations Eq.~\eqref{eq:fullscalar} and \eqref{eq:fulltensor} in the evolving background discussed in Sec.~\ref{sec:example}.

Recalling that $ v_k^{\zeta,h}(x)$ is a function of $x$ only, it is convenient to introduce
\begin{equation}
 \Delta_s^2 = \frac{k^3}{2 \pi^2} {\cal P}_\zeta(k) \,, \qquad  (\Delta_t^\lambda)^2 = \frac{k^3}{2 \pi^2} {\cal P}_{h^\lambda}(k) \,,
\end{equation}
such that
\begin{align}
 \Delta_s^2 & = \left( \frac{H_*}{2 \pi_*}\right)^2 \left( \frac{H_*}{\dot \phi_*}\right)^2  \left(  x \, |w^{(\phi)}_0(x)| \right)^2 \bigg|_{x \ll 1}  \,, \label{eq:Ds}\\
 (\Delta_t^\pm)^2 & = \left( \frac{H_*}{2 \pi}\right)^2 \left( \frac{2}{M_P}\right)^2  \left(  x \, |w^{(\gamma)}_{\pm2}(x)| \right)^2 \bigg|_{x \ll 1}  \,, \label{eq:Dt}
\end{align}
where $H_*$ and $\dot \phi_*$ denote the Hubble parameter and inflaton velocity at the point in time when the mode in question crosses the horizon $(x = 1)$. To ensure that we are fully in the freeze-out regime, the last parenthesis in Eqs.~\eqref{eq:Ds} and \eqref{eq:Dt} is evaluated at $x = 0.1$. We emphasize that the numerical evaluation of the scalar and tensor power spectrum presented below should be taken with a grain of salt, due to the lack of a clear hierarchy between the gauge field background and its fluctuations in the non-abelian regime for this particular parameter point. In other parts of the parameter space, where this problem does not arise, the analysis below applies without this caveat.

\subsubsection{Scalar power spectrum \label{subsec:PS}}

The scalar power spectrum computed in this way is subject to the caveat described around Eq.~\eqref{eq:caveat}:  to linear order in $\delta A$, it is sourced only by helicity zero objects, including the helicity zero gauge field fluctuations. However at ${\cal O}(\delta A^2)$, the helicity $+2$ gauge fluctuations contribute too, and due to their strong enhancement, can become the dominant source.
Generalizing the procedure of Ref.~\cite{Linde:2012bt} (see also \cite{Barnaby:2011qe}) to the non-abelian case, we can obtain an estimate for the full scalar power spectrum including this second order contribution. For simplicity, let us consider only the  helicity 0 mode associated with $\delta \phi$, whose coupling to $\langle F \tilde F \rangle$ in the action contains a coupling to two enhanced helicity $+2$ gauge field modes. In real space, the equation of motion for $\delta \phi$ reads,
\begin{equation}
 \ddot{\delta \phi} + 3 H \dot{\delta \phi} - \frac{\nabla^2}{a^2} \delta \phi + m_{\phi \phi}^2 \delta \phi = - \frac{\alpha}{4 \Lambda} \delta(\langle F \tilde F \rangle) \,,
 \label{eq:dphiRealSpace}
\end{equation}
where we have neglected here the couplings to the helicity 0 gauge field modes, which are discussed in depth in Sec.~\ref{sec:AllFluc}. The right-hand side is the variation of the Chern-Simons term with respect to the average value entering in the 0th order equation, taking also into account the $\dot \phi$-dependence of $\langle F \tilde F\rangle$ through the parameter $\xi$,
\begin{align}
 \delta (\langle F \tilde F \rangle ) & = [F \tilde F - \langle F \tilde F \rangle ]_{\delta \phi = 0} + \frac{\partial \langle F \tilde F \rangle }{\partial \dot \phi} \dot{\delta \phi} \,. 
\end{align}
As demonstrated in App.~\ref{app:variance_computation}, this can be re-expressed as 
\begin{align}
 \delta (\langle F \tilde F \rangle ) = \frac{H^4}{ \pi^2} e^{2 \pi(\kappa - \mu)} \left[  \left(\frac{ \tilde T_1 + \tilde T_2}{5} \right)^{1/2} + \frac{  2 \pi \alpha (\kappa - \mu ) }{2 \Lambda  \xi} \,  \tilde T_0  \, \delta \phi\right] \,.
\end{align}
with 
\begin{align}
\tilde T_1 = 0.0082 \cdot \xi^8   \,, \quad \tilde T_2 = 0.051 \cdot  \xi^6   \,, \quad \tilde T_0 = - 0.24 \cdot  \xi^3\,,
\end{align}
where $\kappa$ and $\mu$ are defined below Eq.~\eqref{eq:whittakeragain}. At horizon crossing the first and third term in Eq.~\eqref{eq:dphiRealSpace} cancel. Moreover, neglecting the slow-roll suppressed term proportional to $m_{\phi \phi}^2$ and using $\dot{\delta \phi} \simeq H \delta \phi$, we obtain
\begin{equation}
 \delta \phi \simeq - \frac{\alpha}{12 \Lambda H^2} \delta(\langle F \tilde F \rangle) \simeq \frac{- \frac{e^{2 \pi(\kappa - \mu)}}{{12} \pi^2} \frac{\alpha}{\Lambda} H^2 \sqrt{\frac{\tilde T_1 + \tilde T_2}{5}}}{1 + \frac{e^{2 \pi(\kappa - \mu)}}{24 \pi^2} \left(\frac{\alpha}{\Lambda} \right)^2  \frac{2 \pi (\kappa - \mu )}{\xi} H^2 \tilde T_0} \,.
 \label{eq:deltaphi_2nd}
\end{equation}
Note that for sufficiently large $\xi$, when the right-hand side dominates Eq.~\eqref{eq:dphiRealSpace} simply becomes
\begin{equation}
 \delta \phi \simeq \frac{-2 \xi \Lambda}{ 2 \pi \alpha (\kappa - \mu )} \left( \frac{\tilde T_1 + \tilde T_2}{5} \right)^{1/2} \frac{1}{\tilde T_0 } \,.
\end{equation}
More generally, the two limiting cases of Eq.~\eqref{eq:deltaphi_2nd} are
\begin{align}
 \left(\Delta_s^2\right)^\text{2nd} & \sim \left( \frac{H}{\dot \phi} \delta \phi \right)^2   \simeq \begin{cases}
    \left( \frac{\alpha H}{2 \pi \Lambda} \right)^4  \frac{\tilde T_1 + \tilde T_2}{180 \, \xi^2}  e^{4 \pi(\kappa - \mu)}
     &\qquad  \text{for     } \xi^3 e^{2 \pi(\kappa - \mu)} \ll 270  \left( \frac{\Lambda}{\alpha H}\right)^2    \label{eq:Ds_2nd} \\
  \frac{\tilde T_1 + \tilde T_2}{5 \, \tilde T_0^2}   \left[2 \pi (\kappa - \mu )\right]^{-2}   &\qquad \text{for    } \xi^3 e^{2 \pi(\kappa - \mu)} \gg 270  \left( \frac{\Lambda}{\alpha H}\right)^2  
   \end{cases} \,.
\end{align}

\begin{figure}
\centering
\includegraphics{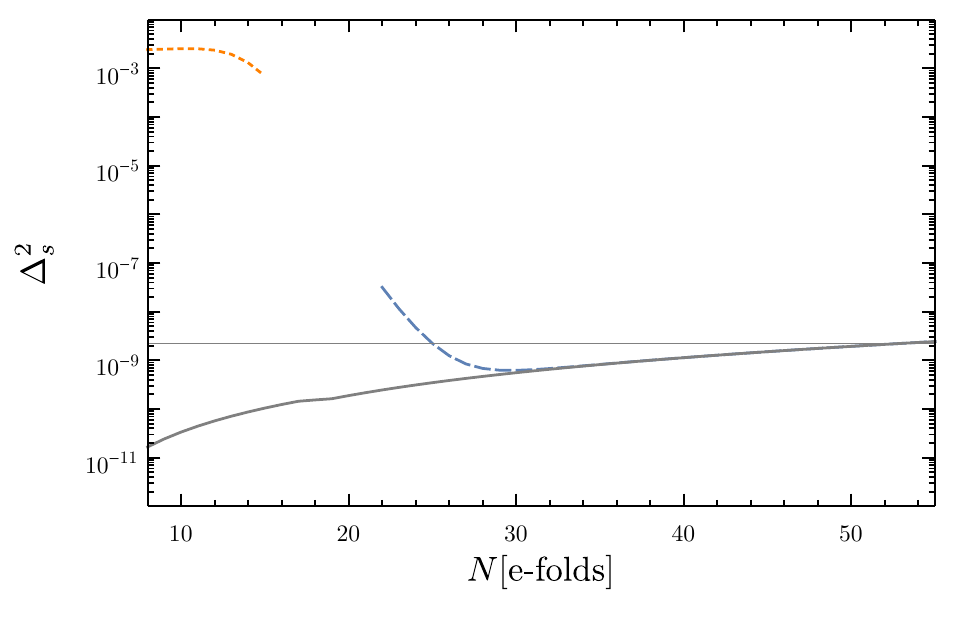}
\caption{Scalar power spectrum. Our semi-analytical estimate~\eqref{eq:Ds_2nd}  in the non-abelian regime is shown as a dotted orange line. For reference, we show the standard contribution of the vacuum fluctuations which also well describe the results of the linearized analysis (solid gray) and the (non-linear) contribution in the abelian regime (dashed blue). Same parameters as in Fig.~\ref{fig:phi-evolution}.}
 \label{fig:spectraPS}
\end{figure}

We note that initially, for small values of $\xi$, this grows as $\exp[4 \pi (\kappa - \mu)] \simeq \exp[4 \pi \xi (2 - \sqrt{2})]$ whereas for large values of $\xi$ we find the same $1/\xi^2$ dependence as in the abelian case~\cite{Linde:2012bt}. The total scalar power spectrum is the sum of the vacuum contribution~\eqref{eq:Ds} and the contribution sourced by the enhanced gauge field mode~\eqref{eq:Ds_2nd}.\footnote{{As this work was being finalized, Refs.~\cite{Dimastrogiovanni:2018xnn,Papageorgiou:2018rfx} appeared, which also consider these nonlinear couplings. The main focus of Ref.~\cite{Dimastrogiovanni:2018xnn} is the three-point correlators between the scalar and tensor perturbations, whereas Ref.~\cite{Papageorgiou:2018rfx} is a dedicated study of the leading nonlinear contribution to the scalar power spectrum.  A direct comparison of our results is difficult due to the different background evolution (see discussion in Sec.~\ref{subsec:dynamical_background}) and due to the fact that in our parameter space we typically encounter larger values of $\xi$ than encountered in \cite{Papageorgiou:2018rfx}. Using a very different methodology than presented here, Ref.~\cite{Papageorgiou:2018rfx} concludes that the non-linear contributions to the scalar power spectrum begin to dominate over the vacuum contributions at  $m_Q \simeq 2.7$, where  $m_Q \simeq c_2 \, \xi$, {and hence $\xi (m_Q = 2.7) \simeq 3$. From Eq.~\eqref{eq:Ds_2nd}, we can estimate that in our analysis, the non-linear term comes to dominate at $\xi \simeq 2.5$.   Moreover, the exponential sensitivity on $\xi$, $\Delta_s^2 \sim e^{2 \pi m_Q}$ in \cite{Papageorgiou:2018rfx}, is similar to what we find here. 
Within the uncertainties inherent to both methods, we consider this a  good agreement.}  For related work in the abelian case see also Refs.~\cite{Barnaby:2010vf,Barnaby:2011vw}.}

} In  Fig.~\ref{fig:spectraPS}, we show the resulting estimate for scalar spectrum in the non-abelian regime (dotted orange). The `strong backreaction regime', where the simplified expression in the second line of Eq.~\eqref{eq:Ds_2nd} applies is reached only around $N \simeq 10$. For reference, we show also corresponding estimate in the abelian regime (dashed blue)~\cite{Linde:2012bt},
\begin{align}
 (\Delta_s^2)_\text{ab.}  = (\Delta_s^2)_\text{vac} +  (\Delta_s^2)_\text{gauge} = \left( \frac{H^2}{2 \pi \dot \phi}\right)^2 + \left(\frac{ \alpha \langle \vec E \vec B \rangle}{2 \beta \Lambda H \dot \phi} \right)^2 \,, 
\end{align}
with
\begin{equation}
 \beta = 1 - 2 \pi \xi \frac{\alpha \langle \vec E \vec B \rangle}{3 \Lambda H \dot \phi} \,, \quad  \langle \vec E \vec B \rangle = - 2.0 \cdot 10^{-4} \frac{H^4}{\xi^4} e^{2 \pi \xi} \,,
\end{equation}
as well as  the standard vacuum contribution (solid gray line, obtained by setting the last parenthesis in Eq.~\eqref{eq:Ds} to 1), which agrees well with the results obtained from the linearized analysis. The horizontal gray line indicates the observed value at the CMB scales.

\subsubsection{Gravitational wave spectrum \label{subsec:TS}}
Next we turn to the tensor power spectrum. For the purpose of direct gravitational wave searches (Pulsar timing arrays (PTAs) and interferometers), it is customary to express the stochastic gravitational wave background (SGWB) as the energy in gravitational waves per logarithmic frequency interval normalized to the critical energy density~$\rho_c$ \cite{Turner:1993vb,Seto:2003kc,Smith:2005mm}, 
\begin{equation}
 \Omega_{\rm GW}(k) = \frac{1}{\rho_c} \frac{\partial \rho_{GW}(k)}{\partial \ln k}  = \frac{(\Delta_t^+)^2 + (\Delta_t^-)^2}{24} \Omega_r \frac{g_*^k}{g_*^0} \left( \frac{g_{*,s}^0}{g_{*,s}^k} \right)^{4/3} \,,
\end{equation}
for modes entering during the radiation dominated epoch of the universe, where $g_*^{k,0}$ ($g_{*,s}^{k,0}$) denotes the effective number of degrees of freedom contributing to the energy (entropy) of the thermal bath at the point in time when the mode $k$ entered the horizon and today, respectively. $\Omega_r = 8.5 \times 10^{-5}$ denotes the fraction of radiation energy today. Neglecting the change in the number of degrees of freedom, this leads to
\begin{align}
 \Omega_{\rm GW}(k) & = \frac{\Omega_r}{24} \left( \frac{H}{2 \pi}\right)^2 \left( \frac{2}{M_P}\right)^2 \left[  \left(  x \, |w^{(\gamma)}_{-2}(x)| \right)^2  +  \left( x \, |w^{(\gamma)}_{+2}(x)| \right)^2 \right]\bigg|_{x = 1}  \\
 & = \frac{\Omega_r}{24} \left( \frac{H}{2 \pi}\right)^2 \left( \frac{2}{M_P}\right)^2 \left[  \left( 1  +  x \, |w^{(\gamma)}_{+2}(x)| \right)^2 \right]\bigg|_{x = 1}  \,,
 \label{eq:OmegaGW} 
\end{align}
where we have made use of the observation that the $w_{-2}^{(\gamma)}$ mode is not enhanced and hence is given by the usual solution of the Mukhanov-Sasaki equation, $w_{-2}^{(\gamma)}(x = 1) = 1 $.

\begin{figure}
\centering
\includegraphics{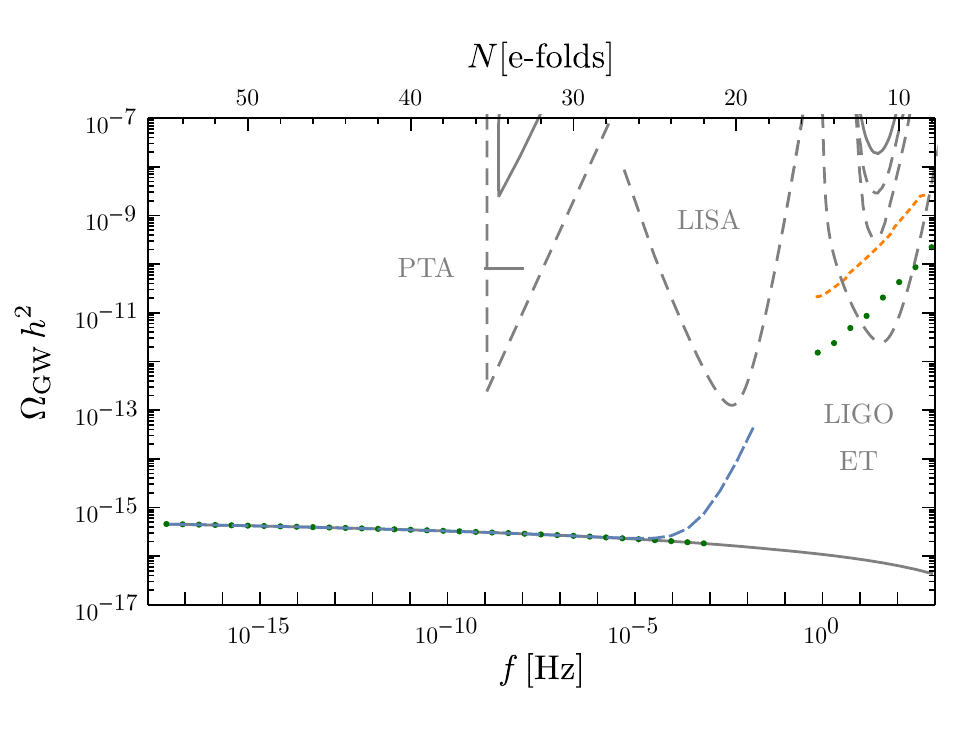}
\caption{Stochastic gravitational wave background. Our semi-analytical estimate~\eqref{eq:GWapprox} in the non-abelian regime  is shown as dotted orange line. For reference, we show the standard contribution of the vacuum fluctuations (solid gray)and the (non-linear) contribution in the abelian regime (dashed blue). The results of the linearized analysis are shown as green dots. Same parameters as in Fig.~\ref{fig:phi-evolution}.}
 \label{fig:spectraGW}
\end{figure}

In Fig.~\ref{fig:spectraGW} we show the resulting  SGWB compared to current and upcoming experimental constraints from pulsar timing arrays~\cite{vanHaasteren:2011ni,Kramer:2004rwa} and from the interferometer experiments LIGO~\cite{TheLIGOScientific:2016wyq}, LISA~\cite{Audley:2017drz} and the Einstein Telescope~\cite{ET}. For reference, we also show the standard vacuum solution (in gray), obtained by setting $  w^{(\gamma, \phi)}(x = 1) = 1$, as well as the analytical results from the abelian regime (dashed blue line),
\begin{align}
 (\Omega_{\rm GW})_\text{ab.}  = \frac{1}{24} \Omega_r \left( \frac{H}{\pi M_P} \right)^2 \left( 2 + 4.3 \cdot 10^{-7} \frac{2 H^2}{M_P^2} \frac{e^{4 \pi \xi}}{\xi^6}\right) \,.
 \label{eq:GWab}
\end{align}
 {A key feature of the non-abelian regime is that the gravitational waves couple to the enhanced gauge-field mode at the linear level, resulting in the enhanced SGWB at large frequencies~\cite{Adshead:2013nka}. In the abelian regime, such a source term is absent at the linear level and only appears at the non-linear level. This is simply because the energy-momentum tensor, the source term of the GW equation of motion, is bi-linear in the gauge fields and we do not have a background gauge field in the abelian regime. This non-linear term is of course not captured by the linearized analysis performed here, and hence we simply include the non-linear contribution in the abelian regime given by Eq.~\eqref{eq:GWab} a posteriori. Note that in the non-abelian regime this non-linear contribution is sub-dominant as long as $\delta A \ll f$. For the parameter point studied here, we find the SGWB to be out of reach of the GW interferometers LIGO and LISA, and barely reachable with the Einstein Telescope. However we stress that a different parameter choice, leading e.g.\ to an earlier matching between the abelian and non-abelian regime, could change this picture. A full-fledged study of the parameter space is beyond the scope of the current paper. We note that as in Ref.~\cite{Agrawal:2017awz}, we expect this gravitational wave background to be very non-gaussian, which may be used as a powerful model discriminator in future observations~\cite{Bartolo:2018qqn}.

Combining Eqs.~\eqref{eq:GWanalytical} and \eqref{eq:OmegaGW} we can obtain an analytical estimate of the SGWB in the non-abelian regime,
\begin{equation}
 \Omega_{\rm GW} \simeq \frac{1}{24} \Omega_r \left( \frac{\xi^3 H}{\pi M_P} \right)^2_{\xi = \xi_\text{cr}}  \left( \frac{2^{7/4} H}{e} \xi^{-1/2} e^{(2 - \sqrt{2}) \pi \xi}  \right)^2_{\xi = \xi_\text{ref}} \,.
 \label{eq:GWapprox}
\end{equation}
Here, to take into account the variation of $\xi$, we evaluate the first parenthesis (related to the GW - gauge field coupling at horizon crossing) at $\xi_\text{cr} = \xi(x = 1)$, whereas we evaluate the second parenthesis (describing the helicity $+2$ gauge field solution) at a the reference value $x = (2 + \sqrt{2}) \xi_\text{cr}$ roughly corresponding to the onset of the instability in the gauge field mode. The resulting estimate for the GW spectrum is depicted by the dotted orange line in  Fig.~\ref{fig:spectraGW}. The discrepancy to the numerical result can be traced back to Eq.~\eqref{eq:GWanalytical}, which overestimates the GW amplitude by about a factor of three.}

In summary, we reproduce the vacuum contribution for both the scalar and tensor power spectrum at large $N$,  thus ensuring agreement with all CMB observations. As we approach smaller scales, we observe an enhancement of both spectra as we pass first through the abelian and then through the non-abelian regime. For the specific parameter point studied in this section, the effects seem out of reach of current and upcoming experiments, though a more careful analysis is required to correctly account for the uncertainties due the lack of hierarchy between the gauge field background and its fluctuations. We generically expect larger signals with increasing $\alpha/\Lambda$ and/or increasing gauge coupling, see discussion at the end of Sec.~\ref{subsec:gaugefluctuations}. This part of the parameter space however comes with even larger gauge field fluctuations, requiring a different computational strategy to even reliably describe the background evolution in this regime.

\section{Conclusions and Outlook \label{sec:conclusions}}

A shift-symmetric coupling between the inflation sector and a (non-abelian) gauge sector opens up new avenues to probe the microphysics of inflation. Based on earlier works studying the coupling of the inflaton field to abelian~\cite{Barnaby:2010vf} and non-abelian~\cite{Adshead:2012kp,Dimastrogiovanni:2012st,Dimastrogiovanni:2012ew,Adshead:2013qp,Adshead:2013nka} gauge groups, we demonstrate here how the former can be understood as a limit of the latter in the far past. To this end, we extend and complement the analysis of CNI (see Refs.~\cite{Dimastrogiovanni:2012ew,Adshead:2013qp,Adshead:2013nka}) in the following ways: (i) Most importantly, we propose a mechanism for the dynamical generation of the homogeneous isotropic gauge field background, a key ingredient to CNI, starting from standard Bunch--Davies initial conditions in the infinite past. The evolution of the background is then governed by two competing effects: the classical background motion and the stochastic motion due to the strong growth of gauge field perturbations. We demonstrate how the system converges to an attractor solution, where the gauge field background is dominated by the classical motion. (ii) This attractor solution is more general than the solution employed in~\cite{Adshead:2012kp}, and in particular also allows us to study the parameter space in which the backreaction of the gauge fields on the equation of motion of the inflaton is small. (iii) We provide relatively simple explicit expressions for the equations of motion of the Fourier components of all physical perturbations modes, after integrating out all gauge degrees of freedom and constraint equations. Due to a slightly different choice of basis, our expressions are somewhat simpler than those found in the literature, in particular when studying only the gauge field fluctuations. (iv) Although the bulk of our analysis is performed for the linearized system of perturbations, we highlight the importance of non-linear contributions in the scalar sector, giving an estimate for the resulting contribution to the scalar power spectrum. Finally, we emphasize that the accuracy and efficiency of the proposed mechanism for the generation of the gauge field background should be verified in numerical lattice simulations. Research in this direction is ongoing (see e.g.\ Refs.~\cite{Adshead:2015pva,Cuissa:2018oiw} for recent developments) and we hope that this paper will arouse interest in these questions.

In the inherently non-abelian regime of our analysis, we reproduce the phenomenology of CNI: a single tensor (i.e.\ helicity 2) gauge field mode is tachyonically enhanced. A linear coupling of this mode to one helicity of the tensor metric perturbation sources a strongly enhanced, chiral gravitational wave background.  We point out that the instability observed in the scalar sector in Ref.~\cite{Dimastrogiovanni:2012ew} only occurs in a tiny part of the parameter space (achievable only for large gauge couplings when the transition to the non-abelian regime can occur for very low inflaton velocities close to the theoretical threshold value), and is in particular absent for the parameter point studied in this paper. Large contributions to the scalar power spectrum are however sourced by non-linear contributions. We propose a method to  estimate this non-linear contribution based on a generalization of the procedure described in~\cite{Linde:2012bt} and it should be seen as an order-of-magnitude estimate only. Interestingly, we have found that the non-linear corrections are expected to quickly dominate over vacuum and linear contributions inducing a strong enhancement in the scalar power spectrum at small scales. Such an enhancement may lead to a rich phenomenology (e.g.\ PBH formation~\cite{Linde:2012bt, Domcke:2017fix,Garcia-Bellido:2016dkw}, $\mu$-distortions~\cite{Meerburg:2012id}, etc).

The transition between an abelian limit at early times and a non-abelian regime at later times provides a natural way to obtain results in agreement with all CMB constraints while obtaining a phenomenologically very interesting enhanced scalar and tensor power spectra at small scales. The key parameter driving this transition is the instability parameter $\xi$ (see Eq.~\eqref{eq:rev_xi}), measuring the velocity of the inflaton. In single-field slow-roll inflation models, this parameter is small at early times but increases over the course of inflation. This parameter controls the tachyonic enhancement of the gauge field perturbations, as well as the amplitude of the homogeneous background solution. For the parameter point studied in this paper, the predicted enhancement of the scalar and tensor power spectrum is out of reach of current and upcoming experiments. A comprehensive exploration of the parameter space of these models, in particular regarding the prospects of detecting the gravitational wave signal with LIGO, LISA or the Einstein Telescope, is left for future work.

For the numerical study in this paper, we focused on a parameter point for which the gauge friction on the inflaton motion was small at the point of transition between the abelian and non-abelian regime. We leave a study of the opposite regime, corresponding to the usual so-called magnetic drift regime of CNI to future work. In this case, the dynamics at the transition point change more violently, requiring a more sophisticated modeling of this transition regime. In a similar spirit, the transition region itself deserves more attention. In this paper we simply refrained from making any statements very close to the transition region, to avoid sensitivity to the precise modeling procedure of this transition. Clearly, this calls for further improvement. Finally, the restriction to a linearized system of perturbations in the bulk of our analysis limits our ability to explore parameters for which the transition to the non-abelian regime occurs significantly before the end of inflation. This constraint is already problematic for the parameter point studied here, and likely becomes even more relevant in the parameter space most promising for direct gravitational wave searches, calling for a full non-linear treatment of the system. 

Our analysis is a further step towards the embedding of axion inflation models in a fully realistic particle physics description of the early Universe. In particular, our unified framework for abelian and non-abelian couplings may be applied to more complex gauge groups such as $\mathrm{SU}(2) \times \mathrm{U}(1)$ and eventually the full Standard Model of Particle Physics.

\vspace{1cm}
\subsubsection*{Acknowledgements}

We thank P.~Adshead, J.~Beltran, D.~Figueroa, A.~Maleknejad, M.~Peloso and S.~Rodrigues Sandner for valuable comments and discussions at various stages of this work. M.P. acknowledges the support of the Spanish MINECOs ``Centro de Excelencia Severo Ocho'' Programme under grant SEV-2012-0249. This project has received funding from the European Unions Horizon 2020 research and innovation programme under the Marie Sk\l{}odowska-Curie grant agreement No 713366.

\newpage
\appendix
\addtocontents{toc}{\protect\setcounter{tocdepth}{1}}

\section{Full equations of motion \label{app:fulleom}}

\noindent In this Appendix we report the full computation of the linearized equations of motion for the gauge fields and metric perturbations, including the non-dynamical metric degrees of freedom that were discarded in Sec.~\ref{sec:linearized} of the main text, labeled lapse $N$ and shift $N^i$ in the ADM formalism~\cite{Arnowitt:1962hi}. We found some minor discrepancies with respect to the results in \cite{Adshead:2013nka} that however do not affect their results.

We start from the expressions for the action and the metric reported in Sec.~\ref{sec:linearized}. The metric in matrix form is
\begin{equation}
g_{\mu \nu} =
\begin{pmatrix}
-N^2 + h_{ij} N^i N^j & h_{ji} N^i \\
h_{ij} N^j & h_{ij}
\end{pmatrix}\,,
\end{equation}
and it can be easily established that the components of the inverse matrix are
\begin{equation}
g^{00} = - \frac{1}{N^2} \quad g^{0i} = \frac{N^i}{N^2} \quad g^{ij} = h^{ij} - \frac{N^i N^j}{N^2} \,.
\end{equation}
Moreover, it is useful to notice that $\sqrt{-g} = N \sqrt{h}$,\footnote{The four-dimensional metric is always denoted by the letter $g$, while the three-dimensional metric is denoted by $h$. $h_{ij}/h^{ij}$ is used in order to raise/lower indices in the three-dimensional space.} where $h = \text{det} \left(h_{ij}\right)$ and $h^{ij} = \left(h_{ij}\right)^{-1}$. Using these definitions, the Einstein-Hilbert term takes the standard form derived in the original paper on the ADM formalism~\cite{Arnowitt:1962hi}
\begin{equation}
\label{eq:FullEH}
\mathcal{L}_{\text{EH}} = \sqrt{h} \left[N R^{(3)} + \frac{1}{N} \left(E^{ij} E_{ij} - E^2\right) \right]\,,
\end{equation}
where $R^{(3)}$ is the spatial curvature (computed using $h_{ij}$), while\footnote{The covariant derivative $\nabla$ is computed using $h_{ij}$.}
\begin{equation}
E_{ij} = \frac{1}{2} \left(h'_{ij} - \nabla_i N_j - \nabla_j N_i\right) \,, \quad E = E^i_i\,.
\end{equation}

\noindent Expressing the Lagrangian in terms of its individual contributions as in Eq.~\eqref{eq:action}, we obtain for the contribution ${\cal L}_\phi$,
\begin{equation}
\mathcal{L}_{\phi} = \sqrt{h} \, \left[\frac{1}{2 N} \left(\phi' - N^i \partial_i \phi\right)^2 - \frac{N}{2} h^{ij} \partial_i \phi \partial_j \phi - N V(\phi)\right]\,,
\end{equation}
while, after using some simple algebra, we split the Yang--Mills Lagrangian as follows
\begin{align}
\mathcal{L}_{\text{YM}} & = \frac{\sqrt{h}}{2 N} \left[h^{ij} \left(F^a_{0i} F^a_{0j} + 2 F^a_{0i} F^a_{jl} N^l + F^a_{ik} F^a_{jl} N^k N^l\right)\right] - \frac{\sqrt{h} N}{4} h^{ij} h^{kl} F^a_{ik} F^a_{jl}  \nonumber \\
& \equiv \mathcal{L}_{\text{YM},1} + \mathcal{L}_{\text{YM},2} + \mathcal{L}_{\text{YM},3} + \mathcal{L}_{\text{YM},4} \,.
\label{eq:LagrangianYMF}
\end{align}

Expanding each of the four terms in Eq.~\eqref{eq:LagrangianYMF} we get:
\begin{align}
\mathcal{L}_{\text{YM},1} & = \frac{\sqrt{h}}{2 N} h^{ij} F^a_{0i} F^a_{0j} = \nonumber \\
& = \frac{\sqrt{h}}{2 N} h^{ij} \left[\partial_0 A^a_i \partial_0 A^a_j + \partial_i A^a_0 \partial_j A^a_0 - \partial_0 A^a_i \partial_j A^a_0  - \partial_0 A^a_j \partial_i A^a_0 + \right.\nonumber \\
& + e \varepsilon^{abc} \left(\partial_0 A^a_i A^b_0 A^c_j - \partial_i A^a_0 A^b_0 A^c_j + \partial_0 A^a_j A^b_0 A^c_i - \partial_j A^a_0 A^b_0 A^c_i\right) \nonumber \\
& \left.+ e^2 \left(A^b_0 A^b_0 A^c_i A^c_j - A^b_0 A^b_j A^c_0 A^c_i\right)\right] \,,
\end{align}
\begin{align}
\mathcal{L}_{\text{YM},2} & = \frac{\sqrt{h}}{N} h^{ij} F^a_{0i} F^a_{jl} N^l = \nonumber \\
& = \frac{\sqrt{h}}{N} h^{ij} \left[\partial_0 A^a_i \partial_j A^a_l - \partial_0 A^a_i \partial_l A^a_j - \partial_i A^a_0 \partial_j A^a_l + \partial_i A^a_0 \partial_l A^a_j +\right. \nonumber \\
& + e \varepsilon^{abc} \left(\partial_0 A^a_i A^b_j A^c_l - \partial_i A^a_0 A^b_j A^c_l + \partial_j A^a_l A^b_0 A^c_i - \partial_l A^a_j A^b_0 A^c_i\right) + \nonumber \\
& \left. + e^2 \left(A^b_0 A^b_j A^c_i A^c_l - A^b_0 A^b_l A^c_i A^c_j\right)\right] N^l \,,
\end{align}

\begin{align}
\mathcal{L}_{\text{YM},3} & = \frac{\sqrt{h}}{2 N} h^{ij} F^a_{ik} F^a_{jl} N^k N^l = \nonumber \\
& = \frac{\sqrt{h}}{N} h^{ij} \left[\partial_i A^a_k \partial_j A^a_l - \partial_i A^a_k \partial_l A^a_j - \partial_k A^a_i \partial_j A^a_l + \partial_k A^a_i \partial_l A^a_j +\right. \nonumber \\
& + e \varepsilon^{abc} \left(\partial_i A^a_k A^b_j A^c_l - \partial_k A^a_i A^b_j A^c_l + \partial_j A^a_l A^b_i A^c_k - \partial_l A^a_j A^b_i A^c_k\right) + \nonumber \\
& \left. + e^2 \left(A^b_i A^b_j A^c_k A^c_l - A^b_i A^b_l A^c_k A^c_j\right)\right] N^k N^l \,,
\end{align}

\begin{align}
\mathcal{L}_{\text{YM},4} & = - \frac{\sqrt{h}}{4 N} h^{ij} h^{kl} F^a_{ik} F^a_{jl} = \nonumber \\
& = - \frac{\sqrt{h} N}{4} h^{ij} h^{kl} \left[\partial_i A^a_k \partial_j A^a_l - \partial_i A^a_k \partial_l A^a_j - \partial_k A^a_i \partial_j A^a_l + \partial_k A^a_i \partial_l A^a_j +\right. \nonumber \\
& + e \varepsilon^{abc} \left(\partial_i A^a_k A^b_j A^c_l - \partial_k A^a_i A^b_j A^c_l + \partial_j A^a_l A^b_i A^c_k - \partial_l A^a_j A^b_i A^c_k\right) + \nonumber \\
& \left. + e^2 \left(A^b_i A^b_j A^c_k A^c_l - A^b_i A^b_l A^c_k A^c_j\right)\right]\,.
\end{align}

The Chern-Simons term is
\begin{align}
\mathcal{L}_{\text{CS}} & = - \frac{\alpha}{8 \Lambda} \phi F^a_{\mu \nu} F^{a}_{\rho \sigma} \varepsilon^{\mu \nu \rho \lambda} = -\frac{\alpha}{2 \Lambda} \phi \left[2 \varepsilon^{0ijk} \partial_0 A^a_i \partial_j A^a_k - \right. \nonumber \\
& \left. - 2 \varepsilon^{0ijk} \partial_i A^a_0 \partial_j A^a_k  + e f^{abc} \varepsilon^{0ijk} \left(\partial_0 A^a_i A^b_j A^c_k - \partial_i A^a_0 A^b_j A^c_k + 2 \partial_i A^a_j A^b_0 A^c_k\right)  \right]\,.
\end{align}

Given the expressions for each term of the action, we can compute the expansion at  the quadratic order in the field fluctuations. We expand the three-dimensional metric to second order as follows:
\begin{equation}
h_{ij} = a^2 \left(\delta_{ij} + \gamma_{ij} + \frac{\gamma_{ik} \gamma_{kj}}{2} \right) \,,
\end{equation}
where $\gamma_{ij}$ is a transverse and traceless quantum fluctuation: $\gamma_{ii} = \partial_i \gamma_{ij} = 0$.\footnote{Notice that spatial indices are raised and lowered by $h_{ij}/h^{ij}$, while gauge indices are raised and lowered with just a $\delta_{ij}$. Everytime we write a gauge field fluctuation with two lowered indices we imply that the first one is the gauge index.} We also notice that, using the standard expression $\delta g = g g^{ij} \delta g_{ij}$, where $g = \text{det}(g)$ the determinant $h = \text{det}(h_{ij})$ can be expanded as
\begin{equation}
\frac{h}{a^2} = \text{det} \left(e^{\text{Tr} \left[\ln (h_{ij}/a^2)\right]}\right) = \text{det} \left(e^{\text{Tr} \left[ \gamma_{ij} + \frac{\gamma_{ik}\gamma_{kj}}{2} - \frac{\gamma_{ik}\gamma_{kj}}{2} \right]}\right) = 1 \,.
\end{equation}
We also notice that in spatially flat gauge the Christoffel symbols vanish ${\Gamma^i}_{jk} = 0$ and then $\nabla_i = \partial_i$. Finally we expand the lapse and shift around the FRW background\footnote{{Formally, lapse and shift have to be expanded up to second order. It turns out that the second order can be eliminated using the background equations of motion (see~\cite{Adshead:2013nka}), and does not affect the subsequent results.}}
\begin{equation}
N = a \left(1+ \delta N\right) \,, \qquad N^i \equiv \delta N^i\,.
\end{equation}

We proceed with the computation of the quadratic action and we report the results term by term to make it easier tracking back the various terms. We start from the Einstein-Hilbert component
\begin{align}
S_{\rm EH} & = \int \textrm {d}^4 x \, \frac{a^2}{2} \left[  \frac{ \gamma_{ij} \partial_l \partial_l \gamma_{ij} }{4} - 6 \mathcal{H}^2 + \frac{\gamma^\prime_{ij} \gamma^\prime_{ij}}{4}  + \right. \nonumber \\ 
& \left. + \frac{4 \mathcal{H}}{a^2} (1 - \delta N) \partial_{i} N_{i} + 6 \mathcal{H}^2 \delta N (1- \delta N) - \left( \frac{\partial_{i} N_{i}}{a^2} \right)^2 + \frac{\partial_{(i} N_{j)} \partial_{(i} N_{j)}}{a^4} \right] \,,
\end{align}
where the contribution from the three-dimensional curvature $R^{(3)}$ is just the first term $a^2/8 \, \gamma_{ij} \partial_l \partial_l \gamma_{ij}$. Concerning the scalar field, the second-order action takes the form:
\begin{align}
\delta^2 S_\phi & = \int \textrm{d}^4 x \, \left[\frac{a^2}{2} \left( \left(\delta \phi'\right)^2 - \partial_i \delta \phi \partial_i \delta \phi - a^2 V'' \left(\delta \phi \right)^2 \right) - \right. \nonumber \\
& \left.- \frac{a^2}{2} \left(\frac{2}{a^2} \langle \phi' \rangle N_i \partial_i \delta \phi + 2 \langle \phi' \rangle\delta N \delta \phi' + \langle \phi' \rangle^2 \left(\delta N\right)^2 - 2 a^2 V' \delta \phi \delta N\right)\right]\,.
\end{align}

Finally, we report the various contributions to the Yang--Mills $\delta^2 \mathcal{S}_{\rm YM}$ and Chern-Simon $\delta^2 S_{\text{CS}}$ quadratic actions
\begin{align}
\delta^2 \mathcal{S}_{\text{YM},1}  = \int \textrm{d}^4 x\, & \left[\frac{1}{4} \left(f'\right)^2 \gamma^{jk} \gamma^{kj} - f' \gamma^{aj} \partial_0 \delta A^{(a}_{j)} - \right. \nonumber \\
& -\frac{1}{2} \delta A^a_i \partial_0\partial_0 \delta A^a_i  - \frac{1}{2} \delta A^a_0 \partial_i \partial_i  \delta A^a_0 +  \delta A^a_0 \partial_0 \partial_i \delta A^a_i + \nonumber \\
& + e \varepsilon^{abi}  f \partial_0 \delta A^a_i \delta A^b_0 + e \varepsilon^{jbc} f' \delta A^b_0 \delta A^c_j - e \varepsilon^{abi}  f \partial_i \delta A^a_0 \delta A^b_0 + e^2 f^2 \delta A^b_0 \delta A^b_0 \nonumber \\
& \left.+ \frac{3}{2} \left(f'\right)^2 \delta N^2 - f' \partial_0 \delta A^i_i \delta N + f' \partial_j \delta A^j_0 \delta N \right] \,,
\end{align}
\begin{align}
\delta^2 \mathcal{S}_{\text{YM},2} & = \int \textrm{d}^4 x \, \frac{1}{a^2} \left[f' \partial_a \delta A^a_l N_l -  f' \partial_l \delta A^a_a N_l + \right. \nonumber \\
& \left. + e f^2 \varepsilon_{ail} \partial_0 \delta A^a_i N_l + e f f' \varepsilon_{ibl} \delta A^b_i N_l - e \varepsilon_{ail} f^2 \partial_i \delta A^a_0 N_l - 2 e^2 f^3 \delta A^l_0 N_l \right] \,, \\
\delta^2 \mathcal{S}_{\text{YM},3} & = \int d^4 x\, e^2 f^4 N^k N^k \,,
\end{align}
\begin{align}
\delta^2 \mathcal{S}_{\text{YM},4} & =  \int \textrm{d}^4 x \, \left[-\frac{3}{4} e^2 f^4 \delta N^2 -\frac{e^2 f^4}{4} \gamma^{jk} \gamma^{kj} -\right. \nonumber \\
& - e f^2 \varepsilon^{abc} \gamma^{ij} \left( \delta^b_{(i} \partial_{j)} \delta A^a_{c} - \delta^b_{(i} \partial_{c} \delta A^a_{j)} \right) + e^2 f^3 \gamma^{bi} \delta A^{(b}_{i)} - \nonumber \\
& - e f^2 \varepsilon^{abc}  \partial_{[b} \delta A^a_{c]} \delta N - 2 e^2 f^3 \delta^b_i \delta A^b_i \delta N + \nonumber \\
& + \frac{1}{2} \delta A^a_j \partial_i \partial_i \delta A^a_j + \frac{1}{2} (\partial_i \delta A^a_i)^2 - e^2 f^2 \left( \delta A^a_a \delta A^b_b + \frac{1}{2} \delta A^b_i \delta A^b_i - \frac{1}{2} \delta A^b_i \delta A^i_b \right) - \nonumber \\
& \left. - e f \varepsilon^{abc} \left( \delta^b_i \partial_i \delta A^a_k \delta A^c_k + \delta^c_k \partial_i \delta A^a_k \delta A^b_i \right) \right] \,,
\end{align}
\begin{align}
\delta^2 \mathcal{S}_{\rm CS} & = \int \textrm{d}^4 x \, \left[-\frac{\alpha}{2 \Lambda} \langle \phi \rangle \left(2 \varepsilon^{0ijk} \partial_0 \delta A^a_i \partial_j \delta A^a_k- 2 \varepsilon^{0ijk} \partial_i \delta A^a_0 \partial_j \delta A^a_k+ \right.\right. \nonumber \\
& + 2 e f \partial_0 \delta A^a_a \delta A^k_k - 2ef \partial_0 \delta A^c_k \delta A^k_c+ e f' \delta A^b_b \delta A^c_c - e f' \delta A^b_j \delta A^j_b - \nonumber \\
& \left. - 2 e f \partial_i \delta A^i_0 \delta A^b_b + 2 e f \partial_i \delta A^j_0 \delta A^i_j + 2 e f \partial_i \delta A^i_j \delta A^j_0 - 2 e f \partial_i \delta A^j_j \delta A^i_0 \right) - \nonumber \\
& - \frac{\alpha}{2 \Lambda} \delta \phi \left(2 f' \varepsilon^{0ijk} \delta^a_i \partial_j \delta A^a_k + 4 e f f' \delta A^b_b + 2 e f^2 \partial_0 \delta A^a_a- 2 e f^2 \partial_i \delta A^i_0\right)\Big] \,.
\end{align}

From these expressions we can infer the full equations of motion and the constraints, which we report below.

\noindent - \textit{Gauss' law:}
\begin{align}
- \frac{f'}{a^2} \partial_j \gamma_{gj} - f' \partial_g \delta N - \partial_i \partial_i \delta A_{g0} + \partial_0 \partial_i \delta A_{gi}
+ e \varepsilon_{agi} f \partial_0 \delta A_{ai} + e \varepsilon_{jgc} f' \delta A_{cj} &+ \nonumber \\
+ 2 e f \varepsilon_{gbi} \partial_i \delta A_{b0} + 2 e^2 f^2 \delta A_{g0} 
+ \frac{e f^2}{a^2} \varepsilon_{gil} \partial_i N_l - \frac{2 e^2 f^3}{a^2} N_g  - \frac{\alpha e f^2}{\Lambda} \partial_g \delta \phi  &= 0 \,. 
\end{align}

\noindent - $\delta N$-\textit{constraint:}
\begin{align}
- 6 a^2 \mathcal{H}^2 \delta N - 2 \mathcal{H} \partial_i N_i - a^2 \langle\phi'\rangle \delta\phi' + a^2 \langle\phi'\rangle^2 \delta N - a^4 V' \delta \phi & + \nonumber \\
 + 3 \left(f'\right)^2 \delta N - f' \partial_0 \delta A^i_i + f' \partial_j \delta A^j_0 - \frac{3}{2} e^2 f^4 \delta N- e f^2 \varepsilon^{aik} \partial_{[i} \delta A^a_{k]} - 2 e^2 f^3 \delta A^i_i & = 0 \, . 
\end{align}

\noindent - $N_l$-\textit{constraint:}
\begin{align}
2 \mathcal{H} \partial_l \delta N + \frac{1}{2 a^2} \left[ \partial_l \partial_i N_i - \partial_i \partial_i N_l \right] -\langle\phi'\rangle \partial_i \delta \phi  + \frac{ 2 e^2 f^4  N_l } {a^4} & + \nonumber \\
+ \frac{1}{a^2} \left[  f' \left( \partial_i \delta A^i_{l} -  \partial_l \delta A^i_{i} \right) + e f^2 \varepsilon_{lai} \partial_0 \delta A^a_{i} - e f f' \varepsilon_{lbi} \delta A^b_{i} - e f^2  \varepsilon_{lai} \partial_i \delta A^a_{0} - 2 e^2 f^3 \delta A_{l0}  \right] & = 0 \,.
\end{align}

\noindent - \textit{Gauge field equation of motion}:
\begin{align}
& -\partial_0 \partial_0 \delta A^a_i - 2 e \varepsilon^{abi} f' \delta A^b_0 + \partial_l \partial_l \delta A^a_i - \partial_i \partial_j \delta A^a_j - \nonumber \\
& - e^2 f^2 (2 \delta^a_i \delta A^b_b  +  \delta A^a_i -\delta A^i_a  ) + 2 e f \varepsilon^{abc} \partial_b \delta A^c_i + e f \varepsilon^{abc} \partial_i \delta A^b_c + e f \varepsilon^{abi} \partial_l \delta A^b_l + \nonumber \\
& + \langle \phi^\prime \rangle  \frac{\alpha}{\Lambda}  \left( \varepsilon^{ijk}  \partial_j \delta A^a_k +  \delta^a_i e f \delta A^b_b - ef \delta A^i_a \right)  + \frac{\alpha}{\Lambda} \left[  f^\prime \varepsilon^{aji}\partial_j \delta \phi  + e f^2 \delta^a_i \delta \phi^\prime  \right] + \nonumber \\
& + f^{\prime \prime} \gamma^{a}_i + f^{\prime} \gamma^{a \ \prime}_i + e f^2  \varepsilon^{ajk} \partial_k \gamma_{ij} + e^2 f^3 \gamma^a_{i} + \nonumber \\ 
& - f^{\prime } \partial_i  N^a +  f^{\prime } \delta^a_i \partial_l  N^l - e \varepsilon^{ail} (3 f f^{\prime} N^l + f^2 N^{l \, \prime} ) + \nonumber \\
& + \delta^a_i (f^{\prime \prime} \delta N + f^{\prime } \delta N^\prime ) + e f^2 \varepsilon^{abi} \partial_b \delta N - 2 e^2 f^3 \delta^a_i \delta N \,  = 0 . 
\end{align}

\noindent - \textit{Inflaton equation of motion:}
\begin{align}
\dAlemb \delta \phi - 2 a' a^{-3} \delta \phi' + a^{-2} \partial_0( \langle \phi' \rangle \delta N) +  2 \frac{H}{a} \langle \phi' \rangle \delta N  + \frac{\langle \phi' \rangle}{a^4} \partial_i N_i - V_{,\phi \phi} \delta \phi - V_{, \phi} \delta N & - \nonumber \\
- \frac{\alpha}{2 \Lambda a^4} \left[ 2 f' \varepsilon^{ijk} \partial_j \delta A^i_k + 4 e f f' \delta A^b_b + 2 e f^2 \partial_0 \delta A^a_a - 2 e f^2 \partial_i \delta A^i_0 \right] & = 0 \, ,
\end{align}
where the $\dAlemb$-operator is expressed in co-moving coordinates.

\noindent - \textit{Gravitational wave equation of motion:}
\begin{align}
&\frac{a}{4}\left[(a \gamma_{ij})'' + \left(-\partial_l\partial_l - \frac{a''}{a} \right) (a \gamma_{ij})\right] = \nonumber \\
&+ \frac{1}{2 a}( f^{\prime \ 2}- e^2 f^4 ) (a \gamma_{ij} ) -  f^\prime \partial_0 \delta A^{(i}_{j)}  + f^\prime \partial_{(i} \delta A^{j)}_0 - 2e f^2 \varepsilon^{a(ic}  \partial_{[j)} \delta A^a_{c]}  + e^2  f^3  \delta A^{(i}_{j)}  \, . 
\end{align}
\section{Basis vectors for the gauge fields in the helicity basis \label{app:basis} }

In this appendix we collect the explicit forms of the basis vectors derived in Sec.~\ref{GaugeAndBasis}. The six physical degrees of freedom are spanned by the helicity states $\hat e_\lambda$,
\begin{equation*}
 (\hat e_{01})^{b}{}_{\mu} = \frac{1}{\sqrt{2}}
 \begin{pmatrix}
 0 & 0 & 0 & 0 \\ 0 & 0 & 1 & 0 \\  0 &0 & 0 & 1
 \end{pmatrix}  \,,  \quad 
 (\hat e_{02})^{b}{}_{\mu} = \frac{1}{\sqrt{2 + 4 y_k(x)^2}}
 \begin{pmatrix}
  0 & 2 i y_k(x) & 0 & 0 \\  0 & 0 & 0 & -1 \\ 0 & 0 & 1 & 0 
 \end{pmatrix} \,,\nonumber \\
 \end{equation*}
 
 \begin{align}
 (\hat  e_{-1})^{b}{}_{\mu} & = \frac{1}{\sqrt{2 + 4 y_k(x) + 4 y_k(x)^2}}
 \begin{pmatrix}
  0 & 0 & i(1 + y_k(x)) & 1 + y_k(x) \\ 0 &  i y_k(x) & 0 & 0 \\ 0 &  y_k(x) & 0 & 0 
 \end{pmatrix} \,,
   \nonumber \\
 (\hat e_{+1})^{b}{}_{\mu} &= \frac{1}{\sqrt{2 - 4 y_k(x) + 4 y_k(x)^2}}
 \begin{pmatrix}
0 &   0 & -i(1 - y_k(x))  & 1 - y_k(x) \\ 0 & i y_k(x) & 0 & 0 \\ 0 & - y_k(x) & 0 & 0
 \end{pmatrix} \,, 
 \nonumber \\
 &(\hat e_{-2})^{b}{}_{\mu} = \frac{1}{2}
 \begin{pmatrix}
   0 & 0 & 0 & 0 \\ 0 &  0 & i & 1 \\ 0 & 0 & 1 & -i 
 \end{pmatrix} \,,
 \quad
 (\hat e_{+2})^{b}{}_{\mu} = \frac{1}{2}\begin{pmatrix}
  0 & 0 & 0 & 0 \\ 0 & 0 & -i & 1 \\ 0 & 0 & 1 & i 
 \end{pmatrix}  \,,
 \label{eq:basis}
\end{align}
with $y_k(x) \equiv e f(\tau)/k$. Note that in any background which is a fixed point of Eq.~\eqref{eq:symmetry3} (e.g.\ if the background follows the $c_2$-solution), we can drop the index $k$ on $y_k(x)$ as this quantity becomes a function of $(- k \tau)$ only: $y_k(x) = e k^{-1} f(- k^{-1} x) = e f(-x) \equiv y(x)$.

The gauge degrees of freedom (simultaneously satisfying Eqs.~\eqref{eq:gaugebasis0} and \eqref{eq:helicity_operator}) read
\begin{align*}
\left(\hat g_{-1}\right)^{b}{}_{\mu} & =\left(\begin{array}{cccc}
0 & 0 & -y_k(x) & i\,y_k(x)\\
i\,\frac{\textrm{d}}{\textrm{d}x} & (1+y_k(x)) & 0 & 0\\
\frac{\textrm{d}}{\textrm{d}x} & -i(1+y_k(x)) & 0 & 0
\end{array}\right)\,,\\
\left(\hat g_{0}\right)^{b}{}_{\mu} & =\left(\begin{array}{cccc}
i\,\frac{\textrm{d}}{\textrm{d}x} & 1 & 0 & 0\\
0 & 0 & 0 & -i\,y_k(x)\\
0 & 0 & i\,y_k(x) & 0
\end{array}\right)\,,\\
\left(\hat g_{+1}\right)^{b}{}_{\mu} & =\left(\begin{array}{cccc}
0 & 0 & y_k(x) & i\,y_k(x)\\
i\,\frac{\textrm{d}}{\textrm{d}x} & (1-y_k(x)) & 0 & 0\\
-\frac{\textrm{d}}{\textrm{d}x} & i(1-y_k(x) & 0 & 0
\end{array}\right) \,,
\end{align*}
where the entry $\textrm{d}/\textrm{d}x$ indicates that the corresponding coefficient $w_\lambda^{(g)}(x)$ is replaced by $\frac{\textrm{d}}{\textrm{d}x} w_\lambda^{(g)}(x)$. 

Finally the  basis vectors encoding the constraint equations (Gauss' law) are given by
\begin{align*}
\left(\hat f_{-1}\right)^{b}{}_{\mu} & =\frac{1}{\sqrt{2}}\left(\begin{array}{cccc}
0 & 0 & 0 & 0\\
1 & 0 & 0 & 0\\
-i & 0 & 0 & 0
\end{array}\right)\,,\\
\left(\hat f_{0}\right)^{b}{}_{\mu} & =\frac{1}{\sqrt{2}}\left(\begin{array}{cccc}
1 & 0 & 0 & 0\\
0 & 0 & 0 & 0\\
0 & 0 & 0 & 0
\end{array}\right)\,,\\
\left(\hat f_{+1}\right)^{b}{}_{\mu} & =\frac{1}{\sqrt{2}}\left(\begin{array}{cccc}
0 & 0 & 0 & 0\\
1 & 0 & 0 & 0\\
+i & 0 & 0 & 0
\end{array}\right)\,.
\end{align*}

\section{ Computation of the variance of \texorpdfstring{$F \tilde F$}{FFdual} \label{app:variance_computation}}
In this appendix we are interested in computing the variation of $\langle F \tilde F \rangle$ with respect to its background value, which we employ to estimate the non-linear contribution to the scalar power spectrum in Sec.~\ref{subsec:PS}. In general this can be expressed as
\begin{equation}
 \delta (\langle F \tilde F \rangle ) = [F \tilde F - \langle F \tilde F \rangle ]_{\delta \phi = 0} + \frac{\partial \langle F \tilde F \rangle }{\partial \dot \phi} \dot{\delta \phi} \equiv \delta_{\vec{EB}} + \frac{\partial \langle F \tilde F \rangle }{\partial \dot \phi} \dot{\delta \phi} \,,
\end{equation}
where we are using the sign convention $\dot \phi > 0$, $\langle F \tilde F \rangle >  0$ (see \cite{Linde:2012bt}). Notice that this corresponds to considering the variation of $\langle F \tilde F \rangle$ due to variations of $\delta \phi$ (second term) plus the variations of $\langle F \tilde F \rangle$ due to the variations of the gauge fields. In order to perform the computation of the latter it is useful to introduce the electric and magnetic fields (in conformal time)
\begin{equation}
  F^{a}_{0i} \equiv - a^2(\tau) E_{i}^a \, , \qquad \qquad F^{a}_{ij} \equiv \varepsilon^{ijk} a^2(\tau) B_{k}^a \, .
\end{equation}
With these definitions we can trivially show that:
\begin{eqnarray}
  \label{eq:FF_EB}
  F^a_{\mu \nu} F^a_{\rho \sigma} g^{\mu \rho} g^{\nu \sigma } & = & - \frac{2}{a^4} F^a_{0 i} F^a_{0 i} + \frac{1}{a^4} F^a_{i j} F^a_{i j} = - 2 \left[ \left( \vec{E}^a \right)^2 - \left( \vec{B}^a \right)^2 \right] \, . \\
  \label{eq:FFdual_EB}
  F^a_{\mu \nu} \tilde{F}^a_{\rho \sigma} & = & \frac{4}{2 \sqrt{-g}} F^a_{0 i} F^a_{j k} \varepsilon^{ijk} = - 2 E^a_{i} B_{l}^a \varepsilon^{ijk}  \varepsilon^{ljk} = - 4 \vec{E}^a \cdot \vec{B}^a \, .
\end{eqnarray}
We can then proceed by computing the expressions of $E^a_{i}$ and $B^a_{i}$ (neglecting terms depending on $A_0^a$) up to second order:
\begin{eqnarray}
  E^a_{i} & = & - \frac{1}{a^2} \left[\delta^a_i \partial_0 f + \partial_0 \delta A^a_i  \right]\,, \\
  B^a_{i} & = & \frac{1}{ a^2} \left[ \varepsilon_{ijk} \partial_j \delta A^a_k + e f^2 \delta^a_i  - e f \delta A^i_a  +e \frac{\varepsilon_{ijk} \varepsilon^{abc}}{2} \delta A^b_j  \delta A^c_k  \right]\, , 
\end{eqnarray}
where we have neglected terms proportional to the helicity 0 quantity $\delta A^a_a$ since they do not feature the exponential enhancement present in the $e_{+2}$ mode. Notice that setting $e = 0$ we can easily recover the usual abelian terms.\\

\noindent With this notation it is now trivial to check that the expectation value of $F \tilde{F}$ can be expressed as:
\begin{equation}
  \langle F^a_{\mu \nu} \tilde{F}^{a \, \mu \nu} \rangle = - 4 \langle \vec{E}^a \cdot \vec{B}^{a } \rangle \equiv  - 4 T_0 \, ,
\end{equation}
and the variance of $F \tilde{F}$ can be expressed as
\begin{equation}
  \delta_{\vec{EB}}^2 = \langle ( F^a \tilde{F}^a )^2 \rangle  -  \langle  F^a \tilde{F}^a  \rangle ^2 = 16\left[  \langle  E_i^a E_j^b \rangle \langle B_i^a B_j^b \rangle + \langle E_i^a B_j^b \rangle \langle B_i^a E_j^b \rangle \right] \, ,
\end{equation}
which has exactly the same shape as in the abelian limit (for comparison see appendix A of~\cite{Linde:2012bt}). At this point it is useful to introduce 
\begin{eqnarray}
  T_1  \equiv \left \langle E^a_{i}  E^b_{j} \right \rangle \left \langle B^a_{i}  B^b_{j} \right \rangle \, , \qquad \qquad T_2  \equiv \left \langle E^a_{i} B^b_{j} \right \rangle \left \langle B^a_{i}  E^b_{j} \right \rangle \, ,
\end{eqnarray}
so that the two contributions can be computed independently. Before substituting the explicit expressions of $E$ and $B$ it is important to notice that (i) in order to compute the variance we only need terms that are exactly quadratic in the fluctuations and (ii)  the base vectors of the $\pm 2$ helicity modes are traceless in all bases (implying $\delta^a_i \left[e_{\pm2} (\vec{k})\right]^a_i = 0 $ for all $\vec{k}$). 
We can now  directly compute 
\begin{align}
  \left \langle E^a_{i}  E^b_{j} \right \rangle  =& \frac{1}{a^4} \left \langle \partial_0 \delta A^a_i \partial_0 \delta A^b_j \right \rangle \, , \\
  \left \langle B^a_{i}  B^b_{j} \right \rangle =& \frac{1}{a^4} \left \langle  \varepsilon_{ilk} \varepsilon_{jnm} \partial_l   \delta A^a_k  \partial_n \delta A^b_m + e^2 f^2 \delta A^i_a \delta A^j_b  - e f \varepsilon_{ilk} \partial_l \delta A^k_a \delta A^b_j  - e f \varepsilon_{jnm}  \delta A^i_a \partial_n \delta A^b_m \right \rangle \nonumber \\
  \equiv&  \left \langle B^a_{i}  B^b_{j} \right \rangle_1 +  \left \langle B^a_{i}  B^b_{j} \right \rangle_2 +  \left \langle B^a_{i}  B^b_{j} \right \rangle_3 +  \left \langle B^a_{i}  B^b_{j} \right \rangle_4 , \\
  \left \langle E^a_{i} B^b_{j} \right \rangle =& - \frac{1}{a^4} \left \langle \partial_0 \delta A^a_i \varepsilon_{jnm} \partial_n \delta A^b_m  - e f  \partial_0 \delta A^a_i \delta A^j_b +  e \delta^a_i \partial_0 f \frac{\varepsilon_{jnm} \varepsilon^{bdh}}{2} \delta A^d_n  \delta A^h_m \right \rangle \nonumber \\
  \equiv&  \left \langle E^a_{i}  B^b_{j} \right \rangle_1 +  \left \langle E^a_{i}  B^b_{j} \right \rangle_2 +  \left \langle E^a_{i}  B^b_{j} \right \rangle_3 \, .
\end{align}
Notice that again setting $e = 0$ our expressions reduce to the abelian case. 

At this point we can expand $ \delta A^a_\nu(\tau, \vec x)$ in terms of its Fourier modes as in Eq.~\eqref{eq:Fourier_A}. Notice that in general the basis vectors satisfy 
\begin{equation}
   \left[ e_{\lambda, \nu}^a (\hat{k}) \right]^* = e_{\lambda, \nu}^a (-\hat{k}) \, , \qquad \qquad i \varepsilon_{ijl} k_j e_{\lambda, l}^a (\hat{k}) = \text{sgn}(\lambda) |\vec{k}| e_{\lambda, i}^a (\hat{k}) \, .
\end{equation}
Moreover, the helicity $\pm2$ vectors are symmetric, \emph{i.e.,} $e_{\pm2, i}^a = e_{\pm2, a}^i$. Using the properties of the basis vectors it is possible to show that (from now on we restrict our analysis to the $+2$ mode only)
\begin{equation}
  \left \langle E^a_{i}  E^b_{j} \right \rangle  = \frac{1}{a^4} \int \frac{\textrm{d} k \textrm{d} \Omega_{\vec{k}} }{(2 \pi)^{3}} \, k^2   \frac{ \partial_0 \tilde{\delta A}_{+2}(\tau,  k) \left[ \partial_0 \tilde{\delta A}_{+2}(\tau,  k) \right]^*  e_{+2, i}^a (\hat{k}) \left[ e_{+2, j}^b (\hat{k}) \right]^*  +h. c.}{2}  \, ,
\end{equation}
and analogously for the terms of $ \left \langle E^a_{i} B^b_{j} \right \rangle$ and $ \left \langle B^a_{i} B^b_{j} \right \rangle$
\begin{align}
  2 \left \langle E^a_{i} B^b_{j} \right \rangle_1  = & - \frac{1}{a^4} \int \frac{\textrm{d} k \textrm{d} \Omega_{\vec{k}} }{(2 \pi)^{3}} \, k^3 \left\{ \partial_0 \tilde{\delta A}_{+2}(\tau,  k) \left[ \tilde{\delta A}_{+2}(\tau,  k) \right]^*  e_{+2, i}^a (\hat{k}) \left[ e_{+2, j}^b (\hat{k}) \right]^*  +h. c. \right \} \, , \\
  2 \left \langle E^a_{i} B^b_{j} \right \rangle_2  = & \frac{ e f }{a^4} \int \frac{\textrm{d} k \textrm{d} \Omega_{\vec{k}} }{(2 \pi)^{3}} \, k^2 \left\{ \partial_0 \tilde{\delta A}_{+2}(\tau,  k) \left[ \tilde{\delta A}_{+2}(\tau,  k) \right]^*  e_{+2, i}^a (\hat{k}) \left[ e_{+2, j}^b (\hat{k}) \right]^*  +h. c. \right \} \, ,\\
  2 \left \langle E^a_{i} B^b_{j} \right \rangle_3  = & - \frac{e \partial_0 f}{a^4} \int \frac{\textrm{d} k \textrm{d} \Omega_{\vec{k}} }{(2 \pi)^{3}} \, k^2 \left\{ \tilde{\delta A}_{+2}(\tau,  k) \left[ \tilde{\delta A}_{+2}(\tau,  k) \right]^*  
  \times \right. \nonumber   \\  & \left. 
  \delta^a_i \frac{\varepsilon_{jnm} \varepsilon^{bdh}}{2} \left[ e_{+2, n}^d (\hat{k}) \right] \left[ e_{+2, m}^h (\hat{k}) \right]^*  +h. c. \right \} \, , \\
  2 \left \langle B^a_{i}  B^b_{j} \right \rangle_1  = & \frac{1}{a^4} \int \frac{\textrm{d} k \textrm{d} \Omega_{\vec{k}} }{(2 \pi)^{3}} \, k^4 \left\{ \tilde{\delta A}_{+2}(\tau,  k) \left[ \tilde{\delta A}_{+2}(\tau,  k) \right]^*  e_{+2, i}^a (\hat{k}) \left[ e_{+2, j}^b (\hat{k}) \right]^*  +h. c. \right \} \, , \\
   2 \left \langle B^a_{i}  B^b_{j} \right \rangle_2  = & \frac{e^2 f^2}{a^4} \int \frac{\textrm{d} k \textrm{d} \Omega_{\vec{k}} }{(2 \pi)^{3}} \, k^2 \left\{ \tilde{\delta A}_{+2}(\tau,  k) \left[ \tilde{\delta A}_{+2}(\tau,  k) \right]^*  e_{+2, i}^a (\hat{k}) \left[ e_{+2, j}^b (\hat{k}) \right]^*  +h. c. \right \} \, , \\
   2 \left \langle B^a_{i}  B^b_{j} \right \rangle_3  = & \frac{-ef}{a^4} \int \frac{\textrm{d} k \textrm{d} \Omega_{\vec{k}} }{(2 \pi)^{3}} \, k^3 \left\{ \tilde{\delta A}_{+2}(\tau,  k) \left[ \tilde{\delta A}_{+2}(\tau,  k) \right]^*   e_{+2, i}^a (\hat{k}) \left[ e_{+2, j}^b (\hat{k}) \right]^*  +h. c. \right \}  \, ,\\
  2 \left \langle B^a_{i}  B^b_{j} \right \rangle_4  = & \frac{-ef}{a^4} \int \frac{\textrm{d} k \textrm{d} \Omega_{\vec{k}} }{(2 \pi)^{3}} \, k^3 \left\{ \tilde{\delta A}_{+2}(\tau,  k) \left[ \tilde{\delta A}_{+2}(\tau,  k) \right]^*  e_{+2, i}^a (\hat{k}) \left[ e_{+2, j}^b (\hat{k}) \right]^*  +h. c. \right \} \, .
\end{align}
Since $T_1 $ and $T_2 $ are respectively given by 
\begin{eqnarray}
  T_1 &=& \left \langle E^a_{i} E^b_{j} \right \rangle \left(  \left \langle B^a_{i} B^b_{j} \right \rangle_1 +\left \langle B^a_{i} B^b_{j} \right \rangle_2 +\left \langle B^a_{i} B^b_{j} \right \rangle_3 +\left \langle B^a_{i} B^b_{j} \right \rangle_4 \right) \, , \\ 
  T_2 &=& \left( \left \langle E^a_{i} B^b_{j} \right \rangle_1 +\left \langle E^a_{i} B^b_{j} \right \rangle_2 +\left \langle E^a_{i} B^b_{j} \right \rangle_3 \right) \left( \left \langle B^a_{i} E^b_{j} \right \rangle_1 +\left \langle B^a_{i} E^b_{j} \right \rangle_2 +\left \langle B^a_{i} E^b_{j} \right \rangle_3 \right) \, , \hspace{1cm}
\end{eqnarray}
we can immediatly see that, all the angular integrals reduce to three combinations:
\begin{eqnarray}
  \Omega_1 &\equiv& \int \textrm{d} \Omega_{\vec{k}} \ \textrm{d} \Omega_{\vec{q}} \ e_{+2, i}^a (\hat{k}) \left[ e_{+2, j}^b (\hat{k}) \right]^*  e_{+2, i}^a (\hat{q}) \left[ e_{+2, j}^b (\hat{q}) \right]^* \, , \\
  \Omega_2 &\equiv& \frac{3}{4} \int \textrm{d} \Omega_{\vec{k}} \ \textrm{d} \Omega_{\vec{q}} \ \varepsilon_{jnm} \varepsilon^{bdh} \left[ e_{+2, n}^d (\hat{k}) \right] \left[ e_{+2, m}^h (\hat{k}) \right]^* \varepsilon_{jlu} \varepsilon^{bpr} \left[ e_{+2, l}^p (\hat{q}) \right] \left[ e_{+2, u}^r (\hat{q}) \right]^* , \hspace{1cm} \\
  \Omega_3 &\equiv& \int \textrm{d} \Omega_{\vec{k}} \ \textrm{d} \Omega_{\vec{q}} \ e_{+2, i}^a (\hat{k}) \left[ e_{+2, j}^b (\hat{k}) \right]^* \delta^a_i \frac{\varepsilon_{jnm} \varepsilon^{bdh}}{2} \left[ e_{+2, n}^d (\hat{k}) \right] \left[ e_{+2, m}^h (\hat{k}) \right]^* \, ,
\end{eqnarray}
whose direct evaluation gives
\begin{equation}
   \Omega_1 = \frac{(4\pi)^2}{5} \, ,\qquad \qquad \Omega_2 = \frac{4\pi^2}{3} \, ,\qquad \qquad \Omega_3 = 0 \, . 
\end{equation}
Notice that $\Omega_3 = 0$ can be used to simplify $T_2$ as:
\begin{equation}
  T_2 = \left( \left \langle E^a_{i} B^b_{j} \right \rangle_1 +\left \langle E^a_{i} B^b_{j} \right \rangle_2  \right) \left( \left \langle B^a_{i} E^b_{j} \right \rangle_1 +\left \langle B^a_{i} E^b_{j} \right \rangle_2 \right) +\left \langle E^a_{i} B^b_{j} \right \rangle_3 \left \langle B^a_{i} E^b_{j} \right \rangle_3 \, ,
\end{equation}
At this point we are finally left with only integrals over the absolute value of the momenta. In order to perform these integrals, we will employ Eq.~\eqref{eq:NonabelianAsymptotics2}, defining $\tilde w(x)$ as follows:
\begin{align}
  w_{+2}(x) & \simeq e^{(\kappa-\mu)\pi}  \sqrt{4\pi} \left(\frac{\zeta(x)}{V(x)}\right)^{1/4}\mathrm{Ai}\left(\zeta(x)\right) \\
 &  \equiv e^{(\kappa-\mu)\pi} \tilde w(x) \,.
\end{align}
We have verified that integrating this approximate expression over the range $0 \leq x \leq 2 \xi$ agrees with the integral over the exact expression~\eqref{eq:AnalyticalSolution} extremely well, and we are moreover insensitive to the choice of the UV-cutoff. Let us start by computing for example:
\begin{equation}
\begin{aligned}
  4 \left \langle E^a_{i} E^b_{j} \right \rangle \left \langle B^a_{i} B^b_{j} \right \rangle_1   = \frac{\Omega_1 e^{4 \pi (\kappa - \mu)}}{(2 \pi)^6 a^8} \int  \frac{ \textrm{d} k  \, k^2 }{k} \left\{ \partial_0 \tilde w( x_k)  \partial_0 \tilde w( x_k)  \right \}      \int  \frac{\textrm{d} q  \, q^4}{ q} \left\{ \tilde w( x_q) \tilde w( x_q) \right \} \,,
\end{aligned}
\end{equation}
where we use the notation $x_k = - k \tau$ to keep track of the different momentum variables. We can then use $\partial_0 \tilde w( x_k)  = - k \tilde w( x_k)$ ($\prime$ here is used to denote the derivative with respect to $x_k$), $k = - x_k/ \tau$ and $\tau = - 1/(a H)$ to get
\begin{equation}
\begin{aligned}
  4 \left \langle E^a_{i} E^b_{j} \right \rangle \left \langle B^a_{i} B^b_{j} \right \rangle_1   = \frac{\Omega_1 e^{4 \pi (\kappa - \mu)}}{(2 \pi)^6} H^8 \int  \textrm{d} x_k  \, x_k^3  \left\{\tilde w'( x_k)  \tilde w'( x_k)   \right \}     \int \textrm{d} x_q  \, x_q^3 \ \left\{\tilde w( x_q)  \tilde w( x_q)   \right \}  \, .
\end{aligned}
\end{equation}
Since all the other terms can also be expressed as a product of two integrals it is useful to introduce the six integrals:
\begin{align}
  \mathcal{I}_1 & =  \int  \textrm{d} x  \, x^3 \  \tilde w'( x)  \tilde w'( x) \,, \qquad && 
  \mathcal{I}_2  =  \int  \textrm{d} x  \, x^3 \ \tilde w( x) \tilde w( x)  \,, \qquad 
  \mathcal{I}_3 & =  \int  \textrm{d} x  \, x \ \tilde w( x) \tilde w( x)  \,, \\ 
  \mathcal{I}_4 &  =  \int  \textrm{d} x  \, x^2 \ \tilde w( x) \tilde w( x) \,, \qquad &&
  \mathcal{I}_5  =  \int  \textrm{d} x  \, x^3 \ \tilde w'( x)  \tilde w( x) \,, \qquad 
  \mathcal{I}_6  &=  \int  \textrm{d} x  \, x^2 \  \tilde w'( x) \tilde w( x)   \,,
\end{align}
so that we can easily express
\begin{align}
  4 \left \langle E^a_{i} E^b_{j} \right \rangle \left \langle B^a_{i} B^b_{j} \right \rangle_1   = & \frac{\Omega_1 e^{4 \pi (\kappa - \mu)}}{ (2 \pi)^6} H^8  \mathcal{I}_1 \mathcal{I}_2  \, , \quad 
  && 4 \left \langle E^a_{i} E^b_{j} \right \rangle \left \langle B^a_{i} B^b_{j} \right \rangle_2   =  \frac{\Omega_1 e^{4 \pi (\kappa - \mu)}}{ (2 \pi)^6} H^8 \xi^2  \mathcal{I}_1 \mathcal{I}_3  \, , \nonumber  \\
  4 \left \langle E^a_{i} E^b_{j} \right \rangle \left \langle B^a_{i} B^b_{j} \right \rangle_3  = & \frac{\Omega_1 e^{4 \pi (\kappa - \mu)}}{ (2 \pi)^6} H^8 (-\xi) \mathcal{I}_1  \mathcal{I}_4  \, , \quad
  && 4 \left \langle E^a_{i} E^b_{j} \right \rangle \left \langle B^a_{i} B^b_{j} \right \rangle_4   =  \frac{\Omega_1 e^{4 \pi (\kappa - \mu)}}{ (2 \pi)^6} H^8 (-\xi)  \mathcal{I}_1  \mathcal{I}_4  \, , \nonumber \\
  4 \left \langle E^a_{i} B^b_{j} \right \rangle_1 \left \langle B^a_{i} E^b_{j} \right \rangle_1   = & \frac{\Omega_1 e^{4 \pi (\kappa - \mu)}}{ (2 \pi)^6} H^8  \mathcal{I}_5^2  \, , \quad 
  && 4 \left \langle E^a_{i} B^b_{j} \right \rangle_1 \left \langle B^a_{i} E^b_{j} \right \rangle_2   =  \frac{\Omega_1 e^{4 \pi (\kappa - \mu)}}{ (2 \pi)^6} H^8 (-\xi)  \mathcal{I}_5  \mathcal{I}_6  \, , \nonumber \\
  4 \left \langle E^a_{i} B^b_{j} \right \rangle_2 \left \langle B^a_{i} E^b_{j} \right \rangle_1   = & \frac{\Omega_1 e^{4 \pi (\kappa - \mu)}}{ (2 \pi)^6} H^8 (-\xi)  \mathcal{I}_5 \mathcal{I}_6  \, , \quad
  && 4 \left \langle E^a_{i} B^b_{j} \right \rangle_2 \left \langle B^a_{i} E^b_{j} \right \rangle_2   =  \frac{\Omega_1 e^{4 \pi (\kappa - \mu)}}{ (2 \pi)^6} H^8 \xi^2  \mathcal{I}_6^2 \, , \nonumber \\
  4 \left \langle E^a_{i} B^b_{j} \right \rangle_3 \left \langle B^a_{i} E^b_{j} \right \rangle_3   = & \frac{\Omega_2 e^{4 \pi (\kappa - \mu)}}{ (2 \pi)^6} H^8 \xi^2 \mathcal{I}_3^2 \, , 
\end{align}
where we have also used $ef = -\xi/\tau$ and $\textrm{d}k k^2 = - \textrm{d}x_k x_k^2 /\tau^3 $. At this point we have all in hand to compute $ T_0 $, $ T_1 $ and $ T_2 $. Let us start with $T_0$:
\begin{equation}
  T_0 = \left \langle E^a_{i} B^b_{j} \right \rangle_1 + \left \langle E^a_{i} B^b_{j} \right \rangle_2 + \left \langle E^a_{i} B^b_{j} \right \rangle_3 = - \frac{H^4}{4 \pi^2} e^{2 \pi (\kappa - \mu)} \,  \left( \mathcal{I}_5 + \frac{\xi}{2} \mathcal{I}_3 - \xi \mathcal{I}_6 \right) \, ,
\end{equation}
which to a good approximation is given by
\begin{equation}
  T_0 = - 0.24 \times \frac{H^4}{4 \pi^2} \times \xi^3  e^{2 \pi (\kappa - \mu)}  \equiv   \frac{H^4}{4 \pi^2} \, e^{2 \pi (\kappa - \mu)}  \, \tilde T_0 \, .
\end{equation}
Analogously we can compute $T_1$ and $T_2$
\begin{align}
  T_1 & = \frac{\Omega_1 H^8}{ 4 (2\pi)^6} e^{4 \pi (\kappa - \mu)} \  \mathcal{I}_1\left[ \mathcal{I}_2  + \xi^2 \mathcal{I}_3 - 2 \xi \mathcal{I}_4 \right]  \,, \\
T_2 & = \frac{\Omega_1 H^8}{ 4 (2\pi)^6} e^{4 \pi (\kappa - \mu)} \ \left[\mathcal{I}_5^2 - 2 \xi \mathcal{I}_5 \mathcal{I}_6+ \xi^2 \mathcal{I}_6^2 + \xi^2 \Omega_2/\Omega_1 \mathcal{I}_3^2 \right]  \, .
\end{align}
Performing the integrals yields
\begin{align}
  T_1 & \simeq 0.0082 \times \frac{H^8}{80 \pi^4} \times \xi^8  e^{4 \pi (\kappa - \mu)}  \equiv  \frac{H^8}{80 \pi^4} \,  e^{4 \pi (\kappa - \mu)} \, \tilde T_1  \, , \\
  T_2 & \simeq 0.051 \times \frac{H^8}{80 \pi^4} \times  \xi^6 e^{4 \pi (\kappa - \mu)} \equiv  \frac{H^8}{80 \pi^4} \,  e^{4 \pi (\kappa - \mu)} \, \tilde T_2  \, ,
\end{align}
where we have also substituted the values of $\Omega_1$ and $\Omega_2$.

\section{\label{app:sec3}Supplemental material for \prettyref{sec:background}}

In \prettyref{app:Jacobi} we describe the properties of the Jacobi $\sn$ function, which is used to describe the oscillatory regime for solutions to \prettyref{eq:chromonatural_eom_final} for the homogeneous and isotropic gauge-field background. In \prettyref{app:pf-bg-future} we qualitatively describe the asymptotic behaviour of solutions in the infinite past and infinite future. \prettyref{app:t_to_0} gives asymptotic expansions for solutions in the infinite future, while \prettyref{app:pf-approx-w} proves asymptotics in the far past. Along the way, we prove the theorems stated in \prettyref{subsec:solutions}.

\subsection{\label{app:Jacobi}Solution of the quartic oscillator equation via the Jacobi \texorpdfstring{$\sn$}{sn} function}

In the case $\xi=0$ (or in the limit $\tau\to-\infty$ when the last term can be neglected), the equation of motion \prettyref{eq:chromonatural_eom_final} is a quartic oscillator, which is solved by the Jacobi \texorpdfstring{$\sn$}{sn} function. We review some basic facts about this function while providing a derivation. Let $v(\tau)$ denote any solution to the $\xi=0$ version of \prettyref{eq:chromonatural_eom_final}. For comparison, let $h(\tau)$ be any solution to the harmonic oscillator equation. Then 
\begin{align}
v''(\tau)+2\left(v(\tau)\right)^{3} & =0, & h''(\tau)+h(\tau) & =0.\label{eq:jacobi-eom}
\end{align}
 Multiplying \prettyref{eq:jacobi-eom} by $2v'(\tau)$ and $2h'(\tau)$ respectively, after integration we obtain 
\begin{align}
v'(\tau)^{2}+v(\tau)^{4} & =\omega^{4}, & h'(\tau)^{2}+h(\tau)^{2} & =A^{2},\label{eq:jacobi-energy}
\end{align}
 where $\omega^{4}$ and $A^{2}$ are integration constants. We will identify $\omega$ and $A$ with the respective amplitudes of oscillation. Identifying the second terms on the left-hand sides of \eqref{eq:jacobi-energy} as twice the potential energy, we note that the first equation in \eqref{eq:jacobi-energy} describes a quartic oscillator. Separating variables and integrating, 
\begin{align}
\int\mathrm{d}\tau & =\int\frac{\mathrm{d}v}{\sqrt{\omega^{4}-v^{4}}}, & \int\mathrm{d}\tau & =\int\frac{\mathrm{d}h}{\sqrt{A^{2}-h^{2}}}\nonumber \\
\tau+u_{0}/\omega & =F\left(\sin^{-1}\left(v/\omega\right)\mid-1\right)/\omega & \tau+\theta_{0} & =\sin^{-1}\left(h/A\right),\label{eq:elliptic-F}
\end{align}
 where $u_{0}/\omega$ and $\theta_{0}$ are constants of integration, and $F(\phi\mid m)\equiv\int_{0}^{\phi}\left(1-m\,\sin^{2}(\theta)\right)^{-1/2}\,\mathrm{d}\theta$ is the \emph{incomplete elliptic integral of the first kind} with \emph{elliptic parameter} $m$. For fixed $m$, the inverse function of $u=F(\phi\mid m)$ is defined as the \emph{Jacobi amplitude} $\phi=\am(u\mid m)$. Solving for $v(\tau)$ and $h(\tau)$ respectively, we obtain the general solution 
\begin{equation}
v(\tau)=\omega\,\sn\left(\omega\tau+u_{0}\mid-1\right),\qquad h(\tau)=A\,\sin\left(\tau+\theta_{0}\right),\label{eq:quartic-solution}
\end{equation}
 where $\sn(u\mid m)\equiv\sin\left(\am\left(u\mid m\right)\right)$. Since we will be concerned only with the case $m=-1$, we define as shorthand
\[
\sn(u)\equiv\sn(u\mid-1),\quad\sn'(u)\equiv\tfrac{\mathrm{d}}{\mathrm{d}u}\sn(u\mid-1).
\]
 The functions $\sn(u)$ and $\sn'(u)$ both oscillate with unit amplitude and satisfy $\sn(0)=0$ and $\sn'(0)=1$. Just as the quarter-period of $\sin(\theta)$ is given by $\sin^{-1}(1)=\pi/2$, it follows from \prettyref{eq:elliptic-F} that the quarter-period of $\sn(u)$ is $K(-1)\approx5.244$, where $K(m)\equiv F(\sin^{-1}(1)\mid m)$ denotes the \emph{complete elliptic integral of the first kind}. From \prettyref{eq:jacobi-eom} and the $\omega=1$ version of \prettyref{eq:jacobi-energy}, we obtain the identities
\begin{equation}
\tfrac{\mathrm{d}}{\mathrm{d}u}\sn'(u)=-2\sn(u)^{3},\qquad\sn'(u)^{2}+\sn(u)^{4}=1.\label{eq:sn-identities}
\end{equation}
 Note that the general solution \eqref{eq:quartic-solution} to \prettyref{eq:jacobi-eom} has amplitude $\omega$ and frequency $\tfrac{1}{4}\omega/K(-1)$. 

For any function $s(\tau)$, it will be useful to define the Jacobi version of polar coordinates in phase space. We introduce the radial coordinate $\omega_{s}(\tau)\geq0$ and the angular coordinate $u_{s}(\tau)$ according to 
\begin{equation}
\left(s(\tau),s'(\tau)\right)=\left(\omega_{s}(\tau)\sn(u_{s}(\tau)),\ \omega_{s}(\tau)^{2}\sn'(u_{s}(\tau))\right).\label{eq:polar-def}
\end{equation}
 It follows from the second identity in \eqref{eq:sn-identities} that 
\[
\omega_{s}(\tau)=\sqrt[4]{s'(\tau)^{2}+s(\tau)^{4}}.
\]
 Moreover if $s(\tau)$ satisfies an equation of the form $s''(\tau)+2s(\tau)^{3}=F(s(\tau),s'(\tau))$, then the corresponding polar equations are 
\begin{align}
\omega_{s}'(\tau)  =\frac{\sn'(u_{s}(\tau))}{2\omega_{s}(\tau)}F,  \qquad \qquad u_{s}'(\tau)  =\omega_{s}(\tau)-\frac{\sn(u_{s}(\tau))}{2\omega_{s}(\tau)^{2}}F, \nonumber \\
F  =F\left(\omega_{s}(\tau)\sn(u_{s}(\tau)),\ \omega_{s}(\tau)^{2}\sn'(u_{s}(\tau))\right).  \hspace{1.cm}  \label{eq:Jacobi-polar}
\end{align}
 In particular, if $F=0$ as in the case $s(\tau)=v(\tau)$ given in \prettyref{eq:jacobi-eom}, then $\omega_{v}'(\tau)=0$ and $u_{v}'(\tau)=\omega$, which agrees with \prettyref{eq:quartic-solution}. 

\subsection{\label{app:pf-bg-future}The limiting behaviour of solutions to the homogeneous background gauge-field equation in the infinite past and future}

Here we analyze the asymptotic behaviour of the equations of motion \prettyref{eq:q-p-eom}. Solutions are given by trajectories of the vector field \eqref{eq:vec-field}, as pictured in \prettyref{fig:trajectories-panel}. In summary, we prove that the observations made in \prettyref{subsec:phase-space} are generally true. We consider two asymptotic limits: the infinite past ($N\to+\infty)$ and the infinite future\footnote{Although inflation ends at $N=0,$ it is still mathematically useful to analyze our background equation of motion by taking limits. The appropriate future limit corresponding to $\tau\to0^{-}$ is $N\to-\infty$. } $(N\to-\infty)$. As we shall see, Bendixson's criterion and the Poincar\'e--Bendixson theorem \cite{bendixson1901} allow us to give a complete classification of all trajectories. Namely, Bendixson's criterion rules out pathologies such as periodic or homoclinic\footnote{A \emph{homoclinic orbit} is a trajectory for which the same (saddle) point is the limit in both the infinite past and infinite future.} orbits, and then the Poincar\'e--Bendixson theorem implies that some limit exists, being either a zero of the vector field or $\bm{\infty}$.\footnote{We say that the limit of a trajectory is $\bm{\infty}$ when the magnitude of the trajectory limits to infinity. Equivalently, when $\mathbb{R}^{2}$ is compactified to $S^{2}$ by adding a point denoted by $\bm{\infty}$, then the limit converges to $\bm{\infty}$. } Further use of Bendixson's criterion will allow us in \prettyref{lem:asymptotic-table} to exactly count the number of trajectories as classified by their asymptotic limits.

The vector field \eqref{eq:vec-field} has at most three zeroes denoted $\mathbf{c}_{0}$, $\mathbf{c}_{1}$, and $\mathbf{c}_{2}$ (see \eqref{eq:ci-points}). We will make use of the following readily-verified properties of the vector field \eqref{eq:vec-field}:
\begin{itemize}
\item The divergence $\nabla\cdot$ is strictly positive (the constant function $+3$). 
\item When $0\leq\xi<2$, there is a single zero at the point $\mathbf{c}_{0}$. When $\xi>2$ there are three zeroes at $\mathbf{c}_{0}$, $\mathbf{c}_{1}$ and $\mathbf{c}_{2}$. (To avoid irrelevant complications arising from degenerate zeroes\footnote{A zero of the vector field $V$ is \emph{degenerate} if its linearization $\frac{\partial V_{i}}{\partial x_{j}}$ at that zero is given by a matrix whose determinant is $0$.} of a vector field, we exclude from analysis the case when $\xi$ is exactly equal to $2$.) In either case, the zeroes are non-degenerate, and $\mathbf{c}_{1}$ is a saddle point while $\mathbf{c}_{0}$ and $\mathbf{c}_{2}$ are attractors (they arise as limit points in the infinite future $N\to-\infty$ but not in the infinite past $N\to+\infty$). 
\item Whenever it is positive, the auxiliary quantity $D(N)$ given by 
\begin{equation}
D(q,p)\equiv p^{2}+q^{4}-\tfrac{4}{3}\xi q^{3},\qquad D(N)\equiv D(q(N),p(N)),\label{eq:def-D}
\end{equation}
 is decreasing in time along trajectories. Namely by \prettyref{eq:q-p-eom}, it satisfies 
\begin{equation}
-\frac{\mathrm{d}D}{\mathrm{d}N}=-3D-(p^{2}+q^{4})<0\ \textrm{ whenever }D>0.\label{eq:D-decreasing}
\end{equation}
\end{itemize}
To explain the significance of $D(N)$, note that 
\begin{equation}
\tau^{-4}D(N)=\left(ef'(\tau)\right)^{2}+\left(ef(\tau)\right)^{4}-\tfrac{4}{3}\xi(ef(\tau))^{3}/(-\tau)=\omega_{ef}(\tau)^{4}\left(1+\mathcal{O}\left(\frac{\xi}{-\tau\omega_{ef}(\tau)}\right)\right).\label{eq:DN-asymptotic}
\end{equation}
 When $\tau\ll-\xi/\omega_{ef}(\tau)$ we have $D(N)\approx(\tau\omega_{ef}(\tau))^{4}$. Thus $D(N)$ encodes the approximate amplitude $\omega_{ef}(\tau)$ of oscillations while also satisfying the monotonicity property \eqref{eq:D-decreasing}.  
\begin{lem}
\label{lem:oscillation-condition}For a given trajectory, the following conditions are equivalent:
\begin{enumerate}
\item \label{enu:inf-past-trajectory-infinity}The trajectory converges to $\bm{\infty}$ in the infinite past.
\item \label{enu:inf-past--tomega-infinity}$-\tau\omega_{ef}(\tau)\to+\infty$ as $\tau\to-\infty$. 
\item \label{enu:inf-past-D-infinity}$D(N)\to+\infty$ as $N\to+\infty$. 
\end{enumerate}
Moreover, if $D(N)$ is bounded, then the trajectory is bounded.
\end{lem}

We will show later that these three conditions are also equivalent to $\omega>0$. 
\begin{proof}
Consider the limit in the infinite past. A trajectory converges to $\bm{\infty}$ if and only if $p^{2}+q^{4}\to+\infty$. Since $p^{2}+q^{4}=(-\tau\omega_{ef}(\tau))^{4}$ it follows that items \ref{enu:inf-past-trajectory-infinity} and \ref{enu:inf-past--tomega-infinity} are equivalent. Supposing that $-\tau\omega_{ef}(\tau)\to+\infty$, it follows from \prettyref{eq:DN-asymptotic} that $D(N)\to+\infty$. One also sees from \prettyref{eq:DN-asymptotic} that $D(N)$ is large only when $-\tau\omega_{ef}(\tau)$ is. Thus item \ref{enu:inf-past-D-infinity} is also equivalent. Finally, if $D(N)$ is bounded, then $-\tau\omega_{ef}(\tau)$ is bounded and hence $p^{2}+q^{4}$ and thus the trajectory are bounded.
\end{proof}
\begin{lem}[Bendixson's criterion]
\label{lem:Bendixson-criterion}The vector field \eqref{eq:vec-field} has no periodic trajectory. More generally, there are no bounded sets of positive area which are invariant under the flow of \eqref{eq:vec-field}. In particular, no finite union of bounded trajectories forms a loop.
\end{lem}

\begin{proof}
Since the divergence is $+3$, the area of any bounded set is proportional to $e^{3N}$ under the flow of \eqref{eq:vec-field}. Since the area of any invariant set must be constant under the flow, any bounded invariant set must have zero area. No periodic orbit is possible: the interior of the enclosed region would be bounded, positive area, and invariant. Similarly, the union of finitely many bounded trajectories cannot form any loop. This general argument goes back to \cite{bendixson1901}. 
\end{proof}
Note that \prettyref{lem:Bendixson-criterion} forbids loops formed by any bounded trajectories, regardless of the direction of the flow. Just as it is impossible to flow $\mathbf{a}\to\mathbf{b}$ and $\mathbf{b}\to\mathbf{a}$, it is also forbidden to flow $\mathbf{a}\to\mathbf{b}$ along two distinct trajectories, because the enclosed region would be invariant and have positive area. 
\begin{lem}[Poincar\'e--Bendixson theorem]
\label{lem:Poincare-Bendixson} In either limit $N\to\pm\infty$, any trajectory either converges to $\mathbf{c}_{0}$, $\mathbf{c}_{1}$ or $\mathbf{c}_{2}$, or the trajectory is unbounded. 
\end{lem}

\begin{proof}
According to the Poincar\'e--Bendixson theorem \cite{bendixson1901}, given any smooth vector field on $\mathbb{R}^{2}$ which has a finite number of non-degenerate zeroes, the limiting behaviour of any trajectory as $N\to+\infty$ or respectively as $N\to-\infty$ is given by one of the following five distinct possibilities: 
\begin{enumerate}
\item \label{enu:pb-zero}The trajectory limits to a zero of the vector field. 
\item \label{enu:pb-unbounded}The trajectory is unbounded. 
\item \label{enu:pb-periodic}The trajectory is periodic.
\item \label{enu:pb-almost-periodic}The trajectory spirals towards a periodic trajectory.
\item \label{enu:pb-broken}The trajectory spirals towards some limiting loop which is a union of finitely many trajectories which connect subsequent saddle-point zeroes of the vector field to each other. 
\end{enumerate}
In our case, Bendixson's criterion (\prettyref{lem:Bendixson-criterion}) applies. Thus \ref{enu:pb-periodic}, \ref{enu:pb-almost-periodic} and \ref{enu:pb-broken} are impossible, so \ref{enu:pb-zero} or \ref{enu:pb-unbounded} must be true. 
\end{proof}
The next two lemmas imply that $\bm{\infty}$ behaves as a repeller (it arises as a limit in the infinite past $N\to+\infty$ for all sufficiently large initial conditions, but it is never a limit in the infinite future).

\begin{lem}
\label{lem:D-unbounded} The forward evolution (in the direction $N\to-\infty$) of any trajectory converges to $\mathbf{c}_{0}$, $\mathbf{c}_{1}$ or $\mathbf{c}_{2}$. 
\end{lem}

\begin{proof}
For any trajectory beginning at $N_{0}$, consider $D(N)$ for values with $N\leq N_{0}$. Since $D(N)$ cannot be simultaneously positive and increasing with $-N$, it satisfies $D(N)\leq\max(0,D(N_{0}))$ for all $N\leq N_{0}$. Thus $D(N)$ is bounded in the future direction. It follows from \prettyref{lem:oscillation-condition} that the trajectory itself is bounded (see \prettyref{lem:oscillation-condition}), and thus \prettyref{lem:Poincare-Bendixson} implies that the limit is a zero of the vector field. 
\end{proof}
\begin{lem}
\label{lem:dichotomy}Every trajectory satisfies exactly one of the following two conditions:
\begin{enumerate}
\item \label{enu:bounded-traj}$D(N)\leq0$ for all $N$, and the trajectory converges to $\mathbf{c}_{0}$, $\mathbf{c}_{1}$ or $\mathbf{c}_{2}$ in the infinite past.
\item \label{enu:conv-to-inf}The trajectory converges to $\bm{\infty}$ in the infinite past, and $-\tau\omega_{ef}(\tau)\to+\infty$ as $\tau\to-\infty$.
\end{enumerate}
\end{lem}

We will see in \prettyref{lem:D-negative-omega-zero} that the first case corresponds to $\omega=0$, which consists of constant trajectories and, by \prettyref{lem:asymptotic-table}, also the instanton-type trajectories. In \prettyref{lem:omega-positive} we will see that the second case corresponds to $\omega>0$, which are the solutions which oscillate in the far past.
\begin{proof}
If $D(N)\leq0$ for all $N$ then by \prettyref{lem:oscillation-condition} the trajectory is bounded, and hence by \prettyref{lem:Poincare-Bendixson} it converges to $\mathbf{c}_{0}$, $\mathbf{c}_{1}$ or $\mathbf{c}_{2}$ in the infinite past. For the remainder of the assertion, it suffices to show that if $D(N_{1})>0$ for some $N_{1}$ then $D(N)\to+\infty$ as $N\to\infty$, because then \prettyref{lem:oscillation-condition} implies \ref{enu:conv-to-inf}. From \prettyref{eq:D-decreasing}, it follows that $D(N)$ is increasing as $N$ increases for $N\geq N_{1}$. Thus either $D(N)\to+\infty$ or $D(N)$ increases to some finite positive limit. Suppose for contradiction that $D(N)$ increases to a finite positive limit. Then $D(N)$ and hence also the trajectory are bounded as $N\to\infty$. By \prettyref{lem:Poincare-Bendixson} it follows that the trajectory must limit to some $\mathbf{c}_{i}$. However, as is easily verified, $D(\mathbf{c}_{i})\leq0$ so this is impossible. In summary, if $D(N)\leq0$ then the trajectory is bounded and so \ref{enu:bounded-traj} holds. Otherwise there exists some $N_{1}$ such that $D(N_{1})>0$ from which it follows that $D(N)\to+\infty$, and \ref{enu:conv-to-inf} holds. 
\end{proof}
\begin{lem}
\label{lem:D-negative-omega-zero}If $D(N_{1})\leq0$ then $\omega_{ef}(\tau_{1})\leq\tfrac{4}{3}\xi/(-\tau_{1})$, where $N_{1}$ and $\tau_{1}$ correspond to the same time. In particular, if $D(N)\leq0$ for all $N$, then $\omega=0$. 
\end{lem}

\begin{proof}
From \prettyref{eq:DN-asymptotic}, $D(N_{1})\leq0$ is equivalent to 
\[
\omega_{ef}(\tau_{1})^{4}=\left(ef'(\tau_{1})\right)^{2}+\left(ef(\tau_{1})\right)^{4}\leq\tfrac{4}{3}\xi(ef(\tau_{1}))^{3}/(-\tau_{1}).
\]
 Thus 
\[
\omega_{ef}(\tau_{1})^{4}\leq\tfrac{4}{3}\xi\omega_{ef}(\tau_{1})^{3}/(-\tau_{1})\implies\omega_{ef}(\tau_{1})\leq\tfrac{4}{3}\xi/(-\tau_{1}).
\]
 If this inequality holds for all $\tau$, then $\omega\equiv\lim_{\tau\to-\infty}\omega_{ef}(\tau)=0$.
\end{proof}
\begin{lem}
\label{lem:asymptotic-table}When $0\leq\xi<2$, all non-constant trajectories limit to $\bm{\infty}$ in the infinite past and to $\mathbf{c}_{0}$ in the infinite future. If we identify trajectories which differ by a shift in $N$ so that each point $(q,p)$ belongs to a unique trajectory, then for $\xi>2$, the number of non-constant trajectories classified by their asymptotic behaviour is given by \prettyref{tab:asymptotic-classification}.
\begin{table}[h]
\begin{centering}
\begin{tabular}{|l|>{\centering}p{3cm}|>{\centering}p{3cm}|>{\centering}p{3cm}|}
\hline 
\diagbox[width=5.9cm,height=1.81cm]{infinite past}{infinite future} & $\mathbf{c}_{0}$ as $N\to-\infty$\linebreak{}
($c_{0}$\nobreakdash-type) & $\mathbf{c}_{1}$ as $N\to-\infty$\linebreak{}
($c_{1}$\nobreakdash-type) & $\mathbf{c}_{2}$ as $N\to-\infty$\linebreak{}
($c_{2}$\nobreakdash-type)\tabularnewline
\hline 
$\bm{\infty}$ as $N\to+\infty$ ($\omega>0$) & $\infty$ & $2$ & $\infty$\tabularnewline
\hline 
$\mathbf{c}_{1}$ as $N\to+\infty$ (instanton-type) & $1$ & $0$ & $1$\tabularnewline
\hline 
\end{tabular}
\par\end{centering}
\caption{\label{tab:asymptotic-classification}Number of non-constant trajectories for $\xi>2$ as classified by their behaviour in the infinite past ($N\to+\infty)$ and the infinite future ($N\to-\infty$).}
\end{table}
\end{lem}

\begin{proof}
First assume that $\xi>2$. Then by Lemmas~\ref{lem:D-unbounded} and \ref{lem:dichotomy}, the limit in the infinite past or future of any trajectory exists and belongs to the set $\left\{ \mathbf{c}_{0},\mathbf{c}_{1},\mathbf{c}_{2},\bm{\infty}\right\} $. Since $\mathbf{c}_{0}$ and $\mathbf{c}_{2}$ are attractors, they cannot arise as $N\to+\infty$ except as constant trajectories. In \prettyref{lem:D-unbounded}, we showed that $\infty$ is not a limit as $N\to-\infty$. Thus the rows and columns listed in \prettyref{tab:asymptotic-classification} correspond to all possibilities. 

Now we deduce the values listed in \prettyref{tab:asymptotic-classification}. Since $\mathbf{c}_{1}$ is a saddle point, exactly two trajectories limit to $\mathbf{c}_{1}$ as $N\to+\infty$ (resp. $N\to-\infty$). We call the pair converging to $\mathbf{c}_{1}$ as $N\to+\infty$ the ``instanton-type trajectories'' and to $\mathbf{c}_{1}$ as $N\to-\infty$ the ``$c_{1}$\nobreakdash-type trajectories.'' First consider the possibility that a trajectory is simultaneously an instanton-type and $c_{1}$\nobreakdash-type trajectory. No non-constant trajectory can limit to $\mathbf{c}_{1}$ both as $N\to+\infty$ and $N\to-\infty$ since that would form a loop, which is forbidden by Bendixson's criterion (\prettyref{lem:Bendixson-criterion}). Thus the only simultaneously instanton-type and $c_{1}$\nobreakdash-type trajectory is constant, so the corresponding table entry is $0$. Since the middle column corresponds to $c_{1}$\nobreakdash-type solutions, of which there are two, the remaining table entry must be $2$, and thus both $c_{1}$\nobreakdash-type solutions must converge to $\bm{\infty}$ as $N\to+\infty$. Next we consider the two instanton-type solutions which correspond to the bottom row of the table. It is impossible for any entry in the bottom row to be $2$, because that would lead to a loop. Therefore the remaining two entries along the bottom row must be $1$. Finally, since $\mathbf{c}_{0}$ and $\mathbf{c}_{2}$ are attractors, infinitely many trajectories limit to them as $N\to-\infty$. Since the corresponding columns must sum to $\infty$, the remaining two entries must be $\infty$. 

In order to deduce the result for the case $\xi<2$, note that the $\mathbf{c}_{1}$ and $\mathbf{c}_{2}$ points do not exist, so the corresponding version of \prettyref{tab:asymptotic-classification} collapses to the single entry in the upper-left. 
\end{proof}
The previous lemmas lead easily to the proof of \prettyref{thm:bg-future} and \prettyref{thm:phase-intervals}:
\begin{proof}[Proof of \prettyref{thm:bg-future}]
\label{pf:bg-future} We must show that $\lim_{\tau\to0^{-}}-\tau\,ef(\tau)=c_{i}\xi$ for some $c_{i}$, and that $\lim_{\lambda\to0^{+}}\lambda\,ef(\lambda\tau)=c_{i}\xi/(-\tau)$. For the former, 
\[
\lim_{\tau\to0^{-}}-\tau\,ef(\tau)=\lim_{N\to-\infty}q(N)=c_{i}\xi,
\]
 where $q(\tau)$ is defined in \eqref{eq:cov}, the change of variables to $N$ is defined in \eqref{eq:g-def}, and the last equality follows from \prettyref{lem:asymptotic-table}: the trajectory limits to one of the points $\mathbf{c}_{i}$, and the $q$-coordinate of $\mathbf{c}_{i}$ is equal to $c_{i}\xi$ by \eqref{eq:ci-points}. Finally, 
\[
\lim_{\lambda\to0^{+}}\lambda\,ef(\lambda\tau)=(-\tau)^{-1}\lim_{\lambda\to0^{+}}q(N+\ln\lambda)=c_{i}\xi/(-\tau).
\]
 
\end{proof}
\begin{proof}[Proof of \prettyref{thm:phase-intervals}]
 Fix any real number $C>0$. For any fixed $\xi>2$, we first observe that the set of all trajectories (modulo shifts in $N$) which limit to $\bm{\infty}$ in the infinite past is parameterized by the level set $\left\{ (q,p)\mid D(q,p)=C\right\} $, which is topologically a circle. In particular, each trajectory must intersect this level set exactly once. This is because $D$ is decreasing in time when positive, $D(N)\to+\infty$ in the infinite past, and by \prettyref{lem:asymptotic-table}, $D(N)\to D(\mathbf{c}_{i})\leq0$ for some $\mathbf{c}_{i}$ in the infinite future. Also from \prettyref{lem:asymptotic-table}, exactly two trajectories of this type correspond to $c_{1}$\nobreakdash-type solutions. The non-constant $c_{0}$\nobreakdash-type and $c_{2}$\nobreakdash-type solutions comprise the remaining points, which topologically are two disjoint open intervals. Since $c_{0}$\nobreakdash-type and $c_{2}$\nobreakdash-type solutions are basins of attraction for the attractors $\mathbf{c}_{0}$ and $\mathbf{c}_{2}$, they are open subsets of the $q$-$p$ plane. Thus their intersections with the level set are non-empty open subsets of the circle. The only way to partition two disjoint intervals into two non-empty open subsets is when each subset corresponds to an interval. Thus the two points on the circle corresponding to $c_{1}$\nobreakdash-type solutions divide the complement of the circle into two open intervals of respective $c_{0}$\nobreakdash-type and $c_{2}$\nobreakdash-type solutions. 

The phase $u_{0}$ also gives a continuous parameterization of all trajectories (modulo shifts in $N$), and it can be shown that $u_{0}$ parameterizes each level set. Thus for each fixed $\xi>2$, there are two phases corresponding to $c_{1}$\nobreakdash-type solutions, and the two complementary phase intervals correspond to intervals of respective $c_{0}$\nobreakdash-type and $c_{2}$\nobreakdash-type solutions.
\end{proof}

\subsection{\label{app:t_to_0}Asymptotics for the general solution of the background gauge field equations in the infinite future (\texorpdfstring{$\tau\to0^-$}{\tau\to0-}) }

We provide here a complete description of the asymptotic behaviour of solutions to \prettyref{eq:chromonatural_eom_final} as $\tau\to0^{-}$. According to \prettyref{thm:bg-future}, for any fixed $\xi$, the limit 
\[
\lim_{\tau\to0^{-}}-\tau\,ef(\tau)
\]
achieves at most three distinct values as $ef(\tau)$ ranges over all solutions: $c_{i}\xi$ for $i\in\left\{ 0,1,2\right\} $, where $c_{i}$ is defined in \prettyref{eq:c-definitions}. In this way, each solution falls into one of three families: two two-parameter families ($c_{0}$\nobreakdash-type and $c_{2}$\nobreakdash-type) and a one-parameter family ($c_{1}$\nobreakdash-type). For each of these three families we provide the leading terms of an asymptotic series solution to \prettyref{eq:chromonatural_eom_final} around $\tau=0$ from which all the parameters can be determined. For brevity, we omit the degenerate case $\xi=2$. The following expressions can be verified by substituting them into \prettyref{eq:chromonatural_eom_final}. Higher-order expressions can be derived by the method of undetermined coefficients.
\begin{enumerate}
\item The $c_{0}$\nobreakdash-type solutions with parameters $\beta$ and $\eta$ are 
\[
ef(\tau)=\beta+2\xi\beta^{2}\cdot(-\tau)\ln(-\tau)+\eta\cdot(-\tau)+\mathcal{O}\left(\left(\xi^{2}\left|\beta^{3}\ln(-\tau)\right|+\xi\eta^{2}\right)(-\tau)^{2}\right)\ \mathrm{as}\ \tau\to0^{-}.
\]
 Under the transformation \eqref{eq:symmetry3}, the parameters transform as $(\beta,\eta)\mapsto(\lambda\beta,\lambda^{2}(\eta+2\xi\beta^{2}\ln\lambda)).$
\item The $c_{1}$\nobreakdash-type solutions with the single parameter $\rho$ are 
\begin{equation}
ef(\tau)=\frac{c_{1}\xi}{-\tau}\left(1+\rho\left(-\tau\right)^{\tfrac{1}{2}\left(3+\sqrt{d_{1}}\right)}+\mathcal{O}\left(\rho^{2}\left(-\tau\right)^{3+\sqrt{d_{1}}}\right)\right)\ \mathrm{as}\ \tau\to0^{-},\label{eq:c1-family}
\end{equation}
 where 
\[
d_{1}=25-8c_{1}\xi^{2}.
\]
 As $\xi$ increases from $2$ to $\infty$, $\sqrt{d_{1}}$ increases from $3$ to $\sqrt{17}$. The parameter $\rho$ transforms under \eqref{eq:symmetry3} as $\rho\mapsto\lambda^{\tfrac{1}{2}\left(3+\sqrt{d_{1}}\right)}\rho$. 
\item The $c_{2}$\nobreakdash-type solutions when\footnote{If $2<\xi<\sqrt{\frac{625}{136}}$ then the square root is negative, so \prettyref{eq:c2-sol} must be rewritten in overdamped form. The case $\xi=\sqrt{\frac{625}{136}}$ corresponds to critical damping. Qualitatively, the only difference in these cases is that the perturbations around the $c_{2}$ solution decay without oscillating. } $\xi>\sqrt{\frac{625}{136}}\approx2.14$ with parameters $\upsilon$ and $\theta$ are of the form 
\begin{equation}
ef(\tau)=\frac{c_{2}\xi}{-\tau}\left(1+\nu(-\tau)^{3/2}\cos\left(\tfrac{1}{2}\sqrt{-d_{2}}\ln(-\tau)+\theta\right)+\mathcal{O}\left(\nu^{2}(-\tau)^{3}\right)\right)\ \mathrm{as}\ \tau\to0^{-},\label{eq:c2-sol}
\end{equation}
 where 
\[
d_{2}:=25-8c_{2}\,\xi^{2}.
\]
 The parameters transform under \eqref{eq:symmetry3} as $(\nu,\theta)\mapsto(\lambda^{3/2}\nu,\theta+\tfrac{1}{2}\sqrt{-d_{2}}\ln(\lambda))$. 
\end{enumerate}
In accordance with \prettyref{thm:bg-future}, it is clear from these formulas and the corresponding transformation laws for the parameters that in the limit $\lambda\to0$ of \eqref{eq:symmetry3} (corresponding to the infinite future), we recover the respective $c_{i}$-solutions $ef(\tau)=c_{i}\xi/(-\tau)$. Note that in this same limit the error terms also vanish, and so these asymptotic formulas become exact.

Closely related to the $c_{1}$ solutions are the instanton-type solutions 
\begin{equation}
ef(\tau)=\frac{c_{1}\xi}{-\tau}\left(1+\rho\left(-\tau\right)^{\tfrac{1}{2}\left(3+\sqrt{d_{1}}\right)}+\mathcal{O}\left(\rho^{2}\left(-\tau\right)^{3+\sqrt{d_{1}}}\right)\right)\label{eq:instanton-type}
\end{equation}
 which are vacuum-to-vacuum transitions which tunnel from the $c_{1}$ solution in the infinite past to either the $c_{0}$ solution or $c_{2}$ solution in the infinite future.

\subsection{\label{app:pf-approx-w}Convergence of \texorpdfstring{$ef(\tau)$}{ef(\tau)} to the Jacobi \texorpdfstring{$\sn$}{sn} function in the far past}

In this subsection, we prove \prettyref{thm:oscillatory} by studying the asymptotic behaviour of solutions to \prettyref{eq:chromonatural_eom_final} as $\tau\to-\infty$. 

As is visible from the envelope of solutions shown in \prettyref{fig:omega_eq_1}, the mean value of the oscillations is slightly positive, and decaying to zero. To very good approximation, this mean value is given by $\xi/(-3\tau)$ as $\tau\to-\infty$. This is verified by defining $s(\tau)$ with this value subtracted from $ef(\tau)$: 
\begin{equation}
s(\tau)\equiv ef(\tau)-\xi/(-3\tau),\label{eq:def-s}
\end{equation}
 We compare the equations of motion for $ef(\tau)$ and $s(\tau)$: 
\begin{equation}
(ef)''(\tau)+2\left(ef(\tau)\right)^{3}=\frac{2\,\xi}{-\tau}(ef(\tau))^{2}\implies s''(\tau)+2s(\tau)^{3}=\frac{2}{3}\left(\frac{\xi}{-\tau}\right)^{2}\left(1+\left(\frac{2}{9}-\frac{1}{\xi^{2}}\right)\left(\frac{\xi}{-\tau s(\tau)}\right)\right)s(\tau).\label{eq:s-eom}
\end{equation}
The right-hand sides can be viewed as forcing terms of the quartic oscillator. The forcing term for $ef(\tau)$ decays in proportion to $(-\tau)^{-1}$, while the forcing term for $s(\tau)$ decays faster as $\tau^{-2}$. Thus the Jacobi $\sn$ function approximates $s(\tau)$ to higher order as $\tau\to-\infty$. Indeed, this allows us to prove asymptotics by switching to Jacobi polar coordinates for $s(\tau)$. 

We shall compute $\omega$ by taking the limit of $\omega_{s}(\tau)$ in the asymptotic past. First we show that working with $\omega_{s}(\tau)$ yields the correct result:
\begin{lem}
\label{lem:omega-ef-equivalent-to-omega-s}If either $\lim_{\tau\to-\infty}\omega_{ef}(\tau)$ or $\lim_{\tau\to-\infty}\omega_{s}(\tau)$ exists, then both limits exist and equal $\omega$. Furthermore, $\lim_{\tau\to-\infty}-\tau\omega_{ef}(\tau)=+\infty$ if and only if $\lim_{\tau\to-\infty}-\tau\omega_{s}(\tau)=+\infty$. 
\end{lem}

\begin{proof}
Recall that $\omega$ is by definition the first limit. Expanding out using the definitions, $\omega_{ef}(\tau)^{4}-\omega_{s}(\tau)^{4}=\mathcal{O}(\xi\omega_{ef}(\tau)^{3}/(-\tau))$, and symmetrically $\omega_{ef}(\tau)^{4}-\omega_{s}(\tau)^{4}=\mathcal{O}(\xi\omega_{s}(\tau)^{3}/(-\tau))$. If either limit exists, then the respective asymptotic estimate shows that the difference of the limits vanishes. A similar argument applied to the ratio $-\tau\omega_{s}(\tau)/(-\tau\omega_{ef}(\tau))\to1$ proves the last assertion. 
\end{proof}
In order to estimate $\omega_{s}(\tau)$ and $u_{s}(\tau)$, we use their differential equations. Substituting \prettyref{eq:s-eom} into \prettyref{eq:Jacobi-polar}, 
\begin{align}
\omega_{s}'(\tau) & =\frac{1}{3}\left(\frac{\xi}{-\tau}\right)^{2}\left(\sn(u_{s}(\tau))+\left(\frac{2}{9}-\frac{1}{\xi^{2}}\right)\left(\frac{\xi}{-\tau\omega_{s}(\tau)}\right)\right)\sn'(u_{s}(\tau)),\label{eq:omega-eom}\\
u_{s}'(\tau) & =\omega_{s}(\tau)\left(1-\frac{1}{3}\left(\frac{\xi}{-\tau\omega_{s}(\tau)}\right)^{2}\left(\sn(u_{s}(\tau))^{2}+\left(\frac{2}{9}-\frac{1}{\xi^{2}}\right)\frac{\xi\sn(u_{s}(\tau))}{-\tau\omega_{s}(\tau)}\right)\right).\label{eq:u-eom}
\end{align}

\begin{thm}
\label{thm:omega-s}The numbers $\omega$ and $u_{0}\equiv\lim_{\tau\to-\infty}\left(u_{s}(\tau)-\omega\tau\right)$ are well-defined. If $\omega=0$ then $\omega_{ef}(\tau)\leq\tfrac{4}{3}\xi/(-\tau)$ for all $\tau$. If $\omega>0$ then 
\begin{align*}
\omega_{s}(\tau_{2}) & =\left(1+\mathcal{O}\left(\frac{\xi^{2}}{\omega^{2}\tau_{2}^{2}}\right)\right)\omega,\ \textrm{ and } & u_{s}(\tau)= & \omega\tau+u_{0}+\mathcal{O}(\xi^{2}/(-\omega\tau))\ \textrm{ as }\tau\to-\infty.
\end{align*}
\end{thm}

\begin{proof}
In what follows, we prove several lemmas which, when taken together, imply this result. According to \prettyref{lem:omega-ef-equivalent-to-omega-s}, we can replace $\omega_{ef}(\tau)$ with $\omega_{s}(\tau)$. We already know from \prettyref{lem:dichotomy} that there are two cases to consider: either $D\leq0$ for all $\tau$ or $\lim_{\tau\to-\infty}-\tau\omega_{s}(\tau)=+\infty$. \prettyref{lem:D-negative-omega-zero} showed that the former case leads to $\omega=0$ with the desired estimate. For the remainder, it suffices to assume $\lim_{\tau\to-\infty}-\tau\omega_{s}(\tau)=+\infty$. Under this assumption, we prove in \prettyref{lem:ws-limit-exists} that $\omega$ exists, in \prettyref{lem:omega-positive} that $\omega>0$, and the desired estimates in \prettyref{lem:omega-s-estimate} and \prettyref{lem:u0-error}. 
\end{proof}
\begin{lem}
\label{lem:ws-limit-exists}In the case $\lim_{\tau\to-\infty}-\tau\omega_{s}(\tau)=+\infty$, the number $\omega\equiv\lim_{\tau\to-\infty}\omega_{ef}(\tau)$ is well-defined.
\end{lem}

\begin{proof}
By \prettyref{lem:omega-ef-equivalent-to-omega-s}, it suffices to show that $\lim_{\tau\to-\infty}\omega_{s}(\tau)$ converges. This is equivalent to showing that for all $\tau_{1}\leq\tau_{2}$, 
\[
\left|\omega_{s}(\tau_{2})-\omega_{s}(\tau_{1})\right|\to0\textrm{ as }\tau_{2}\to-\infty.
\]
 We can estimate the difference $\omega_{s}(\tau_{2})-\omega_{s}(\tau_{1})=\int_{\tau_{1}}^{\tau_{2}}\omega_{s}'(\tau)\,\mathrm{d}\tau$ by using the differential equation \eqref{eq:omega-eom}. Since $\left|\sn(u)\right|\leq1$ and $\left|\sn'(u)\right|\leq1$, 
\[
\left|\omega_{s}'(\tau)\right|\leq\frac{1}{3}\left(\frac{\xi}{-\tau}\right)^{2}\left(1+\left|\frac{2}{9}-\frac{1}{\xi^{2}}\right|\left(\frac{\xi}{-\tau\omega_{s}(\tau)}\right)\right).
\]
 Since $-\tau\omega_{s}(\tau)\to\infty$ as $\tau\to-\infty$, for any $\epsilon>0$, there exists a $\tau_{*}$ such that for all $\tau\leq\tau_{*},$ 
\[
\left|\omega_{s}'(\tau)\right|\leq\frac{1}{3}\left(\frac{\xi}{-\tau}\right)^{2}\left(1+\epsilon\right).
\]
 Here we choose $\tau_{*}$ corresponding to $\epsilon=2$. It follows that for any $\tau_{1}\leq\tau_{2}\leq\tau_{*}$, 
\[
\left|\omega_{s}(\tau_{2})-\omega_{s}(\tau_{1})\right|\leq\int_{\tau_{1}}^{\tau_{2}}\left|\omega_{s}'(\tau)\right|\,\mathrm{d}\tau\leq\frac{\xi^{2}}{-\tau_{2}}\to0\ \textrm{ as }\ensuremath{\tau_{2}\to-\infty}.
\]
\end{proof}
\begin{lem}
\label{lem:omega-positive}If $\lim_{\tau\to-\infty}-\tau\omega_{s}(\tau)=+\infty$ then $\omega>0$ and the function $s(\tau)$ is $\mathcal{O}(\omega)$ as $\tau\to-\infty$. 
\end{lem}

Suppose for contradiction that $\omega=0$. Then the proof of \prettyref{lem:ws-limit-exists} gives 
\[
-\tau_{2}\omega_{s}(\tau_{2})=-\tau_{2}\omega_{s}(\tau_{2})-\tau_{2}\omega=-\tau_{2}\int_{-\infty}^{\tau_{2}}\omega_{s}'(\tau)\,\mathrm{d}\tau\leq\xi^{2},
\]
 which contradicts $-\tau\omega_{s}(\tau)\to+\infty$. Thus $\omega>0$. In this case, $\left|s(\tau)\right|\leq2\omega$ for sufficiently negative $\tau$. Thus $s(\tau)$ is $\mathcal{O}(\omega)$. 
\begin{lem}
\label{lem:omega-s-estimate}If $\omega>0$ then 
\begin{equation}
\omega_{s}(\tau)=\left(1+\mathcal{O}\left(\frac{\xi^{2}}{\omega^{2}\tau^{2}}\right)\right)\omega\ \textrm{ as }\tau\to-\infty.\label{eq:omega-quadratic-asymptotic}
\end{equation}
\end{lem}

\begin{proof}
For this proof it is easier not to use polar coordinates. We need the identity 
\[
\frac{\mathrm{d}}{\mathrm{d}\tau}\omega_{s}(\tau)^{4}=2s'(\tau)\left(s''(\tau)+2s(\tau)^{3}\right)=\frac{2}{3}\left(\frac{\xi}{-\tau}\right)^{2}\left(s(\tau)^{2}\right)'+\frac{4}{3}\left(\frac{2}{9}-\frac{1}{\xi^{2}}\right)\left(\frac{\xi}{-\tau}\right)^{3}s'(\tau),
\]
 which follows from \prettyref{eq:s-eom}. By the fundamental theorem of calculus, 
\begin{align*}
\omega_{s}(\tau_{2})^{4}-\omega^{4} & =\int_{-\infty}^{\tau_{2}}\frac{\mathrm{d}}{\mathrm{d}\tau}\omega_{s}(\tau)^{4}\,\mathrm{d}\tau=\int_{-\infty}^{\tau_{2}}\left(\frac{2}{3}\left(\frac{\xi}{-\tau}\right)^{2}\left(s(\tau)^{2}\right)'+\frac{4}{3}\left(\frac{2}{9}-\frac{1}{\xi^{2}}\right)\left(\frac{\xi}{-\tau}\right)^{3}s'(\tau)\right)\,\mathrm{d}\tau\\
 & =\frac{2\xi^{2}s(\tau_{2})^{2}}{3\tau_{2}^{2}}+\left(\frac{2}{9}-\frac{1}{\xi^{2}}\right)\frac{4\xi^{3}s(\tau_{2})}{-3\tau_{2}^{3}}-\int_{-\infty}^{\tau_{2}}\left(\frac{4\xi^{2}s(\tau)^{2}}{-3\tau{}^{3}}+\left(\frac{2}{9}-\frac{1}{\xi^{2}}\right)\frac{4\xi^{3}s(\tau)}{\tau^{4}}\right)\,\mathrm{d}\tau,
\end{align*}
 where the last equality follows from integration by parts. From \prettyref{lem:omega-positive} we obtain 
\[
\omega_{s}(\tau_{2})^{4}-\omega^{4}=\omega^{4}\left(\mathcal{O}\left(\frac{\xi^{2}}{\omega^{2}\tau_{2}^{2}}\right)+\int_{-\infty}^{\tau_{2}}\mathcal{O}\left(\frac{\xi^{2}}{-\omega^{2}\tau^{3}}\right)\,\mathrm{d}\tau\right)=\omega^{4}\mathcal{O}\left(\frac{\xi^{2}}{\omega^{2}\tau_{2}^{2}}\right).
\]
 Solving for $\omega_{s}(\tau_{2})$ and using $\sqrt[4]{1+\epsilon}=1+\mathcal{O}\left(\epsilon\right)$, we obtain \prettyref{eq:omega-quadratic-asymptotic}.  
\end{proof}
\begin{lem}
\label{lem:u0-error}If $\omega>0$ then the limit $u_{0}\equiv\lim_{\tau\to-\infty}\left(u_{s}(\tau)-\omega\tau\right)$ is well-defined (modulo the period of $\sn$), and $u_{s}(\tau)=\omega\tau+u_{0}+\mathcal{O}(\xi^{2}/(-\omega\tau)).$ 
\end{lem}

\begin{proof}
From \prettyref{eq:u-eom}, 
\[
u_{s}'(\tau)=\omega\left(1+\mathcal{O}\left(\frac{\xi^{2}}{\omega^{2}\tau_{2}^{2}}\right)\right)\left(1+\mathcal{O}\left(\frac{\xi^{2}}{\omega^{2}\tau_{2}^{2}}\right)\mathcal{O}\left(1\right)\right)=\omega\left(1+\mathcal{O}\left(\frac{\xi^{2}}{\omega^{2}\tau_{2}^{2}}\right)\right).
\]
 Thus if $\tau_{1}\leq\tau_{2},$ then 
\[
\left|(u(\tau_{2})-\omega\tau_{2})-(u(\tau_{1})-\omega\tau_{1})\right|=\int_{\tau_{1}}^{\tau_{2}}\mathcal{O}\left(\frac{\xi^{2}}{\omega\tau_{2}^{2}}\right)\,\mathrm{d}\tau=\mathcal{O}\left(\frac{\xi^{2}}{-\omega\tau_{2}}\right).
\]
 Since this approaches zero as $\tau_{2}\to-\infty$, the limit $u_{0}$ exists, and the claimed asymptotic holds.  
\end{proof}
\begin{proof}[Proof of \prettyref{thm:oscillatory}]
\label{pf:oscillatory} To show that $(\omega,u_{0})\mapsto(\lambda\omega,u_{0})$, one applies the transformation \eqref{eq:symmetry3} to \prettyref{eq:approx-quartic}. It remains to prove the error estimate. From the definitions \eqref{eq:def-s} of $s(\tau)$ and \eqref{eq:polar-def} of Jacobi polar coordinates, 
\[
ef(\tau)=\xi/(-3\tau)+s(\tau)=\omega_{s}(\tau)\sn(u_{s}(\tau))+\mathcal{O}(\xi/(-\tau))\ \textrm{ as }\tau\to-\infty.
\]
 Note that since $\sn'(u)$ is bounded, it follows that $\sn(u+\epsilon)=\sn(u)+\mathcal{O}(\epsilon)$. Combining this with the estimates of \prettyref{thm:omega-s}, 
\begin{align*}
ef(\tau) & =\omega\left(1+\mathcal{O}\left(\frac{\xi^{2}}{\omega^{2}\tau^{2}}\right)\right)\left(\sn(\omega\tau+u_{0})+\mathcal{O}(\xi^{2}/(-\omega\tau))\right)+\mathcal{O}(\xi/(-\tau))\\
 & =\omega\sn(\omega\tau+u_{0})+\mathcal{O}((\xi+\xi^{2})/(-\tau)),
\end{align*}
 as desired. 
\end{proof}
\begin{proof}[Proof of \prettyref{thm:approx-w}]
\label{pf:approx-w} The parameter $\omega$ which appeared in the proof of \prettyref{thm:oscillatory} was indeed the same $\omega$ defined as $\lim_{\tau\to-\infty}\omega_{ef}(\tau)$. In order to show that $\omega>0$ for the two given cases, by \prettyref{lem:oscillation-condition} and \prettyref{lem:omega-positive} it suffices to show that the trajectory converges to $\bm{\infty}$ in the infinite past. In the case $0\leq\xi<2$, \prettyref{lem:asymptotic-table} implies that all non-zero trajectories limit to $\bm{\infty}$ in the infinite past, and $\omega_{ef}(\tau)>0$ implies that the trajectory is nonzero. For general $\xi$, if $\omega_{ef}(\tau_{1})>\tfrac{4}{3}\xi/(-\tau_{1})$ then \prettyref{lem:D-negative-omega-zero} implies that $D(N_{1})>0$, and then \prettyref{lem:dichotomy} implies that the trajectory limits to $\bm{\infty}$. 
\end{proof}

\section{\label{app:gaugefields}Homogeneous gauge fields}

Throughout this appendix, we investigate the properties of general $\mathrm{SU}(2)$ gauge fields $A=A_{\mu}^{a}(\tau,\vec{x})$ which are not subject to any particular equation of motion. We make use of the following notation: 
\[
A^{a}=(A_{0}^{a},\vec{A}^{a}),\ \implies\ A=(A_{0},\vec{A}).
\]
 Thus $\vec{A}=\vec{A}_{i}^{a}$ denotes a $3\times3$ matrix, while $A=A_{\mu}^{a}$ denotes a $3\times4$ matrix.

In \prettyref{app:homo-gauge} we introduce definitions of homogeneity and isotropy for a gauge field. Then we prove that when $A$ is homogeneous, there is a gauge where $A_{0}=0$ and $\vec{A}(\tau,\vec{x})=\vec{A}(\tau)$. In \prettyref{app:iso-gauge} we prove that if $A$ is moreover isotropic, then there is a gauge in which the standard\footnote{This ansatz dates back at least to 1989 (see \cite{Verbin:1989sg}). A convincing justification of this ansatz appears in Section 5.1 of \cite{Maleknejad:2012fw}. Our mathematically rigorous proof extends this work, in the case of $\mathrm{SU}(2)$, by accounting for additional edge cases and explicitly constructing the necessary gauge transformations. } ansatz $A_{i}^{a}(\tau)=f(\tau)\,\delta_{i}^{a}$ holds. We describe in \prettyref{app:global-symmetries} the diagonal $\mathrm{SO}(3)$ subgroup which fixes a homogeneous and isotropic gauge-field background, and how this is useful for the decomposition of perturbations. In \prettyref{app:quant-aniso}, we introduce some observables for any homogeneous gauge field, and we prove that they are gauge-invariant (with a very minor caveat). We use these observables to quantify isotropy and anisotropy. These observables are also used in \prettyref{subsec:Anisotropic-background-fields} to plot the time-evolution of a non-isotropic background gauge field. Finally in \prettyref{app:isotrop-proof}, we use these observables to complete the proof of \prettyref{thm:isotropy} from \prettyref{app:iso-gauge}. 

\subsection{\label{app:homo-gauge}Homogeneous and temporal gauge }

In this subsection, we define homogeneity for gauge fields, and we construct a gauge suitable for studying such fields.

A gauge field is homogeneous when every translation is equivalent to a gauge transformation, i.e.
\begin{defn}
\label{def:homogeneous}A gauge field $A$ is \emph{homogeneous} if for every spatial translation by $\Delta\vec{x}$, there exists a gauge transformation $g_{\Delta\vec{x}}(\tau,\vec{x})$ such that 
\begin{equation}
A(\tau,\vec{x}+\Delta\vec{x})=g_{\Delta\vec{x}}(\tau,\vec{x})\cdot A(\tau,\vec{x}),\label{eq:homogeneous-definition}
\end{equation}
where $\cdot$ denotes a gauge transformation: 
\begin{equation}
g(\tau,\vec{x})\cdot A_{\mu}(\tau,\vec{x})\equiv g(\tau,\vec{x})A_{\mu}(\tau,\vec{x})g(\tau,\vec{x})^{-1}+(i/e)\left(\partial_{\mu}g(\tau,\vec{x})\right)g(\tau,\vec{x})^{-1}.\label{eq:gauge-formula}
\end{equation}
\end{defn}

Obvious examples of gauge fields which are homogeneous are those which satisfy the following condition.
\begin{defn}
\label{def:homogeneous-gauge}A gauge field $A(\tau,\vec{x})$ is said to be in \emph{homogeneous gauge} if $A(\tau,\vec{x})=A(\tau)$, i.e. $A(\tau,\vec{x})$ does not depend on $\vec{x}$.
\end{defn}

The following lemma is unsurprising yet not completely obvious:
\begin{lem}
\label{lem:homogeneous}If $A(\tau,\vec{x})$ is homogeneous in the sense of \prettyref{def:homogeneous}, then there exists some gauge transformation which puts it into homogeneous gauge.
\end{lem}

\begin{proof}
The strategy will be to use the definition of homogeneity to construct a gauge transformation $h(\tau,\vec{x})$ such that the gauge-transformed field $h(\tau,\vec{x})\cdot A(\tau,\vec{x})$ is spatially constant. Equivalently, $h(\tau,\vec{x})$ should satisfy 
\begin{equation}
h(\tau,\vec{x}+\Delta\vec{x})\cdot A(\tau,\vec{x}+\Delta\vec{x})=h(\tau,\vec{x})\cdot A(\tau,\vec{x})\label{eq:const-gf}
\end{equation}
 for all $\Delta\vec{x}$. 

By comparing a single translation by $\Delta\vec{x}_{1}+\Delta\vec{x}_{2}$ with two translations $\Delta\vec{x}_{1}$ followed by $\Delta\vec{x}_{2}$, we have 
\begin{align*}
g_{\Delta\vec{x}_{1}+\Delta\vec{x}_{2}}(\tau,\vec{x})\cdot A(\tau,\vec{x}) & =A(\tau,\vec{x}+\Delta\vec{x}_{1}+\Delta\vec{x}_{2})\\
 & =g_{\Delta\vec{x}_{2}}(\tau,\vec{x}+\Delta\vec{x}_{1})\cdot A(\tau,\vec{x}+\Delta\vec{x}_{1})\\
 & =g_{\Delta\vec{x}_{2}}(\tau,\vec{x}+\Delta\vec{x}_{1})\,g_{\Delta\vec{x}_{1}}(\tau,\vec{x})\cdot A(\tau,\vec{x}).
\end{align*}
 It follows that\footnote{Technically, we are assuming here that a gauge transformation is determined by its action on $A$. It is however possible that there exist global symmetries which fix $A$. In this case, the gauge transformation is determined only up to this subgroup. This issue is easily remedied by declaring that all equalities of gauge transformations are understood modulo this subgroup of symmetries. Then no further modifications to the proof are necessary.} 
\[
g_{\Delta\vec{x}_{1}+\Delta\vec{x}_{2}}(\tau,\vec{x})=g_{\Delta\vec{x}_{2}}(\tau,\vec{x}+\Delta\vec{x}_{1})\,g_{\Delta\vec{x}_{1}}(\tau,\vec{x}).
\]
 By making appropriate substitutions for $\Delta\vec{x}_{1}$, $\Delta\vec{x}_{2}$ and $\vec{x}$ in this identity, immediate consequences are 
\begin{align}
g_{0}(\vec{x}) & =\mathrm{Id},\nonumber \\
g_{\Delta\vec{x}}(\vec{x})^{-1} & =g_{-\Delta\vec{x}}(\vec{x}+\Delta\vec{x}),\nonumber \\
g_{\Delta\vec{x}}(\vec{x}) & =g_{\vec{x}-\vec{x}_{0}+\Delta\vec{x}}(\vec{x}_{0})\,g_{\vec{x}-\vec{x}_{0}}(\vec{x}_{0})^{-1},\label{eq:gauge-id3}
\end{align}
 for all values of $\vec{x}$, $\Delta\vec{x}$ and $\vec{x}_{0}$. 

Now fix a basepoint $\vec{x}_{0}$ and define $h(\tau,\vec{x})\equiv g_{\vec{x}-\vec{x}_{0}}(\tau,\vec{x}_{0})^{-1}$. It follows from \prettyref{eq:gauge-id3} that 
\[
g_{\Delta\vec{x}}(\tau,\vec{x})=h(\tau,\vec{x}+\Delta\vec{x})^{-1}h(\tau,\vec{x})
\]
 for all values of $\vec{x}$ and $\Delta\vec{x}$. Substituting this identity into \prettyref{eq:homogeneous-definition} we obtain \prettyref{eq:const-gf} as desired.
\end{proof}
\begin{thm}
\label{thm:hom-temp-gauge}If $A(\tau,\vec{x})$ is a homogeneous gauge field in the sense of \prettyref{def:homogeneous}, then it may be gauge-transformed simultaneously into homogeneous gauge and temporal gauge, so that $A_{0}(\tau,\vec{x})=0$ and $\vec{A}(\tau,\vec{x})=\vec{A}(\tau)$. 
\end{thm}

\begin{proof}
By \prettyref{lem:homogeneous} we may assume that $A$ is in homogeneous gauge. To put $A(\tau)$ into temporal gauge, one solves for the gauge transformation $g(\tau)$ in 
\[
\tfrac{\mathrm{d}}{\mathrm{d}\tau}g(\tau)=-ieg(\tau)A_{0}(\tau).
\]
 Since $g(\tau)$ does not depend on $\vec{x}$, it preserves homogeneous gauge.
\end{proof}
We warn that sometimes homogeneous and temporal gauge is incomplete as gauge fixing: we shall see that local gauge freedom remains in the case of \prettyref{exa:local-gt}. 

\subsection{\label{app:iso-gauge}Homogeneous and isotropic \texorpdfstring{$\mathrm{SU}(2)$}{SU(2)} gauge fields}

Suppose that $A_{\mu}^{a}(\tau,\vec{x})$ is a homogeneous gauge field, as described in the previous subsection. If every rotation is equivalent to a gauge transformation, then we say that $A$ is isotropic:
\begin{defn}
\label{def:isotropic-gf}A homogeneous gauge field $A(\tau,\vec{x})$ is \emph{isotropic} if for every $R\in\mathrm{SO}(3)_{\mathrm{spatial}}\subset\mathrm{SO}(4)$ there exists a gauge transformation $g_{R}(\tau,\vec{x})$ such that 
\begin{equation}
A_{\nu}^{a}(\tau,\vec{x})\,R_{\mu}^{\nu}=\left(g_{R}(\tau,\vec{x})\cdot A(\tau,\vec{x})\right)_{\mu}^{a},\label{eq:iso-gt}
\end{equation}
 where again $\cdot$ denotes a gauge transformation as in \eqref{eq:gauge-formula}. 
\end{defn}

For the remainder of \prettyref{app:iso-gauge} we assume $A$ to be isotropic, and moreover given by $\vec{A}(\tau)$ in the homogeneous and temporal gauge of \prettyref{thm:hom-temp-gauge}. Our goal is to prove the following theorem. 
\begin{thm}
\label{thm:isotropy}Let $A_{\mu}^{a}(\tau,\vec{x})$ be a homogeneous and isotropic $\mathrm{SU}(2)$ gauge field (which is real-analytic\footnote{This technical condition is used to rule out the non-physical pathology described in \prettyref{exa:no-ansatz}. It is satisfied by the solutions of any well-behaved system of ODEs and constraint equations, and thus it applies to all solutions considered in this paper. }). Then up to a gauge transformation, $A$ is of the form 
\begin{equation}
A_{0}^{a}(\tau)=0,\quad A_{i}^{a}(\tau)=f(\tau)\,\delta_{i}^{a}\label{eq:iso-diag}
\end{equation}
 for some real-valued function $f(\tau)$. 
\end{thm}

\begin{proof}
By \prettyref{thm:hom-temp-gauge} we may assume that $A$ is in homogeneous and temporal gauge. In \prettyref{lem:SVD-SOn} we introduce an $\mathrm{SO}(n)$-version of singular value decomposition (SVD) to be used with the $3\times3$ matrix $\vec{A}(\tau)=A_{i}^{a}(\tau)$. \prettyref{def:isotropic-matrix} introduces a notion of isotropy for $3\times3$ matrices, and \prettyref{lem:matrix-isotropy} characterizes such matrices. In \prettyref{exa:local-gt} we observe that despite $A$ being an isotropic gauge field, the matrix $\vec{A}(\tau)$ can be rank-one instead of isotropic. Later in \prettyref{app:quant-aniso} we develop the tools needed to deal with this subtlety, and in \prettyref{app:isotrop-proof} we prove \prettyref{thm:rk1-is-gauge-artifact}, that the rank-one case may be eliminated by a gauge transformation, and we may therefore assume that $\vec{A}(\tau)$ is always an isotropic matrix. \prettyref{lem:non-analytic} proves that $\vec{A}(\tau)=f(\tau)S$ for some $S\in\mathrm{SO}(n)$ along any interval where $f(\tau)\neq0$, but \prettyref{exa:no-ansatz} shows that different intervals may require different $S$. \prettyref{thm:analytic} proves that when $A$ is real-analytic, a single matrix $S$ suffices for all $\tau$. Finally, a global gauge transformation whose adjoint is $S^{-1}$ brings $A$ into the desired form.  
\end{proof}
To prove the supporting lemmas, we first introduce an $\mathrm{SO}(n)$-version of SVD.
\begin{lem}
\label{lem:SVD-SOn}Let $M$ be any real $n\times n$ matrix with $n$ odd. Then there exist real $n\times n$ matrices $G$, $\Sigma$, and $R$ such that $G,R\in\mathrm{SO}(n)$, $\Sigma$ is diagonal, and $M=G\Sigma R^{T}$. This decomposition is not unique. The diagonal entries $(\sigma_{1},\ldots,\sigma_{n})$ of $\Sigma$ are called the \emph{singular values} of $M$. The singular values can be chosen to be non-negative (resp. negative) when $\det M$ is non-negative (resp. negative). Moreover, they can be chosen to satisfy $\left|\sigma_{1}\right|\geq\cdots\geq\left|\sigma_{n}\right|$. When subject to these two constraints, the singular values are uniquely determined. 
\end{lem}

\begin{proof}
The standard SVD for real $n\times n$ matrices produces matrices $G_{0}$, $\Sigma_{0}$, and $R_{0}$ with $G_{0},R_{0}\in\mathrm{O}(n)$ and $\Sigma_{0}$ diagonal. Moreover, we may choose the diagonal entries $\sigma_{1}^{0},\ldots,\sigma_{n}^{0}$ of $\Sigma_{0}$ to satisfy $\sigma_{1}^{0}\geq\cdots\geq\sigma_{n}^{0}\geq0$ in which case the singular values are uniquely determined. Upon setting 
\[
G=\begin{cases}
+G_{0} & \textrm{ if }\det G_{0}=1,\\
-G_{0} & \textrm{ if }\det G_{0}=-1,
\end{cases}\quad R=\begin{cases}
+R_{0} & \textrm{ if }\det R_{0}=1,\\
-R_{0} & \textrm{ if }\det R_{0}=-1,
\end{cases}\quad\Sigma=\begin{cases}
+\Sigma_{0} & \textrm{ if }\det M\geq0,\\
-\Sigma_{0} & \textrm{ if }\det M\leq0,
\end{cases}
\]
 the hypotheses of \prettyref{lem:SVD-SOn} are satisfied. 
\end{proof}
In analogy with \prettyref{def:isotropic-gf} of an isotropic gauge field, we make a similar definition for matrices:
\begin{defn}
\label{def:isotropic-matrix}A $3\times3$ matrix $M$ is said to be \emph{isotropic} if for each $R\in\mathrm{SO}(3)$ there exists some $G_{R}\in\mathrm{SO}(3)$ such that $MR=G_{R}M$.
\end{defn}

In the case $M=\vec{A}(\tau_{0})$, the matrix $M$ is isotropic if and only if every spatial rotation is equivalent to a \emph{global} gauge transformation for the restriction of $\vec{A}(\tau)$ to the time-slice $\tau_{0}$. 

We have the following characterization of isotropic matrices:
\begin{lem}
\label{lem:matrix-isotropy}A $3\times3$ matrix $M$ is isotropic if and only if $M$ is a scalar multiple of an orthogonal matrix, or equivalently if the singular values of $M$ satisfy $\left|\sigma_{1}\right|=\left|\sigma_{2}\right|=\left|\sigma_{3}\right|$. 
\end{lem}

\begin{proof}
We will demonstrate the following logical implications:
\begin{align*}
\left|\sigma_{1}\right|=\left|\sigma_{2}\right|=\left|\sigma_{3}\right| & \implies M\textrm{ is a scalar multiple of an orthogonal matrix}\\
 & \implies M\textrm{ is isotropic}\\
 & \implies\left|\sigma_{1}\right|=\left|\sigma_{2}\right|=\left|\sigma_{3}\right|.
\end{align*}
 Since these three conditions will form a loop of implications, they are logically equivalent. 

If $\left|\sigma_{1}\right|=\left|\sigma_{2}\right|=\left|\sigma_{3}\right|$, then it is easy to express $\Sigma$ from \prettyref{lem:SVD-SOn} as $\lambda\equiv\sqrt[3]{\det M}$ times some diagonal matrix in $\mathrm{SO}(n)$. In this case, \prettyref{lem:SVD-SOn} expresses $M$ as $\lambda$ times a product of three elements of $\mathrm{SO}(n)$, establishing the first implication. 

For the next implication, suppose that $M=\lambda S$ for $S\in\mathrm{SO}(n)$. Then for any $R\in\mathrm{SO}(n)$, we verify that $MR=\lambda SR=SRS^{T}(\lambda S)$, so that $M$ is isotropic by \prettyref{def:isotropic-matrix} with $G_{R}=SRS^{T}$. 

For the final implication, suppose that $M$ is isotropic, and that $M=G_{1}\Sigma_{1}R_{1}^{T}$ is an SVD in the sense of \prettyref{lem:SVD-SOn}. By taking $R=R_{1}$ in the definition of isotropy, we conclude that $G_{R_{1}}M=MR_{1}=G_{1}\Sigma_{1}$. Since left-multiplication by orthogonal matrices (here $G_{R_{1}}$ and $G_{1}$) preserves the norms of matrix columns, we conclude that the columns of $M$ have norms $\left(\left|\sigma_{1}\right|,\left|\sigma_{2}\right|,\left|\sigma_{3}\right|\right)$. Next we consider the result of taking $R=R_{1}C$ where 
\begin{gather*}
C=\left(\begin{array}{ccc}
0 & 0 & 1\\
1 & 0 & 0\\
0 & 1 & 0
\end{array}\right).
\end{gather*}
 We similarly conclude that the columns of $M$ have norms $\left(\left|\sigma_{2}\right|,\left|\sigma_{3}\right|,\left|\sigma_{1}\right|\right)$. Therefore $\left|\sigma_{1}\right|=\left|\sigma_{2}\right|=\left|\sigma_{3}\right|$. 
\end{proof}
We note that if the uniqueness constraints of \prettyref{lem:SVD-SOn} are imposed, then $\left|\sigma_{1}\right|=\left|\sigma_{2}\right|=\left|\sigma_{3}\right|$ implies that $\sigma_{1}=\sigma_{2}=\sigma_{3}$. Thus if $M$ is isotropic, then it is possible to find an SVD of the type in \prettyref{lem:SVD-SOn} where $\Sigma$ is a scalar multiple of the identity matrix.

At this stage, it is tempting to mistakenly claim that since $A(\tau)$ is an isotropic gauge field in the sense of \prettyref{def:isotropic-gf}, it follows that for each $\tau_{0}$ the matrix $\vec{A}(\tau_{0})$ must be isotropic in the sense of \prettyref{def:isotropic-matrix}. However, the following example demonstrates that an isotropic gauge field can indeed have a non-isotropic matrix.
\begin{example}
\label{exa:local-gt}For any real number $\sigma_{1}$, consider the gauge transformation $\exp\left(-ie\sigma_{1}\mathbf{T}_{1}x^{1}\right)$ applied to the zero gauge field $A_{\mu}^{a}=0$. The transformed gauge field has $A_{1}^{1}=\sigma_{1}$ with all other components zero. Thus a nonzero constant gauge field with singular values $(\sigma_{1},0,0)$ is gauge-equivalent to the zero gauge field.
\end{example}

The following definition is useful for characterizing \prettyref{exa:local-gt} in a coordinate-independent manner:
\begin{defn}
\label{def:rank-one}A $3\times3$ matrix $M$ is said to be \emph{rank-one} if exactly one of the singular values of $M$ is non-zero.
\end{defn}

Thus \prettyref{exa:local-gt} shows that if $\vec{A}(\tau)$ is any rank-one matrix which is constant in $\tau$, then $\vec{A}(\tau)$ is gauge-equivalent to zero. In \prettyref{app:isotrop-proof} we show that this is the only example in which a non-isotropic matrix can arise from an isotropic gauge field. Since the zero matrix is isotropic, we may assume, up to gauge, that the matrix $\vec{A}(\tau)$ is always isotropic (see \prettyref{thm:rk1-is-gauge-artifact}).

In order to analyze $\vec{A}(\tau)$, it is useful to work with the electric field matrix 
\[
E_{i}^{b}(\tau,\vec{x})\equiv-a(\tau)^{-2}F_{0i}^{b}(\tau,\vec{x}).
\]
For convenience, we introduce the comoving quantity 
\begin{align*}
\vec{\mathcal{E}}(\tau,\vec{x}) & \equiv-a(\tau)^{2}\vec{E}(\tau,\vec{x})=F_{0i}^{b}(\tau,\vec{x}).
\end{align*}
 When $\vec{A}(\tau)$ is in homogeneous and temporal gauge, it follows that 
\begin{equation}
\vec{\mathcal{E}}(\tau)=\tfrac{\mathrm{d}}{\mathrm{d}\tau}\vec{A}(\tau).\label{eq:E-is-derivative}
\end{equation}
 Moreover, $\vec{\mathcal{E}}(\tau)$ is always an isotropic matrix when $A$ is homogeneous and isotropic:
\begin{lem}
\label{lem:elec-is-iso}If $\vec{A}(\tau)$ is a homogeneous and isotropic $\mathrm{SU}(2)$ gauge field in homogeneous and temporal gauge, then the electric field matrix $\vec{E}(\tau_{0})$ and the corresponding comoving quantity $\vec{\mathcal{E}}(\tau_{0})$ are isotropic matrices for all $\tau_{0}$. 
\end{lem}

\begin{proof}
First note that $\vec{\mathcal{E}}(\tau,\vec{x})$ satisfies the tensorial transformation property $\vec{\mathcal{E}}(\tau,\vec{x})\mapsto G(\tau,\vec{x})\,\vec{\mathcal{E}}(\tau,\vec{x})$ under a gauge transformation $A\mapsto g(\tau,\vec{x})\cdot A$, where $G(\tau,\vec{x})$ denotes the adjoint $\mathrm{SO}(3)$ matrix corresponding to $g(\tau,\vec{x})$. Next we observe that $\vec{\mathcal{E}}(\tau_{0})$ is an isotropic matrix for each $\tau_{0}$, which follows from the isotropy of $A$ as follows. For any $R\in\mathrm{SO}(3)$ we have 
\begin{equation}
g_{R}(\tau_{0},\vec{x})\cdot\vec{A}(\tau_{0})=\vec{A}(\tau_{0})\,R^{T},\label{eq:isotropy-A_sp}
\end{equation}
 and consequently 
\[
G_{R}(\tau_{0},\vec{x})\,\vec{\mathcal{E}}(\tau_{0})=\vec{\mathcal{E}}(\tau_{0})\,R^{T}
\]
 for each $\vec{x}$. Choosing an arbitrary point $\vec{x}_{0}$, it follows that $G_{R}(\tau_{0},\vec{x}_{0})$ provides the necessary group element for matrix isotropy (\prettyref{def:isotropic-matrix}) to be satisfied. Since $\vec{E}(\tau_{0})$ and $\vec{\mathcal{E}}(\tau_{0})$ are proportional, the same $G_{R}(\tau_{0},\vec{x}_{0})$ applies also to $\vec{E}(\tau)$. 

We note that this proof uses \prettyref{def:isotropic-gf} of an isotropic gauge field without assuming that $\vec{A}(\tau_{0})$ is an isotropic matrix. Thus this lemma will be useful in \prettyref{app:isotrop-proof} for understanding the rank one case. The same argument straightforwardly extends to prove that the magnetic field is isotropic.
\end{proof}
\begin{lem}
\label{lem:non-analytic}Let $\vec{A}(\tau)$ be a homogeneous and isotropic $\mathrm{SU}(2)$ gauge field in homogeneous and temporal gauge such that $\vec{A}(\tau)$ is isotropic. For any interval along which $\vec{A}(\tau)\neq0$, there exists some $S\in\mathrm{SO}(n)$ such that 
\[
\vec{A}(\tau)=f(\tau)S.
\]
\end{lem}

\begin{proof}
From the characterization of \prettyref{lem:matrix-isotropy}, we know that $\vec{A}(\tau)=f(\tau)S(\tau)$, where $f(\tau)$ is nowhere zero along the interval in question. We wish to show that $S(\tau)$ is constant. Since $\vec{\mathcal{E}}(\tau)$ is an isotropic matrix and $S(\tau)\in\mathrm{SO}(3)$, we know from \prettyref{lem:matrix-isotropy} that $S(\tau)^{-1}\vec{\mathcal{E}}(\tau)$ is also an isotropic matrix, and 
\[
S(\tau)^{-1}\vec{\mathcal{E}}(\tau)=S(\tau)^{-1}\tfrac{\mathrm{d}}{\mathrm{d}\tau}\vec{A}(\tau)=f'(\tau)\,I+f(\tau)\,S^{-1}(\tau)\tfrac{\mathrm{d}}{\mathrm{d}\tau}S(\tau).
\]
 In particular, this must be a scalar multiple of an orthogonal matrix for each $\tau$. Recall that the eigenvalues of an orthogonal matrix all have absolute-value one. Thus the eigenvalues of $S(\tau)^{-1}\vec{\mathcal{E}}(\tau)$ must all have the same absolute value. The matrix $S^{-1}(\tau)\tfrac{\mathrm{d}}{\mathrm{d}\tau}S(\tau)$ is antisymmetric, and thus its components can be written as $\xi^{b}(\tau)\varepsilon_{aib}$ for some $\xi^{b}(\tau)$. The three eigenvalues of $S(\tau)^{-1}\vec{\mathcal{E}}(\tau)$ are 
\[
\left\{ f'(\tau),f'(\tau)\pm if(\tau)\sqrt{\xi^{b}(\tau)\xi^{b}(\tau)}\right\} .
\]
 For these to have the same absolute value, it must be that $\xi^{b}(\tau)=0$ since by assumption $f(\tau)\neq0$. Thus $S^{-1}(\tau)\tfrac{\mathrm{d}}{\mathrm{d}\tau}S(\tau)=0$ and hence $S(\tau)$ is constant.
\end{proof}
The following example shows that without some additional hypothesis, $S$ need not be constant where $f(\tau)=0$.
\begin{example}
\label{exa:no-ansatz}For any fixed time $\tau_{0}$, any $f(\tau)$ satisfying $f(\tau_{0})=0$, and any distinct $S_{+},S_{-}\in\mathrm{SO}(n)$, consider 
\[
A_{0}(\tau)=0,\quad\vec{A}(\tau)=\begin{cases}
f(\tau)S_{+} & \textrm{if }\tau\geq\tau_{0},\\
f(\tau)S_{-} & \textrm{if }\tau\leq\tau_{0}.
\end{cases}
\]
 This gauge field is homogeneous, isotropic, and not gauge-equivalent to the ansatz of \prettyref{thm:isotropy}. 
\end{example}

Such examples are not expected to occur in practice. If $\vec{A}(\tau)$ solves some well-behaved system of ODEs and constraint equations, then $\vec{A}(\tau)$ is real-analytic, meaning that it has a locally-convergent power-series expansion. The following theorem sharpens \prettyref{lem:non-analytic} by excluding cases such as \prettyref{exa:no-ansatz}.
\begin{thm}
\label{thm:analytic}Let $\vec{A}(\tau)$ be a real-analytic, homogeneous and isotropic $\mathrm{SU}(2)$ gauge field in homogeneous and temporal gauge such that $\vec{A}(\tau)$ is isotropic. There exists some $S\in\mathrm{SO}(n)$ such that 
\begin{equation}
\vec{A}(\tau)=f(\tau)S\ \textrm{ for all }\tau.\label{eq:const-son}
\end{equation}
\end{thm}

\begin{proof}
According to the principle of unique continuation for real-analytic functions, if a real-analytic function $g(\tau)$ vanishes along any interval of positive width, then $g(\tau)$ is identically zero for all $\tau$. For any interval along which $f(\tau)\neq0$, \prettyref{lem:non-analytic} implies that $\vec{A}(\tau)$ is confined to a one-dimensional subspace of $3\times3$ matrices. Thus there are eight complementary linear combinations of components of $\vec{A}(\tau)$ which vanish along the interval. By the principle of unique continuation, these components vanish for all $\tau$, and thus \prettyref{eq:const-son} holds for all $\tau$. 
\end{proof}
Up to the proof of \prettyref{thm:rk1-is-gauge-artifact} given in \prettyref{app:isotrop-proof}, these results complete the proof of \prettyref{thm:isotropy}.

\subsection{\label{app:global-symmetries}Global symmetries of homogeneous and isotropic \texorpdfstring{$\mathrm{SU}(2)$}{SU(2)} gauge fields}

Understanding the global symmetries of a given field is important for perturbation theory. Namely, if some field is fixed by a global symmetry group, then linear perturbations around that field decompose into irreducible representations of that group. For example, suppose that $A(\tau)$ is any homogeneous $\mathrm{SU}(2)$ gauge field. Then $A(\tau)$ transforms under a pair of global $\mathrm{SO}(3)$ symmetries denoted by $\mathrm{SO}(3)_{\mathrm{gauge}}$ and $\mathrm{SO}(3)_{\mathrm{spatial}}$. The group of spatial rotations is $\mathrm{SO}(3)_{\mathrm{spatial}}$, while $\mathrm{SO}(3)_{\mathrm{gauge}}$ is the adjoint group of global $\mathrm{SU}(2)$ gauge transformations. If $(G,R)\in\mathrm{SO}(3)_{\mathrm{gauge}}\times\mathrm{SO}(3)_{\mathrm{spatial}}$, then the transformation is given by 
\begin{align}
A^{a}{}_{i}(\tau) & \mapsto G^{a}{}_{b}A^{b}{}_{j}(\tau)\left(R^{T}\right)^{j}{}_{i}\quad\left(\vec{A}(\tau)\mapsto G\vec{A}(\tau)R^{T}\right),\label{eq:diagonal-action-space}\\
A^{a}{}_{0}(\tau) & \mapsto G^{a}{}_{b}A^{b}{}_{0}(\tau)\quad\left(A_{0}(\tau)\mapsto GA_{0}(\tau)\right).\label{eq:diagonal-action-time}
\end{align}
Suppose now that $A(\tau)$ is isotropic (and not gauge-equivalent to zero). By \prettyref{thm:isotropy} we can take $A(\tau)$ to be of the form $\vec{A}(\tau)=f(\tau)I$, $A_{0}(\tau)=0$. The subgroup of $\mathrm{SO}(3)_{\mathrm{gauge}}\times\mathrm{SO}(3)_{\mathrm{spatial}}$ which leaves $A(\tau)$ fixed is evidently the diagonal $\mathrm{SO}(3)$ subgroup consisting of pairs $(G,R)$ such that $G=R$. Thus for any homogeneous perturbation $\vec{A}(\tau)+\vec{P}(\tau)\epsilon+\mathcal{O}(\epsilon^{2})$ we may decompose 
\begin{equation}
P_{i}^{a}(\tau)=s(\tau)\delta_{i}^{a}+v^{j}(\tau)\varepsilon_{ija}+T_{i}^{a}(\tau)\label{eq:one-three-five}
\end{equation}
 where $v^{j}$ is a vector, $\varepsilon$ is the Levi-Civita symbol, and $T_{i}^{a}$ is a traceless symmetric tensor, corresponding to the decomposition $\mathbf{3}\otimes\mathbf{3}=\mathbf{1}\oplus\mathbf{3}\oplus\mathbf{5}$. Two remarks are in order:

Firstly, although this decomposition contains a scalar, vector and tensor, it differs from the SVT/helicity decomposition described in \prettyref{sec:helbas}. The latter applies to inhomogeneous perturbations, and it is defined in terms of charges under the $\mathrm{SO}(2)$ subgroup of the diagonal $\mathrm{SO}(3)$ corresponding to rotations around the axis specified by a Fourier mode. In contrast, \prettyref{eq:one-three-five} is a decomposition into irreducible $\mathrm{SO}(3)$ representations. 

Secondly, it's important to note that in the case where $A(\tau)$ is the zero gauge field (i.e. $f(\tau)=0$ for all~$\tau$), it is impossible to sensibly decompose perturbations (at least without some additional structure). This is because all of $\mathrm{SO}(3)_{\mathrm{gauge}}\times\mathrm{SO}(3)_{\mathrm{spatial}}$ acts trivially on the zero gauge field, and $\mathbf{3}_{\mathrm{gauge}}\otimes\mathbf{3}_{\mathrm{spatial}}$ is an irreducible representation of this full group. In contrast, when $f(\tau)\neq0$, $\vec{A}(\tau)$ determines an identification between spatial directions and Lie algebra directions, leading to a distinguished diagonal subgroup and enabling the previous decomposition to proceed.\footnote{For this reason, the analysis of non-abelian gauge theories is actually much easier when a non-zero background field is present. }

\subsection{\label{app:quant-aniso}Quantifying anisotropy}

The purpose of this subsection is to develop the tools used in \prettyref{subsec:Anisotropic-background-fields} to measure the anisotropy of a homogeneous $\mathrm{SU}(2)$ gauge-field background. These are the same tools needed to complete the proof of \prettyref{thm:isotropy}, which is carried out in \prettyref{app:isotrop-proof}. Throughout this subsection, we take $\vec{A}(\tau)$ to be an $\mathrm{SU}(2)$ gauge field in homogeneous and temporal gauge, but is not necessarily isotropic.

Recall from \prettyref{lem:matrix-isotropy} that a $3\times3$ matrix $M$ is \emph{isotropic} when its singular values are all equal in absolute value. We wish to apply this to the case $M=\vec{A}(\tau)$. In this case, left-multiplication by an $\mathrm{SO}(3)$ matrix corresponds to the adjoint action of a spatially-constant $\mathrm{SU}(2)$ gauge transformation. Right-multiplication by (the transpose of) an $\mathrm{SO}(3)$ matrix corresponds to a spatial rotation. We seek scalars which are invariant under both left and right $\mathrm{SO}(3)$ transformations. Since $\dim(\mathrm{SO}(3))=3$, assuming that the two $\mathrm{SO}(3)$ symmetries are nondegenerate, we expect a total of $9-2\times3=3$ independent scalars. From \prettyref{lem:SVD-SOn}, such scalars must be functions of the three singular values. A convenient choice is the polynomials 
\begin{align}
I_{2}(M) & \equiv\left|M\right|^{2}\equiv M_{i}^{a}M_{i}^{a}=\sigma_{1}^{2}+\sigma_{2}^{2}+\sigma_{3}^{2},\label{eq:def-invariants}\\
I_{3}(M) & \equiv\det M_{i}^{a}=\sigma_{1}\sigma_{2}\sigma_{3},\nonumber \\
I_{4}(M) & \equiv\left(\tfrac{1}{2}\varepsilon_{ijk}\varepsilon^{abc}M_{i}^{a}M_{j}^{b}\right)^{2}=\left(\sigma_{2}\sigma_{3}\right)^{2}+\left(\sigma_{3}\sigma_{1}\right)^{2}+\left(\sigma_{1}\sigma_{2}\right)^{2}.\nonumber 
\end{align}
 As shown in \prettyref{thm:invariant-scalars}, any invariant scalar is determined as some function of these three quantities. (They also occur in Sec.~4.1.2 of \cite{mares2010}, where this problem occurs in a slightly different context.)

In the isotropic case $\sigma_{1}=\sigma_{2}=\sigma_{3}=f$, we have 
\begin{align*}
I_{2}(M) & =3f^{2},\qquad I_{3}(M)=f^{3}\qquad I_{4}(M)=3f^{4}.
\end{align*}

In the general case, these scalars satisfy certain inequalities. The inequality of arithmetic and geometric means implies that 
\[
3\sqrt{3}\left|\det M\right|\leq\left|M\right|^{3},
\]
 where the inequality is saturated when $M$ is isotropic (i.e. $\left|\sigma_{1}\right|=\left|\sigma_{2}\right|=\left|\sigma_{3}\right|$). Thus 
\[
-1\leq3\sqrt{3}\frac{I_{3}(M)}{\left|M\right|^{3}}\leq1,
\]
 with isotropy when $3\sqrt{3}\frac{I_{3}(M)}{\left|M\right|^{3}}=\pm1$. 

Note that $I_{4}(M)\geq0$ since it is a sum of squares, and $I_{4}(M)=0$ precisely when $M$ is rank-one or zero. Furthermore, note that 
\[
\left|M\right|^{4}-3I_{4}(M)=\tfrac{1}{6}(2\sigma_{1}^{2}-\sigma_{2}^{2}-\sigma_{3}^{2})^{2}+\textrm{cyclic permutations}\geq0,
\]
 where this inequality is saturated exactly when $M$ is isotropic. In summary,
\begin{gather}
0\leq\frac{3I_{4}(M)}{\left|M\right|^{4}}\leq1,\\
I_{4}(M)=0\textrm{ if and only if \ensuremath{M} is rank one or zero,}\label{eq:two-sv-vanish}\\
3I_{4}(M)-I_{2}(M)^{2}=0\textrm{ if and only if \ensuremath{M} is isotropic.}\label{eq:isotropy-characterization}
\end{gather}

As a sort of polar decomposition, we can consider the radial coordinate $\left|M\right|=\sqrt{I_{2}(M)}$ together with two other quantities which are invariant under scaling. As such, we define 
\begin{equation}
\left(D,E\right)\equiv\left(3\sqrt{3}\frac{I_{3}(M)}{\left|M\right|^{3}},\frac{3I_{4}(M)}{\left|M\right|^{4}}\right)\in\left[-1,1\right]\times\left[0,1\right].\label{eq:DE-def}
\end{equation}
 (Here $E$ is a scalar, and should not be confused with the electrical field $\vec{E}$ used in \prettyref{lem:rank-one-const}.) Not all points inside this rectangle can be realized. The points which are realized belong to the enclosed region in \prettyref{fig:DE-triangle} which resembles a triangle, but with curved edges. A nonzero matrix is isotropic precisely when it corresponds to either the left or right vertex, and the bottom vertex corresponds to rank-one matrices.

In \prettyref{eq:F-norm-fn} of \prettyref{subsec:Anisotropic-background-fields} we introduce the radial quantity $F\equiv-\tau|M|/\sqrt{3}\xi$ in the context of $M=\vec{A}(\tau)$. In order to know whether the quantities $(D,E,F)$ are physically meaningful in this context, we must worry about gauge invariance. From \prettyref{exa:local-gt} \vpageref{exa:local-gt}, we note that $I_{2}(\vec{A}(\tau))$ can fail to be gauge-invariant at any time $\tau_{0}$ when $\vec{A}(\tau_{0})$ is rank-one or zero, i.e. when $I_{4}(\vec{A}(\tau_{0}))=0$. We show in \prettyref{thm:scalars-are-gauge-invt} that this is the only such case. Thus $I_{3}(\vec{A}(\tau))$ and $I_{4}(\vec{A}(\tau))$ are always gauge-invariant, while $I_{2}(\vec{A}(\tau))$ is gauge-invariant at all times $\tau_{0}$ for which $I_{4}(\vec{A}(\tau_{0}))\neq0$. In summary, for the case $M=\vec{A}(\tau)$, gauge-invariance of the quantities $(D,E,F)$ fails only at the rank-one point $(D,E)=(0,0)$ or when $F=0$. 

In any case of physical interest, this slight lack of gauge invariance presents no difficulties: trajectories with generic non-isotropic initial conditions should never pass through these bad points. Even when a non-generic trajectory passes through a bad point, a unique meaningful value of $F$ is determined by continuity. 

\begin{figure}
\begin{centering}
\includegraphics{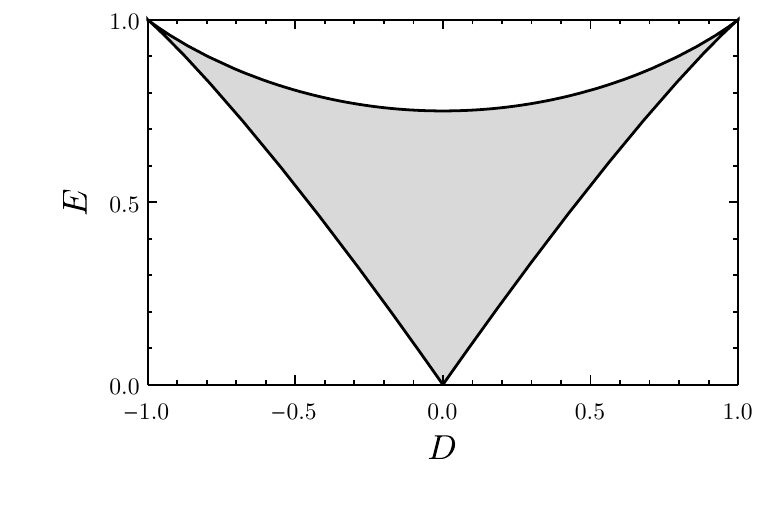}
\par\end{centering}
\caption{\label{fig:DE-triangle}Only values inside the shaded region can be realized as pairs $(D,E)$ of an actual matrix. The left vertex corresponds to isotropy with $\sigma_{i}<0$. The right vertex corresponds to isotropy with $\sigma_{i}>0$. The bottom vertex corresponds to rank one. The boundary corresponds to when two or more singular values coincide. The upper edge corresponds to a coincidence in the larger two singular values, so it is traced out by $(1,1,\sigma_{3})$ for $\sigma_{3}\in\left[-1,1\right]$. The lower two edges correspond to a coincidence of the smaller singular values, i.e. ($\pm1$,$\sigma_{3}$,$\sigma_{3}$) for $\sigma_{3}\in\left[0,1\right]$. The defining equation of this region is $\zeta\protect\geq0$, where $\zeta$ is defined in \prettyref{eq:zeta-sv}.}
\end{figure}

\subsubsection*{Completeness of the scalars $I_{n}(M)$}
\begin{thm}
\label{thm:invariant-scalars}Any scalar function of a $3\times3$ matrix $M$ which is invariant under both left and right multiplication by $\mathrm{SO}(3)$ is a function of $I_{2}(M)$, $I_{3}(M)$ and $I_{4}(M)$. 
\end{thm}

\begin{proof}
From the singular value decomposition of \prettyref{lem:SVD-SOn}, any function of $M$ which is invariant under both left and right multiplication by $\mathrm{SO}(3)$ must be expressible in terms of the singular values of $M$. When considering functions of the singular values $(\sigma_{1},\sigma_{2},\sigma_{3})$, there are two possible approaches. The first possibility is to impose a uniqueness condition on the singular values (such as the one stated in \prettyref{lem:SVD-SOn}), so that any possible function of the singular values makes sense. Thus $\sigma_{1}$ denotes the largest singular value of $M$ (in absolute value). This approach is awkward because there is no algebraic expression for $\sigma_{1}(M)$ in terms of the components of $M$, so it is difficult to compute. 

The more natural approach is to consider only combinations of singular values which are invariant under rearrangements. The singular values are determined only up to permutation and flipping pairs of signs. Such freedom to rearrange the singular values may be seen explicitly by multiplying $\Sigma$ on the left and right by $\mathrm{SO}(3)$ matrices as follows:
\begin{align}
\left(\begin{array}{ccc}
0 & 1 & 0\\
-1 & 0 & 0\\
0 & 0 & 1
\end{array}\right)\left(\begin{array}{ccc}
\sigma_{1}\\
 & \sigma_{2}\\
 &  & \sigma_{3}
\end{array}\right)\left(\begin{array}{ccc}
0 & -1 & 0\\
1 & 0 & 0\\
0 & 0 & 1
\end{array}\right) & =\left(\begin{array}{ccc}
\sigma_{2}\\
 & \sigma_{1}\\
 &  & \sigma_{3}
\end{array}\right),\nonumber \\
\left(\begin{array}{ccc}
-1\\
 & -1\\
 &  & 1
\end{array}\right)\left(\begin{array}{ccc}
\sigma_{1}\\
 & \sigma_{2}\\
 &  & \sigma_{3}
\end{array}\right)\left(\begin{array}{ccc}
1\\
 & 1\\
 &  & 1
\end{array}\right) & =\left(\begin{array}{ccc}
-\sigma_{1}\\
 & -\sigma_{2}\\
 &  & \sigma_{3}
\end{array}\right).\label{eq:SVD-signs}
\end{align}
Assuming no degeneracies, this gives a total of 24 possible ways that $\Sigma$ can be rearranged (6 permutations times 4 possibilities for signs). Correspondingly, these rearrangements are given by the action of a group with 24 elements.\footnote{This group is the symmetric group $S_{4}$ in disguise: our group action is isomorphic to the standard action of the Weyl group of $\mathfrak{so}(6)$ on a Cartan subalgebra, $\mathfrak{so}(6)$ is isomorphic to $\mathfrak{su}(4)$, and the Weyl group of $\mathfrak{su}(n)$ is the symmetric group $S_{n}$.} Since this group action can rearrange an arbitrary triple of singular values $(\sigma_{1},\sigma_{2},\sigma_{3})$ so that they satisfy the uniqueness constraints given in \prettyref{lem:SVD-SOn}, we conclude that this is the complete group of all possible rearrangements.

Next we explain how the singular values of $M$ can be reconstructed from $I_{2}(M)$, $I_{3}(M)$ and $I_{4}(M)$. Equivalently, given some numbers $I_{2}$, $I_{3}$ and $I_{4}$, we wish to solve the algebraic system 
\begin{align}
\sigma_{1}^{2}+\sigma_{2}^{2}+\sigma_{3}^{2} & =I_{2},\nonumber \\
\sigma_{1}\sigma_{2}\sigma_{3} & =I_{3},\nonumber \\
\left(\sigma_{2}\sigma_{3}\right)^{2}+\left(\sigma_{3}\sigma_{1}\right)^{2}+\left(\sigma_{1}\sigma_{2}\right)^{2} & =I_{4}\label{eq:bezout-system}
\end{align}
 for $\sigma_{1}$, $\sigma_{2}$ and $\sigma_{3}$. It can happen that this system has no real solutions, for instance when $I_{2}<0$. However, we assume that the numbers $\left(I_{2},I_{3},I_{4}\right)$ arise from some matrix $M$, so that we are guaranteed that the system has at least one real solution. Since the $24$-element group acts on the solutions of this equation, there must be at \emph{least} 24 real solutions (when the solutions are counted with multiplicity). By B\'ezout's theorem (see for example \cite{shafarevich1994}), the number of real solutions is at \emph{most} $\deg I_{2}\cdot\deg I_{3}\cdot\deg I_{4}=2\cdot3\cdot4=24$ (counting multiplicities). Therefore, the number of real solutions to the system \eqref{eq:bezout-system} is \emph{exactly} 24 (counting multiplicities). Thus the information encoded in $I_{2}(M)$, $I_{3}(M)$ and $I_{4}(M)$ is precisely that of the 24 possible rearrangements of the singular values of the original matrix $M$. 

Finally, any invariant scalar function of $M$ can be written as an invariant function of the singular values of $M$, and the above procedure gives a recipe to solve for those singular values of $M$ in terms of $I_{2}(M)$, $I_{3}(M)$ and $I_{4}(M)$. Therefore, any invariant scalar function can be written as a function of $I_{2}(M)$, $I_{3}(M)$ and $I_{4}(M)$.
\end{proof}
We also present the following short abstract proof of \prettyref{thm:invariant-scalars}. (Details about the relevant combinatorics and Molien's Theorem can be found in Chapter 1 of \cite{mukai2003}.)
\begin{proof}
The polynomials $I_{n}(M)$ are readily verified to be algebraically independent functions of the $\sigma_{i}$ by the Jacobian criterion. Thus the $I_{n}(M)$ freely generate a subring of polynomials which are invariant under the group action, and the number of linearly independent polynomials in each degree is given by the generating function $\prod_{p=2}^{4}(1-t^{p})^{-1}$. By Molien's Theorem, a tedious but completely straightforward computation shows that the generating function for the full ring of invariant polynomials is given by the same expression. Thus the subring of invariant polynomials generated by $I_{2}(M)$, $I_{3}(M)$ and $I_{4}(M)$ spans the whole ring of invariant polynomials. Since the problem of determining invariant quantities is algebraic, if there are no missing polynomials then there are no missing functions. Thus all invariant scalars of $M$ are functions of $I_{2}(M)$, $I_{3}(M)$ and $I_{4}(M)$. 
\end{proof}

\subsubsection*{Gauge invariance of the scalars $I_{n}(\vec{A}(\tau))$}

We saw in \prettyref{exa:local-gt} that $I_{2}(\vec{A}(\tau))$ fails to be gauge-invariant, so the scalars $I_{n}(\vec{A}(\tau))$ are of potentially dubious physical significance. However, it turns out that $I_{n}(\vec{A}(\tau))$ is gauge-invariant most of the time:
\begin{thm}
\label{thm:scalars-are-gauge-invt}For any $\mathrm{SU}(2)$ gauge field $\vec{A}(\tau)$ in homogeneous and temporal gauge, the quantities $I_{3}(\vec{A})$ and $I_{4}(\vec{A})$ are always gauge-invariant. Furthermore, $I_{2}(\vec{A}(\tau_{0}))$ is gauge-invariant for all $\tau_{0}$ such that $I_{4}(\vec{A}(\tau_{0}))\neq0$.
\end{thm}

To prove this theorem, it suffices to show that each $I_{n}(\vec{A}(\tau_{0}))$ is expressible in terms of scalars which are known to be gauge-invariant. This can mostly be accomplished by considering the magnetic field matrix as in the following lemma:
\begin{lem}
\label{lem:scalars-gauge-invt}For any $\tau_{0}$, the quantities $\left|I_{3}(\vec{A}(\tau_{0}))\right|$ and $I_{4}(\vec{A}(\tau_{0}))$ are gauge-invariant. If $I_{3}(\vec{A}(\tau_{0}))\neq0$ then $I_{2}(\vec{A}(\tau_{0}))$ is also gauge-invariant.
\end{lem}

\begin{proof}
We introduce the magnetic field matrix 
\begin{equation}
B_{i}^{b}(\tau_{0})\equiv\tfrac{1}{2}a(\tau_{0})^{-2}\varepsilon_{ijk}F_{jk}^{b}(\tau_{0}),\label{eq:B-def}
\end{equation}
 and the corresponding comoving quantity 
\begin{equation}
\vec{\mathcal{B}}(\tau)=a(\tau)^{2}\vec{B}(\tau)/e.\label{eq:cal-B}
\end{equation}
Since $\vec{\mathcal{B}}(\tau_{0})$ transforms as a tensor under gauge transformations (see the proof of \prettyref{lem:elec-is-iso}), each $I_{n}(\vec{\mathcal{B}}(\tau_{0}))$ is gauge-invariant. If the SVD of $\vec{A}(\tau)$ is $\vec{A}(\tau)=G_{1}(\tau)\Sigma^{A}(\tau)R_{1}(\tau)^{T}$ then
\begin{gather}
\vec{\mathcal{B}}(\tau)=G_{1}(\tau)\Sigma^{\mathcal{B}}(\tau)R_{1}(\tau)^{T},\quad\sigma_{1}^{\mathcal{B}}(\tau)=\sigma_{2}^{A}(\tau)\sigma_{3}^{A}(\tau)\ \textrm{ and cyclic permutations}.\label{eq:SV-AB}
\end{gather}
A short computation gives $\left|I_{3}(\vec{A})\right|=\sqrt{I_{3}(\vec{\mathcal{B}})}$ and $I_{4}(\vec{A})=I_{2}(\vec{\mathcal{B}})$. Since the expressions involving $I_{n}(\vec{\mathcal{B}})$ are gauge-invariant, so are the corresponding expressions involving $I_{n}(\vec{A})$. Finally, for all $\tau_{0}$ such that $I_{3}(\vec{A}(\tau_{0}))\neq0$, we have $I_{3}(\vec{\mathcal{B}}(\tau_{0}))\neq0$, and so 
\[
I_{2}(\vec{A}(\tau_{0}))=\frac{I_{2}(\vec{\mathcal{B}}(\tau_{0}))^{2}-I_{4}(\vec{\mathcal{B}}(\tau_{0}))}{3I_{3}(\vec{\mathcal{B}}(\tau_{0}))}.
\]
\end{proof}
Because we are in the context of non-abelian gauge theory, it is impossible to fully reconstruct $\vec{A}(\tau_{0})$ from $\vec{\mathcal{B}}(\tau_{0})$ up to gauge. For instance, $\vec{\mathcal{B}}(\tau_{0})$ carries no information about the sign of $I_{3}(\vec{A}(\tau_{0}))$. To sharpen the result of \prettyref{lem:scalars-gauge-invt} in order to prove \prettyref{thm:scalars-are-gauge-invt}, it is necessary to introduce a new quantity: 
\begin{proof}[Proof of \prettyref{thm:scalars-are-gauge-invt}]
 We introduce the tensorial quantity 
\[
\vec{\mathcal{C}}\equiv e^{-1}\,\vec{\nabla}^{(A)}\!\times\vec{\mathcal{B}}\quad\left(\mathcal{C}_{i}=e^{-1}\varepsilon_{ijk}\mathbf{D}_{j}\left(\mathcal{B}_{k}^{b}\mathbf{T}_{b}\right),\quad\mathcal{C}_{i}^{a}=e^{-2}F_{ij;j}^{a}\right),
\]
 where $\vec{\nabla}^{(A)}\times$ denotes the gauge-covariant curl operator, and the subscript $_{;j}$ denotes components corresponding to the gauge-covariant derivative $\mathbf{D}_{j}\mathbf{F}$. If the SVD of $\vec{A}(\tau)$ is $\vec{A}(\tau)=G_{1}(\tau)\Sigma^{A}(\tau)R_{1}(\tau)^{T}$ then $\vec{\mathcal{C}}(\tau)=G_{1}(\tau)\Sigma^{\mathcal{C}}(\tau)R_{1}(\tau)^{T}$, where 
\[
\sigma_{1}^{\mathcal{C}}=\left(\left(\sigma_{2}^{A}\right)^{2}+\left(\sigma_{3}^{A}\right)^{2}\right)\sigma_{1}^{A}\ \textrm{ and cyclic permutations}.
\]
A quick computation shows that the gauge-invariant scalar $\mathcal{B}_{i}^{a}\mathcal{C}_{i}^{a}$ satisfies 
\[
\mathcal{B}_{i}^{a}\mathcal{C}_{i}^{a}=2I_{2}(\vec{A})I_{3}(\vec{A}),
\]
 and thus $I_{3}(\vec{A})$ is always given by the gauge-invariant expression 
\[
I_{3}(\vec{A})=\mathrm{sign}\left(\mathcal{B}_{i}^{a}\mathcal{C}_{i}^{a}\right)\,\sqrt{I_{3}(\vec{\mathcal{B}})}.
\]
 Finally, for all $\tau_{0}$ such that $I_{4}(\vec{A}(\tau_{0}))\neq0$, a quick computation shows that 
\[
I_{2}(\vec{A}(\tau_{0}))=\frac{I_{2}(\vec{\mathcal{C}}(\tau_{0}))-3I_{3}(\vec{A}(\tau_{0}))^{2}}{I_{4}(\vec{A}(\tau_{0}))},
\]
 and the right-hand side involves only quantities which have been shown to be gauge-invariant. 
\end{proof}

\subsection{\label{app:isotrop-proof}The case when $A$ is isotropic but the matrix $\vec{A}(\tau)$ is not }

In this subsection we assume that $A$ is isotropic, and given by $\vec{A}(\tau)$ in homogeneous and temporal gauge (see \prettyref{app:homo-gauge}). Our goal is to prove \prettyref{thm:rk1-is-gauge-artifact}, that it is always possible to find a particular homogeneous and temporal gauge in which $\vec{A}(\tau)$ is an isotropic matrix for all $\tau$. We must deal with the case observed in \prettyref{exa:local-gt}, where $\vec{A}(\tau)$ can be a rank-one matrix. 

In \prettyref{lem:isotropic-gf-relation} we derive a relation satisfied by $\vec{A}(\tau)$. \prettyref{lem:rewrite-isotropic-relation} then rewrites this relation so that in \prettyref{thm:iso-or-rank-one} we conclude for each $\tau_{0}$ that $\vec{A}(\tau_{0})$ is either isotropic or rank-one. In \prettyref{lem:iso-open} we prove that if $\vec{A}(\tau_{0})$ is rank-one, then it is rank-one for all nearby $\tau$. \prettyref{lem:rank-one-const} uses the electrical field to show that $\vec{A}(\tau)$ is constant where it is rank-one. Finally, the proof of \prettyref{thm:rk1-is-gauge-artifact} concludes that if $\vec{A}(\tau_{0})$ is rank one for any $\tau_{0}$, then $\vec{A}(\tau)$ is constant for all $\tau$. Thus it coincides with \prettyref{exa:local-gt} and can be made isotropic (and moreover zero) with a gauge transformation. 

First we prove a relation which holds for $\vec{A}(\tau)$: 
\begin{lem}
\label{lem:isotropic-gf-relation}Let $\vec{A}(\tau)$ be a homogeneous and isotropic $\mathrm{SU}(2)$ gauge field in homogeneous and temporal gauge. Then $\vec{A}(\tau)$ satisfies 
\[
I_{4}(\vec{A}(\tau))^{2}-3I_{2}(\vec{A}(\tau))I_{3}(\vec{A}(\tau))^{2}=0
\]
 for all $\tau$. 
\end{lem}

\begin{proof}
The matrix $\vec{{\cal B}}(\tau)$ is isotropic for all $\tau$ by the proof of \prettyref{lem:elec-is-iso}. From \prettyref{eq:isotropy-characterization} it follows that $I_{2}(\vec{\mathcal{B}}(\tau))^{2}-3I_{4}(\vec{\mathcal{B}}(\tau))=0$. From \prettyref{eq:SV-AB}, 
\[
I_{2}(\vec{\mathcal{B}}(\tau))^{2}-3I_{4}(\vec{\mathcal{B}}(\tau))=I_{4}(\vec{A}(\tau))^{2}-3I_{2}(\vec{A}(\tau))I_{3}(\vec{A}(\tau))^{2}.
\]
\end{proof}
We rewrite the characterization given in \prettyref{lem:isotropic-gf-relation} into a form which more clearly implies that the matrix is either isotropic or rank-one:
\begin{lem}
\label{lem:rewrite-isotropic-relation}Suppose $M_{0}$ is a matrix which satisfies $I_{4}(M_{0})^{2}-3I_{2}(M_{0})I_{3}(M_{0})^{2}=0$. Then 
\[
I_{4}(M_{0})\left(I_{2}(M_{0})^{2}-3I_{4}(M_{0})\right)=0.
\]
\end{lem}

\begin{proof}
We will make use of the quantity 
\[
\zeta(I_{2},I_{3},I_{4})\equiv\left(I_{2}I_{4}\right)^{2}+18I_{2}I_{3}^{2}I_{4}-4\left(I_{4}^{3}+I_{2}^{3}I_{3}^{2}\right)-27I_{3}^{4}.
\]
 It is easily verified that for any matrix $M$, 
\begin{equation}
\zeta(M)\equiv\zeta(I_{2}(M),I_{3}(M),I_{4}(M))=\left(\left(\sigma_{1}^{2}-\sigma_{2}^{2}\right)\left(\sigma_{2}^{2}-\sigma_{3}^{2}\right)\left(\sigma_{3}^{2}-\sigma_{1}^{2}\right)\right)^{2}\geq0.\label{eq:zeta-sv}
\end{equation}
 As an aside, one can also show the converse: if $\zeta(I_{2},I_{3},I_{4})\geq0$ then there exists a matrix $M$ such that $I_{n}=I_{n}(M)$. Thus the triangular region shown in \prettyref{fig:DE-triangle} is characterized by $\zeta\geq0$. 

Suppose now that $M_{0}$ is any matrix which satisfies 
\begin{equation}
I_{4}(M_{0})^{2}-3I_{2}(M_{0})I_{3}(M_{0})^{2}=0.\label{eq:A-iso-poly}
\end{equation}
 Consider the quantity $\zeta(M_{0})I_{2}(M_{0})^{2}$. From \prettyref{eq:zeta-sv} we know that $\zeta(M_{0})I_{2}(M_{0})^{2}\geq0$. However, using \prettyref{eq:A-iso-poly} to eliminate $I_{3}(M_{0})$, we obtain 
\begin{equation}
\zeta(M_{0})I_{2}(M_{0})^{2}=-\tfrac{1}{3}\left(\left(I_{2}(M_{0})^{2}-3I_{4}(M_{0})\right)I_{4}(M_{0})\right)^{2}\leq0.\label{eq:eliminate-I3}
\end{equation}
 Since $0\leq\zeta(M_{0})I_{2}(M_{0})^{2}\leq0$, we conclude that $\zeta(M_{0})I_{2}(M_{0})^{2}=0$, and thus by the equality in \prettyref{eq:eliminate-I3}, 
\[
\left(I_{2}(M_{0})^{2}-3I_{4}(M_{0})\right)I_{4}(M_{0})=0.
\]
\end{proof}
We therefore have the following conclusion:
\begin{thm}
\label{thm:iso-or-rank-one}Let $\vec{A}(\tau)$ be a homogeneous and isotropic $\mathrm{SU}(2)$ gauge field in homogeneous and temporal gauge. Then for each time $\tau_{0}$, either $\vec{A}(\tau_{0})$ is isotropic or $\vec{A}(\tau_{0})$ is rank-one. 
\end{thm}

\begin{proof}
By \prettyref{lem:isotropic-gf-relation} combined with \prettyref{lem:rewrite-isotropic-relation}, 
\begin{equation}
\left(I_{2}(\vec{A}(\tau))^{2}-3I_{4}(\vec{A}(\tau))\right)I_{4}(\vec{A}(\tau))=0.\label{eq:reduced-isotropy-A}
\end{equation}
 Thus for any $\tau_{0}$, either $I_{4}(\vec{A}(\tau_{0}))=0$ or $I_{2}(\vec{A}(\tau_{0}))^{2}-3I_{4}(\vec{A}(\tau_{0}))=0$, and the result follows from \eqref{eq:two-sv-vanish} and \eqref{eq:isotropy-characterization}.
\end{proof}
The following two lemmas will imply that if $\vec{A}(\tau_{0})$ is rank one for some $\tau_{0}$ then $\vec{A}(\tau)$ is constant. 
\begin{lem}
\label{lem:iso-open}Let $\vec{A}(\tau)$ be a homogeneous and isotropic $\mathrm{SU}(2)$ gauge field in homogeneous and temporal gauge. If there is some $\tau_{0}$ such that $\vec{A}(\tau_{0})$ is rank-one, then $\vec{A}(\tau)$ is rank-one along some interval containing $\tau_{0}$ in its interior. 
\end{lem}

\begin{proof}
First, note that by solving \prettyref{eq:reduced-isotropy-A} for $I_{4}(\vec{A}(\tau_{1}))$ at any time $\tau_{1}$, the number $I_{4}(\vec{A}(\tau_{1}))$ may equal either $0$ or $\tfrac{1}{3}I_{2}(\vec{A}(\tau_{1}))^{2}$. Now suppose that there is some time $\tau_{0}$ for which the matrix $\vec{A}(\tau_{0})$ is not an isotropic matrix. By \prettyref{thm:iso-or-rank-one}, $\vec{A}(\tau_{0})$ must be rank-one. Thus from \prettyref{eq:two-sv-vanish}, $I_{4}(\vec{A}(\tau_{0}))=0$ but $I_{2}(\vec{A}(\tau_{0}))\neq0$. Assuming the continuity of $\vec{A}(\tau)$, it follows that there is some interval containing $\tau_{0}$ along which $I_{2}(\vec{A}(\tau))\neq0$. Also by continuity, $I_{4}(\vec{A}(\tau))$ along this same interval must be equal to either the positive branch $\tfrac{1}{3}I_{2}(\vec{A}(\tau))^{2}$ or the zero branch. Since $I_{4}(\vec{A}(\tau_{0}))=0$, we conclude that it must be the zero branch. Therefore, $\vec{A}(\tau)$ is rank-one for this whole interval along which $I_{2}(\vec{A}(\tau))\neq0$.
\end{proof}
\begin{lem}
\label{lem:rank-one-const}Let $\vec{A}(\tau)$ be a homogeneous and isotropic $\mathrm{SU}(2)$ gauge field in homogeneous and temporal gauge. Along any interval where $\vec{A}(\tau)$ is rank-one, it is constant. 
\end{lem}

\begin{proof}
Recall from \prettyref{eq:E-is-derivative} that $\tfrac{\mathrm{d}}{\mathrm{d}\tau}\vec{A}(\tau)=\vec{\mathcal{E}}(\tau)$. Thus in order to show that $\vec{A}(\tau)$ is constant along some interval, it suffices to show that $\vec{\mathcal{E}}(\tau)=0$ along the same interval. 

From \prettyref{lem:elec-is-iso} it follows that $\vec{\mathcal{E}}(\tau)$ is always an isotropic matrix. By \prettyref{lem:matrix-isotropy}, $\vec{\mathcal{E}}(\tau_{0})$ is a scalar multiple of an orthogonal matrix for each $\tau_{0}$. Since orthogonal matrices are invertible, the only scalar multiple of an orthogonal matrix which has a nonzero nullspace is the zero matrix. Thus the lemma follows if we show that $\vec{\mathcal{E}}(\tau_{0})$ has a nonzero nullspace. 

The singular value decomposition of $\vec{A}(\tau)$ in the rank-one case (see \prettyref{def:rank-one}) gives the $3\times3$ matrix equation 
\[
\vec{A}(\tau)=\vec{v}(\tau)\sigma_{1}(\tau)\vec{w}(\tau)^{T},
\]
 where $\vec{v}(\tau)$ denotes the first column of $G(\tau)$ and $\vec{w}(\tau)$ is the first column of $R(\tau)$. Thus 
\begin{equation}
\vec{\mathcal{E}}(\tau)=\vec{v}(\tau)\,\tfrac{\mathrm{d}}{\mathrm{d}\tau}\left(\sigma_{1}(\tau)\vec{w}(\tau)^{T}\right)+\left(\tfrac{\mathrm{d}}{\mathrm{d}\tau}\vec{v}(\tau)\right)\sigma_{1}(\tau)\vec{w}(\tau)^{T}.\label{eq:E-rk-1-as-derivative}
\end{equation}
 Choosing any nonzero vector $\vec{s}(\tau)$ in $\mathbb{R}^{3}$ which is orthogonal to both $\vec{v}(\tau)$ and $\tfrac{\mathrm{d}}{\mathrm{d}\tau}\vec{v}(\tau)$, it is clear from \prettyref{eq:E-rk-1-as-derivative} that $\vec{s}(\tau)^{T}\vec{\mathcal{E}}(\tau)$ vanishes. Thus $\vec{\mathcal{E}}(\tau)$ has a nonzero left-nullspace, proving the lemma.
\end{proof}
\begin{thm}
\label{thm:rk1-is-gauge-artifact}Let $A$ be a homogeneous and isotropic $\mathrm{SU}(2)$ gauge field. There exists a homogeneous and temporal gauge for $A$ in which $\vec{A}(\tau)$ is an isotropic matrix for all $\tau$. 
\end{thm}

\begin{proof}
By \prettyref{thm:hom-temp-gauge} we may assume that $A$ is in homogeneous and temporal gauge. If $\vec{A}(\tau)$ is isotropic for all $\tau$ then we are done. Alternatively, if $\vec{A}(\tau)$ is a constant rank-one matrix for all $\tau$ as in \prettyref{exa:local-gt} then we are done because \prettyref{exa:local-gt} is gauge-equivalent to zero, and zero is isotropic. By \prettyref{thm:iso-or-rank-one}, the only other possibility is that $\vec{A}(\tau_{0})$ is rank-one at some time $\tau_{0}$, but $\vec{A}(\tau)$ is non-constant. However, this cannot happen by \prettyref{lem:iso-open} and \prettyref{lem:rank-one-const}. 
\end{proof}
As a result of \prettyref{thm:rk1-is-gauge-artifact}, when $A$ is homogeneous and isotropic we may assume without loss of generality that $\vec{A}(\tau)$ is an isotropic matrix for all $\tau$.

\section{Asymptotics of the Whittaker \texorpdfstring{$W$}{W} function\label{app:asymptotics}}

The purpose of this appendix is to understand the asymptotics of the Whittaker function, which describes the enhanced gauge field modes.

For parameters $k$, $m$, and $b$, and a positive real variable $x$, the Whittaker functions 
\[
C_{1}W_{k,m}(bx)+C_{2}M_{\kappa,\mu}(bx)
\]
 for constants $C_{1}$ and $C_{2}$ provide the general solution to the differential equation 
\[
\frac{\mathrm{d}^{2}}{\mathrm{d}x^{2}}w(x)+\left(-\left(\frac{b}{2}\right)^{2}+\frac{kb}{x}+\frac{\tfrac{1}{4}-m^{2}}{x^{2}}\right)w(x)=0.
\]
 In general, given any differential equation of the form 
\[
\frac{\mathrm{d}^{2}}{\mathrm{d}x^{2}}w(x)+\left(A+\frac{B}{x}+\frac{C}{x^{2}}\right)w(x)=0,
\]
 it is obviously possible to solve for (possibly complex) Whittaker function parameters $k$, $m$ and $b$ which match general parameters $A$, $B$, and $C$. For our purposes, it will be convenient to make the transformation $k\mapsto-i\kappa$, $m\mapsto-i\mu$ and $b\mapsto-i\beta$. Under the assumption that $\kappa$, $\mu$ and $\beta$ are real and positive, it will be sufficient for our purposes to take $C_{2}=0$ and consider only the ``negative-imaginary Whittaker $W$ function'' 
\[
C_{1}W_{-i\kappa,-i\mu}(-i\beta x),
\]
 which solves 
\begin{equation}
\frac{\mathrm{d}^{2}}{\mathrm{d}x^{2}}w(x)+\left(\left(\frac{\beta}{2}\right)^{2}-\frac{\kappa\beta}{x}+\frac{\tfrac{1}{4}+\mu^{2}}{x^{2}}\right)w(x)=0.\label{eq:whiteqn}
\end{equation}
 The case of \prettyref{eq:+2mode} for $w_{+2}^{(e)}$ with the $c_{2}$-solution corresponds to the negative-imaginary Whittaker $W$ function with parameters
\begin{align}
\kappa & =(1+c_{2})\xi\approx2\xi,\label{eq:whitparams}\\
\mu & =\sqrt{2c_{2}-(2\xi)^{-2}}\xi\approx\sqrt{2}\xi,\nonumber \\
\beta & =2.\nonumber 
\end{align}

\subsection{Summary of asymptotics}

The phenomenology in Sections \ref{sec:linearized} and \ref{sec:example} depends strongly on the behavior of the mode functions both in the infinite future $x\to0^{+}$ and around the global maximum. We summarize here the corresponding estimates for the Whittaker function. Derivations are given in the subsequent subsections.

We impose the boundary condition 
\begin{align}
w(x) & =C_{1}W_{-i\kappa,-i\mu}(-i\beta x)\sim e^{i\beta x/2}\textrm{ as }x\to\infty\label{eq:whittaker-with-coeff}
\end{align}
 which, with our normalization conventions for mode functions, corresponds to the Bunch--Davies vacuum \prettyref{eq:BDsimple}. We show in \prettyref{subsec:asymp-normalization} that 
\[
C_{1}=e^{\kappa\pi/2}e^{i\phi_{0}},
\]
 where the phase $\phi_{0}$ is undetermined. 

In \prettyref{subsec:asymp-x-0} we study the resulting behavior as $x\to0^{+}$. As is the case in \prettyref{eq:whitparams}, if $\kappa>\mu$ then upon taking $\phi_{0}$ as in \prettyref{eq:phi-0}, the real part is enhanced and the imaginary part is suppressed. Specifically, 
\begin{align}
\textrm{Re}(w(x)) & \approx2e^{(\kappa-\mu)\pi}\sqrt{\frac{\beta x}{2\mu}}\cos\left(\mu\ln(\beta x)+\theta_{0}\right),\label{eq:wh-re}\\
\textrm{Im}(w(x)) & \approx\tfrac{1}{2}e^{-(\kappa-\mu)\pi}\sqrt{\frac{\beta x}{2\mu}}\sin\left(\mu\ln(\beta x)+\theta_{0}\right).\label{eq:wh-im}
\end{align}
 Note that this is a wave which decays in proportion to $\sqrt{x}$. Furthermore, this wave oscillates $\mu/(2\pi)$ times per e-fold. In our case of interest $w(x)=w_{+2}^{(e)}(x)$, e-folds in $x$ are equivalent to e-folds in $\tau$, and the frequency is $\approx\sqrt{2}\xi/(2\pi)$ oscillations per e-fold. 

While for large $\xi$, we could approximate $\kappa$ and $\mu$ to obtain 
\begin{align}
w_{+2}^{(e)}(x) & \approx2^{-1/4}\sqrt{x}\left(2e^{(2-\sqrt{2})\pi\xi}\cos\theta+\tfrac{1}{2}e^{-(2-\sqrt{2})\pi\xi}i\sin\theta\right),\label{eq:as_wp2}\\
\theta & \approx\sqrt{2}\xi\ln(2x)+\theta_{0},\nonumber 
\end{align}
 it is far more accurate to use \prettyref{eq:wh-re} rather than \prettyref{eq:as_wp2} due to the exponential sensitivity on $\kappa$ and $\mu$. 

Finally in \prettyref{subsec:Airy-Whittaker}, we apply the WKB approximation as reviewed in \prettyref{subsec:WKB-review} and \prettyref{subsec:WKB-formulas} to obtain \prettyref{eq:airy-full-approx}, which accurately approximates $\textrm{Re}(w(x))$ around its maximum value in terms of the Airy function.

\subsection{\label{subsec:asymp-normalization}Asymptotics as \texorpdfstring{$x\to\infty$}{x\to\infty}}

Here we determine the magnitude $\left|C_{1}\right|$ of the coefficient in \prettyref{eq:whittaker-with-coeff} by matching the Bunch--Davies vacuum with the asymptotics of the Whittaker function. 

Let $x$ be a positive real variable. For large $x$, the Whittaker function satisfies 
\begin{align*}
W_{k,m}(bx) & =e^{-bx/2}(bx)^{\kappa}\left(1+\epsilon_{1}(bx)\right),\\
\epsilon_{1}(bx) & =0+\mathcal{O}\left(\frac{1+k^{2}+m^{2}}{bx}\right).
\end{align*}
 This asymptotic expression remains valid upon replacing $k\mapsto-i\kappa$, $m\mapsto-i\mu$ and $b\mapsto-i\beta$, and evaluating complex exponents $a^{p}=e^{p\ln a}$ using the principal branch of the logarithm. Thus 
\begin{align}
W_{-i\kappa,-i\mu}(-i\beta x) & =e^{i\beta x/2}(-i\beta x)^{-i\kappa}\left(1+\epsilon_{1}(\beta x)\right)\nonumber \\
 & =e^{-\kappa\pi/2}\,e^{i\beta x/2-i\kappa\ln(\beta x)}\left(1+\epsilon_{1}(\beta x)\right)\label{eq:log-drift}
\end{align}
 as $x\to\infty$. Therefore $\left|C_{1}\right|=e^{\kappa\pi/2}$. Note the (inconsequential) logarithmic drift in complex phase, which prevents us at this stage from selecting a distinguished phase.

\subsection{\label{subsec:asymp-x-0}Asymptotics as \texorpdfstring{$x\to0^+$}{x\to0+}}

Let us now consider the behavior as $x\to0^{+}$, corresponding to the infinite future. In this limit we have 
\begin{align*}
W_{-i\kappa,-i\mu}(-i\beta x) & =\sqrt{\beta x}e^{-i\pi/4}\sum_{\pm}\frac{\Gamma\left(\mp2i\mu\right)}{\Gamma\left(\tfrac{1}{2}+i\kappa\mp i\mu\right)}(-i\beta x)^{\pm i\mu}\left(1+\epsilon_{2}(\beta x)\right),\\
\epsilon_{2}(\beta x) & =0+\mathcal{O}\left(\frac{1+\kappa+\mu}{1+\mu}\beta x\right).
\end{align*}
 To make sense of this expression, we may rewrite it in the form 
\begin{align}
w(x)\equiv e^{\kappa\pi/2}e^{i\phi_{0}}W_{-i\kappa,-i\mu}(-i\beta x) & =C_{0}\sqrt{\beta x}\left(\lambda\cos\theta+i\lambda^{-1}\sin\theta+\epsilon_{3}(\beta x)\right),\label{eq:w-exact}\\
\theta & \equiv\mu\ln(\beta x)+\theta_{0},
\end{align}
 where the four constants $C_{0}$, $\lambda$, $\phi_{0}$ and $\theta_{0}$ depend only on the parameters $\kappa$ and $\mu$. After some algebra, we find the exact expressions
\begin{align}
C_{0} & \equiv(2\mu)^{-1/2},\label{eq:whit-consts}\\
\lambda & \equiv e^{(\kappa-\mu)\pi}\left(\frac{\sqrt{1+e^{-2(\kappa-\mu)\pi}}+\sqrt{1+e^{-2(\kappa-\mu)\pi}e^{-4\mu\pi}}}{\sqrt{1-e^{-4\mu\pi}}}\right),\label{eq:whit-lambda}\\
\theta_{0} & \equiv\tfrac{1}{2}\left(\phi_{\Gamma,1/2}(\kappa+\mu)-\phi_{\Gamma,1/2}(\kappa-\mu)\right)-\phi_{\Gamma,0}(2\mu)\\
\phi_{0} & \equiv\tfrac{\pi}{4}+\tfrac{1}{2}\left(\phi_{\Gamma,1/2}(\kappa+\mu)+\phi_{\Gamma,1/2}(\kappa-\mu)\right),\label{eq:phi-0}\\
\phi_{\Gamma,a}(b) & \equiv\arg\Gamma(a+ib).\\
\epsilon_{3}(\beta x) & =0+\mathcal{O}\left(\frac{1+\kappa+\mu}{1+\mu}(\lambda+\lambda^{-1})\beta x\right).
\end{align}
 To derive these parameters, we have made use of the following polar decomposition identities for the gamma function. For $b\in\mathbb{R}$, 
\[
\Gamma\left(ib\right)=\sqrt{\frac{\pi}{b\sinh(\pi b)}}e^{i\phi_{\Gamma,0}(b)},\quad\Gamma\left(\tfrac{1}{2}+ib\right)=\sqrt{\frac{\pi}{\cosh(\pi b)}}e^{i\phi_{\Gamma,1/2}(b)}.
\]

In the parameter range of interest (see \prettyref{eq:whitparams}), 
\[
\kappa\approx2\xi,\qquad\mu\approx\sqrt{2}\xi,\qquad\beta=2,\qquad\xi\geq2,
\]
 we have $e^{-2\pi(\kappa-\mu)}\ll1$ and $e^{-4\pi\mu}\ll1$. This allows us to very accurately approximate the part of the expression for $\lambda$ inside the parentheses in \prettyref{eq:whit-lambda} by the number $2$. Specifically, 
\begin{align}
\lambda & =2e^{(\kappa-\mu)\pi}\left(1+\epsilon_{4}\right),\label{eq:l-approx}\\
\textrm{with }\ 0 & <\epsilon_{4}\leq e^{-2(\kappa-\mu)\pi}+e^{-4\mu\pi}\textrm{ when }\mu\geq\tfrac{1}{25}.
\end{align}
 In the worst case with our parameters (when $\xi=2$), the relative error is only $\epsilon_{4}\approx3\times10^{-4}$. 

Finally, when $e^{-2\pi(\kappa-\mu)}\ll1$ and $e^{-4\pi\mu}\ll1$ we obtain the formulas \prettyref{eq:wh-re} and \prettyref{eq:wh-im} by plugging \prettyref{eq:l-approx} into \prettyref{eq:w-exact}.

\subsection{WKB approximation }

\subsubsection{Review of WKB approximation\label{subsec:WKB-review}}

Although widely known, we briefly summarize the technique of WKB approximation as utilized here. The solutions to a differential equation of the form 
\begin{equation}
\frac{\mathrm{d}^{2}}{\mathrm{d}x^{2}}w(x)=V(x)w(x)\label{eq:WKB-orig}
\end{equation}
 are often not straightforward when $V(x)$ is a general function. Important exceptions are when $V(x)$ is one of the following model potentials $V_{0}(x)$. 
\begin{itemize}
\item If $V_{0}(x)=1$ then $w_{0}(x)=e^{\pm x}$ are solutions. 
\item If $V_{0}(x)=-1$ then $w_{0}(x)=e^{\pm ix}$ are solutions. 
\item If $V_{0}(x)=x$ then the Airy functions $w_{0}(x)=\mathrm{Ai}(x)$ and $w_{0}(x)=\mathrm{Bi}(x)$ are solutions. 
\item If $V_{0}(x)=\pm\left(x^{2}-\alpha^{2}\right)$ for some constant $\alpha$ then the solutions are called parabolic cylinder functions, and they have $\alpha$ as a parameter. 
\end{itemize}
The main idea of the WKB approximation is to perform a change of variables to make \prettyref{eq:WKB-orig} resemble such a model equation. We consider transformations of \prettyref{eq:WKB-orig} which do not introduce a first-derivative term, and are of the form 
\begin{align}
w(x) & \mapsto\vartheta(\zeta(x))\cdot W(\zeta(x)),\label{eq:liouville}
\end{align}
 where $\zeta(x)$ is an increasing change of variables ($\mathrm{d}\zeta/\mathrm{d}x>0$), and $\vartheta(\zeta)$ is constructed to eliminate any first-derivative term introduced by $\zeta$. Such a transformation is called a Liouville transformation. The transformed equation should be of the form 
\begin{equation}
\frac{\mathrm{d}^{2}}{\mathrm{d}\zeta^{2}}W(\zeta)=\left(V_{0}(\zeta)+\epsilon(\zeta)\right)W(\zeta),\label{eq:WKB-transformed}
\end{equation}
 where $V_{0}(\zeta)$ is a model potential, $\epsilon(\zeta)$ is arranged to be suitably negligible. The Liouville transformation condition that ensures no first-derivative term is equivalent to 
\begin{equation}
\vartheta(\zeta)=\left(\frac{\mathrm{d}x}{\mathrm{d}\zeta}\right)^{1/2}.\label{eq:WKB-mult}
\end{equation}
A generally good choice of $\zeta(x)$ to keep $\epsilon(\zeta)$ small is 
\begin{equation}
\frac{\mathrm{d}\zeta}{\mathrm{d}x}=+\sqrt{\frac{V(x)}{V_{0}(\zeta)}}.\label{eq:zeta-def}
\end{equation}
Since this must be real, it must satisfy $\mathrm{sign}(V(x))=\mathrm{sign}(V_{0}(\zeta))$. Thus for each zero of $V(x)$ in the domain of interest, there must be a corresponding zero in $V_{0}(\zeta)$. In this way, the number of roots of $V(x)$ determines the appropriate type of model potential.

Solving \prettyref{eq:zeta-def} gives 
\[
Z(\zeta)=X(x),\ \textrm{where}\ Z(\zeta)\equiv\int\sqrt{\left|V_{0}(\zeta)\right|}\ \mathrm{d}\zeta\ \textrm{ and }\ X(x)\equiv\int\sqrt{\left|V(x)\right|}\ \mathrm{d}x.
\]
In case there are any roots $\left\{ x_{i}\right\} $ of $V(x)$, then the constant of integration and any parameters of the model potential $V(x)$ must be chosen so that $X(x_{i})=Z(\zeta_{i})$, where $\left\{ \zeta_{i}\right\} $ are the corresponding zeroes of $V_{0}(\zeta)$.  Assuming that we can compute the antiderivatives $Z$ and $X$, as well as $Z^{-1}$, we have the formula 
\[
\zeta(x)=Z^{-1}\left(X(x)\right).
\]
 Working backward to get our approximate solution, let $W_{0}(\zeta)$ denote a solution to \prettyref{eq:WKB-transformed} where we ignore\footnote{For detailed error estimates for the WKB approximation, see \cite{olver1997}.} $\epsilon(\zeta)$. Plugging this into \prettyref{eq:liouville} and using \prettyref{eq:WKB-mult} and \prettyref{eq:zeta-def}, we find that 
\begin{equation}
w_{\textrm{approx}}(x)=C\left(\frac{V_{0}(\zeta(x))}{V(x)}\right)^{1/4}W_{0}(\zeta(x))\label{eq:WKB}
\end{equation}
 is an approximate solution to \prettyref{eq:WKB-orig} for any number $C$.

\subsubsection{\label{subsec:WKB-formulas}WKB approximation for various model potentials}

Consider first the simple case where $V(x)$ has no zeroes so that we can take $V_{0}(\zeta)=\pm1=\mathrm{sign}V$. The model solutions are $W_{0}(\zeta)=e^{\pm\sqrt{\mathrm{sign}V}\zeta}$. For the reparameterization, $Z(\zeta)=\zeta$ and so 
\[
\zeta_{\textrm{exp}}(x)=X(x)=\int\sqrt{\left|V(x)\right|}\,\mathrm{d}x.
\]
 Finally, 
\begin{equation}
w_{\textrm{exp-approx}}(x)=C\left|V(x)\right|^{-1/4}e^{\pm\sqrt{\mathrm{sign}V}X(x)},\label{eq:exp-approx}
\end{equation}
 which is the standard WKB approximation. 

Consider next the case where $V(x)$ has a single simple zero at $x=x_{1}$ in the domain of interest. We take $V_{0}(\zeta)=\zeta$. Thus $Z(\zeta)=\mathrm{sign(\zeta)}\tfrac{2}{3}\left|\zeta\right|^{3/2}$, and $X(x)=\int_{x_{1}}^{x}\sqrt{\left|V(y)\right|}\,\mathrm{d}y$, so 
\begin{align}
\zeta_{\textrm{Airy}}(x) & =\mathrm{sign}(x-x_{1})\left|\tfrac{3}{2}X(x)\right|^{2/3},\\
w_{\textrm{Airy-approx}}(x) & =C\left(\frac{\zeta(x)}{V(x)}\right)^{1/4}\mathrm{Ai}\left(\zeta(x)\right).\label{eq:airy-approx}
\end{align}

The case where $V(x)$ has two zeroes at $x_{1}<x_{2}$ is treated in~\cite{Olver137} and applied to the imaginary Whittaker function in~\cite{Olver_1980}. As a summary, the parameter $\alpha$ of $V_{0}$ is determined by $\int_{-\alpha}^{\alpha}\sqrt{\alpha^{2}-\zeta^{2}}=\int_{x_{1}}^{x_{2}}\sqrt{\left|V(y)\right|}\,\mathrm{d}y$, which is necessary for $\zeta(x_{1})=-\alpha$ and $\zeta(x_{2})=+\alpha$. One complication is that although $Z(\zeta)$ has a closed form, its inverse function does not. Furthermore the parabolic cylinder functions are considerably more complicated. Since our region of interest is around the global maximum, which occurs close to the zero\footnote{A slight discrepancy between $x_{1}$ and $x_{\mathrm{min}}$ arises because our $-V(x)$ differs slightly from the mass term $m_{+2}^{2}$ of \prettyref{eq:mass-term-2p}, as explained in \prettyref{fn:modified-potential}. } $x_{1}\approx x_{\mathrm{min}}$ defined in \prettyref{eq:def-xmin-xmax}, the Airy approximation \eqref{eq:airy-approx} suffices for our purposes. 

\subsubsection{\label{subsec:Airy-Whittaker}Airy approximation of the imaginary Whittaker \texorpdfstring{$W$}{W} function around its maximum }

We wish to find the constant coefficient of \prettyref{eq:airy-approx} which will make it agree with $\textrm{Re}(w(x))$ from \prettyref{eq:wh-re} as $x\to0$, or equivalently as $\zeta\to-\infty$. 

In the imaginary Whittaker equation \prettyref{eq:whiteqn} we take\footnote{\label{fn:modified-potential}Since the WKB approximation introduces small errors, there is potential for these errors to cancel. For the Whittaker function, it is advantageous to take $V(x)$ without the $\tfrac{1}{4}x^{-2}$ term in \prettyref{eq:whiteqn}. This is equivalent to using a non-zero error term $\epsilon(x)=\tfrac{1}{4}x^{-2}$, as in \prettyref{eq:WKB-transformed}. (See \cite{olver1997} for details.) As motivation for this modification, consider $\epsilon(x)=ax^{-2}$, so that the remainder is $V(x)\approx-(\tfrac{1}{4}-a+\mu^{2})x^{-2}$ as $x\to0$. Therefore $X(x)\approx\sqrt{\tfrac{1}{4}-a+\mu^{2}}\ln x$. When $V(x)$ is large, the Airy approximation reduces to the exponential approximation, and thus by \prettyref{eq:exp-approx} the approximate solution is of the form 
\[
w_{\mathrm{approx}}(x)\approx C\sqrt{x}e^{\pm i\sqrt{\tfrac{1}{4}-a+\mu^{2}}\ln x}.
\]
 By comparison of the exponent with \prettyref{eq:wh-re}, we see that the coefficient of $\ln x$ should be $\mu$, so it is better to take $a=\tfrac{1}{4}$ rather than $a=0$. } 
\begin{align*}
V(x) & =-\left(\left(\frac{\beta}{2}\right)^{2}-\frac{\kappa\beta}{x}+\frac{\tfrac{0}{4}+\mu^{2}}{x^{2}}\right),
\end{align*}

To determine the proper coefficient $C$ of \prettyref{eq:airy-approx}, we substitute into \prettyref{eq:airy-approx} the asymptotic formula 
\[
\mathrm{Ai}(\zeta)=\pi^{-1/2}(-\zeta)^{-1/4}\sin\left(\frac{\pi}{4}+\frac{2}{3}(-\zeta)^{3/2}\right)+O(-\zeta^{-1})\textrm{ as }\zeta\to-\infty
\]
 to obtain 
\[
w_{\textrm{approx}}(x)\approx C\pi^{-1/2}V(x)^{-1/4}\sin\left(\frac{\pi}{4}-X(x)\right).
\]
 As $x\to0$ we have $V(x)^{-1/4}\approx\sqrt{x/\mu}$. Matching the magnitude with \prettyref{eq:wh-re}, we find $C=\sqrt{2\pi\beta}\ e^{(\kappa-\mu)\pi}$, so 
\begin{equation}
w_{\textrm{approx}}(x)=\sqrt{2\pi\beta}\ e^{(\kappa-\mu)\pi}\left(\frac{\zeta(x)}{V(x)}\right)^{1/4}\mathrm{Ai}\left(\zeta(x)\right).\label{eq:airy-full-approx}
\end{equation}
 The exact expression for $\zeta(x)$ is complicated, but it is well-approximated for $x\approx x_{1}$ by 
\[
\zeta(x)\approx\left(2\mu^{2}-\beta\kappa x_{1}\right)^{1/3}\ln\frac{x}{x_{1}}.
\]
 To give the exact formula for $\zeta(x)$ in the interval $\left(0,x_{1}\right)$ in a concise form, we introduce 
\begin{align*}
\chi\equiv\frac{\beta x}{2\mu},\quad\lambda & \equiv\frac{\kappa}{\mu}>1,\quad\chi_{i}\equiv\lambda+(-1)^{i}\sqrt{\lambda^{2}-1}\textrm{ for }i\in\left\{ 1,2\right\} ,\\
R & \equiv\sqrt{\chi^{2}-2\chi\lambda+1}=\frac{x}{\mu}\sqrt{-V(x)}.
\end{align*}
 It is easy to verify that 
\begin{align*}
x_{i} & =2\mu\chi_{i}/\beta\textrm{ for }i\in\left\{ 1,2\right\} ,\\
V(x) & =V(x)=-\left(\frac{\beta}{2\chi}\right)^{2}\left(\chi^{2}-2\chi\lambda+1\right)=-\left(\frac{\beta}{2\chi}\right)^{2}\left(\chi-\chi_{1}\right)\left(\chi-\chi_{2}\right),\\
X(x) & =\left(R-\ln\left(1-\chi\lambda+R\right)-\lambda\ln\left(\lambda-\chi-R\right)+\ln\chi+\tfrac{1}{2}(1+\lambda)\ln\left(\lambda^{2}-1\right)\right)\mu,\\
\zeta(x) & =\mathrm{sign}(x-x_{1})\left|\tfrac{3}{2}X(x)\right|^{2/3}.
\end{align*}
 Due to the limited domain of the logarithm, the expression given for $X(x)$ applies only in the interval $(0,x_{1})$. 

We use \prettyref{eq:airy-full-approx} to compute the gauge-field contribution to the scalar power spectrum in \prettyref{app:variance_computation}.

\textcolor{white}{}

\bibliographystyle{JHEP}
\bibliography{refs_non-abelian}{}

\end{document}